\documentclass[a4paper, journal,twocolumn]{IEEEtran}
\usepackage{amsfonts}
\usepackage{algorithm}
\usepackage[noend]{algpseudocode}

\usepackage[singlelinecheck=false]{subcaption}

\usepackage[titles]{tocloft}
\usepackage{graphics}
\usepackage{hyperref}
\usepackage{amsmath}

\usepackage{stfloats}

\usepackage{multirow}

\usepackage{float}
\usepackage{tabularray}

\newcommand{\Ua}{U_{a,N,m}}

\usepackage{pgfplots}
\usepackage{physics}

\usepackage{enumerate}
\usepackage{hyperref}
\usepackage{ulem}
\usepackage{mathtools}
\mathtoolsset{showonlyrefs}
\mathtoolsset{showonlyrefs=true}

\usepackage{cite}

\usepackage{titlesec}

\newcommand*{\vac}{{\rm vac}}

\usepackage{graphicx}
\usepackage{enumerate}
\usepackage{amssymb}
\usepackage{color}
\usepackage{amsthm,amsmath}
\usepackage{thmtools}
\usepackage{thm-restate}
\usepackage{subcaption}
\usepackage{graphicx}
\newcommand{\appref}[1]{\hyperref[#1]{{Appendix~\ref*{#1}}}}
\newcommand{\be}{\begin{eqnarray} \begin{aligned}}
		\newcommand{\ee}{\end{aligned} \end{eqnarray} }
\newcommand{\benn}{\begin{eqnarray*} \begin{aligned}}
		\newcommand{\eenn}{\end{aligned} \end{eqnarray*}}
\newcommand*{\textfrac}[2]{{{#1}/{#2}}}

\newcommand*{\cA}{\mathcal{A}} 
\newcommand*{\cB}{\mathcal{B}}
\newcommand*{\cC}{\mathcal{C}}

\newcommand*{\cH}{\mathcal{H}}

\newcommand*{\cL}{\mathcal{L}}

\newcommand*{\cQ}{\mathcal{Q}}

\newcommand*{\cP}{\mathcal{P}}

\newcommand*{\cX}{\mathcal{X}}

\newcommand*{\supp}{\mathrm{supp}}

\newcommand{\bc}{\begin{center}}
	\newcommand{\ec}{\end{center}}

\newcommand{\id}{\mathbb{I}}

\newcommand{\e}{\mathrm{e}}

\newtheorem{lemmamain}{Lemma}
\newtheorem{theorem}{Theorem}[section]

\newtheorem{lemma}{Lemma}[section]
\newtheorem{claim}[theorem]{Claim}
\newtheorem{definition}[theorem]{Definition}

\newtheorem{corollary}[theorem]{Corollary}

\usepackage{amsfonts}

\def\id{\mathbb{I}}
\def\01{\{0,1\}}
\newcommand{\ceil}[1]{\lceil{#1}\rceil}

\newcommand*{\ExpE}{\mathbb{E}}

\newcommand*{\gkp}{\mathsf{GKP}}

\newcommand*{\lsb}{\mathsf{LSB}}

\newcommand{\ind}{\mathsf{ind}}

\usepackage[OT2,T1]{fontenc}
\DeclareSymbolFont{cyrletters}{OT2}{wncyr}{m}{n}
\DeclareMathSymbol{\Sha}{\mathalpha}{cyrletters}{"58}

\newcommand{\xmod}[1]{\!\!\mod{#1}\!}

\newcommand*{\round}[1]{\lfloor #1\rceil}

\DeclareMathOperator{\fra}{frac}
\newcommand{\Rem}{\textrm{Rem}}

\usepackage{amsfonts}

\def\id{\mathbb{I}}
\def\01{\{0,1\}}

\newcommand*{\qubit}{{\mathbf{\mathrm{Q}}}}
\newcommand*{\boson}{{\mathbf{\mathrm{B}}}}

\newcommand*{\bosonB}{{\mathbf{\mathrm{B}}}}

\newcommand*{\bosonC}{{\mathbf{\mathrm{C}}}}

\newcommand*{\bosonA}{{\mathbf{\mathrm{A}}}}

\renewcommand*{\id}{\mathsf{id}}

\newcommand*{\good}{\mathsf{Good}}
\newcommand*{\dis}{\textnormal{\textsf{discret}}}

\newcommand*{\shorcircuit}{\cQ^{\textrm{Shor}}}
\newcommand*{\ourcircuit}{\cQ}
\newcommand*{\ourcircuitideal}{\cQ^{\textrm{ideal}}}

\newcommand{\Uan}{U_{a,N,m}}
\newcommand{\Vam}{V_{a,N,m}}

\newcommand{\circWithInitStates}{\mathcal{W}_{a,N}}
\newcommand{\semiideal}{\circWithInitStates(e^{-iRP}\Sha_{R}, \gkp, \ket{0})}
\newcommand{\circuitPhysicalStates}{\circWithInitStates(e^{-iRP}\gkp_{\kappa_\bosonA,\Delta_\bosonA}, \gkp_{\kappa_\bosonB,\Delta_\bosonB},\Psi_{\Delta_\bosonC})}

\usepackage{quantikz}
\newcommand{\kett}[1]{\ket{\vphantom{a^{2^2}} #1}}
\newcommand{\cpsi}[1]{c_{#1}}
\newcommand{\cShor}{\cC^{\mathsf{Shor}}}
\newcommand{\xdistance}[1]{\left\|\Psi^{(#1)}-\Psi^{(\the\numexpr #1 + 1\relax)}\right\|_1}

\DeclareMathOperator{\dev}{dev}

\begin{document}
	
	\pagestyle{plain}
	
	\title{Factoring an integer with\\
		three oscillators and a qubit}
	
	\author{
		\IEEEauthorblockN{Lukas Brenner\IEEEauthorrefmark{1}\IEEEauthorrefmark{3}, Libor Caha\IEEEauthorrefmark{1}\IEEEauthorrefmark{3}, Xavier Coiteux-Roy\IEEEauthorrefmark{1}\IEEEauthorrefmark{2}\IEEEauthorrefmark{3}, and Robert Koenig\IEEEauthorrefmark{1}\IEEEauthorrefmark{3}}
		\vspace{12pt}\\
		\IEEEauthorblockA{\IEEEauthorrefmark{1}\footnotesize School of Computation, Information and Technology, Technical University of Munich, Germany}\\
		\IEEEauthorblockA{\IEEEauthorrefmark{2}\footnotesize School of Natural Sciences, Technical University of Munich, Germany}\\
		\IEEEauthorblockA{\IEEEauthorrefmark{3}\footnotesize Munich Center for Quantum Science and Technology, Germany}
	}

	\maketitle
	\thispagestyle{plain}

	\begin{abstract}
		A common starting point of  traditional quantum algorithm design is the notion of a universal quantum computer with a scalable number of qubits. This convenient abstraction mirrors classical computations manipulating finite sets of symbols, and allows for a device-independent development of algorithmic primitives.
		
		Here we advocate an alternative approach centered on the physical setup and the associated set of natively available operations. We show that these can be leveraged to great benefit by sidestepping the standard approach of reasoning about computation  in terms of individual qubits. 
		
		As an example, we consider hybrid qubit-oscillator systems with linear optics operations augmented by certain qubit-controlled Gaussian unitaries. The continuous-variable (CV) Fourier transform has a native realization in such systems in the form of homodyne momentum measurements. We show that this fact can be put to algorithmic use. Specifically, we give a polynomial-time quantum algorithm in this setup which finds a factor of an $n$-bit integer~$N$.  Unlike Shor's algorithm, or CV implementations thereof based on qubit-to-oscillator encodings, our algorithm relies on the CV (rather than discrete) Fourier transform. The physical system used is independent of  the number~$N$ to be factored: It consists of a single qubit and three oscillators only.

	\end{abstract}
	
	Shor's discovery of an efficient quantum factoring algorithm~\cite{Shor} has profoundly impacted the field of quantum computing: It has amplified the hope for meaningful computational advantages derived from quantum effects, and has played a pivotal role in shaping our community's preconceptions about the  prerequisites for quantum computing. This has given  rise to the notion of a universal, scalable quantum computer as succinctly formulated by DiVicenzo's criteria~\cite{DiVicenzocriteria}. Crucially, this concept---which is by now common lore---stipulates that a computationally useful quantum device needs to provide a number of (logical) qubits scaling extensively with the problem size. For example, this is the case when implementing Shor's algorithm: It requires a number of qubits that is proportional to the number of bits specifying the integer to be factored. 
	
	Following this traditional concept of  scalability, the most established current approach when using continuous-variable (CV) systems for quantum computing is to encode a  single logical qubit in a suitable two-dimensional subspace of a bosonic mode. By associating a single physical information carrier (a boson) to each  individual qubit, this philosophy emphasizes modularity, breaking the engineering challenge into more manageable pieces, and allowing for device-independent approaches in algorithms design. It amounts to a quite literal interpretation of what it means to scale a quantum computation. 
	
	Here we argue that in light of the possibilities offered by CV quantum systems, this simple and seemingly inevitable
	idea of scalability may be too restrictive. We give a polynomial-time quantum algorithm which factors an $n$-bit integer (for any integer~$n$) using a  device relying on three harmonic oscillators (bosons) and a single qubit only. The algorithm uses a polynomial number of elementary operations  which are readily available in present-day experimental setups. In other words, our algorithm  trades problem size (length~$n$ of the binary representation of the integer to be factored) against circuit size  (i.e., number of gates) while keeping the underlying physical system (3~bosons and 1~qubit) fixed. In contrast, in Shor's algorithm, both the number of gates and the number of qubits (i.e., the physical system) scale with the problem size~$n$.

	We note that---as any finite-dimensional space can trivially be embedded into  a single harmonic oscillator---the question  of how to realize a computation with a small number of CV information carriers is only meaningful when restricting to basic, physically realizable operations, and with an estimate on the number of operations used. We show that a polynomial number of the following elementary operations---summarized in Fig.~\ref{fig:elementary gates}--- is sufficient to factor an $n$-bit integer, where we assume that unitaries can be classically controlled by measurement results (or efficiently computable functions thereof):
	\begin{enumerate}[(i)]
		\item {\em Single- and two-mode  linear quantum optics operations} with constant coupling strength:
		This includes the preparation of the single-mode vacuum state, a set of constant-strength Gaussian one- and two-mode unitaries, and homodyne quadrature measurements.\label{it:singletwomodeops}
		\item {\em Qubit-controlled single-mode Gaussian unitaries} of constant strength: We use 
		qubit-controlled phase space displacements by a constant, and qubit-controlled phase space rotations with a constant angle.
		\item
		{\em Single-qubit operations}, namely preparation of the single-qubit computational basis state~$\ket{0}$, and application of the single-qubit Hadamard gate.\label{it:singlequbitoperations}
	\end{enumerate}
	We emphasize that these individual building blocks are physically realizable with current technology in appropriate  experimental setups. While linear optics operations are clearly well-established, qubit-boson interactions, typically modeled by a Jaynes-Cummings-type Hamiltonian, are rapidly becoming a standard tool in various platforms such as superconducting circuits, where qubit-controlled displacements (with a resonant drive~\cite{Eickbusch_2022} and in the dispersive regime~\cite{CampagneEikbushetal20}) as well as qubit-controlled phase space rotation~\cite{DispersiveRegimeSCCircuit} (in the dispersive regime) were experimentally implemented. 
	We refer to~\cite{liu2024hybridoscillatorqubitquantumprocessors} for an up-to-date and thorough review of the state of the art, and a detailed discussion of physical realizations of the operations we use here.
	\begingroup
	\begin{figure}
		\centering
		\hspace*{-0.5cm}\includegraphics[width=0.9\linewidth]{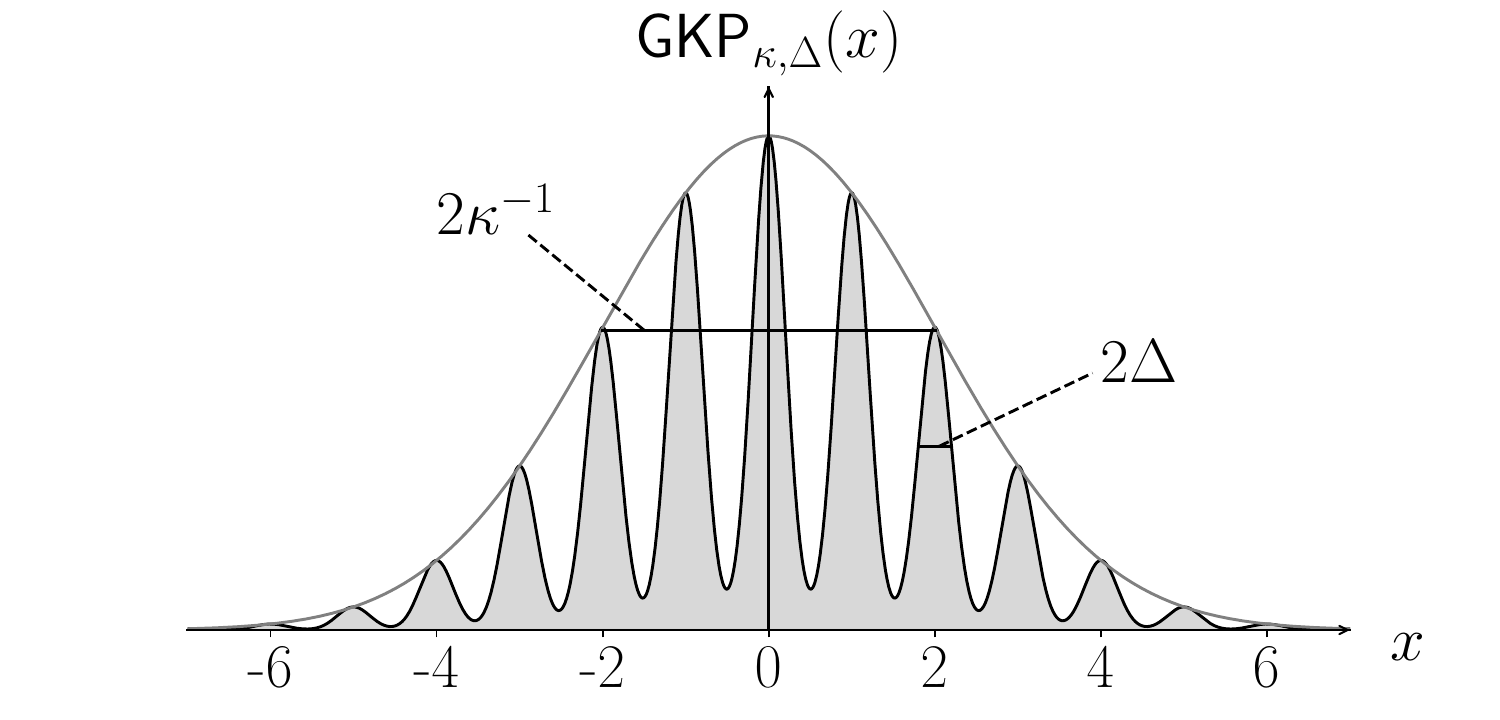}
		\caption{
			The approximate GKP state~$\ket{\gkp_{\kappa,\Delta}}$ defined in Eq.~\eqref{eq:gkpstatedefinitionapproximate}. In the limit~$(\kappa,\Delta)\rightarrow (0,0)$, this becomes a ``comb-state''~$\ket{\Sha}$, a formal superposition of Dirac-$\delta$-distributions centered on integers. For convenience, our convention differs slightly from the error-correction literature, where peaks are typically centered on integer multiplies of $\sqrt{2\pi}$ (instead of integers).\label{fig:approximategkpstates}}
	\end{figure}
	\endgroup

	\begingroup
	\begin{figure*}[bp]
		\centering
		\includegraphics[width=1\textwidth]{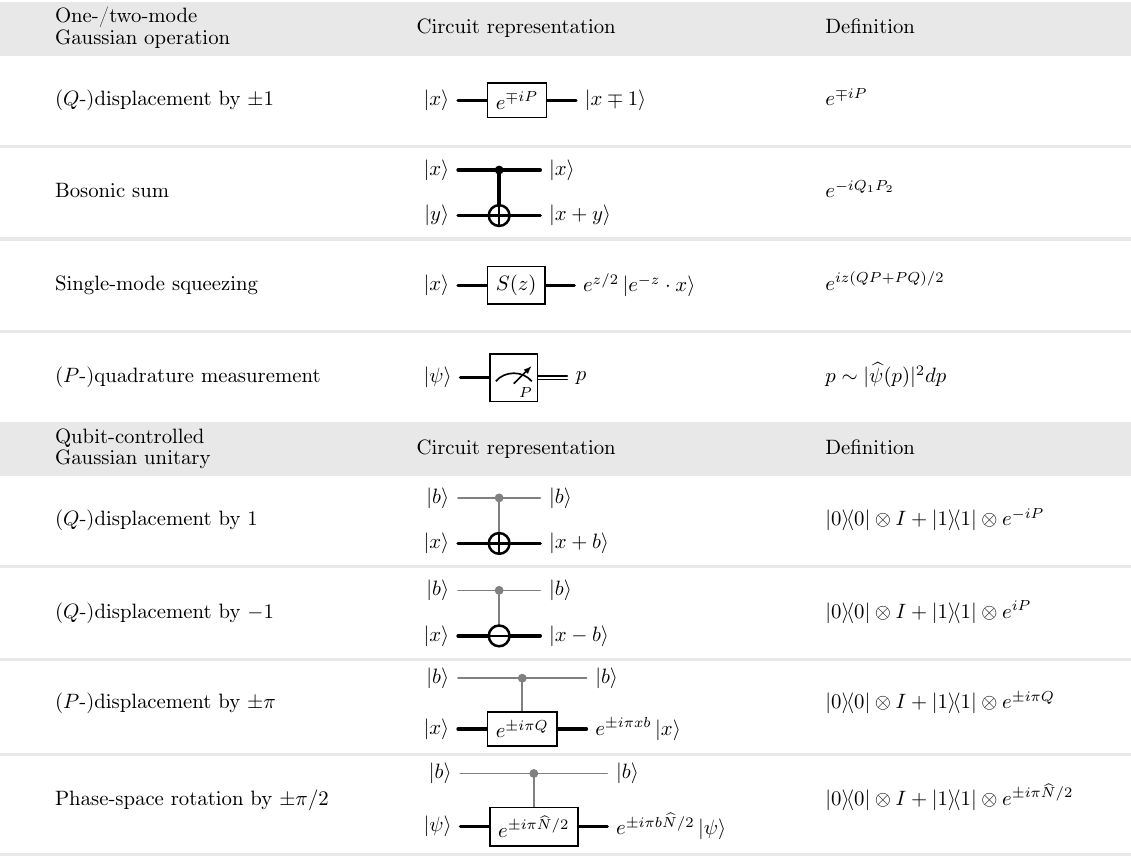}
		\vspace*{0.1cm}
		\caption{ 
			Elementary operations that we use in our circuits (CV systems are depicted by thick wires and qubit systems by thin grey wires).
			The circuit representations illustrate the actions for position-eigenstates~$\{\ket{x}\}_{x\in\mathbb{R}}$ of bosonic modes and the computational basis state~$\{\ket{b}\}_{b\in\{0,1\}}$ of qubit. The two-mode bosonic addition gate~$e^{-iQ_1P_2}$ transforms mode operators according to $(Q_1,P_1,Q_2,P_2)\mapsto (Q_1,P_1-P_2,Q_1+Q_2,P_2)$. A detailed discussion of how this gate can be realized in terms of beam-splitters and single mode squeezing  can be found in Ref.~\cite{liu2024hybridoscillatorqubitquantumprocessors}. A homodyne $P$-quadrature (measurement) produces a sample~$p\in\mathbb{R}$ from the distribution with density function~$p\mapsto |\widehat{\Psi}(p)|^2$, where $\widehat{\Psi}(p):=\frac{1}{(2\pi)^{1/2}}\int \Psi(x)e^{ipx}dx$ denotes the Fourier transform of~$\Psi$ when applied to a state~$\ket{\Psi}\in L^2(\mathbb{R})$. Homodyne $Q$-quadrature measurement is defined similarly. Controlled-phase space rotations are defined in terms of the number operator~$\widehat{N}=(Q^2+P^2-I)/2$.}\label{fig:elementary gates}
	\end{figure*}
	\endgroup
	For suitable choices of parameters (i.e., phase space displacement vectors, squeezing parameters and/or rotation angles), the physical operations~\eqref{it:singletwomodeops}--\eqref{it:singlequbitoperations} naturally realize certain 
	(real) arithmetic operations when the action on bosonic position-eigenstates~$\ket{x}$, $x\in\mathbb{R}$ and qubit computational basis states~$\ket{b}$, $b\in\{0,1\}$ is considered, see Fig.~\ref{fig:elementary gates}. Our algorithm relies on an extended arithmetic toolbox associated with certain 
	composite unitaries. These realize specific  arithmetic functionalities, see Fig.~\ref{fig:composite gates}. One of these unitaries  coherently performs a form of modular exponentiation:  It  computes~$x\mapsto f_{a,N,m}(x)$ for a function~$f_{a,N,m}$ such that \begin{align}
		f_{a,N,m}(x)\equiv a^x\pmod{N}\ \textrm{ for all } x \in \{0,...,2^m-1\}\ .
	\end{align}

	Depending on the setting, such a function can be 
	understood as a proxy for the  modular exponentiation map~$x \mapsto a^x\xmod N$  (with $a, x\in\mathbb{N}$). Correspondingly, we refer to~$f_{a,N,m}$ as a \emph{pseudomodular power}.

	\begingroup
	\begin{figure*}[hbp]
		\centering\includegraphics[width=\textwidth]{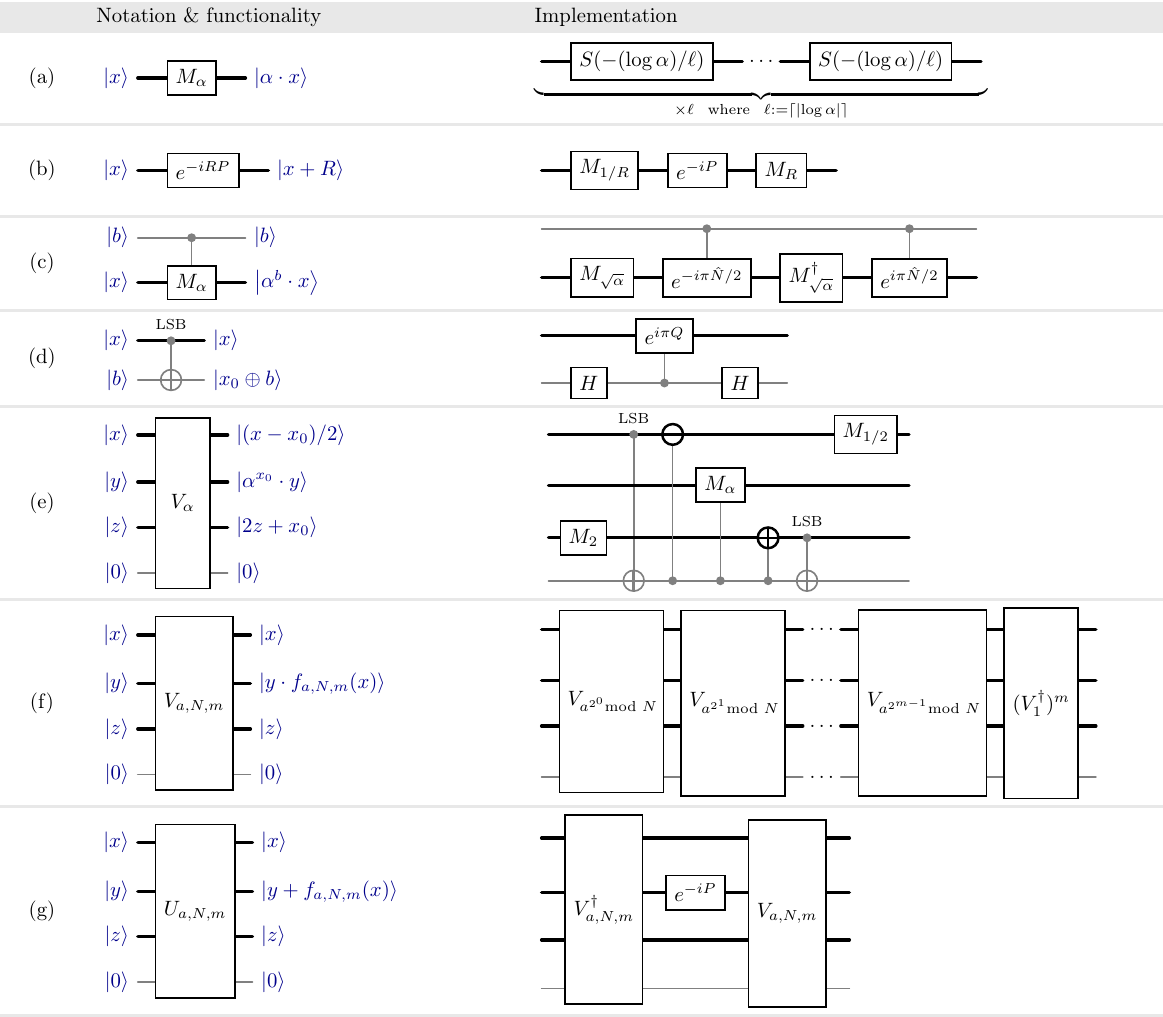}
		\caption{Composite unitaries: Their actions 
			are specified in blue (suppressing normalization factors) for the special case where $x=\sum_{i=0}^{m-1}2^i x_i$ is an $m$-bit integer. They include: (a)~scalar multiplication by a real number $\alpha>0$ (realized by a succession of $\ell\in\mathbb{N}$ constant-strength squeezing operations),
			(b)~translation of the position by a real number~$R>0$,  (c) a qubit-controlled version of scalar multiplication, 
			(d) coherent extraction of the least significant bit~$x_0$ of an integer~$x$ (e) an auxiliary unitary~$V_\alpha$ realizing scalar multiplication of the position~$y$ of the second mode by~$\alpha$, controlled by the least significant bit~$x_0$ of the position~$x$ of the first mode. This unitary then moves the bit~$x_0$ to the third ancillary mode and makes $x_1$ the new least significant bit of the position in the first mode. This is used to build (f) the unitary~$V_{a,N,m}$ that multiplies the position of the second mode by the pseudomodular power~$f_{a,N,m}(x)=\prod_{i=0}^{m-1} \left(a^{2^i}\xmod N \right)^{x_i}$ of the position~$x$ of the first mode. Finally, (g) the unitary~$U_{a,N,m}$ implementing a translation by the pseudomodular power $f_{a,N,m}(x)$ of the second bosonic mode controlled by the position of the first mode. Here  $b\in\{0,1\}$,  $a,N,m\in \mathbb{N}$, $z\in \mathbb{N}_0$ and $y \in \mathbb{R}$ are arbitrary.}
		\label{fig:composite gates}
	\end{figure*}
	\endgroup

	In addition to these derived operations, the second main ingredient of our algorithm are approximate, i.e., finitely squeezed, Gottesman-Kitaev-Preskill (GKP) states: Our circuit uses the single-mode (normalized) approximate GKP states $\ket{\gkp_{\kappa,\Delta}}\in L^2(\mathbb{R})$ (for $\kappa,\Delta>0$) defined by
	\begin{align}
		\gkp_{\kappa,\Delta}(x)&\propto\sum_{z\in\mathbb{Z}}
		e^{-\kappa^2 z^2/2}e^{-(x-z)^2/(2\Delta^2)} \ ,\label{eq:gkpstatedefinitionapproximate}
	\end{align}
	for $x\in\mathbb{R}$, 
	see Fig.~\ref{fig:approximategkpstates}. GKP states  have been introduced and extensively studied in the context of CV-quantum error correction, where their approximate translation-invariance allows to detect  displacement errors. 
	Due to their central importance, various proposals for the efficient preparation of approximate GKP states exist~\cite{Terhal_2016, Weigand_2018, Shi_2019, Hastrup_2021}. As the primary focus in this context is on detecting small (i.e., constant) shifts, existing protocols for their preparation typically target constant parameters~$(\kappa,\Delta)$. Furthermore, in the error correction context, it often suffices to measure  the quality of prepared states in terms of  effective squeezing parameters which capture the degree of invariance with respect to small shifts~\cite{Weigand_2018, Hastrup_2021}.  To our knowledge, the use of approximate GKP states in a quantum algorithmic context is new. This application requires significantly more stringent (i.e., non-constant) parameter choices for~$(\kappa,\Delta)$, as well as a fidelity guarantee. In separate work~\cite{brenneretalGKP2024}, we give an efficient (adaptive)  protocol~$\cP^{\gkp}_{\kappa,\Delta}$ with the following property:

	\begin{lemmamain}[{Corollary 5.3 of \cite{brenneretalGKP2024} (specialized)}] \label{lem: P gkp}
		There are two polynomials~$p,\tilde{p}$ with $p(0)=\tilde{p}(0)=0$ such that the following holds for all sufficiently large~$M\in\mathbb{N}$: For $\kappa=\Delta=p(1/M)$, there exists a state preparation protocol $\cP^{\gkp}_{\kappa,\Delta}$ that outputs a classical ``success'' flag with probability at least~$1/10$,
		and has the property that the output state $\rho\in\cB(L^2(\mathbb{R}))$ conditioned on ``success'' satisfies
		\begin{equation}
			\left\| \rho - \gkp_{\kappa,\Delta} \right\|_1 \le \tilde{p}(1/M)\ . 
		\end{equation}
		The protocol~$\cP^{\gkp}_{\kappa,\Delta}$ uses $O(\log M)$ operations of the form~\eqref{it:singletwomodeops}--\eqref{it:singlequbitoperations}, where some of the involved squeezing and displacement operations are classically controlled by (efficiently computable functions of) the measurement result associated with a $Q$-quadrature measurement.   
	\end{lemmamain}
	
	With these preparations, we can complete the description of our algorithm. Its ``quantum subroutine'' is  given by the circuit~$\ourcircuit_{a,N}$ depicted in Fig.~\ref{fig:maincircuit}, but with the approximate initial 
	GKP states~$\ket{\gkp_{\kappa_\bosonA,\Delta_\bosonA}},\ket{\gkp_{\kappa_\bosonB,\Delta_\bosonB}}$ replaced by the output states of the preparation procedure~$\cP_{\kappa_\bosonA,\Delta_\bosonA}^{\gkp}$ and~$\cP_{\kappa_\bosonB,\Delta_\bosonB}^{\gkp}$ (for suitably chosen parameters~$\kappa_\bosonA,\Delta_\bosonA,\kappa_\bosonB,\Delta_\bosonB$), respectively.
	We note that for our choice of parameters (see the supplementary material), both the circuit~$\cQ_{a,N}$ and the preparation procedures~$\cP^{\gkp}_{\kappa,\Delta}$, and hence also our quantum subroutine, use~$O(n^2)$ operations from the set of operations~\eqref{it:singletwomodeops}--\eqref{it:singlequbitoperations}.

	\begin{figure*}[bp]
		\centering
		\includegraphics{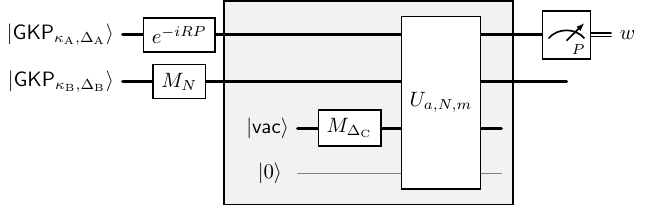}
		\caption{
			The quantum circuit $\ourcircuit_{a,N}$ where $a\in\mathbb{Z}_N^*$---see the supplementary material for a suitable choice of parameters~$R,m,\Delta_\bosonA,\kappa_\bosonA,\Delta_\bosonB,\kappa_\bosonB,\Delta_\bosonC>0$.
			The quantum subroutine used in our algorithm is obtained by
			replacing the initial approximate GKP states on the first and second modes by states prepared by the protocol~$\cP_{\kappa,\Delta}^{\gkp}$ (see Lemma~\ref{lem: P gkp}). 
			It applies a sequence of elementary operations (see  Fig.~\ref{fig:elementary gates}) and the derived unitaries (given in Fig.~\ref{fig:composite gates}) to two approximate GKP states (cf. Eq.~\eqref{eq:gkpstatedefinitionapproximate}) in the first and second mode, a vacuum state~$\ket{\mathsf{vac}}$ in the third mode and a qubit computational basis state~$\ket{0}$ in the qubit system. 
			The output is a sample~$w\in\mathbb{R}$ obtained by performing a  $P$-quadrature measurement on the first mode. The shaded subcircuit 
			approximately computes a pseudomodular power. To provide intuition, we will first discuss the effect of the circuit when this subcircuit is replaced by an ideal unitary~$U_{\mathbb{R},a,N}^{\textrm{ideal}}$ computing the real power, 
			see Fig.~\ref{fig:idealizedcircuit}.  
			\label{fig:maincircuit}
		}
	\end{figure*}
	
	Key to our algorithm is the fact that a single sample from the output distribution of the  circuit~$\ourcircuit_{a,N}$ can be post-processed by an efficient, i.e., polynomial-time classical algorithm, yielding a factor of~$N$ with a substantial probability. That is, we have the following:
	\begin{restatable}{lemmamain}{lemmain}\label{lem:mainlemma}
		Suppose $N$ is an $n$-bit number.
		There is a polynomial-time classical algorithm which---given a uniformly chosen element~$a\sim\mathbb{Z}_N^*$ and a single sample from the output distribution of the circuit~$\ourcircuit_{a,N}$ (see Fig.~\ref{fig:maincircuit})---produces a factor of $N$~with probability at least~$\Omega(1/\log n)$. 
	\end{restatable}

	For a given element~$a\in\mathbb{Z}_N^*$, our quantum subroutine proceeds identically to the circuit~$\ourcircuit_{a,N}$ but with the initial  approximate GKP states~$\ket{\gkp_{\kappa_\bosonA ,\Delta_\bosonA }}$ and
	~$\ket{\gkp_{\kappa_\bosonB ,\Delta_\bosonB}}$ replaced by states prepared using the protocol~$\cP^{\gkp}_{\kappa,\Delta}$. By Lemma~\ref{lem: P gkp}, the corresponding output distribution (conditioned on the preparation being successful, which happens with constant probability) is close to the output distribution of the circuit~$\ourcircuit_{a,N}$. This means that combining this quantum subroutine (for randomly chosen~$a\in\mathbb{Z}_N^*$) with the efficient classical algorithm of Lemma~\ref{lem:mainlemma} produces a factor of~$N$ with comparable probability. By repetition, we obtain the following by a suitable choice of parameters (see the supplementary material):
	
	\begin{restatable}[Efficient quantum algorithm for factoring] {theoremmain}{thmmain}\label{thm:main}
		There is a polynomial-time algorithm which, given an $n$-bit integer $N$,
		\begin{enumerate}[(I)]
			\item repeatedly uses a quantum circuit 
			on $3$~oscillators and $1$~qubit consisting of $O(n^2)$ elementary operations~\eqref{it:singletwomodeops}--\eqref{it:singlequbitoperations}, and
			\item produces a factor of $N$ with constant probability.
		\end{enumerate}
	\end{restatable}
	In the remainder of this paper, we explain the key ideas behind the circuit~$\ourcircuit_{a,N}$ (cf.  Fig.~\ref{fig:maincircuit}), and the corresponding post-processing procedure. We then comment on the technical challenges involved in establishing Lemma~\ref{lem:mainlemma}. The full details of this derivation can be found in the supplementary material.
	
	Our factoring algorithm is the result of  translating Shor's algorithm~\cite{Shor} to CV systems, with specific modifications exploiting their potential. 
	These modifications mean that our construction---although formally quite similar to Shor's algorithm---is not simply obtained by  embedding a finite-dimensional quantum computation into an infinite-dimensional system.\footnote{As we show in separate work, such an embedding can indeed be constructed in a straightforward manner: any polynomial-time quantum computation on $n$~qubits can be realized to any desired accuracy using a polynomial number of operations from~\eqref{it:singletwomodeops}--\eqref{it:singlequbitoperations} acting on a system consisting of a single qubit and a single oscillator only.}  
	Instead, our approach relies on different algebraic structures: We use an approximate GKP state as a proxy for a uniform superposition over all integers instead of a uniform superposition over all $n$-bit integers as in Shor's algorithm. Instead of realizing modular arithmetic by  gates acting on finite-dimensional systems, our gates natively perform real arithmetic. To realize modular arithmetic we exploit the (approximate) stabilization property of (approximate) GKP states under suitable (discrete) displacements, a property that underlies their relevance for quantum error correction. Finally,  our algorithm leverages the Fourier transform on~$\mathbb{R}^n$ instead of the Fourier transform over a finite cyclic group. With these alternative choices, we can show that the physical  operations~\eqref{it:singletwomodeops}--\eqref{it:singlequbitoperations} (together with an efficient classical computation) are sufficient to address the factoring problem. The difference between our approach and Shor's algorithm is further expressed in the different complexities of the resulting algorithms: the quantum subroutine of our algorithm has circuit size~$O(n^2)$, whereas that of Shor's algorithm has size~$O(n^2\log n)$ (when relying on the best currently known classical multiplication algorithm with size~$O(n\log n)$ by~\cite{HarveyvanderHoeven}).

	{\bf Shor's circuit: Modular arithmetic and the Fourier transform over~$\mathbb{Z}_q$.} In more detail, recall that Shor's algorithm~\cite{Shor} builds on 
	Miller's reduction~\cite{Miller1976} of the problem of finding the prime factorization\footnote{Whenever the period~$r$ of $f_{a,N}$ is even  and $x^{r/2}\neq -1 \xmod N$, it can be used to compute a factor of~$N$. For an integer~$a\in \mathbb{Z}_N^*$ chosen uniformly at random,
		this is the case for the period~$r$ of~$f_{a,N}$ except with probability~$2^{-k}$, where $k$~is the number of factors of~$N$ (by an application of the Chinese remainder theorem).}
	of an (odd) integer~$N$ to the problem of  finding the period~$r$ of the function~$f_{a,N}(x)= a^x\xmod N$, where $a\in\mathbb{Z}_N^*$. Given a pair $(a,N)$, Shor's algorithm applies an efficient classical post-processing algorithm to a sample~$c$  produced by a (polynomial-size) quantum circuit~$\shorcircuit_{a,N}$. This gives the period~$r$ of the function~$f_{a,N}$ with probability at least $\Omega(\frac{1}{\log \log r})$, resulting in a polynomial-time factorization algorithm by repeated application. 
	\begin{figure*}[bp]
		\centering
		
		\begin{subfigure}[b]{0.4\linewidth}
			\centering
			\includegraphics{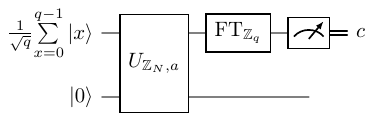}
			\vspace*{0.1cm}
			\caption{Shor's circuit $\shorcircuit_{a,N}$.}\label{fig:shorcircuit}
		\end{subfigure}~~~~~~~~~~~~~~~~~~~~
		\begin{subfigure}[b]{0.5\linewidth}
			\centering
			\includegraphics{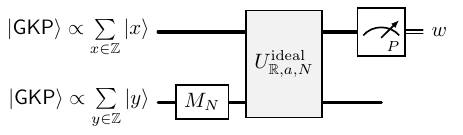}
			\caption{An idealized bosonic circuit $\ourcircuitideal_{a,N}$.}\label{fig:idealizedcircuit}
		\end{subfigure}
		
		\caption{Contrasting Shor's circuit with an idealized bosonic circuit $\ourcircuitideal_{a,N}$: We replace modular by real arithmetic. 
			Instead of modular exponentiation followed by Fourier transform and computational basis measurement, we use standard (real) exponentiation and homodyne detection. To obtain ``modularity'' (i.e., realize $\pmod N$ computations), the second bosonic mode is 
			initialized in a modified GKP state with spacing~$N$.}
		\label{fig:circuitsforfactoring}
	\end{figure*}

	A high-level description of Shor's circuit~$\shorcircuit_{a,N}$ is given in Fig.~\ref{fig:shorcircuit}.  
	It uses a binary encoding of $(\log_2 q)$-bit integers into~$\log_2 q$~qubits, where $q$ is the smallest power of $2$ such that $N^2 < q$. The circuit starts with the first register in the uniform superposition~$\frac{1}{q^{1/2}}\sum_{x=0}^{q-1}\ket{x}$ over all integers $z\in \{0,\ldots,q-1\}=:\mathbb{Z}_{q}$,  coherently computes the function~$f_{a,N}$ into a second register using a ``modular exponentiation unitary'' $U_{\mathbb{Z}_N,a}$ acting as
	\begin{align}
		U_{\mathbb{Z}_N,a}(\ket{x}\otimes\ket{y})=\ket{x}\otimes\ket{y+(a^x \xmod N)}\ ,
	\end{align} applies the (discrete) Fourier transform~$\textrm{FT}_{\mathbb{Z}_{q}}$  on $L^2(\mathbb{Z}_q)$ and finally measures the first register in the computational basis giving an outcome~$c\in\mathbb{Z}_q$. 
	Shor shows that  the output distribution~$p_{\mathbb{Z}_{q}}(c)$ is close to the uniform distribution on the set~$\{\frac{q}{r}\cdot d\ |\ d\in \{0,\ldots,r-1\}\}$, i.e., the output~$c$ is an integer multiple of~$q/r$ with high probability. The classical post-processing algorithm  then proceeds as follows:  With the  continued fraction expansion of~$c/q$, the number~$c/q$ is rounded to the nearest fraction of the form~$d'/r'$ with denominator~$r'$ smaller than~$N$. The value of~$r=r'$ can then be extracted whenever~$d'$ and~$r'$ are coprime.

	Here we argue that---in place of the circuit~$\shorcircuit_{a,N}$---the idealized bosonic circuit~$\ourcircuitideal_{a,N}$ given in Fig.~\ref{fig:idealizedcircuit} can be used, assuming that the classical post-processing procedure is appropriately modified.
	We note that the circuit~$\ourcircuitideal_{a,N}$  involves non-normalizable GKP states, i.e., the
	formal uniform superposition
	\begin{align}
		\ket{\gkp} \propto \sum_{x\in\mathbb{Z}}\ket{x}\ ,\label{eq:gkpidealstateuniform}
	\end{align}
	also called ``comb-states'' cf.~Fig.~\ref{fig:approximategkpstates}.  The circuit~$\ourcircuitideal_{a,N}$ is thus not physical, but it nevertheless illustrates the key ideas underlying our physical circuit~$\ourcircuit_{a,N}$ in Fig.~\ref{fig:maincircuit}. Let us highlight how it differs from  Shor's circuit~$\shorcircuit_{a,N}$.
	
	{\bf Real instead of modular arithmetic. }Clearly, 
	the expression~\eqref{eq:gkpidealstateuniform} is a natural analog of the uniform superposition of basis states used in Shor's circuit~$\shorcircuit_{a,N}$. The circuit~$\ourcircuit_{a,N}$ also    uses a CV-analogue of the modular exponentiation unitary~$U_{\mathbb{Z}_N,a}$.  Concretely, consider the unitary~$U^{\textrm{ideal}}_{\mathbb{R},a,N}$ which acts on pairs of position-eigenstates~$\ket{x},\ket{y}$ as
	\begin{align}
		&U_{\mathbb{R},a,N}^{\textrm{ideal}}(\ket{x}\otimes\ket{y})=\ket{x}\otimes\ket{y+a^x} \label{eq:idealizedcircuitfunctionality}
	\end{align}
	where
	\begin{align}
		a^{x}:=\begin{cases}
			a^{x} & \text{if~$x\ge 0$}\\
			(a^{-1} \xmod N)^{|x|}& \text{if $x<0$}\ .
		\end{cases}&    
	\end{align}
	
	Importantly, the exponentiation  in the definition of~$a^x$ is not taken modulo~$N$, but is to be understood over the non-negative reals. Integer values belonging to~$\mathbb{Z}_{N}$ are only obtained subsequently by means of a modular measurement.
	
	{\bf Homodyne measurement instead of Fourier transform.}
	As already mentioned, our bosonic circuits~$\ourcircuit_{a,N}$
	and $\ourcircuitideal_{a,N}$ exploit a peculiarity of CV-systems that allows to circumvent the need for constructing an implementation of the unitary Fourier transform. Indeed, homodyne $P$-quadrature measurements (natively available in typical quantum-optical systems) are equivalent to a Fourier transform~$\textrm{FT}_{\mathbb{R}}$ on $L^2(\mathbb{R})$ followed by a homodyne $Q$-quadrature measurement. Motivated by this, the circuit~$\ourcircuitideal_{a,N}$ simply applies a homodyne $P$-quadrature measurement on the first mode, giving a sample from the distribution
	\begin{align}
		p'_{\mathbb{R}}(w)\propto \left|\langle \widehat{w}|\otimes I, U_{\mathbb{R},a,N}^{\textrm{ideal}}(\ket{\gkp}\otimes M_N\ket{\gkp})\right|^2\ ,\label{eq:pprimedistributiondefinition}
	\end{align}
	for $w \in \mathbb{R}$.
	Here $\ket{\widehat{w}}:=\textrm{FT}_{\mathbb{R}}^\dagger \ket{w}$  is the momentum-eigenstate to value~$w$, i.e., $P\ket{\widehat{w}}=w\ket{\widehat{w}}$. 
	
	{\bf Recovering modular arithmetic with auxiliary $\gkp$ states.}
	A further key difference between Shor's circuit~$\shorcircuit_{a,N}$ and our bosonic circuits~$\ourcircuitideal_{a,N}$, (and the physically realizable circuit) $\ourcircuit_{a,N}$
	is the choice of initial state in the second register (see Fig.~\ref{fig:circuitsforfactoring}): Whereas in Shor's circuit,  the auxiliary register is initialized in the state~$\ket{0}$ (corresponding to $0\in\mathbb{Z}_q$), the idealized circuit~$\ourcircuitideal_{a,N}$ uses~$M_N\ket{\gkp}$ in its place. This is a uniform superposition~$M_N\ket{\gkp}\propto \sum_{y\in\mathbb{Z}}\ket{y\cdot N}$ of integer multiples of~$N$. This choice is dictated by the need to realize a computation modulo~$N$ (rather than only real arithmetic): Together 
	with the subsequent steps in the algorithm, it essentially realizes a  modular  $P$-quadrature measurement. 
	
	We note that modular measurements have traditionally been used in syndrome extraction circuits for GKP-codes~\cite{GKP}. An example is shown in Fig.~\ref{fig:gkpsyndromeextraction}. Our  circuits~$\ourcircuitideal_{a,N}$  and $\ourcircuit_{a,N}$
	are partly motivated by such modular measurement circuits. Indeed, by combining real arithmetic with this approach, they effectively realize an analogue of modular exponentiation followed by measurement.
	\begingroup
	\begin{figure}
		\centering
		\includegraphics{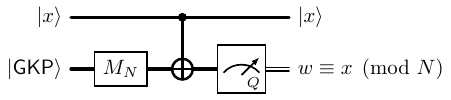}
		\caption{Modular position measurement using an auxiliary GKP state: The modular value~$x\xmod N$ of a position-eigenstate~$\ket{x}$ is obtained by ``adding'' $x$ to an auxiliary system in the state~$M_N\ket{\gkp}$, and then measuring the $Q$-quadrature of that system. Because of the fact that $M_N\ket{\gkp}\propto\sum_{y\in\mathbb{Z}}\ket{y\cdot N}$ is the uniform superposition of position-eigenstates with spacing~$N$, this circuit provides~$w\equiv x \pmod N$ (and no other information on~$x$). This reasoning has been used to construct syndrome extraction circuits for GKP-codes~\cite{GKP}.
			\label{fig:gkpsyndromeextraction}
		} 
	\end{figure}
	\endgroup
	
	{\bf Output distribution and classical post-processing.} A brief computation which we give below shows that the output distribution $p'_{\mathbb{R}}$ of the circuit~$\ourcircuitideal_{a,N}$  (see Eq.~\eqref{eq:pprimedistributiondefinition}) is (formally) 
	the uniform distribution on the set~$\left\{j/r\ |\ j\in\mathbb{Z}\right\}$.
	A suitable classical post-processing procedure is then evident: whenever the outcome~$j/r$ has the property that~$j$ and $r$ are coprime, the period can be recovered immediately from~$j/r$. Following Shor's analysis, such a favorable outcome is obtained with probability at least~$\Omega(1/\log \log r)$, leading to a polynomial overall runtime when the circuit is used repeatedly\footnote{The integers~$j$ and $r$ are coprime with probability at least $\phi(r)/r$, where $\phi(\cdot)$ is Euler's totient function. Using the bound given in~\cite[Thm. 328]{HardyWright}, one can conclude that we recover the period $r$ of $a$ modulo $N$ with probability at least $\delta/\log\log r$ for some $\delta = \Theta(1)$. Therefore, it is enough to repeat the procedure $\log\log N$ times, i.e.\@ logarithmically many times in the number of bits of $N$ to succeed with probability at least $1-e^{-1}$.}.

	To show the claim that the output distribution~$p'_{\mathbb{R}}$ is (formally) uniform on the set~$\{ j/r\ |\ j\in\mathbb{Z}\}$, first observe that  the state before the measurement can be written as
	\begin{align}
		U_{\mathbb{R},a,N}^{\textrm{ideal}}(\ket{\gkp}\otimes M_N\ket{\gkp})\label{eq:statetobemeasuredxam}\\
		\propto \sum_{x\in\mathbb{Z}}\sum_{y\in\mathbb{Z}}\ket{x}\otimes &\ket{y\cdot N+a^x}\\
		\propto \sum_{x\in \mathbb{Z}}\sum_{y\in\mathbb{Z}}\ket{x}\otimes&\ket{y\cdot N+(a^x\xmod N)}\ .
	\end{align}
	We note that  modular arithmetic arises here  because of the invariance of the state~$M_N\ket{\gkp}$ under translations by integer multiples of~$N$.
	
	To analyze the output distribution~$p'_{\mathbb{R}}$, we may without loss of generality assume that the $Q$-quadrature of the second system of the state~\eqref{eq:statetobemeasuredxam} is measured (instead of simply tracing out the system). Since~$r$ is the order of the element~$a$ in the multiplicative group~$\mathbb{Z}_N^*$, the measurement outcome $m\in\mathbb{R}$ has the form $m=\ell\cdot N+a^k$ for some $\ell\in\mathbb{Z}$ and $k\in \{0,\ldots,r-1\}$. The corresponding post-measurement state of the first mode then is the state
	\begin{align}
		\sum_{\substack{x\in\mathbb{Z}\\ a^x\equiv a^k \pmod{N}}}\ket{x}
		&= \sum_{\substack{x\in\mathbb{Z}\\ x\equiv k \pmod{r}}} \ket{x}\\
		&= \sum_{j\in\mathbb{Z}} \ket{j\cdot r+k}=:\ket{\Sha_{k,r}}\  .
	\end{align}
	Because the state~$\ket{\Sha_{k,r}}=e^{-ikP}\ket{\Sha_{0,r}}$ is the result of a translation in position applied to the 
	state~$\ket{\Sha_{r}}:=\ket{\Sha_{0,r}}$, its Fourier transform~$\widehat{\Sha}_{k,r}$ differs from that of the Dirac comb~$\Sha_{r}(x)\propto\sum_{j\in\mathbb{Z}} \delta(x-jr)$  (which is identical to~$M_r\ket{\gkp}$) only by a $p$-dependent phase factor, i.e., ~$\widehat{\Sha}_{k,r}(p)=e^{-ikp}\widehat{\Sha}_{0,r}(p)$. In particular, this implies that when measuring the momentum of the first mode, the resulting distribution over outcomes~$c\in\mathbb{R}$ is independent of~$k$ and equal to $p'_{\mathbb{R}}(c) \propto \left|\widehat{\Sha}_{r}(c)\right|^2$. The claim then follows from the fact that the Fourier transform of the Dirac comb is
	$\widehat{\Sha}_r(x)=\frac{1}{r}\Sha_{\frac{1}{r}}(x)$, see e.g.,~\cite{bracewell2000fourier}.
	
	{\bf A physically realizable circuit.} This formal discussion of the idealized circuit~$\ourcircuitideal_{a,N}$ merely illustrates the basic ideas. The circuit~$\ourcircuitideal_{a,N}$ falls short of being a real, i.e., physically amenable circuit in two important ways. First, it relies on idealized (infinitely squeezed) GKP states which are unnormalizable and hence unphysical. Second, we have not provided a circuit decomposition of the unitary~$U_{\mathbb{R},a,N}^{\textrm{ideal}}$ (defined in Eq.~\eqref{eq:idealizedcircuitfunctionality}) into elementary operations from the list of operations~\eqref{it:singletwomodeops}--\eqref{it:singlequbitoperations}. These issues are addressed by the circuit~$\ourcircuit_{a,N}$  (see Fig.~\ref{fig:maincircuit}) and Lemma~\ref{lem:mainlemma}, which summarizes the result of the corresponding analysis.

	The main building blocks of the circuit~$\ourcircuit_{a,N}$ in Fig.~\ref{fig:maincircuit}  are:
	\begin{enumerate}[(1)]
		\item Two approximate GKP states $\ket{\gkp_{\kappa_\bosonA, \Delta_\bosonA}}$ and $\ket{\gkp_{\kappa_\bosonB, \Delta_\bosonB}}$ with parameters $\kappa_\bosonA = \Delta_\bosonA = 2^{-\Theta(n)}$ and $\kappa_\bosonB=\Delta_\bosonB = 2^{-\Theta(n^2)}$, respectively. (As explained above, high-fidelity approximations to these states can be produced by a protocol having polynomial complexity in~$n$, see Lemma~\ref{lem: P gkp}.)

		\item A displacement $e^{-iRP}=M_Re^{-iP}M_{R^{-1}}$ with $R = 2^{\Theta(n)}$ applied to the first mode (which can be realized by a polynomial number of constant-strength operations): This 
		shifts the center (expectation value of~$Q$) of the initial state~$\ket{\gkp_{\kappa_\bosonA ,\Delta_\bosonA}}$ to the right,  ensuring that the subsequent 
		arithmetic operations (corresponding to the coherent evaluation of a pseudomodular power) are applied to positive arguments 
		only---up to a negligible part having to do with the shape of the envelope of the state $\gkp$.
		\item The operation $M_N$ applied to the second mode is used to prepare the state $M_N\ket{\gkp}$ (used to implement a modular measurement, cf. Fig.~\ref{fig:gkpsyndromeextraction}).
		\item The unitary $M_{\Delta_\bosonC}$ applied to the vacuum state in the third mode, where $\Delta_\bosonC = 2^{-\Theta(n)}$. (We note that this unitary can be realized by a polynomial number of constant-strength squeezing operations, cf. Fig.~\ref{fig:elementary gates}.)

		\item
		The unitary~$U_{a,N,m}$ with parameter $m = \Theta(n)$ acting on~$L^2(\mathbb{R})^{\otimes 3} \otimes\mathbb{C}^2$ constitutes the rest of the circuit before the measurement. This unitary acts  on position-eigenstates~$\ket{x}\otimes \ket{y}\otimes \ket{z}$ where $z=0$
		and a qubit in the state~$\ket{0}$ as 
		\begin{align}
			U_{a,N,m}\left(\ket{x}\otimes\ket{y}\otimes\ket{0}\otimes \ket{0}\right)\qquad\qquad\qquad \\
			\quad \qquad =\ket{x}\otimes \ket{y + f_{a,N,m}(x)}\otimes\ket{0} \otimes \ket{0}\ ,
		\end{align}
		when $x\in\{0,..., 2^m-1\}$ and $y\in\mathbb{R}$ is arbitrary, i.e., it coherently computes the pseudomodular power~$f_{a,N,m}(x)$ (see Fig.~\ref{fig:composite gates}). It takes the role of~$U_{\mathbb{R},N}^{\mathsf{ideal}}(a)$ in the idealized circuit~$\ourcircuitideal_{a,N}$ (see Fig.~\ref{fig:idealizedcircuit}). 
		\item A homodyne $P$-quadrature measurement of the first bosonic mode  realizing the Fourier transform on $L^2(\mathbb{R})$ of the first mode followed by a homodyne $Q$-quadrature measurement.
	\end{enumerate}

	{\bf Analysis of the physically realizable circuit.} The use of physically realizable (approximate) GKP states~$\ket{\gkp_{\kappa_\bosonA ,\Delta_\bosonA  }}$ and~$\ket{\gkp_{\kappa_\bosonB ,\Delta_\bosonB }}$ with (finite) squeezing parameters $(\kappa_\bosonA ,\Delta_\bosonA)$ and $(\kappa_\bosonB ,\Delta_\bosonB)$ in the circuit~$\ourcircuit_{a,N}$ necessitates a detailed analysis of arithmetic function evaluation by unitaries on~$L^2(\mathbb{R})^{\otimes 3} \otimes\mathbb{C}^2$. Specifically, we have to consider  non-integer inputs, i.e., cases when position states~$\ket{x}$ with $x\in\mathbb{R}\backslash \mathbb{Z}$ are involved. By continuity arguments and using the definition of approximate  GKP states, we argue that for functions~$\Psi\in L^2(\mathbb{R})^{\otimes 3} \otimes\mathbb{C}^2$ whose support is concentrated sufficiently close to the set of integers, the considered unitaries approximately implement a rounded version of function evaluation. For example, we argue that---when acting on a product  state~$\Psi_\bosonA\otimes\Psi_\bosonB\otimes\Psi_\bosonC\otimes\ket{0}_\qubit$ such that 
	(i) $\Psi_\bosonA$ is (in position-space) fully supported in the neighborhood of an integer~$x\in\mathbb{N}_0$, and (ii)~$\Psi_\bosonC$ is supported near $z=0$ ---the unitary~$U_{a,N,m}$ displaces the position of the second mode by the pseudomodular power~$f_{a,N,m}(x)$ of~$x$ (see Lemma~\ref{lem: Ua shor circuit error} in the supplementary material for a quantitative statement). Applied to the circuit~$\ourcircuit_{a,N}$, this analysis shows that the associated output distribution on~$\mathbb{R}$ still has the property that the samples can still be post-processed to find the period~$r$ with a significant probability.
	
	{\bf Conclusions and Outlook.} 
	Our polynomial-time factoring algorithm exemplifies the benefits of hardware-aware quantum algorithms design: We show how to algorithmically leverage a set of natively available operations in hybrid qubit-oscillator-systems. This draws attention to a physically motivated computational model which deserves further study from a complexity-theoretic perspective. Despite the efficiency guarantee we establish, our proposal is---at present---a proof-of-principle proposal of theoretical nature. To realize our algorithm experimentally, significant hurdles still need to be overcome. 
	
	First, although our algorithm only uses a polynomial number of active single-mode (squeezing) operations, the states produced in the course of our algorithm have extensive energy. This fact is hardly surprising---indeed, it is necessary to meaningfully encode the integer to be factored. We note that the challenge of producing highly squeezed state is certainly not a new one, but a primary focus of experimentally oriented research. In this sense, our proposal aligns with central objectives in the field.
	
	Second, while our algorithm establishes a new connection between continuous-variable quantum error-correcting codes (in the form of GKP states) and quantum algorithms, it does not incorporate fault-tolerance considerations. We hope that---similar to the way Shor’s algorithm spurred interest in quantum error-correction---our work will motivate further research in fault-tolerant quantum information-processing in hybrid qubit-oscillator systems. 
	
	{\bf Acknowledgements.}
	LB, LC and RK acknowledge support by the European Research Council
	under grant agreement no. 101001976 (project EQUIPTNT), as well as the Munich Quantum Valley, which is supported by the Bavarian state government through the Hightech Agenda Bayern Plus. XCR thanks the Swiss National Science Foundation (SNSF) for their financial support and acknowledges funding from the BMW endowment fund.
	\normalem
	
	\addtocontents{toc}{\protect\setcounter{tocdepth}{-3}}
	
	\addtocontents{toc}{\protect\setcounter{tocdepth}{3}}
	\ULforem

	\clearpage
	\renewcommand*{\thepage}{\normalfont\arabic{page}}
	\fontsize{11}{15}\selectfont
	\onecolumn 
	\titleformat{\section}{\Large\bfseries}{\thesection}{1em}{}
	\titleformat{\subsection}{\large\bfseries}{\thesubsection}{1em}{}
	\titleformat{\subsubsection}{\bfseries}{\thesubsubsection}{1em}{}
	\titlespacing*{\section}{0pt}{*2.5}{*1.5}
	\titlespacing*{\subsection}{0pt}{*2.5}{*1.5}
	\titlespacing*{\subsubsection}{0pt}{*2}{*1}
	\def\thesection{\arabic{section}}
	\def\thesubsection{\arabic{section}.\arabic{subsection}}
	\def\thesubsubsection{\arabic{section}.\arabic{subsection}.\arabic{subsubsection}}
	
	\section*{\huge{Supplementary material}}
	In this supplementary material, we provide additional intuition on our factorization algorithm and 
	establish our main result (Theorem~\ref{thm:main} in the main text). 
	Recall that the quantum circuit $\cQ_{a,N}$  (see Fig.~\ref{fig:maincircuit} in the main text) is the quantum subroutine used in the factorization algorithm. We reproduce it here for the reader's convenience:
	\begin{align}
		\cQ_{a,N} ~&:=~
		\vcenter{\hbox{\includegraphics{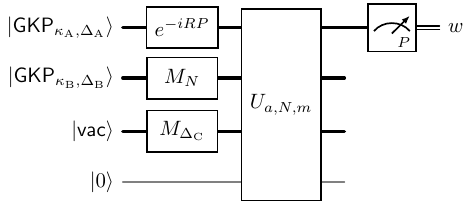}}}
	\end{align}
	Here $N\in\mathbb{N}$ is an $n$-bit integer to be factored (i.e., $2^{n-1}\leq N\leq 2^n-1$), and~$a\in\mathbb{Z}_N^*$ is a randomly chosen integer coprime to~$N$. We denote the first, second and third modes and the qubit in the circuit $\cQ_{a,N}$ by the labels $\bosonA,\bosonB,\bosonC$ and $\qubit$, respectively.
	The circuit~$\cQ_{a,N}$ additionally depends on the parameters~$(R,m)$
	and~$(\kappa_\bosonA,\Delta_\bosonA,\kappa_\bosonB,\Delta_\bosonB, 
	\Delta_\bosonC)$. In this work, we  make the choice of parameters given by Table~\ref{tab: parameters}.
	\begin{table}[ht]
		\centering
		\renewcommand{\arraystretch}{1.3}
		\normalsize
		\begin{tabular}{c||l}
			\multirow{2}{*}{Unitaries} & $m=17n$ \\
			& $R 
			=2^{17n-1} $\\
			\hline\hline
			system $\bosonA$ & $\kappa_{\bosonA}=\Delta_{\bosonA}=2^{-16n}$\\
			\hline
			system $\bosonB$ & $\kappa_{\bosonB}=\Delta_{\bosonB}=2^{-18n^2}$\\
			\hline
			system $\bosonC$ & $\Delta_{\bosonC} = 
			2^{-50n}$
		\end{tabular}
		\caption{The parameters~$(N,a,m,R)$ determine 
			what unitaries are applied in the circuit~$\cQ_{a,N}$. The parameters~$(\kappa_\bosonA,\Delta_\bosonA)$, $(\kappa_\bosonB,\Delta_\bosonB)$ and~$\Delta_\bosonC$ determine the amount of squeezing in the modes~$\bosonA,\bosonB$ and~$\bosonC$, respectively. The table gives a choice of these parameters such that output samples of~$\cQ_{a,N}$ can be used to factor~$N$, and~$\cQ_{a,N}$ can be realized by a circuit of depth and size polynomial in~$n$.}
		\label{tab: parameters}
	\end{table}
	
	We proceed as follows: To develop intuition, we first discuss a version of the quantum subroutine where the 
	initial  approximate 
	GKP states~$\ket{\gkp_{\kappa_\bosonA,\Delta_\bosonA}}$,~$\ket{\gkp_{\kappa_\bosonB,\Delta_\bosonB}}$  are replaced by idealized GKP states and where the vacuum state~$\ket{\vac}$ is replaced by a position-eigenstate localized at~$x=0$ (see Section~\ref{sec: circuit explanation}).  In Section~\ref{sec: approximate analysis}, we prove our main result (Theorem~\ref{thm:main} in the main text). We show
	that the circuit~$\ourcircuit_{a,N}$ 
	can be used to find a factor of~$N$ even when physically meaningful initial states (approximate GKP states) are used. This proof relies on three propositions, which we show in subsequent sections. They can be  summarized as follows:
	\begin{LaTeXdescription}
		\item[Proposition 1] (shown in Section~\ref{sec: prop zero proof}) identifies 
		a sufficient condition for when 
		access to samples of an ($N$-dependent) probability distribution on~$\mathbb{R}$ can be used to efficiently find a factor of~$N$. 
		\item[Proposition 2] (shown in Section~\ref{sec: prop one proof}) derives an expression for (an approximation to) the state of the circuit~$\cQ_{a,N}$ before the measurement.
		\item[Proposition 3] (shown in Section~\ref{sec: prop two proof}) analyzes the 
		output distribution of~$\cQ_{a,N}$, and shows that the condition of Proposition~1 is satisfied.
	\end{LaTeXdescription}
	Combined with the efficient preparation procedure for approximate GKP states  from~\cite{SMbrenneretalGKP2024} (see Lemma~\ref{lem: P gkp} in the main text), Propositions~$1$--$3$ imply 
	our main result (Theorem~\ref{thm:main}). 
	Indeed, Propositions~$1$--$3$ mean that the output of the circuit~$\cQ_{a,N}$  can be efficiently post-processed to obtain a factor of~$N$ (see Lemma~\ref{lem:mainlemma}). To obtain a circuit which does not use approximate GKP initial states as a resource, 
	we use the GKP state preparation protocol from~\cite{SMbrenneretalGKP2024} to prepare initial states $\sigma_{\kappa_\bosonA,\Delta_\bosonA}, \sigma_{\kappa_\bosonB,\Delta_\bosonB}\in\cB(L^2(\mathbb{R}))$ close to approximate GKP states in modes $\bosonA$ and $\bosonB$, respectively. Finally, we achieve the constant success probability stated in Theorem~\ref{thm:main} by repetition.

	We need a number of auxiliary concepts and results.
	In Section~\ref{sec: approx function evaluation} we discuss how to perform approximate function evaluation using CV quantum systems. 
	In Section~\ref{sec: properties appx GKP}, we summarize relevant properties of approximate GKP states. Combinatorial bounds on Gaussian sums used in our analysis are given in Section~\ref{sec: combinatorial bounds}.

	\subsection*{Circuits analyzed in this supplementary material}
	It is convenient to view the circuit~$\cQ_{a,N}$, as well as other circuits appearing in our analysis, as elements of a family of circuits parametrized by a triple~$(\rho_\bosonA,\rho_\bosonB,\rho_\bosonC)$ of single-mode initial states which may or may not be mixed. For each such triple, we define a circuit $\circWithInitStates(\rho_\bosonA,\rho_\bosonB,\rho_\bosonC)$ 
	by  replacing the initial states (up to the unitaries~$e^{-iRP}$ on~$\bosonA$ and $M_{\Delta_\bosonC}$ on~$\bosonC$) 
	in the circuit~$\cQ_{a,N}$ by $(\rho_\bosonA,\rho_\bosonB,\rho_\bosonC)$,  i.e., we set
	\vspace*{0.3cm}
	\begin{align}
		\circWithInitStates(\rho_\bosonA,\rho_\bosonB,\rho_\bosonC)~&:=~
		\vcenter{\hbox{\includegraphics{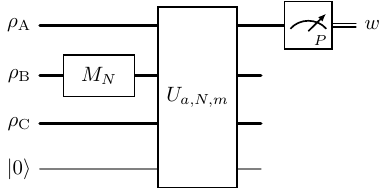}}}
		\label{eq:cW circuit}
	\end{align}
	\vspace*{0.3cm}
	
	\noindent
	With this convention, the circuits we analyze can be described as follows:
	\begin{enumerate}
		\item The circuit~$\cQ_{a,N}$ 
		is equal to
		\begin{align}
			\cQ_{a,N}=\circWithInitStates(e^{-i R P }\gkp_{\kappa_\bosonA,\Delta_\bosonA},\gkp_{\kappa_\bosonB,\Delta_\bosonB},\Psi_{\Delta_\bosonC})\ ,
		\end{align}
		where $\Psi_{\Delta_\bosonC}:=M_{\Delta_\bosonC}\ket{\vac}$ is a squeezed vacuum state.
		\item 
		The idealized circuit we discuss in Section~\ref{sec: main circuit ideal initial states} to develop intuition is 
		\begin{align}
			\circWithInitStates(\e^{-iRP}\Sha_{R}, \gkp, \ket{0})\ ,
		\end{align}
		where 
		$\ket{\Sha_R}\propto \sum_{x=-R}^{R-1}\ket{x}$ is a (formal) superposition of position-eigenstates~$\ket{x}$, where $x\in \{-R,\ldots,R-1\}$,~$\ket{\gkp}\propto \sum_{x\in\mathbb{Z}}\ket{x}$ is an idealized GKP state,  and~$\ket{0}$ is the position-eigenstate with eigenvalue~$x=0$. 
		\item
		The circuit used to obtain Theorem~\ref{thm:main} uses states~$\sigma_{\kappa_\bosonA,\Delta_\bosonA}$ and $\sigma_{\kappa_\bosonB,\Delta_\bosonB}$ produced by our approximate GKP state preparation protocol~\cite{SMbrenneretalGKP2024}  in modes~$\bosonA$ and~$\bosonB$, with a displacement applied to mode $\bosonA$, as well as a squeezed vacuum state~$\Psi_{\Delta_\bosonC}$ in  mode~$\bosonC$, i.e., it amounts to the circuit
		\begin{align}
			\circWithInitStates(e^{-iRP} \sigma_{\kappa_\bosonA,\Delta_\bosonA} e^{iRP},\sigma_{\kappa_\bosonB,\Delta_\bosonB},\Psi_{\Delta_\bosonC})\ .
		\end{align}
	\end{enumerate}
	In the proof of Proposition~\ref{prop:propositionone}, we will also use the following intermediate circuit, which depends on three ``truncation'' parameters~$(\varepsilon_\bosonA,\varepsilon_\bosonB,\varepsilon_\bosonC)$:
	\begin{enumerate}\setcounter{enumi}{3}
		\item 
		Given~$\varepsilon_\bosonA,\varepsilon_\bosonB,\varepsilon_\bosonC\in (0,1/2)$, we consider the circuit
		\begin{align}
			\circWithInitStates(c\cdot \Pi_{[-1/2,2R-1/2]} e^{-i R P }\gkp_{\kappa_\bosonA,\Delta_\bosonA}^{\varepsilon_\bosonA},\gkp_{\kappa_\bosonB,\Delta_\bosonB}^{\varepsilon_\bosonB},\Psi_{\Delta_\bosonC}^{\varepsilon_\bosonC})\ ,
		\end{align}
		where we define $\Pi_{\cC}$ as the orthogonal projection onto the subspace of $L^2(\mathbb{R})$ of functions supported on the set $\cC \subseteq\mathbb{R}$, i.e., $\Pi_{\cC}$ acts on $L^2(\mathbb{R})$ as pointwise multiplication by the characteristic function of the set $\cC$. 
		The factor $c$ is a normalization constant.
		We define the union of intervals of length~$2\varepsilon$ centered around integers as $\mathbb{Z}(\varepsilon)=\mathbb{Z}+[-\varepsilon,\varepsilon]=\{x+y\mid x\in\mathbb{Z},\ y\in\mathbb{R} ,\ |y|\le\varepsilon\}$.
		The state $\gkp^\varepsilon_{\kappa,\Delta} =
		\Pi_{\mathbb{Z}(\varepsilon)}\gkp_{\kappa,\Delta}/\|\Pi_{\mathbb{Z}(\varepsilon)}\gkp_{\kappa,\Delta}\|$
		is then obtained by truncating the support 
		of the approximate GKP state~$\gkp_{\kappa,\Delta}$ to $\mathbb{Z}(\varepsilon)$. 
		Similarly, $\Psi^{\varepsilon}_{\Delta}=\Pi_{[-\varepsilon,\varepsilon]}\Psi_\Delta/\|\Pi_{[-\varepsilon,\varepsilon]}\Psi_\Delta\|$ is the result of truncating the support of the Gaussian state~$\Psi_\Delta$  to the interval~$[-\varepsilon,\varepsilon]$.
	\end{enumerate}
	
	\clearpage
	\tableofcontents
	
	\clearpage
	\section{Action of the  circuit~$\semiideal$}\label{sec: circuit explanation}
	In this section, we give a high-level  overview of the 
	functionality of our  circuit~$\ourcircuit_{a,N}$ (Fig.~\ref{fig:maincircuit} in the main text), as well as its  various building blocks. To do so, we consider idealized initial states which are defined in terms of the position-``eigenstates''~$\ket{x}$ of the position (quadrature) operator~$Q$. Here~$\ket{x}$ denotes the eigenstate of~$Q$ with eigenvalue~$x\in\mathbb{R}$. 
	Strictly speaking, these are distributions, and so are most of the expressions we deal with in this section. This discussion is for illustrative purposes only. We shall give a rigorous version of these arguments dealing with normalizable functions (i.e., elements of~$L^2(\mathbb{R})$) in Section~\ref{sec: approximate analysis}.

	Recall that we obtain the circuit $\semiideal$ from the circuit~$\ourcircuit_{a,N}$ by replacing the initial state of 
	\begin{enumerate}
		\item  mode~$\bosonA$  by the formal uniform superposition
		$e^{-iRP} \ket{\Sha_R}= e^{-iRP} \sum_{x=-R}^{R-1} \ket{x}=\sum_{x=0}^{2R-1}\ket{x}$ of position-eigenstates~$\{\ket{x}\}_{x=0}^{2R-1}$, 
		\item mode~$\bosonB$ by the idealized GKP state
		$\ket{\gkp}\propto \sum_{y\in\mathbb{Z}}\ket{y}$,
		defined as the (formal) uniform superposition of position-eigenstates with integer positions~$y\in\mathbb{Z}$,
		\item  mode~$\bosonC$ by the position-eigenstate~$\ket{0}$ localized at $x=0$.
	\end{enumerate}
	The corresponding circuit~$\semiideal$ is shown in Fig.~\ref{fig:idealizedcircuitforanalysis} with intermediate states $\Psi_0,\Psi_1,\Psi_2$.
	
	\begin{figure}[h!]
		\begin{center}
			\includegraphics{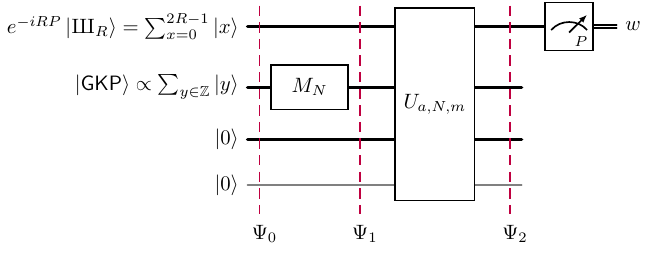}
		\end{center}
		\caption{The circuit
			$\semiideal$ which uses idealized initial states. We indicate intermediate states (distributions)~$\Psi_0,\Psi_1,\Psi_2$ that are used in our analysis in Section~\ref{sec: main circuit ideal initial states}.
			\label{fig:idealizedcircuitforanalysis}}
	\end{figure}
	
	\noindent
	To compute the state~$\Psi_2$ produced in this circuit (i.e., the state before measurement, see Fig.~\ref{fig:idealizedcircuitforanalysis}), we 
	analyze the unitary $\Uan$. This  entails an analysis of the elementary gates in Section~\ref{sec: building blocks individual} as well as an analysis of the unitaries composing the unitary $\Uan$ (in Sections~\ref{sec: ideal V gate} and~\ref{sec: VaNm idealized}). This will allow us to verify that~$\semiideal$ implements what can be seen as a CV version of the quantum subroutine of Shor's algorithm, see Section~\ref{sec: main circuit ideal initial states}.
	
	\subsection{Action of composite unitaries}\label{sec: building blocks individual}
	
	In this section, we discuss the action of the composite unitaries introduced in Fig.~\ref{fig:composite gates} of the main paper.  Specifically, we study their effect on position-eigenstates and explain how the corresponding action arises from the elementary operations introduced in Fig.~\ref{fig:elementary gates} of the main paper. We also analyze the corresponding circuit size (i.e., the number of elementary gates used) and depth (i.e., the number of gate layers) in each case.

	\subsubsection{Scalar multiplication}\label{sec: ideal scalar multi}
	Let $\alpha>0$ be arbitrary. The multiplication operator $M_\alpha$ is defined in terms of the squeezing operator~$S(z)=e^{iz(QP+PQ)/2}$ as
	\begin{align}
		M_\alpha = S( -\log \alpha) \ .
	\end{align}
	The action of~$M_\alpha$
	on an element~$\Psi\in L^2(\mathbb{R})$ is given 
	by
	\begin{align}
		(M_\alpha\Psi)(x)&=\frac{1}{\sqrt{\alpha}}\Psi(x/\alpha)\qquad\textrm{ for }\qquad x\in\mathbb{R}\ .\label{eq:malphaactionexplained}
	\end{align}
	This implies that for~$x\in\mathbb{R}$, the position-eigenstate~$\ket{x}$ is transformed as
	\begin{align}
		M_\alpha\ket{x}&\propto \ket{\alpha \cdot x}\qquad\textrm{ for any }x\in\mathbb{R}\ .\label{eq:positionbasisactionmalapha}
	\end{align}
	(Here and in the following we omit normalization constants since position-eigenstates are not normalizable in the first place. We note that these normalization constants are fixed by the action of the corresponding function being evaluated, i.e., scalar multiplication by~$\alpha$ in this case, see Section~\ref{sec: approx function evaluation}.) 
	In other words, the operator~$M_\alpha$ acts as a scalar multiplication by~$\alpha$ in position-space. We represent this as 
	\begin{align}
		\vcenter{\hbox{\includegraphics{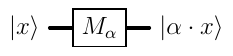}}}
		\qquad\textrm{ for }\qquad x\in \mathbb{R}\ .
	\end{align}
	
	We note that this action can equivalently be understood by the fact that squeezing is a Gaussian unitary, and the action on the mode operators $(Q,P)$ in the Heisenberg picture is given by 
	\begin{align}
		\begin{matrix}
			M_\alpha^\dagger Q M_\alpha&=&\alpha Q\\
			M_\alpha^\dagger P M_\alpha&=&\frac{1}{\alpha} P\ .
		\end{matrix}\label{eq:heisenbergactionmalpha}
	\end{align}
	Eqs.~\eqref{eq:malphaactionexplained},~\eqref{eq:positionbasisactionmalapha} and~\eqref{eq:heisenbergactionmalpha} are all essentially equivalent expressions defining the action of~$M_\alpha$. We will use these different formalisms interchangeably, but will mostly rely on terms of the form~\eqref{eq:positionbasisactionmalapha}.

	We can find the following decomposition of $M_\alpha$ in terms of elementary operations. Let 
	\begin{align}
		\ell:= \ceil{\abs{\log \alpha}}\ .
	\end{align}
	We decompose unitary $M_\alpha$ as
	\begin{align}
		M_\alpha&=S(-(\log\alpha)/\ell)^\ell\ ,\label{eq:malphadecomposition}
	\end{align}
	where the power $\ell$ indicates repeated application of~$\ell$ single-mode squeezing operations $S(-(\log \alpha)/\ell)$  of a constant strength, i.e., the absolute value of the squeezing parameter is bounded by $\abs{(\log\alpha)/\ell} \le 1$. Therefore the unitary~$M_\alpha$ can be implemented by a circuit of size and depth~$\ell=O(\abs{\log \alpha})$ for $\alpha\to\infty$ (respectively $\alpha\to 0$).

	\subsubsection{Translation by $R$}\label{sec: translation by R}
	Let $R\geq 1$. Recall that the unitary~$e^{-iRP}$ is Gaussian and acts as a translation in position-space, i.e.,
	\begin{align}
		\begin{matrix}
			e^{iRP}Qe^{-iRP}&=&Q+R\cdot I\\
			e^{iRP}Pe^{-iRP}&=&P\ 
		\end{matrix}\label{eq:expresstranslation}\ , 
	\end{align}
	where $I$ denotes the identity on~$L^2(\mathbb{R})$. We express this as
	\begin{align}
		\vcenter{\hbox{\includegraphics{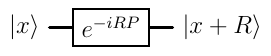}}}
		\qquad \textrm{ for \qquad $x\in \mathbb{R}$} \ .
	\end{align}
	Combining~\eqref{eq:heisenbergactionmalpha} and 
	~\eqref{eq:expresstranslation}, it is straightforward to verify that 
	\begin{align}
		e^{-iRP} =  M_{R} \cdot e^{-iP}\cdot M_{1/R}\ . \label{eq:eirpdecompos}
	\end{align}
	By decomposing both~$M_R$ and $M_{1/R}$
	as in Eq.~\eqref{eq:malphadecomposition}, 
	Eq.~\eqref{eq:eirpdecompos} implies that the unitary~$e^{-iRP}$ can be realized by a circuit of size and depth~$O(\log R)$.

	\subsubsection{Qubit-controlled scalar multiplication}\label{sec: ideal controlled scalar mult} 
	The qubit-controlled version of 
	the unitary~$M_\alpha$ is denoted $\mathsf{ctrl}M_\alpha$. It  multiplies the position $x\in\mathbb{R}$ of a bosonic mode by a scalar $\alpha>0$ if the control qubit is in the state~$\ket{1}$.  
	This operation is introduced in  Fig.~\ref{fig:composite gates} of the main paper as
	\begin{align}
		\vcenter{\hbox{\includegraphics[scale=1]{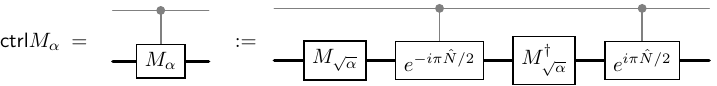}}} \ \ \ ,
	\end{align}
	where $\hat{N}=(Q^2+P^2-I)/2$.
	
	We can show that the unitary $\mathsf{ctrl}M_\alpha$ has the following action on a product state~$\ket{b}\otimes\ket{x}$, where~$\ket{b}\in \{\ket{0},\ket{1}\}$ is a computational basis state of a qubit, and $\ket{x}$ for $x\in\mathbb{R}$ denotes the position-eigenstate: We have
	\begin{align}
		\mathsf{ctrl}M_\alpha \left(\ket{b}\otimes\ket{x}\right)
		&\propto
		\begin{cases} 
			\ket{0}\otimes\ket{x}& \qquad \textrm{ if } b=0\\
			\ket{1}\otimes\ket{\alpha\cdot x} & \qquad \textrm{ if } b=1\ .
		\end{cases}    
	\end{align}
	We denote this as
	\begin{align}
		\vcenter{\hbox{\includegraphics[scale=1]{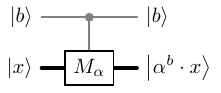}}}
		\qquad \textrm{ for}\qquad b\in\{0,1\} \ \textrm{ and }\  x\in\mathbb{R}\ . 
	\end{align}
	This claim is trivial for~$b=0$. For $b=1$, it follows  from the identity
	\begin{align}
		e^{-i\pi \hat{N}/2}M_{\sqrt{\alpha}}^\dagger e^{i\pi \hat{N}/2}
		&=M_{\sqrt{\alpha}}
	\end{align}
	because the Gaussian unitary~$e^{i\pi \hat{N}/2}$ acts as 
	\begin{align}
		\begin{matrix}
			e^{i\pi \hat{N}/2}Qe^{-i\pi\hat{N}/2}&=&-P\\
			e^{i\pi \hat{N}/2}Pe^{-i\pi\hat{N}/2}&=&Q\ .
		\end{matrix}
	\end{align}

	By definition, the circuit realizing $\mathsf{ctrl}M_\alpha$ is composed of two controlled $e^{\pm i \pi \hat{N}/2}$ gates and two unitaries $M_{\sqrt{\alpha}}$ and $M_{\sqrt{\alpha}}^\dagger=M_{1/\sqrt{\alpha}}$ (see Section~\ref{sec: ideal scalar multi}).
	Therefore, both the circuit size as well as the circuit depth of the gate $\mathsf{ctrl}M_\alpha$ are dominated by that of $M_{\sqrt{\alpha}}$ and scale as $O(|\log \alpha|)$ for $\alpha\to\infty$ (respectively $\alpha\to 0$).

	\subsubsection{Coherent extraction of the least significant bit (LSB)}\label{sec: ideal LSB}
	Let $x\in \mathbb{Z}$ be an integer. We define the least significant bit (LSB) of $x$ as $x_0=x\xmod 2$.

	The following circuit (see Fig.~\ref{fig:composite gates} of the main paper) defines a unitary~$U^{\lsb}_{\bosonA \rightarrow \qubit}$ acting on mode $\bosonA$ and qubit $\qubit$:
	\begin{align}
		U^{\lsb}_{\bosonA \rightarrow \qubit}~=~%
		\raisebox{-42pt}{\mbox{\includegraphics{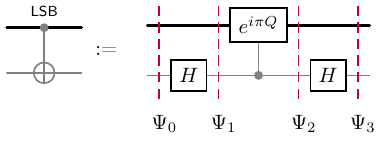}}}\ .\label{circ: LSB}
	\end{align}
	We show that it extracts the LSB of the position of mode $\bosonA$ into the qubit $\qubit$. That is, for any integer~$x\in\mathbb{Z}$ encoded into the position-eigenstate $\ket{x}$ of a mode $\bosonA$,
	this unitary exactly computes the LSB $x_0$ of~$x$ into the qubit $\qubit$, i.e., we have 
	\begin{align}
		U^{\lsb}_{\bosonA\rightarrow \qubit}\left(\ket{x}_\bosonA\otimes\ket{b}_Q\right)=\ket{x}_\bosonA\otimes\ket{x_0\oplus b}_\qubit\qquad\textrm{ for all }x\in \mathbb{Z} \textrm{ and }b\in \{0,1\}\ ,\label{eq:lsbgateactiondef}
	\end{align}
	where $\oplus$ denotes addition modulo~$2$.
	\begin{proof}
		We consider the intermediate states $\Psi_0$, $\Psi_1$ and $\Psi_2$, see \eqref{circ: LSB}.
		Starting from $\ket{\Psi_0}\propto \ket{x}\otimes\ket{b}$, we derive
		\begin{align}
			\ket{\Psi_1}&\propto (I\otimes H)\ket{\Psi_0}\\
			&\propto \ket{x} \otimes \big(\ket{0} + (-1)^{b} \ket{1} \big)\ ;
		\end{align}
		hence
		\begin{align}
			\ket{\Psi_2}&\propto \big(I\otimes I+e^{i \pi Q} \otimes \proj{1}\big) \ket{\Psi_1}\\
			&\propto \ket{x}\otimes\ket{0} + (-1)^{b} e^{i \pi x} \ket{x}\otimes\ket{1} \\
			&\propto \ket{x}\otimes\ket{0} + (-1)^{b} e^{i \pi x_0} \ket{x}\otimes\ket{1}\label{eq: x to x0}\\
			&\propto \ket{x}\otimes\ket{0} + (-1)^{x_0+b} \ket{x}\otimes\ket{1} \ .
		\end{align}
		We note that Eq.~\eqref{eq: x to x0} relies on the identity~$e^{i\pi x}=(-1)^x=(-1)^{x_0}$ for any $x\in\mathbb{Z}$, i.e., the fact that the LSB $x_0$ determines whether or not $x$ is even or odd.
		
		It follows that the output state is 
		\begin{align}
			\ket{\Psi_3}&\propto(I\otimes H)\ket{\Psi_2}\\
			&\propto\ket{x}\otimes\ket{x_0\oplus b}\ ,
		\end{align}
		which is the claim.
	\end{proof}
	We will represent this as 
	\begin{align}
		\vcenter{\hbox{\includegraphics{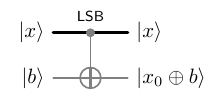}}}
		\qquad\textrm{ for }\qquad x\in\mathbb{Z} \textrm{ and }b\in \{0,1\}\ . 
	\end{align}
	
	The unitary $U_{\bosonA\to\qubit}^\lsb$ is implemented by a circuit of a constant size and constant depth.

	\subsection{The unitary $V_\alpha$}\label{sec: ideal V gate}
	This building block $V_\alpha$ is used repeatedly (with $\alpha>0$) in the definition of the unitary~$\Vam$ (see 
	Fig.~\ref{fig:composite gates} in the main paper). It is defined as follows, where we labeled the respective bosonic modes by $\bosonA,\bosonB,\bosonC$ and the qubit by $\qubit$. 
	
	\begin{align}
		\vcenter{\hbox{\includegraphics[scale=1]{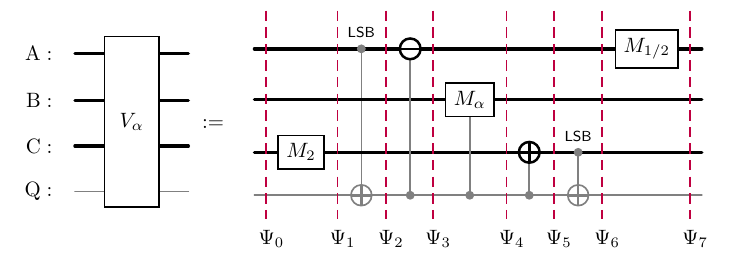}}}
		\label{circ: V alpha}
	\end{align}

	We show that this unitary acts on a tensor product of position-eigenstates $\ket{x},\ket{y},\ket{z}$ and of a qubit in the state~$\ket{0}$ as
	\begin{align}
		V_\alpha\big(\ket{x}_\bosonA\otimes\ket{y}_\bosonB\otimes\ket{z}_\bosonC\otimes\ket{0}_\qubit\big)    
		&~\propto~ \ket{(x-x_0)/2}_\bosonA\otimes\ket{\alpha^{x_0}\cdot y}_\bosonB\otimes\ket{2z+x_0}_\bosonC\otimes\ket{0}_\qubit \label{eq: V alpha action}
	\end{align}
	for all $x\in\mathbb{Z}$, $y\in\mathbb{R}$ and $z\in\mathbb{Z}$.

	\begin{proof}
		The unitary~$V_\alpha$ is the result of applying a sequence of $7$~unitaries, see the illustration~\eqref{circ: V alpha} for the definition of intermediate states~$\Psi_0,\ldots,\Psi_{7}$. 
		
		Starting with the initial state 
		\begin{align}
			\ket{\Psi_0}&\propto \ket{x}\otimes\ket{y}\otimes\ket{z}\otimes\ket{0}\ ,
		\end{align}
		we first squeeze mode~$\bosonC$, obtaining
		\begin{align}
			\ket{\Psi_1}&\propto (I\otimes I \otimes M_2 \otimes I) \Psi_0\\
			&\propto \ket{x}\otimes\ket{y}\otimes\ket{2z}\otimes\ket{0}\ .
		\end{align}
		We then extract the LSB $x_0$ of the position $x$ of mode~$\bosonA$ to the qubit, resulting in
		\begin{align}
			\ket{\Psi_2}&\propto U^{\lsb}_{\bosonA \rightarrow \qubit} \ket{\Psi_1}\\
			&\propto \ket{x}\otimes\ket{y}\otimes\ket{2z}\otimes\ket{x_0}\ .
		\end{align}
		Having extracted the LSB $x_0$ to the qubit, we can now subtract $x_0$ from the position-eigenstate of mode~$\bosonA$ by applying a qubit-controlled translation, resulting in
		\begin{align}
			\ket{\Psi_3}&\propto \big(I \otimes I\otimes I \otimes \proj{0}+e^{i P}\otimes I\otimes I \otimes \proj{1} \big)  \ket{\Psi_2}\\
			&\propto \ket{x-x_0}\otimes\ket{y}\otimes\ket{2z}\otimes\ket{x_0}\ .
		\end{align}
		After this step, the position $x-x_0$ of mode~$\bosonA$ is even. We then use $x_0$ (stored in the qubit~$\qubit$) as a control to multiply the position-eigenstate of mode~$\bosonB$ by $\alpha$, giving
		\begin{align} 
			\ket{\Psi_4}&\propto \big(I \otimes I\otimes I \otimes \proj{0}+I \otimes M_\alpha\otimes I \otimes \proj{1} \big) \ket{\Psi_3}\\
			&\propto \ket{x-x_0}\otimes\ket{\alpha^{x_0}\cdot y}\otimes\ket{2z}\otimes\ket{x_0}\ ,
		\end{align}
		and add $x_0$ to the position of mode~$\bosonC$ by a qubit-controlled translation. This results in
		\begin{align}
			\ket{\Psi_5}&\propto \big(I \otimes I\otimes I \otimes \proj{0}+I \otimes I \otimes e^{-iP} \otimes \proj{1} \big)  \ket{\Psi_4}\\
			&\propto \ket{x-x_0}\otimes\ket{\alpha^{x_0}\cdot y}\otimes\ket{2z+x_0}\otimes\ket{x_0}\ .
		\end{align}
		We then restore the qubit~$\qubit$ to the computational basis state $\ket{0}$ by application of the unitary $U^{\lsb}_{\bosonC\to \qubit}$, that is,
		\begin{align}
			\ket{\Psi_6}&\propto U^{\lsb}_{\bosonC \rightarrow \qubit} \ket{\Psi_5}\\
			&\propto \ket{x-x_0}\otimes\ket{\alpha^{x_0}\cdot y}\otimes\ket{2z+x_0}\otimes\ket{0}\ .
		\end{align}
		Here we used that $2z+x_0$ has LSB $x_0$.
		In the final step, we apply squeezing to mode~$\bosonA$, obtaining
		\begin{align}
			\ket{\Psi_7}&\propto  \big(M_{1/2} \otimes I\otimes I \otimes I\big)\ket{\Psi_6}\\
			&\propto \ket{(x-x_0)/2}\otimes\ket{\alpha^{x_0}\cdot y}\otimes\ket{2z+x_0}\otimes\ket{0}\ .
		\end{align}
		This proves Eq.~\eqref{eq: V alpha action}.
	\end{proof}

	To summarize, the unitary $V_\alpha$ acts as
	\begin{align}
		\vcenter{\hbox{\includegraphics{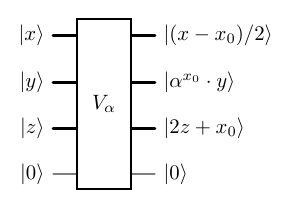}}} \qquad \textrm{ for \qquad $x\in\mathbb{Z}$, $y\in\mathbb{R}$ and $z\in\mathbb{Z}$ \ .}
	\end{align}

	Furthermore, both the size as well as the depth of  the circuit implementing the unitary~$V_\alpha$ scale as $O(|\log\alpha|)$ for $\alpha \rightarrow 0$ respectively $\alpha \rightarrow \infty$. Indeed, these quantities are dominated by the circuit size and depth of the unitary~$\mathsf{ctrl}M_\alpha$, whereas the other gates only contribute a constant.

	\subsection{Controlled multiplication by a pseudomodular power}\label{sec: VaNm idealized}
	
	At the core of the circuit $\semiideal$ is the composed unitary $\Vam$. It is defined by iterative application of the unitaries~$V_\alpha$ with changing parameter $\alpha$ as follows (cf. Fig.~\ref{fig:composite gates} of the main paper). 
	\begin{align}
		\raisebox{-68pt}{\mbox{\includegraphics{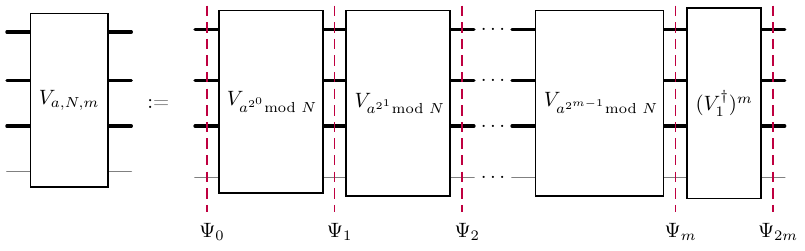}}}
		\label{eq: VaNm}
	\end{align}
	Since each $V_\alpha$ is given by a sequence of elementary gates (see Eq.~\eqref{circ: V alpha}) and this set is closed under taking inverses, the unitary $V_{a,N,m}$ can be decomposed into elementary gates.
	
	We will argue that
	the unitary~$\Vam$ multiplies the position $y$ of the mode~$\bosonB$ by the value~$f_{a,N,m}(x)$ where $x\in\mathbb{N}_0$ is the position of the mode~$\bosonA$. Here
	\begin{align}
		f_{a,N,m}(x)&=\prod_{i=0}^{m-1} \left(a^{2^i}\xmod N \right)^{x_i}\ ,\label{eq:pseudomodularpowerfirstdef}
	\end{align}
	where we used the $m$ least significant bits in the binary representation of $x\in\mathbb{N}_0$
	\begin{align}
		x&=\sum_{j=0}^{\infty} x_j 2^j\ .
	\end{align}
	We call the function~$f_{a,N,m}(x)$ 
	the pseudomodular power. 
	In particular, we have $a^x=a^{\sum_{i=0}^{m-1}x_i2^{i}}=\prod_{i=0}^{m-1}\left(a^{2^i}\right)^{x_i}$ for any integer $x\in\{0,\ldots, 2^m-1\}$ and thus 
	thus
	\begin{align}
		f_{a,N,m}(x)&\equiv a^x \pmod{N} \qquad \textrm{ for all $x\in\{0,\ldots, 2^m-1\}$}\ . \label{eq: pseudomodular power mod N}
	\end{align}
	We show the following action of~$\Vam$ for all $x\in\mathbb{N}_0$, $y\in\mathbb{R}$ and $z\in\mathbb{N}_0$. We have
	\begin{align}
		\Vam \big(\ket{x}\otimes\ket{y}\otimes\ket{z}\otimes\ket{0}\big) ~\propto~
		\ket{x}\otimes\ket{f_{a,N,m}(x)\cdot y}\otimes\ket{z}\otimes\ket{0} \ .
		\label{eq: action ctrl M pseudo}
	\end{align}

	\begin{proof}
		Let $x\in\mathbb{N}_0$, $y\in\mathbb{R}$ and $z\in\mathbb{N}_0$. Consider the initial state
		\begin{align}
			\ket{\Psi_0}&\propto\ket{x}\otimes\ket{y}\otimes\ket{z}\otimes\ket{0}\ ,
		\end{align}
		see diagram~\eqref{eq: VaNm}.

		Each element $(V_{a^{2^i}\xmod N})_{i\in\{0,\ldots,m-1\}}$ of the first $m$ gates defining~$\Vam$ corresponds to a bit-controlled multiplication by a factor~$(a^{2^i}\xmod)^{x_i}$ in the definition~\eqref{eq:pseudomodularpowerfirstdef}  of the pseudomodular power.

		In more detail, the application of~$V_{a^{2^0}\xmod N}$ to~$\Psi_0$ multiplies the position of the mode~$\bosonB$ by the factor \mbox{$(a^{2^0}~\xmod N)^{x_0}$} and moves $x_0$ (originally the LSB of the the position $x$ in mode~$\bosonA$) to the LSB of the position in mode~$\bosonC$, that is,
		\begin{align}
			\ket{\Psi_1}
			&\propto V_{a^{2^0}\xmod N}\ket{\Psi_0}\\
			&\propto \ket{(x-x_0)/2}\otimes\ket{ (a^{2^0}\xmod N)^{x_0}\cdot y}\otimes\ket{
				2z+x_0}\otimes \ket{0} \ .
		\end{align}
		Subsequently,~$V_{a^{2^1}\xmod N}$ is applied to~$\Psi_1$. This multiplies the position of mode~$\bosonB$ by the second factor $(a^{2^1}\xmod N)^{x_1}$ of the pseudomodular power and moves $x_1$ (previously the LSB of the position $(x-x_0)/2$ of the mode~$\bosonA$) to the LSB of the position in the mode~$\bosonC$, i.e., 
		\begin{align}
			\ket{\Psi_2}&\propto V_{a^{2^1}\xmod N}\ket{\Psi_1}
			\\
			&\propto \ket{(x-2^1x_1-2^0x_0)/2^2}\otimes\ket{(a^{2^1}\xmod N)^{x_1}\cdot (a^{2^0}\xmod N)^{x_0} \cdot y}\otimes\ket{2^2z+2^1x_0+2^0x_1}\otimes\ket{b}\ .\notag
		\end{align}
		Note that $x_0$ is now in the second-least-significant bit of mode~$\bosonC$).
		Inductively, we obtain for all $x\in\mathbb{N}_0$, $y\in\mathbb{R}$ and $z\in \mathbb{N}_0$ that after applying the $j$th unitary $V_{a^{2^{j-1}}\xmod N}$ in the sequence, we have the state
		\begin{align}
			\ket{\Psi_{j}}
			&\propto \ket{x^{(j)}}\otimes\ket{y^{(j)}}\otimes\ket{z^{(j)}}\otimes\ket{0}\ ,
		\end{align}
		where
		\begin{align}
			x^{(j)}&=2^{-j} x-\sum_{i=0}^{j-1} 2^{-j+i} x_i\\
			y^{(j)}&=\left(\prod_{i=0}^{j-1} \left(a^{2^{i}} \xmod N\right)^{x_i} \right)\cdot y\label{eq: xj, yj, zj}\\ 
			z^{(j)}&=2^j z + 
			\sum_{i=0}^{j-1} 2^{j-1-i} x_i \ .
		\end{align}
		Thus applying the first $m$ gates, the resulting state is
		\begin{align}
			\ket{\Psi_{m}}
			&\propto \ket{x^{(m)}}\otimes\ket{y^{(m)}}\otimes\ket{z^{(m)}}\otimes\ket{0}\ \\
			&= \ket{x^{(m)}}\otimes\ket{f_{a,N,m}(x)\cdot y}\otimes\ket{z^{(m)}}\otimes\ket{0}\ ,
		\end{align}
		where we used the definition~\eqref{eq:pseudomodularpowerfirstdef} of the pseudomodular power.
		
		At this stage, the multiplication of the position $y$ of the mode~$\bosonB$ by the pseudomodular power~$f_{a,N,m}(x)$ of the position $x$ of the mode~$\bosonA$ is completed. However, $x$ (or more precisely the $m$ lowest bits of $x$) has been moved from mode~$\bosonA$ to mode~$\bosonC$ (in the reversed binary representation).
		The remaining steps (the $m$-fold application of $V^\dagger_1$) undo that latter, undesired operation.

		Since $(x,y,z) \mapsto (2x + z_0, y, (z-z_0)/2)$ is a bijection on $\mathbb{N}_0\times\mathbb{R}\times\mathbb{N}_0$ which is computed by the unitary $V_1$, it is easy to verify that
		\begin{align}
			V^\dagger_1\left(\ket{x}\otimes\ket{y}\otimes\ket{z}\otimes\ket{0}\right) = \ket{2x+z_0}\otimes\ket{y}\otimes\ket{(z-z_0)/2}\otimes\ket{0} 
		\end{align}
		for all $x\in\mathbb{N}_0$, $y\in\mathbb{R}$ and $z\in\mathbb{N}_0$. Applying the unitary $V^\dagger_1$ to $\Psi^{(m)}$ gives
		\begin{align}
			V^\dagger_1 \Psi_{m} &\propto \ket{2x^{(m)}+z^{(m)}_{0}}\otimes \kett{f_{a,N,m}(x)\cdot y} \otimes \ket{(z^{(m)}-z^{(m)}_{0})/2}\otimes \ket{0} \ .\label{eq: V1dagg single application}
		\end{align}
		By definition of $z^{(m)}$ we have that $z^{(m)}_{0}=x_{m-1}$, thus
		\begin{align}
			2x^{(m)}+z^{(m)}_{0} 
			&= 2\left(2^{-m} x-\sum_{i=0}^{m-1} 2^{-m+i} x_i\right)+x_{m-1} \\
			&=2^{-(m-1)} x-\sum_{i=0}^{(m-1)-1} 2^{-(m-1)+i} x_i\\
			&=x^{(m-1)}\ , \label{eq: V dagger xm}
		\end{align}
		where we used definition of $x^{(j)}$ from~\eqref{eq: xj, yj, zj}. Similarly, we have
		\begin{align}
			(z^{(m)}-z^{(m)}_{0})/2 
			&= \left(\left(2^m z + 
			\sum_{i=0}^{m-1} 2^{m-1-i} x_i \right)-x_{m-1}\right)/2\\
			&=2^{m-1} z +
			\sum_{i=0}^{(m-1)-1} 2^{(m-1)-1-i} x_i\\
			&=z^{(m-1)}\ . \label{eq: V dagger zm}
		\end{align}
		Inserting these into~\eqref{eq: V1dagg single application} we obtain
		\begin{align}
			V^\dagger_1 \ket{\Psi_{m}} &\propto \ket{x^{(m-1)}}\otimes \ket{f_{a,N,m}(x)\cdot y} \otimes \ket{z^{(m-1)}}\otimes \ket{0} \ .
		\end{align}
		By induction, it follows that
		\begin{align}
			(V^\dagger_1)^m\ket{\Psi_{m}}
			&\propto \ket{x}\otimes\ket{f_{a,N,m}(x)\cdot y}\otimes\ket{z}\otimes\ket{0}
		\end{align}
		for all
		$x\in\mathbb{N}_0$, $y\in\mathbb{R}$ and $z\in\mathbb{N}_0$.
		This completes the proof.
	\end{proof}

	We summarize this action as
	\begin{align}
		\vcenter{\hbox{\includegraphics{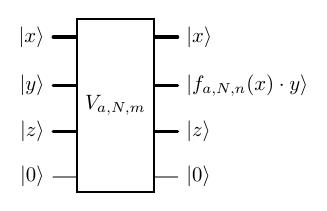}}}
		\qquad \text{ for } \qquad \textrm{  $x\in\mathbb{N}_0$, $y\in\mathbb{R}$  and  $z\in\mathbb{N}_0$}\ .
	\end{align}
	
	By definition, $\Vam$ is composed of $m$ gates $V_\alpha$ with $0 < \alpha < N$ and $m$ gates $V_1^\dagger$. Therefore, both the size as well as the depth of the unitary $\Vam$ scale as $O(m \log N)$. 
	
	\subsection{Controlled translation by a pseudomodular power}\label{sec: explanation Ua}
	The unitary~$\Ua= \Vam^\dagger \cdot (I \otimes e^{-i P} \otimes I\otimes I )\cdot \Vam$ translates the position $y$  of the mode~$\bosonB$ by $f_{a,N,m}(x)$, where $x$ is the position of the mode~$\bosonA$. Diagrammatically it is defined as 
	\begin{align}
		\raisebox{-65pt}{\mbox{\includegraphics{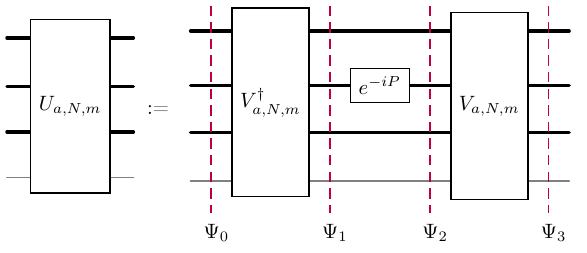}}}\ \ . \label{eq: UaNm gates}
	\end{align}
	We show that the unitary~$\Ua$ acts as
	\begin{align}
		\Ua \left(\ket{x}\otimes\ket{y}\otimes\ket{z}\otimes\ket{0}\right) & \propto \ket{x}\otimes\ket{y+f_{a,N,m}(x)}\otimes\ket{z}\otimes\ket{0} \label{eq: Ua integers action}
	\end{align}
	for all $x\in\mathbb{N}_0$, $y\in\mathbb{R}$ and $z\in\mathbb{N}_0$.
	
	\begin{proof}
		We first note that~\eqref{eq: action ctrl M pseudo} implies
		\begin{align}
			\Vam^\dagger\left(\ket{x}\otimes\ket{y}\otimes\ket{z}\otimes\ket{0}\right) = \ket{x}\otimes\ket{y/f_{a,N,m}(x)}\otimes\ket{z}\otimes\ket{0}
		\end{align}
		for all $x\in\mathbb{N}_0$, $y\in\mathbb{R}$, and $z\in\mathbb{N}_0$.
		
		Thus starting from $\ket{\Psi_0}\propto \ket{x}\otimes\ket{y}\otimes\ket{z}\otimes\ket{0}$,
		we obtain (see Eq.~\eqref{eq: UaNm gates})
		\begin{align}
			\ket{\Psi_1}&\propto (\Vam)^\dagger \ket{\Psi_0}\\
			&\propto \ket{x}\otimes\ket{y/ f_{a,N,m}(x)}\otimes\ket{z}\otimes\ket{0}\ ,\\
			\ket{\Psi_2}&\propto(I \otimes e^{-i P} \otimes I\otimes I)\ket{\Psi_1}\\
			&\propto\ket{x}\otimes\ket{1+y/ f_{a,N,m}(x)}\otimes\ket{z}\otimes\ket{0}\ ,\\
			\ket{\Psi_3}&\propto (\Vam) \ket{\Psi_2}\\
			&\propto\ket{x}\otimes\ket{y+f_{a,N,m}(x)}\otimes\ket{z}\otimes\ket{0}\ ,\label{eq: Ua result}
		\end{align}
		which is the claim.
	\end{proof}
	
	We summarize this as
	\begin{align}
		\vcenter{\hbox{\includegraphics{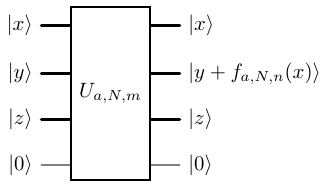}}}
		\qquad \text{ for \qquad $x\in\mathbb{N}_0$, $y\in\mathbb{R}$ and $z\in\mathbb{N}_0$\ .}
	\end{align}
	
	Since both  $\Vam^\dagger$ and $\Vam$ can be realized by a circuit of size $O(m\log N)$ and $e^{-iP}$ belongs to our gate set, both the circuit size as well as the circuit depth of the unitary $\Ua$ scale as $O(m\log N)$.

	\subsection{Analysis of
		the circuit~$\semiideal$}\label{sec: main circuit ideal initial states}

	We now explain the intuition behind the circuit~$\ourcircuit_{a,N}$ (see Fig.~\ref{fig:maincircuit} of the main text) by analyzing its idealized version $\semiideal$ (see Fig.~\ref{fig:idealizedcircuitforanalysis}).
	We show that the circuit $\semiideal$ implements
	a form of (pseudo)modular exponentiation and a modular measurement.
	
	Note that a rigourous analysis of circuit~$\ourcircuit_{a,N}$ is provided in Section~\ref{sec: approximate analysis}.  The following analysis of the circuit~$\semiideal$ is for illustrative purposes only.

	\textbf{Pseudomodular exponentiation.}
	We analyze the action of the circuit $\semiideal$. Recall that the initial state is (see Fig.~\ref{fig:idealizedcircuitforanalysis})
	\begin{align}
		\ket{\Psi_0}
		&\propto e^{-iRP}\ket{\Sha_R}\otimes\ket{\gkp}\otimes\ket{0}\otimes\ket{0}\\
		&= \left(\sum_{x= 0}^{2R-1} \ket{x}\right) \otimes \left(\sum_{y\in\mathbb{Z}}\ket{y}\right)\otimes\ket{0}\otimes\ket{0}\ ,
	\end{align}
	where $R=2^{17n-1}=2^{m-1}$ (see Table~\ref{tab: parameters}). Thus the state $\Psi_1$ is
	\begin{align}
		\ket{\Psi_1}&\propto  \left(I\otimes M_N\otimes I \otimes I\right)\, \ket{\Psi_0}\\
		&\propto  \left(I\otimes M_N\otimes I \otimes I\right)  \left(\sum_{x=0}^{2R-1}\ket{x}\right)\otimes\left(\sum_{y\in\mathbb{Z}}\ket{y}\right)\otimes\ket{0}\otimes\ket{0}\\
		&\propto \left(\sum_{x=0}^{2R-1}\ket{x}\right)\otimes\left(\sum_{y\in\mathbb{Z}}\ket{N\cdot y}\right)\otimes\ket{0}\otimes\ket{0}
	\end{align}
	Subsequently, we apply the unitary $\Ua$. Following Eq.~\eqref{eq: Ua integers action} we obtain
	\begin{align}
		\ket{\Psi_2}&\propto \Ua\,\ket{\Psi_1}\\
		&\propto\sum_{x=0}^{2R-1}\sum_{y\in\mathbb{Z}}\ket{x}\otimes\ket{N\cdot y+ f_{\alpha,N,m}(x)}\otimes\ket{0}\otimes\ket{0}\ .
	\end{align}
	
	\textbf{Modular measurement.}
	Let us trace out both systems $\bosonC$ and $\qubit$. At this point, 
	measuring the mode~$\bosonB$ implements a form of modular measurement of the  position of mode~$\bosonA$, in our case a $\pmod{N}$-measurement. 
	Suppose we measure the mode~$\bosonB$ and obtain the measurement outcome $k \in \mathbb{Z}$. This implies that every $x\in\{0,\ldots 2R-1\}$ contained in the support of the reduced state of the mode~$\bosonA$ has to satisfy
	\begin{align}
		N\cdot y+ f_{\alpha,N,m}(x)=k\qquad \textrm{for~some~}y\in\mathbb{Z} \ .
	\end{align}
	In particular
	\begin{align}
		k &\equiv f_{a,N,m}(x) \pmod{N}\\
		&\equiv a^x \pmod{N}\ ,
	\end{align}
	where we obtained the last equality by~\eqref{eq: pseudomodular power mod N}. Thus after obtaining the measurement outcome $k$ in mode~$\bosonB$, the reduced state $\Psi_2^{(k)}$ of mode~$\bosonA$ is
	\begin{align}
		\ket{\Psi_{2}^{(k)}} \propto \sum_{\substack{x\in\{0,\ldots,2R-1\} :\\
				a^x\equiv k\xmod{N}}}\ket{x} \ .
	\end{align}
	This state is formally similar to the state produced by the quantum subroutine of Shor's algorithm, but with position-encoded integers.
	
	The size (i.e., number of elementary operations) of the circuit $\semiideal$ excluding the preparation of the initial states $(e^{-iRP}\Sha_R,\gkp,\ket{0})$ is  the size of the composed unitaries $M_N$ and $\Ua$, plus the $P$-quadrature measurement. With the  parameters~$(m,R)$ given in Table~\ref{tab: parameters} this adds up to $O(\log N)+O(n^2)+1=O(n^2)$ elementary operations. Ignoring the initial state preparation, the circuit $\circWithInitStates(e^{-iRP}\Sha_R,\gkp,\ket{0})$   is thus composed of $O(n^2)$ elementary operations (of the form~\eqref{it:singletwomodeops}--\eqref{it:singlequbitoperations} in the main paper).

	\clearpage
	\section{Circuit analysis  with finitely squeezed GKP states}\label{sec: approximate analysis}

	In this section, we prove that the circuit $\cQ_{a,N}$ (cf. Fig.~\ref{fig:maincircuit} in the main paper) can be used to factor integers even though it uses approximate (instead of ideal) GKP states and finitely squeezed Gaussian states (instead of position-eigenstates). We use that this remains true even if the circuit is initialized with states prepared by our GKP state preparation protocol~\cite{SMbrenneretalGKP2024}.

	We first give the proof 
	of Lemma~\ref{lem:mainlemma} (of the main paper) which we restate here for convenience:
	\lemmain*

	Finally, we prove that combined with the GKP state preparation protocol from~\cite{SMbrenneretalGKP2024} and repetition of our algorithm leads to the Theorem~\ref{thm:main} (of the main paper), which we restate here.
	\thmmain*

	We first discuss relevant definitions of approximate GKP states  in Section~\ref{sec: approximate gkp}. In Section~\ref{sec: proof main thm}, we give the proof of Lemma~\ref{lem:mainlemma} and Theorem~\ref{thm:main}. The proof of various auxiliary results is deferred to subsequent sections.

	\subsection{Definition of approximate GKP states} \label{sec: approximate gkp}
	
	We define approximate position-eigenstates as well as approximate GKP states in this section and summarize some of their properties.
	
	Let us start with approximate position-eigenstates. For $z\in\mathbb{R}$ let $\ket{\chi_\Delta(z)}$ be an approximate position-eigenstate defined as
	\begin{align}
		\ket{\chi_\Delta(z)}&:=e^{-iPz} S(\log1/\Delta)\ket{\vac}\ , 
	\end{align}
	where $S(\xi) =e^{i\frac{\xi}{2}(QP+PQ)}$  is the single-mode squeezing operation. By definition, $\chi_{\Delta}(z)\in L^2(\mathbb{R})$ is the function
	\begin{align}
		(\chi_\Delta(z))(x)&:=\Psi_\Delta(x-z)\ ,\label{eq:chiDeltadefinition}
	\end{align}
	where $\Psi_\Delta$ is the centered Gaussian
	given by
	\begin{align}
		\Psi_\Delta(x)&=\frac{1}{(\pi \Delta^2)^{1/4}}e^{-x^2/(2\Delta^2)}\ \qquad \textrm{for}\qquad x\in \mathbb{R}.\label{eq:Psi Delta}
	\end{align}
	Notice that $|\Psi_\Delta(\cdot)|^2$ is a centered normal distribution with variance~$\sigma^2=\Delta^2/2$. The vacuum state $\ket{\textrm{vac}}$ is defined as $\Psi_{\Delta}$ where $\Delta=1$.
	
	We define the approximate GKP state 
	\begin{align}
		\ket{\mathsf{GKP}_{\kappa,\Delta}}:=C_{\kappa,\Delta} \sum_{z\in\mathbb{Z}}\eta_\kappa(z)\ket{\chi_\Delta(z)}\label{eq:gkp approximate eta chi}\ 
	\end{align}
	where $C_{\kappa, \Delta}\in\mathbb{R}$ is a normalization constant, and $\eta_\kappa(\cdot)$ a Gaussian envelope specified by a Gaussian distribution with variance~$\sigma^2=\frac{1}{2\kappa^2}$, i.e., for $\kappa>0$ 
	\begin{align}
		\eta_\kappa(z)&=\frac{\sqrt{\kappa}}{\pi^{1/4}}e^{-\kappa^2z^2/2} \ ,\label{eq:eta kappa def}
	\end{align}
	for illustration see Fig.~\ref{fig:gkp-large-squeezing}. We note that our definition~\eqref{eq:gkp approximate eta chi} is slightly different from the convention in literature, see e.g.~\cite{SMGKP}, where a non-integer spacing is used.

	For our analysis it is convenient to define truncated approximate GKP states. This will ensure certain orthogonality properties. To this end let us define the $\pm\varepsilon$-supported position-eigenstate $\ket{\chi_\Delta^{\varepsilon}(z)}$ as a truncated Gaussian having support on $[z-\varepsilon,z+\varepsilon]$ for some small~$\varepsilon\in(0,1/2)$. More precisely,~$\chi^\varepsilon_\Delta(z)\in L^2(\mathbb{R})$ is defined as \begin{align}
		(\chi^\varepsilon_\Delta(z))(x)&:=\Psi^\varepsilon_\Delta(x-z) \qquad\textrm{for}\qquad x\in \mathbb{R}\ ,
	\end{align}
	where
	\begin{align}
		\Psi_\Delta^\varepsilon&=\frac{\Pi_{[-\varepsilon,\varepsilon]}\Psi_\Delta}{\left\|\Pi_{[-\varepsilon,\varepsilon]}\Psi_\Delta\right\|}\ ,    
	\end{align}
	where we denote by $\Pi_\cC$ the orthogonal projection onto  the subspace of~$L^2(\mathbb{R})$ of functions supported on~$\cC \subseteq\mathbb{R}$.
	Since $\varepsilon<1/2$ the $\pm\varepsilon$-supported position-eigenstates
	satisfy the orthogonality property
	\begin{align}
		\langle \chi_\Delta^{\varepsilon}(z), \chi_\Delta^{\varepsilon}(z')\rangle=\delta_{z,z'} \qquad \text{ for all }\qquad z, z'\in \mathbb{Z}\ .
	\end{align}
	Furthermore, the overlap between a $\pm\varepsilon$-supported position-eigenstate and a displaced vacuum state centered at $z$ is close to $1$ for $\Delta$ small and $\varepsilon\in [\sqrt{\Delta},1/2)$, see Lemma~\ref{lem: gaussian gaussian epsilon}
	\begin{align}
		\left|\langle \Psi_\Delta, \Psi_\Delta^\varepsilon\rangle \right|^2 
		&\ge 1-2\Delta \  .\label{eq: Psi Psi eps overlap restricted ep}
	\end{align}

	\begin{figure}[!htb]
		\centering
		\begin{subfigure}[b]{0.45\textwidth}
			\centering
			\hspace*{-0.7cm}\includegraphics[width= 0.8\textwidth]{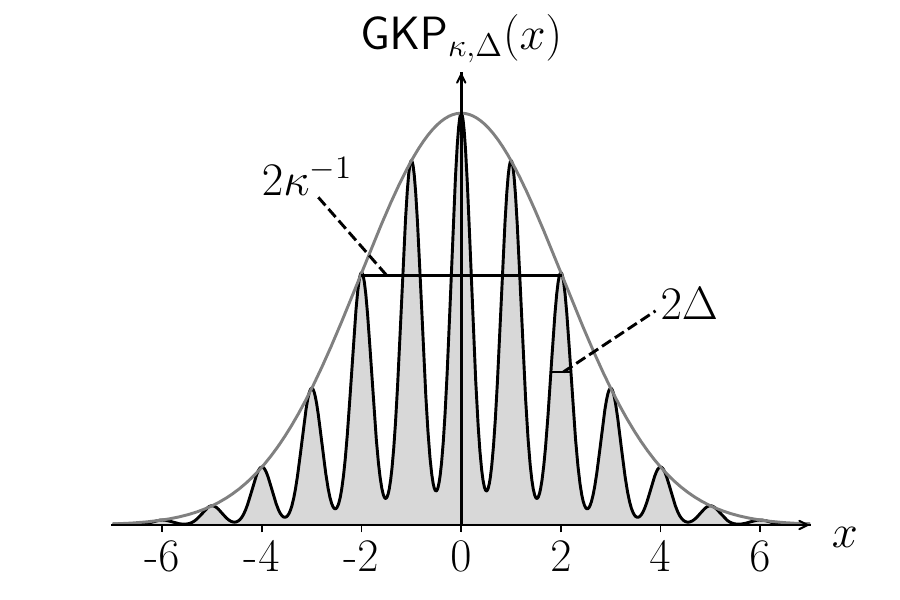}
			\caption{The state $\ket{\gkp_{\kappa,\Delta}}\in L^2(\mathbb{R})$.}
			\label{fig:gkp-large-squeezing}
		\end{subfigure}
		\hfill
		\begin{subfigure}[b]{0.45\textwidth}
			\centering
			\hspace*{-0.7cm}\includegraphics[width=0.8\textwidth]{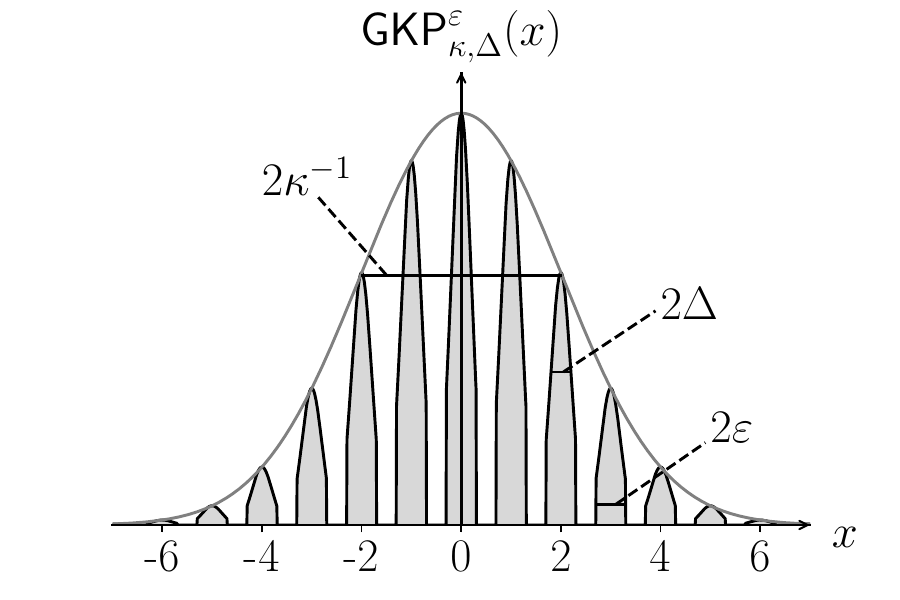}
			\caption{The state $\ket{\gkp_{\kappa,\Delta}^\varepsilon}\in L^2(\mathbb{R})$.}
			\label{fig:truncated-gkp}
		\end{subfigure}
		\caption{The approximate GKP state and its ``truncated'' version represented in position space.}
	\end{figure}

	We define the approximate GKP state $\ket{\gkp_{\kappa,\Delta}^{\varepsilon}}$ with truncated peaks as
	\begin{align}
		\ket{\gkp_{\kappa, \Delta}^{\varepsilon}}&:= C_\kappa \sum_{z\in\mathbb{Z}}\eta_\kappa(z)\ket{\chi_\Delta^\varepsilon(z)} \ , \label{eq:approximate gkp def eta xi eps}
	\end{align}
	where $C_\kappa$ is the normalization factor. A state $\ket{\gkp_{\kappa, \Delta}^{\varepsilon}}$ is shown in Fig.~\ref{fig:truncated-gkp}. The GKP state with truncated peaks is related to its ``non-truncated'' version by 
	\begin{align}
		\gkp_{\kappa, \Delta}^{\varepsilon} = \frac{\Pi_{\mathbb{Z}(\varepsilon)}\gkp_{\kappa, \Delta}}{\left\| \Pi_{\mathbb{Z}(\varepsilon)}\gkp_{\kappa, \Delta}\right\|}\ ,
	\end{align}
	where $\mathbb{Z}(\varepsilon)=\mathbb{Z}+[-\varepsilon, \varepsilon]$ where the addition is to be understood as a Minkowski sum.

	\subsection{Proof of Lemma~\ref{lem:mainlemma} and Theorem~\ref{thm:main}: Efficient quantum algorithm for factoring}\label{sec: proof main thm}
	
	We first provide a sufficient condition for when samples of a real-valued random variable provide a means for factoring. This statement is a straightforward generalization of the analysis of the classical post-processing procedure in Shor's algorithm.
	Recall that $\mathbb{Z}_N=\{0,\ldots, N-1\}$ is the set of elements in the ring of integers modulo $N$ and $\mathbb{Z}_{N}^*=\{a\in\mathbb{Z}_N ~|~ \gcd(a,N)=1\}$ is the set of elements of the multiplicative group of integers modulo $N$. We denote the order of~$a$ in the group~$\mathbb{Z}_N^*$ by~$r(a)$, i.e., this is the smallest positive integer $r$ such that $a^r\equiv 1\xmod{N}$.
	
	The sufficient condition for a family~$\{p_a\}_{a\in\mathbb{Z}_N^*}$ of probability density functions on~$\mathbb{R}$ is expressed by the following definition.
	\begin{definition}\label{def:suitabledistribution}
		Let $N\in \mathbb{N}$ be an integer. Let $q=\min\{2^k\mid k\in\mathbb{N}, N^2<2^k \}$ be the smallest power of $2$ greater than $N^2$. Let $a\in \mathbb{Z}_N^*$, and let  $r(a)$ denote the order of~$a\in\mathbb{Z}_N^*$. For $d\in\mathbb{Z}_{r(a)}$ define the union of intervals
		\begin{align}
			\Gamma_{d}(a):=\bigcup_{j\in \mathbb{Z}}\left[j+\frac{d}{r(a)} - \frac{1}{2q}, j+\frac{d}{r(a)} + \frac{1}{2q} \right]\ . \label{eq:Omega d}
		\end{align}
		A family~$\{p_a\ |\ p_a:\mathbb{R}\rightarrow [0,\infty)\}_{a\in\mathbb{Z}_N^*}$ of probability density functions on~$\mathbb{R}$ is called {\em suitable} if 
		\begin{align}
			\min_{d\in \mathbb{Z}_r} \int_{\Gamma_{d}(a)} p_a(w) dw\ge \Omega\left(1\right) \cdot \frac{1}{r(a)} \qquad\textrm{ for all }\qquad a\in\mathbb{Z}_N^*\ .
		\end{align}
	\end{definition}
	We depict the unions of intervals $\Gamma_d(a)$ in Fig.~\ref{fig:gammas}.
	\begin{figure}[ht]
		\centering
		\vspace*{0.5cm}
		\includegraphics[width=\linewidth]{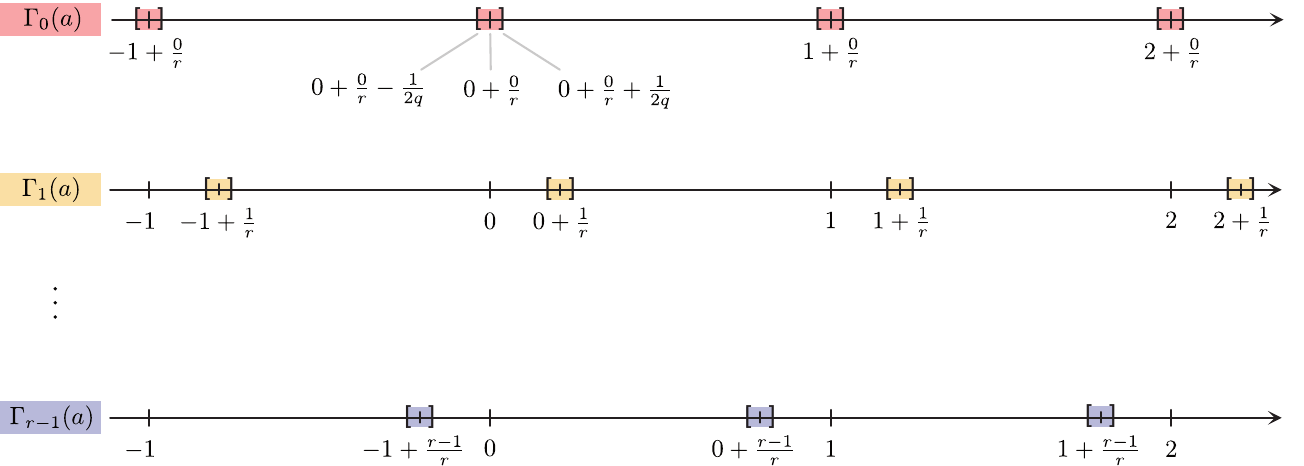}
		\vspace*{0.1cm}
		\caption{An illustration of the union $\Gamma_d(a)$ of intervals  from Definition~\ref{def:suitabledistribution}, for $d\in \{0,\ldots,r-1\}$ depicted on the real axis. We shaded the intervals belonging to the same set~$\Gamma_d(a)$ by the same color.}
		\label{fig:gammas}
	\end{figure}
	
	The following proposition shows that --- as suggested by the terminology --- suitable families of distributions can be used to factor.
	
	\begin{restatable}[Classical post-processing subroutine]{proposition}{propositionzero}\label{prop:propositionzero} 
		There is a polynomial-time classical algorithm~$\cA$ such that the following holds. Given an~$n$-bit integer~$N\in\mathbb{N}$, assume that
		a pair~$(x,a)$ is generated as follows:
		\begin{enumerate}[(i)]
			\item 
			the element $a\sim U_{\mathbb{Z}_N^*}$  is drawn from the uniform distribution $U_{\mathbb{Z}_N^*}$ on $\mathbb{Z}_N^*$, and 
			\item 
			$x\sim p_a$ is drawn from $p_a$, where $\{p_a\ |\ p_a:\mathbb{R}\rightarrow [0,\infty)\}_{a\in\mathbb{Z}_N^*}$ is a suitable family of probability density functions.
		\end{enumerate}
		Then the algorithm~$\cA$ --- on input~$(x,a,N)$ ---  outputs a factor of $N$ with probability at least
		\begin{align}
			\Pr_{\substack{a\sim U_{\mathbb{Z}_N^*}\\ x\sim p_a}}\left[\,\cA(x,a,N) \textrm{ divides } N\,\right] 
			\ge \Omega\left(\frac{1}{\log n}\right)\ . 
		\end{align}
		The runtime of $\cA$ is polynomial in $n$.
	\end{restatable}

	The remainder of the proof of Lemma~\ref{lem:mainlemma} establishes that the output distributions of the circuit $\cQ_{a,N}=\circuitPhysicalStates$ in Fig.~\ref{fig:maincircuit} of the main text 
	(for different~$a\in\mathbb{Z}_N^*$) are suitable, i.e., satisfy Definition~\ref{def:suitabledistribution} and thus can be used for factoring according to Proposition~\ref{prop:propositionzero}. It consists of two parts summarized by Propositions~\ref{prop:propositionone} and~\ref{prop:propositiontwo} below.

	First, we consider the output state of the circuit $\cQ_{a,N}$ in Fig.~\ref{fig:maincircuit} before the measurement. It can be characterized as follows.
	For convenience, we use the shorthand notation $\|\Psi - \Phi\|_1 :=\|\proj{\Psi}-\proj{\Phi}\|_1$ for the variational distance between two rank-$1$-projections corresponding to pure states $\Psi, \Phi \in L^2(\mathbb{R})$.
	
	\begin{restatable}[Output state of the circuit~$\ourcircuit_{a,N}$]{proposition}{propositionone}\label{prop:propositionone} Let $N\ge 8$ (i.e., $n\ge 4$) and $a\in\mathbb{Z}_N^*$. Let $(m,R,\kappa_\bosonA,\Delta_\bosonA,\kappa_\bosonB,\Delta_\bosonB, \Delta_\bosonC)$ be as specified in Table~\ref{tab: parameters}. 
		Let $\Psi_a$ be the final state before the measurement of the circuit $\cQ_{a,N}$ in Fig.~\ref{fig:maincircuit} of the main text. Let~$\Phi_a$ be the state 
		\begin{align}
			\!\!\ket{\Phi_a} &= c_{\Phi_a} \sum_{z\in \mathbb{Z}}
			\eta_{\kappa_\bosonA}(z-R)\ket{\chi_{\Delta_\bosonA }(z)} \otimes 
			e^{-i(a^z\xmod N)P_\bosonB}M_N\ket{\gkp^{\varepsilon_\bosonB}_{\kappa_\bosonB,\Delta_\bosonB}}\otimes \ket{\Psi^{\varepsilon_\bosonC}_{\Delta_\bosonC}} \otimes \ket{0}\  , \label{eq:stateoutputm}
		\end{align}
		where $c_{\Phi_a}$ is a normalization constant and where $\varepsilon_\bosonB = \sqrt{\Delta_\bosonB}$ and $\varepsilon_\bosonC = \sqrt{\Delta_\bosonC}$.
		Then
		\begin{align}
			\left\|\Psi_a-\Phi_a\right\|_1 &\le 364 \cdot 2^{-2n}\ .
		\end{align}
	\end{restatable}
	
	The expression~\eqref{eq:stateoutputm} allows us to show that
	the distribution $p_{\Psi_a}$ obtained by measuring the $P$-quadrature of the state~$\Psi_a$ has the  property
	required by Proposition~\ref{prop:propositionzero}. Indeed, the monotonicity of the $L^1$-norm distance under measurement implies that  the distributions $p_{\rho}$ and $p_{\sigma}$ obtained by measuring any two states $\rho, \sigma \in \cB(L^2(\mathbb{R}))$ satisfy
	\begin{align}
		\|p_{\rho}-p_{\sigma}\|_1 \le \|\rho-\sigma\|_1  \label{eq: bound distribution by trace dist}\ .
	\end{align}
	Furthermore, we have the following statement about the distribution $p_{\Phi_a}$ obtained from a $P$-quadrature measurement of mode~$\bosonA$ of the state $\Phi_a$.
	\begin{restatable}[Distribution of measurement outcomes]{proposition}{propositiontwo}\label{prop:propositiontwo} 
		Let $N\ge 2$. For  $a\in \mathbb{Z}_N^*$ let $\Phi_a \in L^2(\mathbb{R})^{\otimes 3}\otimes \mathbb{C}^2$ be the state defined by Eq.~\eqref{eq:stateoutputm} of Proposition~\ref{prop:propositionone}, with the parameters $(m,R,\kappa_\bosonA,\Delta_\bosonA,\kappa_\bosonB,\Delta_\bosonB, \Delta_\bosonC)$ specified by Table~\ref{tab: parameters} and with $ \varepsilon_\bosonB=\sqrt{\Delta_\bosonB}$ and  $\varepsilon_\bosonC=\sqrt{\Delta_\bosonC}$. Let~$p_{\Phi_a}:\mathbb{R}\rightarrow [0,\infty)$ be the probability density function of the distribution of outcomes when performing a $P$-quadrature measurement on mode~$\bosonA$ of~$\Phi_a$. Then the family of distributions $\{p_{\Phi_a}\}_{a\in \mathbb{Z}_N^*}$ is suitable. 
	\end{restatable}
	
	With these ingredients, we are ready to prove Lemma~\ref{lem:mainlemma}.
	Indeed, Proposition~\ref{prop:propositionone} combined 
	with Eq.~\eqref{eq: bound distribution by trace dist} and Proposition~\ref{prop:propositiontwo} imply that the family of output distributions $\{p_{\Psi_a}\}_{a\in\mathbb{Z}_N^*}$ of the circuit~$\cQ_{a,N}$ is suitable. That is, for $a\in\mathbb{Z}_N^{*}$ we have (for sufficiently large $N$) that
	\begin{align}
		\min_{d\in \mathbb{Z}_r} \int_{\Gamma_{d}(a)} p_{\Psi_a}(w) dw
		& \ge \Bigg( \min_{d\in \mathbb{Z}_r} \int_{\Gamma_{d}(a)} p_{\Phi_a}(w) dw \Bigg) -\left\|p_{\Psi_a}-p_{\Phi_a}\right\|_1\\
		& \ge \Bigg( \min_{d\in \mathbb{Z}_r} \int_{\Gamma_{d}(a)} p_{\Phi_a}(w) dw \Bigg) -\left\|\Psi-\Phi\right\|_1 \label{eq: QaN output distribution prop one} \\
		&\ge \Omega\left(1\right)\cdot\frac{1}{r(a)} - 364\cdot 2^{-2n}\\
		& \ge \Omega\left(1\right)\cdot \frac{1}{r(a)}\ , \label{eq: QaN output distribution prop}
	\end{align}
	where we used the triangle inequality to obtain the first inequality and $r(a) < N < 2^n$, in the last step. We note that $\Omega(1)$ is independent of $a$.
	This proves Lemma~\ref{lem:mainlemma}.
	
	It remains to prove Theorem~\ref{thm:main}. To do so we use the GKP state preparation protocol of~\cite{SMbrenneretalGKP2024} which we denote by $\cP^{\gkp}_{\kappa,\Delta}$ in the following. This protocol prepares a mixed state close (in trace distance) to the GKP state $\ket{\gkp_{\kappa,\Delta}}$, see Theorem~\ref{thm:P gkp} below (stated in its informal variant as Lemma~\ref{lem: P gkp} in the main text).
	\begin{theorem}[{\!\!\!\cite[Theorem 5.1]{SMbrenneretalGKP2024}} (see also Corollary~5.3)]\label{thm:P gkp}
		There exists a protocol $\cP^{\gkp}_{\kappa,\Delta}$
		that outputs a classical ``success'' flag with probability at least~$1/10$,
		and has the property that the output state $\rho\in\cB(L^2(\mathbb{R}))$ conditioned on success 
		satisfies
		\begin{align}
			\left\| \rho - \gkp_{\kappa,\Delta} \right\|_1 \le O(\Delta^{1/2})+O(\kappa^{1/3})\ ,
		\end{align}
		in the limit $(\kappa,\Delta)\to(0,0)$. The protocol uses $O(\log1/\kappa)+O(\log 1/\Delta)$ operations of the form~\eqref{it:singletwomodeops}--\eqref{it:singlequbitoperations}, where some of the involved squeezing- and displacement operations are classically controlled by (efficiently computable functions of) the measurement result associated with a $Q$-quadrature measurement.
	\end{theorem}
	
	Let $\Psi_a$ be the final state before the measurement of the circuit $\cQ_{a,N}$. Let $\rho_a$ be the final (mixed) state before the measurement in the circuit $\circWithInitStates(e^{-iRP} \sigma_{\kappa_\bosonA,\Delta_\bosonA} e^{iRP},\sigma_{\kappa_\bosonB,\Delta_\bosonB},\Psi_{\Delta_\bosonC})$, where we used the approximate GKP state preparation protocol $\cP^{\gkp}_{\kappa_i,\Delta_i}$ to generate the state $\sigma_{\kappa_i,\Delta_i}$ in modes $i = \bosonA, \bosonB$. Using the triangle inequality together with~\eqref{eq: bound distribution by trace dist} yields
	\begin{align}
		\min_{d\in \mathbb{Z}_r} \int_{\Gamma_{d}(a)} p_{\rho_a}(w) dw
		&\ge \Bigg( \min_{d\in \mathbb{Z}_r} \int_{\Gamma_{d}(a)} p_{\Psi_a}(w) dw \Bigg) -\left\|\rho_a-\Psi_a \right\|_1 \ .
	\end{align}
	By the invariance of the $L^1$-norm distance under unitaries we obtain 
	\begin{align}
		\left\|\rho_a-\Psi_a\right\|_1
		&= \left\| \sigma_{\kappa_\bosonA,\Delta_\bosonA} \otimes \sigma_{\kappa_\bosonB,\Delta_\bosonB}  \otimes \Psi_{\Delta_\bosonC}\otimes\proj{0} - \gkp_{\kappa_\bosonA, \Delta_\bosonA}\otimes \gkp_{\kappa_\bosonB, \Delta_\bosonB} \otimes  \Psi_{\Delta_\bosonC}\otimes\proj{0}\right\|_1\\
		&\le  \left\| \sigma_{\kappa_\bosonA,\Delta_\bosonA} -  \gkp_{\kappa_\bosonA, \Delta_\bosonA} \right\|_1 +  \left\|  \sigma_{\kappa_\bosonB,\Delta_\bosonB}   - \gkp_{\kappa_\bosonB, \Delta_\bosonB}  \right\|_1 \label{eq:thmaminStateclosenessTriangleIneq}\\
		&\le O(\Delta_\bosonA^{1/2})+O(\kappa_\bosonA^{1/3}) + O(\Delta_\bosonB^{1/2})+O(\kappa_\bosonB^{1/3}) \label{eq:thmmainStateClosenessProt}\\
		&=O(2^{-8n} +  2^{-16n/3} +  2^{-9n^2/2} + 2^{-18n^2/3})\\
		&\le O(2^{-4n})\ 
	\end{align}
	where we used that~$\|\rho_1\otimes \sigma-\rho_2\otimes\sigma\|_1=\|\rho_1-\rho_2\|_1$ for states~$\rho_1,\rho_2,\sigma$, and that 
	\begin{align}
		\|\rho_\bosonA\otimes\rho_\bosonB-\sigma_\bosonA\otimes\sigma_\bosonB\|_1
		&\le\|\rho_\bosonA\otimes\rho_\bosonB-\sigma_\bosonA\otimes\rho_\bosonB\|_1+\|\sigma_\bosonA\otimes\rho_\bosonB-\sigma_\bosonA\otimes\sigma_\bosonB\|_1 \\
		&\le \|\rho_\bosonA-\sigma_\bosonA\|_1 + \|\rho_\bosonB-\sigma_\bosonB\|_1
	\end{align}
	for states $\rho_\bosonA,\rho_\bosonB,\sigma_\bosonA,\sigma_\bosonB$ to obtain Eq.~\eqref{eq:thmaminStateclosenessTriangleIneq} and Theorem~\ref{thm:P gkp} (with the parameters
	$\kappa_\bosonA, \Delta_\bosonA, \kappa_\bosonB, \\ \Delta_\bosonB, \Delta_\bosonC$ specified in Table~\ref{tab: parameters}) to obtain Eq.~\eqref{eq:thmmainStateClosenessProt}.
	Combining this with~\eqref{eq: QaN output distribution prop} and using the fact that $r(a) < 2^n$ we obtain
	\begin{align}
		\min_{d\in \mathbb{Z}_r} \int_{\Gamma_{d}(a)} p_{\rho_a}(w) dw
		& \ge \Omega\left(1\right) \cdot \frac{1}{r(a)}-O(2^{-4n})
		\ge \Omega\left(1\right)\cdot \frac{1}{r(a)}\ .
	\end{align}
	Thus Proposition~\ref{prop:propositionzero} implies that there exists a classical polynomial-time algorithm which when applied to a single sample from the outcome distribution yields a factor of $N$ with probability at least $\Omega(1/\log n)$.
	A simple amplification with $\log n$ repetitions yields the claim (II) of Theorem~\ref{thm:main}. Claim (I) of Theorem~\ref{thm:main} is obtained by realizing the circuit $\cQ_{a,N}$ (see Fig.~\ref{fig:maincircuit}) together with the preparation of the initial states with the parameters
	$(m,R,\kappa_\bosonA,\Delta_\bosonA,\kappa_\bosonB,\Delta_\bosonB, \Delta_\bosonC)$ given in Table~\ref{tab: parameters}. We analyzed the size and depth of the circuit $\semiideal$ 
	in Section~\ref{sec: main circuit ideal initial states} which did not accounted for the initial state preparation but for the circuit $\circWithInitStates$ only. To prepare GKP states we use the protocol $\cP_{\kappa,\Delta}^\gkp$ of~\cite{SMbrenneretalGKP2024} producing states $\sigma_{\kappa,\Delta}$ close to the state $\gkp_{\kappa,\Delta}$, see Theorem~\ref{thm:P gkp}. Thus the number of operations needed to prepare the initial state $e^{-iRP}\sigma_{\kappa_\bosonA,\Delta_\bosonA}$ on the mode $\bosonA$ is $O(\log 1/\kappa_\bosonA)+O(\log 1/\Delta_{\bosonA})+O(\log R)=O(n)$, where we haven also taken into account the translation by $R$ performed by the decomposed unitary $e^{-iRP}$ (cf. Section~\ref{sec: translation by R}) and used the parameters $\kappa_\bosonA=\Delta_\bosonA=2^{-16n}$. The initial state on the mode~$\bosonB$ is a GKP state prepared by protocol $\cP_{\kappa,\Delta}^\gkp$ with parameters $\kappa_\bosonB=\Delta_\bosonB=2^{-18n^2}$ (cf. Table~\ref{tab: parameters}) which requires $O(n^2)$ elementary operations to prepare. The initial state on mode $\bosonC$ is the squeezed vacuum state $\ket{\Psi_{\Delta_\bosonC}}=M_{\Delta_\bosonC}\ket{\vac}$ whose preparation requires $O(\abs{\log \Delta_{\bosonC}})=O(n)$ elementary operations since $\Delta_\bosonC=2^{-50n}$. Therefore, the number of elementary operations needed to implement $\circWithInitStates(e^{-iRP} \sigma_{\kappa_\bosonA,\Delta_\bosonA} e^{iRP},\sigma_{\kappa_\bosonB,\Delta_\bosonB},\Psi_{\Delta_\bosonC})$ including its initial states scales as $O(n^2)$ and so does its circuit size and depth.

	\clearpage
	\section{Proof of Proposition~\ref{prop:propositionzero}: Classical post-processing}\label{sec: prop zero proof}
	In this section, we prove 
	Proposition~\ref{prop:propositionzero}, which we restate here for the reader's convenience.
	\propositionzero*
	
	Here we prove Proposition~\ref{prop:propositionzero} by relating it to the classical post-processing subroutine of Shor's algorithm~\cite{SMShor}. To this end, let us define the set of good inputs for the post-processing subroutine of Shor's algorithm.
	\begin{definition}\label{def:good set}
		Let $N\in\mathbb{N}$. Let $q$ be the smallest power of $2$ such that $q> N^2$.  We define the following set for all $a\in \mathbb{Z}_N^*$
		\begin{align}
			\good_a&:=\left\{c\in \mathbb{Z}_q \ |\ 
			\exists d\in\mathbb{Z}_{r(a)}\textrm{ such that }
			\left|\frac{c}{q}-\frac{d}{r(a)}\right|\leq \frac{1}{2q}
			\right\}\ ,
		\end{align}
		where $r(a)$ is the order of $a$ in $\mathbb{Z}^*_N$.
	\end{definition}
	We use this set to paraphrase the success probability of the classical post-processing of Shor's algorithm in the following theorem.

	\begin{theorem}[Shor's post-processing subroutine, paraphrased from~\cite{SMShor}]\label{thm:shor factorization postprocessing} There is an efficient classical algorithm~$\cShor$ with the following properties.
		Given an~$n$-bit integer~$N\in\mathbb{N}$, assume that
		a pair~$(c,a)$ is generated as follows:
		\begin{enumerate}[(i)]
			\item 
			the element $a\sim U_{\mathbb{Z}_N^*}$  is drawn from the uniform distribution $U_{\mathbb{Z}_N^*}$ on $\mathbb{Z}_N^*$, and 
			\item 
			$c\sim P_a$ is drawn from $P_a$, where $\{P_a\ |\ P_a:\mathbb{Z}_q\rightarrow [0,1]\}_{a\in\mathbb{Z}_N^*}$ is a family of probability distributions on~$\mathbb{Z}_q$. 
		\end{enumerate}
		Then the algorithm~$\cShor$, given a sample $c\sim P_a$ and the two parameters $a$ and $N$, outputs a factor of $N$ with probability at least
		\begin{align}
			\Pr_{\substack{a\sim U_{\mathbb{Z}_N^*}\\ c\sim P_a}}[\cShor(c, a, N) \textrm{ divides } N]\ge \Omega(1)\cdot \min_{a\in \mathbb{Z}_N^*} \left( \frac{r(a)}{\log\log r(a)} \cdot \min_{c\in\good_a} P_a(c)\right) \label{eq: shor prob bound}
		\end{align}
		in time polynomial in $n$.
	\end{theorem}
	Theorem~\ref{thm:shor factorization postprocessing} shows that samples from the family~$\{P_a\}_{a\in\mathbb{Z}_N^*}$ can be used to factor~$N$ as long as each~$P_a$ is such that every element~$c\in \good_a$ has non-negligible probability~$P_a(c)$.
	
	We can relax this condition somewhat, demanding only that~$P_a(\{c-1,c,c+1\})$ is non-negligible for every~$c\in\good_a$. (Here $c-1,c+1\in\mathbb{Z}_q$, i.e., substraction and addition is to be understood modulo~$q$.) To this end, consider the following modified algorithm~$\cC^{\mathsf{Shor}'}$. It again takes~$a\sim U_{\mathbb{Z}_N^*}$, a sample~$c\sim P_a$ and $N$, and proceeds as follows: For $\ell\in \{-1,0,1\}$, the 
	algorithm computes~$\cC^{\mathsf{Shor}}(c+\ell,a,N)$. It returns the corresponding number if one of the three outputs is a factor of~$N$, and outputs failure otherwise. It is easy to check that (as a consequence of Theorem~\ref{thm:shor factorization postprocessing}), we then also have

	\begin{align}
		\Pr_{\substack{a\sim U_{\mathbb{Z}_N^*}\\ c\sim P_a}}[\cC^{\mathsf{Shor}'}(c, a, N) \textrm{ divides } N]
		&\ge \Omega(1)\cdot \min_{a\in \mathbb{Z}_N^*} \left( \frac{r(a)}{\log\log r(a)} \cdot \min_{c\in\good_a} 
		P_a\left( \{c-1, c, c+1\}\right)\right)\ .
		\label{eq: prob bound min sum}
	\end{align}
	
	This reasoning can trivially be extended to probability density functions on the real line combined with a discretization procedure.
	Let $\dis:\mathbb{R}\to \mathbb{Z}_{q}$ be a ``discretization'' function. 
	Then a probability distribution  $\widetilde{p}:\mathbb{R}\to [0,1]$ 
	on~$\mathbb{R}$ induces a distribution~$p:\mathbb{Z}_q\rightarrow[0,1]$
	by
	\begin{align}
		P(c)&=\int_{\dis^{-1}(\{c\})}
		\tilde{p}(x)dx\qquad\textrm{ for }\qquad c\in\mathbb{Z}_q\ .
	\end{align}
	In particular, statements such as~\eqref{eq: prob bound min sum} can be translated to families~$\{p_a\}_{a\in\mathbb{Z}_N^*}$ of probability distributions on the real line. We argue that doing so with a specific discretization function~$\dis$ leads to the concept of a suitable family of distributions, see Def.~\eqref{def:suitabledistribution}.
	
	Concretely, consider the function
	\begin{align}
		\begin{matrix}
			\dis:& \mathbb{R} & \rightarrow & \mathbb{Z}_q\\
			&x& \mapsto& \min\left( \underset{c\in\{0,\ldots,q\}}{\arg\min}\left|\frac{c}{q}-\fra(x)\right|\right) \xmod{q}
		\end{matrix}
	\end{align}
	where $\fra(x) := x-\lfloor x\rfloor $ is the fractional part of~$x\in\mathbb{R}$. It returns the closest multiple of $1/q$ to the fractional part of~$x$. One can check that the function $\dis(\cdot)$ satisfies the following two properties:
	\begin{enumerate}[(i)]\setlength{\itemsep}{7pt}
		\item $\dis(x) = \dis(x+1) \quad \textrm{ for all } x\in \mathbb{R}\ ,$ i.e., it is translation-invariant, and \label{it: period dis}
		\item $\dis\left(\left\{\frac{c}{q}\right\} + \left(-\frac{1}{2q}, \frac{1}{2q}\right]\right) = \{c\}\, \quad \textrm{ for all } c\in \mathbb{Z}_q$\ , i.e., any argument $1/(2q)$-close to~$c/q$ for $c\in\mathbb{Z}_q$ gets mapped to~$c$.\label{it:interval dis}
	\end{enumerate}
	
	We use these properties
	to show the following relation between the unions of intervals $\Gamma_d(a)$ from Proposition~\ref{prop:propositionzero} and the action of the discretization function.
	\begin{lemma} \label{lem: Omega d dis}
		Let $a\in\mathbb{Z}_N^*$, let $r(a)$ be the order of~$a$ in~$\mathbb{Z}_N^*$, and let $c \in \good_a$. Then there exists $d \in \mathbb{Z}_{r(a)}$ such that
		\begin{align}
			\dis^{-1}\left( \{c -1, c, c+1\}\right) \supset \Gamma_{d}(a)\ ,
		\end{align}
		where $\Gamma_d(a)$ is the union of intervals defined in Definition~\ref{def:suitabledistribution} (see Eq.~\eqref{eq:Omega d}).
	\end{lemma}
	\begin{proof}
		For brevity, let us write~$r=r(a)$.  We partition $\Gamma_d(a)$ as
		\begin{align}
			\Gamma_d(a) &= \bigcup_{j\in\mathbb{Z}}\left[j+\frac d r-\frac{1}{2q},j+\frac d r+\frac{1}{2q}\right] \\
			&=\bigcup_{j\in\mathbb{Z}} \Gamma_{j,d} \qquad \textrm{  where }\qquad \Gamma_{j,d}:=\left[j+\frac d r-\frac{1}{2q},j+\frac d r+\frac{1}{2q}\right]\ ,
		\end{align}
		where we observed that the intervals~$\Gamma_{j,d}$ are in both $j\in\mathbb{Z}$ and $d\in\mathbb{Z}_r$ pairwise disjoint (since $q>r\ge 1$). Furthermore, we have 
		\begin{align}
			\Gamma_{j,d}=\{j\}+\Gamma_{0,d}\qquad\textrm{ for every }\qquad j\in\mathbb{Z}\ \label{eq:omegajdtranslation}
		\end{align}
		where the addition is to be understood as  a Minkowski sum. We claim that
		\begin{align}
			\{c-1,c,c+1\}\supseteq\dis(\Gamma_{0,d})\ .\label{eq:claimcontainmentccplus}
		\end{align}
		With~\eqref{eq:omegajdtranslation}, this implies the claim because of property~\eqref{it: period dis} of the discretization function~$\dis$.
		
		To prove~\eqref{eq:claimcontainmentccplus}, recall that the assumption $c \in \good_a$ means that there is an element $d\in \mathbb{Z}_r$ such that
		\begin{align}
			\abs{\frac{c}{q} - \frac{d}{r}}\le \frac{1}{2q}\ . \label{eq: dist cq dr}
		\end{align}
		This implies that 
		\begin{align}
			\frac{c}{q} \in \Gamma_{0, d}\  .
		\end{align}
		Since~$\Gamma_{0,d}$ is an interval of diameter~$\frac1q<\frac{3}{2q}$, it follows that
		\begin{align}
			\left(\frac{c}{q}-\frac{3}{2q},\frac{c}{q}+\frac{3}{2q}\right]\supset \Gamma_{0,d}\ ,
		\end{align}
		because
		\begin{align}
			\left(\frac{c}{q}-\frac{3}{2q},\frac{c}{q}+\frac{3}{2q}\right]
			&=\left(\left\{\frac{c-1}{q}\right\}+I_q\right)\cup \left(\left\{\frac{c}{q}\right\}+I_q\right)\cup \left(\left\{\frac{c+1}{q}\right\}+I_q\right)\ ,\label{eq:threeintervals}
		\end{align}
		where $I_q=(-1/(2q),1/(2q)]$. We illustrate the covering of $\Gamma_{0,d}$ by the three intervals in the RHS of Eq.~\eqref{eq:threeintervals} in Fig.~\ref{fig:cover by three intervals}.
		The claim~\eqref{eq:claimcontainmentccplus} follows from property~\eqref{it:interval dis} of the discretization function~$\dis$. 
	\end{proof}
	\begin{figure}[ht]
		\centering
		\begin{subfigure}[t]{\textwidth}
			\centering
			\includegraphics[width= \linewidth]{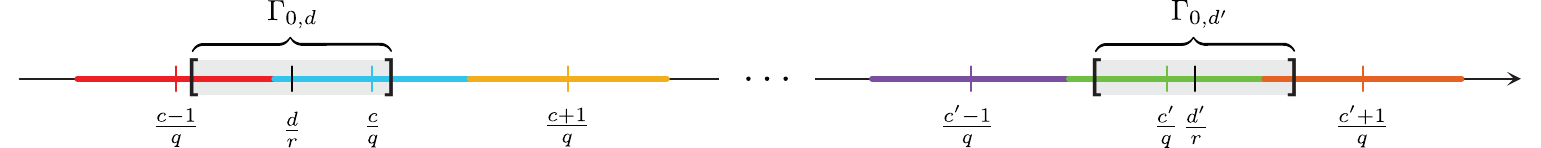}
			\caption{A covering of the interval $\Gamma_{0,d}$ by three intervals $\{\frac{c-1}{q}\}+I_q$, $\{\frac{c}{q}\}+I_q$, $\{\frac{c+1}{q}\}+I_q$ colored in red, blue and yellow, respectively. The function $\dis$ maps 
				these intervals to $c-1,c$ and $c+1$, respectively. 
				In the given example, the set~$\Gamma_{0,d}$ is covered by the two  intervals $\{\frac{c-1}{q}\}+I_q$ and $\{\frac{c}{q}\}+I_q$ (red and blue). On the other hand, the displayed interval $\Gamma_{0,d'}$ is covered by the two intervals $\{\frac{c'}{q}\}+I_q$ and $\{\frac{c'+1}{q}\}+I_q$ (green and orange). 
				In general, using three intervals~$\{\frac{c-1}{q}\}+I_q$, $\{\frac{c}{q}\}+I_q$, $\{\frac{c+1}{q}\}+I_q$ is always sufficient to 
				cover~$\Gamma_{j,d}$ (in this illustration $c,c'\in\good_a$).}
		\end{subfigure}
		\begin{subfigure}[t]{\textwidth}
			\vspace*{0.5cm}
			\includegraphics[width=\linewidth]{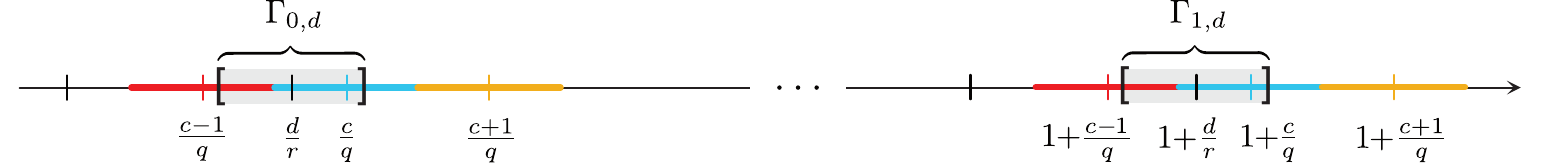}
			\caption{A covering of interval $\Gamma_{j,d}$, for $j\in \{0,1\}$. These intervals are covered by $\{j+\frac{c-1}{q}\}+I_q$, $\{j+\frac{c}{q}\}+I_q$, $\{j+\frac{c+1}{q}\}+I_q$ and mapped by the function $\dis$ to $c-1,c,c+1$, respectively. We can observe the similar behaviour for any $j\in\mathbb{Z}$.}
		\end{subfigure}
		\caption{An illustration of the covering of the interval $\Gamma_{0,d}$ by the three intervals $\{\frac{c-1}{q}\}+I_q$, $\{\frac{c}{q}\}+I_q$, $\{\frac{c+1}{q}\}+I_q$, 
			where $I_q=(-1/(2q),1/(2q)]$, used in the proof of Lemma~\ref{lem: Omega d dis} and explanation of the action of the discretization function $\dis$.
		}
		\label{fig:cover by three intervals}
	\end{figure}

	In the following we prove Proposition~\ref{prop:propositionzero}. By the assumption that the family $\{p_a~|~ p_a:\mathbb{R}\to [0,\infty) \}_{a\in\mathbb{Z}_N^*}$ of probability distributions is suitable (see Def.~\ref{def:suitabledistribution}), 
	we have 
	\begin{align}
		\min_{d\in \mathbb{Z}_r} \int_{\Gamma_{d}(a)} p_a(w) dw\ge \Omega\left(1\right)\cdot \frac{1}{r(a)}\qquad \textrm{ for all }\qquad a\in\mathbb{Z}_N^*\ .
	\end{align}
	It follows from Lemma~\ref{lem: Omega d dis}  that for every~$a\in\mathbb{Z}_N^*$, the random variable~$C_a=\dis(X_a)$, where $X_a\sim p_a$ is drawn from~$p_a$, is distributed according to a distribution~$P_a$ that satisfies
	\begin{align}
		P_a(\{c-1,c,c+1\})\geq p_a(\Gamma_d(a))\geq \Omega\left(1\right) \cdot \frac{1}{r(a)} \qquad\textrm{ for every }c\in\good_a\ .
	\end{align}
	(Here we write $p_a(A)=\int_A p_a(x)dx$ for $A\subseteq \mathbb{R}$.) 
	Applying~\eqref{eq: prob bound min sum}
	to the family~$\{P_a\}_{a\in\mathbb{Z}_N^*}$ of distributions
	yields
	\begin{align}
		\Pr_{\substack{a\sim U_{\mathbb{Z}_N^*}\\ c\sim P_a}}[\cC^{\mathsf{Shor}'}(c, a, N) \textrm{ divides } N]
		&\ge \Omega(1)\cdot \left(\min_{a\in \mathbb{Z}_N^*}\frac{1}{\log\log r(a)}\right)\ .
	\end{align}
	This is the claim because $r(a)\leq N$.
	
	\clearpage
	\section{Proof of Proposition~\ref{prop:propositionone}: Output state of circuit} \label{sec: prop one proof}
	This section is devoted to the proof of Proposition~\ref{prop:propositionone}, which characterizes the output state~$\Psi_a$ of the circuit~$\cQ_{a,N}=\circuitPhysicalStates$  before the measurement (see Fig.~\ref{fig:maincircuit} of the main paper). By definition, this output state~$\Psi_a$ is equal to
	\begin{align}
		\ket{\Psi_a}&=
		\Uan\Big(e^{-iRP_\bosonA}\ket{\gkp_{\kappa_\bosonA,\Delta_\bosonA}}\otimes 
		M_N\ket{\gkp_{\kappa_\bosonB,\Delta_\bosonB}}\otimes \ket{\Psi_{\Delta_\bosonC}}\otimes \ket{0}\Big)\ 
		,\label{eq:stateonex}
	\end{align}
	where $\Uan$ is the unitary specified in Section~\ref{sec: explanation Ua} and the parameters $\kappa_\bosonA,\Delta_\bosonA,\kappa_\bosonB,\Delta_\bosonB,\Delta_\bosonC$ and $R,m$ are as in Table~\ref{tab: parameters}. Proposition~\ref{prop:propositionone}  (which we restate here for convenience) is the following.
	\propositionone*
	\begin{proof}
		To establish Proposition~\ref{prop:propositionone}, we 
		go through a sequence~
		\begin{align}
			\Psi_a=\Psi^{(0)},\Psi^{(1)},\ldots,\Psi^{(5)}=\Phi_a
		\end{align} of states and compare these pairwise.  Their definition is given in Table~\ref{tab:approximatealgorithm}. 
		We establish error bounds of the form
		\begin{align}
			\left\|\Psi^{(j)}-\Psi^{(j-1)}\right\|_1 &\leq \varepsilon^{(j)}\qquad\textrm{ for }\qquad j\in \{1,\ldots,5\}\ .\label{eq:errorboundpsij}
		\end{align}
		With the triangle inequality, this implies that
		\begin{align}
			\left\|\Phi_a-\Psi_a\right\|_1 & \leq \sum_{j=1}^{5}\varepsilon^{(j)} = 364\cdot 2^{-2n}\ ,
		\end{align}
		implying Proposition~\ref{prop:propositionone}.

		\begin{table}[ht]
			\normalsize
			\centering
			{
				\arraycolsep=1.2pt
				\def\arraystretch{1.3}
				$\begin{array}{llrrr} 
					\hline\hline
					\ket{\Psi^{(0)}}= &  \Uan(e^{-iRP_\bosonA}\ket{\gkp_{\kappa_\bosonA,\Delta_\bosonA }}\hspace{-12ex}&\otimes ~M_N\ket{\gkp_{\kappa_\bosonB ,\Delta_\bosonB }}&\otimes  \ket{\Psi_{\Delta_\bosonC}} &\otimes\ket{0})   \\
					\ket{\Psi^{(1)}}=
					& \cpsi{1}\, \Uan(\mathop{{\large\Pi}}_{[-1/2,2R-1/2]}e^{-iRP_\bosonA}\ket{\gkp^{\varepsilon_\bosonA }_{\kappa_\bosonA,\Delta_\bosonA }}\hspace{-12ex}&\otimes~ M_N\ket{\gkp^{\varepsilon_\bosonB }_{\kappa_\bosonB ,\Delta_\bosonB }}&\otimes\ket{\Psi_{\Delta_\bosonC}^{\varepsilon_\bosonC}}&\otimes \ket{0})  \\
					\ket{\Psi^{(2)}}= &\cpsi{2}\sum_{z=0}^{2R-1} 
					\eta_{\kappa_\bosonA}(z-R)\!\ket{\chi^{\varepsilon_\bosonA }_{\Delta_\bosonA }(z)} &\otimes~ 
					e^{-if_{a,N,m}(z)P_\bosonB}~M_N\ket{\gkp^{\varepsilon_\bosonB }_{\kappa_\bosonB ,\Delta_\bosonB }}&\otimes\ket{\Psi_{\Delta_\bosonC}^{\varepsilon_\bosonC}}& \otimes \ket{0}~ \\ 
					\ket{\Psi^{(3)}}= & \cpsi{3}\sum_{z=0}^{2R-1}
					\!\eta_{\kappa_\bosonA}(z-R)\!\ket{\chi^{\varepsilon_\bosonA }_{\Delta_\bosonA }(z)}&\otimes~e^{-i(a^z \xmod{N})P_\bosonB} 
					M_N\ket{\gkp^{\varepsilon_\bosonB }_{\kappa_\bosonB ,\Delta_\bosonB }}&\otimes\ket{\Psi_{\Delta_\bosonC}^{\varepsilon_\bosonC}}&\otimes \ket{0}~ \\
					\ket{\Psi^{(4)}}= & \cpsi{4}\sum_{z\in \mathbb{Z}}
					\eta_{\kappa_\bosonA}(z-R)\ket{\chi^{\varepsilon_\bosonA }_{\Delta_\bosonA }(z)}&\otimes~
					e^{-i(a^z\xmod N)P_\bosonB}M_N\ket{\gkp^{\varepsilon_\bosonB }_{\kappa_\bosonB ,\Delta_\bosonB }}&\otimes\ket{\Psi_{\Delta_\bosonC}^{\varepsilon_\bosonC}}&\otimes \ket{0}~  \\
					\ket{\Psi^{(5)}}= & \cpsi{5}\sum_{z \in \mathbb{Z}}
					\eta_{\kappa_\bosonA}(z-R)\ket{\chi_{\Delta_\bosonA }(z)} &\otimes~
					e^{-i(a^z\xmod N)P_\bosonB}M_N\ket{\gkp^{\varepsilon_\bosonB }_{\kappa_\bosonB ,\Delta_\bosonB }}&\otimes\ket{\Psi_{\Delta_\bosonC}^{\varepsilon_\bosonC}}&\otimes \ket{0}~  \\[0.2cm] \hline
				\end{array}\newline\vspace*{0.5cm}
				\arraycolsep=10pt
				\def\arraystretch{1.3}
				\hspace*{-1.2cm}\vspace*{0.3cm}\begin{array}{@{\hskip 0.01cm}ll@{\hskip 5.3cm}r}
					\xdistance{0}~~~\le &\varepsilon^{(1)} = 9\cdot 2^{-2n}& (\textrm{Lemma~~\ref{lem:truncationerr}})\\
					\xdistance{1}~~~\le&\varepsilon^{(2)} = 352\cdot 2^{-2n}&(\textrm{Lemma~\ref{lem: approx_unitary}})\\
					\xdistance{2}~~~\le&\varepsilon^{(3)} = 2^{-2n}&(\textrm{Lemma~~\ref{lem: untruncation_error_tail}})\\
					\xdistance{3}~~~\le&\varepsilon^{(4)} = 2^{-2n}&(\textrm{Lemma~~\ref{lem:shift_inariance_error}})\\
					\xdistance{4}~~~\le&\varepsilon^{(5)}=2^{-2n}&(\textrm{Lemma~~\ref{lem:untruncation_error_peak}})\\[0.15cm] 
					\hline\hline
				\end{array}
				$}
			\caption{Above, the states used in the analysis. The constants~$\cpsi{j}$ are chosen such that the states are normalized. 
				Below, the results establishing closeness~$\varepsilon^{(j)}$ between the respective consecutive states $\Psi^{(j-1)}$ and $\Psi^{(j)}$ (cf.~\eqref{eq:errorboundpsij}).
			} 
			\label{tab:approximatealgorithm} 
		\end{table}
		Let us briefly sketch the main steps. The corresponding detailed proofs (including the estimates on the errors~$\varepsilon^{(j)}$, $j\in \{1,\ldots,5\}$) are deferred to subsequent sections. 
		
		\noindent {\bf Closeness of $\Psi^{(0)}$ and $\Psi^{(1)}$:} 
		The state $\Psi^{(0)}=\Psi_a$ of the circuit in Fig.~\ref{fig:maincircuit} before the measurement 
		is given by Eq.~\eqref{eq:stateonex}, that is
		\begin{align}
			\ket{\Psi^{(0)}}=\Uan \ket{\Phi^{(0)}}\ \textrm{ where }\  \ket{\Phi^{(0)}}= e^{-iRP_\bosonA}\ket{\gkp_{\kappa_\bosonA,\Delta_\bosonA }}\otimes M_N\ket{\gkp_{\kappa_\bosonB ,\Delta_\bosonB }}\otimes\ket{\Psi_{\Delta_\bosonC}}\otimes \ket{0}\label{eq:stateone}\tag{$\Psi^{(0)}$}\ .
		\end{align}
		For sufficiently small~$\varepsilon_\bosonA,\varepsilon_\bosonB,\varepsilon_\bosonC$, 
		the result
		\begin{align}
			\ket{\widetilde{\Phi}^{(0)}}&=    e^{-iRP_\bosonA}\ket{\gkp^{\varepsilon_{\bosonA}}_{\kappa_\bosonA,\Delta_\bosonA }}\otimes M_N\ket{\gkp^{\varepsilon_{\bosonB}}_{\kappa_\bosonB ,\Delta_\bosonB }}\otimes\ket{\Psi^{\varepsilon_\bosonC}_{\Delta_\bosonC}}\otimes \ket{0}\label{eq:Phiwidetildedefm}
		\end{align} of restricting the support of~$\Phi^{(0)}$ (and renormalizing) is close to~$\Phi^{(0)}$. Indeed, this follows from the fact that the distance measure we consider (the $L^1$-distance) is invariant under unitaries and tensorizes, together with the fact that a truncated approximate GKP state $\gkp^\varepsilon_{\kappa,\Delta}$ is close to an approximate GKP state~$\gkp_{\kappa,\Delta}$, and a truncated squeezed vacuum state~$\Psi^{\varepsilon}_\Delta$ is close to the squeezed vacuum state~$\Psi_\Delta$ (for appropriate parameter choices).
		
		Observe that~$e^{-iRP}\ket{\gkp^{\varepsilon_\bosonA}_{\kappa_\bosonA,\Delta_\bosonA }}$ 
		is centered around~$R$ (with an envelope whose variance is determined by~$\kappa_\bosonA$), and has peaks (local maxima) supported around integers (with interval width~$2\varepsilon_\bosonA $). It follows that 
		$e^{-iRP}\ket{\gkp^{\varepsilon_\bosonA}_{\kappa_\bosonA,\Delta_\bosonA }}$
		has most of its support on the interval~$[-1/2,2R-1/2]$, i.e., the state
		$\widetilde{\Phi}^{(0)}$ (see Eq.~\eqref{eq:Phiwidetildedefm})
		is close to the state
		\begin{align}
			\ket{\Phi^{(1)}}
			&=c_1\cdot \Pi_{[-1/2,2R-1/2]}e^{-iRP_\bosonA}\ket{\gkp^{\varepsilon_\bosonA}_{\kappa_\bosonA,\Delta_\bosonA}}\otimes M_N\ket{\gkp^{\varepsilon_\bosonB}_{\kappa_\bosonB,\Delta_\bosonB}}\otimes\ket{\Psi^{\varepsilon_{\bosonC}}_{\Delta_\bosonC}}\otimes\ket{0}\ ,\label{eq:phioneproofsketchdefhelpful}
		\end{align}
		where $c_1>0$ is a normalization constant. Thus $\Phi^{(0)}$ and $\Phi^{(1)}$ are close (by the triangle inequality). Since $\ket{\Psi^{(b)}}=U_{a,N,m}\ket{\Phi^{(b)}}$ for $b\in \{0,1\}$, it follows that the state~\ref{eq:stateone}  is $\varepsilon^{(1)}$-close to the \hypertarget{eq:statetwo}{state}
		\begin{align}
			\ket{\Psi^{(1)}}=\cpsi{1}\cdot \Uan\left(\large \Pi_{[-1/2,2R-1/2]} e^{-iRP_\bosonA}\ket{\gkp^{\varepsilon_\bosonA }_{\kappa_\bosonA,\Delta_\bosonA }}\otimes M_N\ket{\gkp^{\varepsilon_\bosonB }_{\kappa_\bosonB ,\Delta_\bosonB }}\otimes\ket{\Psi_{\Delta_\bosonC}^{\varepsilon_\bosonC }}\otimes \ket{0}\right)\ 
		\end{align}
		with an error~$\varepsilon^{(1)}$ given in  Lemma~\ref{lem:truncationerr}. 
		We note that~$\ket{\Psi^{(1)}}$
		is simply the output state (i.e., the state before the measurement) of the circuit  $\circWithInitStates(c_1\Pi_{[-1/2,2R-1/2]} e^{-i R P }\gkp_{\kappa_\bosonA,\Delta_\bosonA}^{\varepsilon_\bosonA},\gkp_{\kappa_\bosonB,\Delta_\bosonB}^{\varepsilon_\bosonB},\Psi^{\varepsilon_\bosonC}_{\Delta_\bosonC})$.

		\noindent {\bf Closeness of $\Psi^{(1)}$ and $\Psi^{(2)}$:}
		Arguing that~\hyperlink{eq:statetwo}{$\Psi^{(1)}$} is close to~$\Psi^{(2)}$ is the most technically challenging part of our proof. This is because it involves a detailed analysis of the action of the unitary~$\Uan$  
		when applied to states which are primarily supported around (certain) integers in position-space. To derive corresponding approximation results, we introduce a framework for approximate function evaluation by unitaries in Section~\ref{sec: approx function evaluation}.  Here we only sketch the idea. As argued 
		in Section~\ref{sec: explanation Ua}, the unitary~$\Uan$ acts as
		\begin{align}
			\Uan\left(\ket{x}\otimes\ket{y}\otimes\ket{u}\otimes\ket{0}\right)
			&~\propto~ \ket{x}\otimes \ket{y+f_{a,N,m}(x)}\otimes \ket{u} \otimes \ket{0}
		\end{align}
		for all $x\in \{0,\ldots,2R-1\}$, $y,u\in\mathbb{R}$,
		i.e., it shifts the position~$y$ of the mode~$\bosonB$ by the pseudomodular power~$f_{a,N,m}(x)$ of the position~$x$ of the mode $\bosonA$. 
		We argue that --- in a precise sense only involving $L^2$-functions (i.e., normalizable states) --- this is still approximately the case when the position $x\in\mathbb{R}$ of mode~$\bosonA$ belongs to the interval~$[-1/2,2R-1/2]$, and is additionally close to an integer. If these conditions are satisfied, the unitary~$U_{a,N,m}$ has the effect of (approximately) shifting the position of the mode~$\bosonB$ by~$f_{a,N,m}(\round{x})$, where $\round{x}\in\mathbb{Z}$ denotes the integer closest to the position~$x\in\mathbb{R}$ of the mode~$\bosonA$, see Lemma~\ref{lem: Ua shor circuit error} for a rigorous statement.
		
		We can use this to show that~$\Psi^{(1)}$ is close to~$\Psi^{(2)}$ as follows: First recall that the state~\hyperlink{eq:statetwo}{$\Psi^{(1)}$} is of the form~$\Psi^{(1)}=\Uan\Phi^{(1)}$
		where~$\Phi^{(1)}$ is the product state 
		given in Eq.~\eqref{eq:phioneproofsketchdefhelpful}. 
		In this product state,  the
		state $c_1\cdot \Pi_{[-1/2,2R-1/2]}e^{-iRP}\ket{\gkp^{\varepsilon_\bosonA}_{\kappa_\bosonA,\Delta_\bosonA}}$ on mode~$\bosonA$ is supported on~$[-1/2,2R-1/2]$ because of the projection~$\Pi_{[-1/2,2R-1/2]}$, and 
		additionally, every point~$x$ in the support is~$\varepsilon_\bosonA $-close to an integer (because of the definition of the truncated approximate GKP state~$\ket{\gkp^{\varepsilon_\bosonA }_{\kappa_\bosonA,\Delta_\bosonA }}$).
		Let $M \subseteq\mathbb{R}$, we define the set 
		\begin{align}
			M(\varepsilon) = M + [-\varepsilon,\varepsilon] = \{x+y\ |\ x\in M,\  |y|\le \varepsilon\}\ .  \label{eq: def M(varepsilon)}
		\end{align}
		In more detail, we have (using~$e^{-iRP}\ket{\chi^\varepsilon_\Delta(z)}=\ket{\chi^\varepsilon_\Delta(z+R)}$ and the fact that considered parameter~$R$ is an integer)
		\begin{align}
			\Pi_{[-1/2,2R-1/2]}e^{-iRP_\bosonA}\ket{\gkp^{\varepsilon_\bosonA }_{\kappa_\bosonA,\Delta_\bosonA }}
			&\propto\sum_{z\in\mathbb{Z}}\eta_{\kappa_\bosonA}(z)\Pi_{[-1/2,2R-1/2]}\ket{\chi^{\varepsilon_\bosonA }_{\Delta_\bosonA }(z+R)}\\
			&=\sum_{z\in\mathbb{Z}}\eta_{\kappa_\bosonA}(z-R)
			\Pi_{[-1/2,2R-1/2]}\ket{\chi^{\varepsilon_\bosonA }_{\Delta_\bosonA }(z)}\\
			&=\sum_{z=0}^{2R-1}\eta_{\kappa_\bosonA}(z-R) \ket{\chi^{\varepsilon_\bosonA }_{\Delta_\bosonA }(z)}\ ,\label{eq:hmve}
		\end{align}
		showing that this function has support contained in $\{0,\ldots,2R-1\}(\varepsilon_\bosonA)$. 
		Inserting expression~\eqref{eq:hmve} into the definition of~$\Phi^{(1)}$ gives
		\begin{align}
			\ket{\Phi^{(1)}}&\propto \sum_{z=0}^{2R-1}\sum_{y\in\mathbb{Z}}\eta_{\kappa_\bosonA}(z-R)\eta_{\kappa_\bosonB}(y)\ket{\chi^{\varepsilon_\bosonA}_{\Delta_\bosonA}(z)}\otimes M_N\ket{\chi^{\varepsilon_\bosonB}_{\Delta_\bosonB}(y)}\otimes\ket{\Psi^{\varepsilon_\bosonC}_{\Delta_\bosonC}}\otimes\ket{0}\ .\label{eq:psionelinearcombinationstatement}
		\end{align}
		Now consider a state of the form
		\begin{align}
			\ket{\chi^{\varepsilon_\bosonA}_{\Delta_{\boson_A}}(z)}\otimes M_N\ket{\chi^{\varepsilon_{\bosonB}}_{\Delta_\bosonB}(y)}\otimes\ket{\Psi^{\varepsilon_C}_{\Delta_\bosonC}}\otimes \ket{0}\qquad\textrm{ where }\qquad z\in \{0,\ldots,2R-1\}\textrm{ and }y\in\mathbb{Z}\ .
		\end{align}
		By definition, any position~$x\in\mathbb{R}$ occurring in the support of~$\chi^{\varepsilon_\bosonA}_{\Delta_\bosonA}(z)$ 
		is $\varepsilon_\bosonA$-close to the integer~$z=\round{x}$. 
		It follows from the discussion of~$\Ua$ above 
		that  for $z\in \{0,\ldots,2R-1\}$, we have 
		\begin{align}
			&U_{a,N,m}\left(\ket{\chi^{\varepsilon_\bosonA}_{\Delta_{\bosonA}}(z)}\otimes          M_N\ket{\chi^{\varepsilon_{\bosonB}}_{\Delta_\bosonB}(y)}\otimes\ket{\Psi^{\varepsilon_C}_{\Delta_\bosonC}}\otimes \ket{0}\right)\qquad\qquad\qquad \\
			&\qquad\quad \approx \ket{\chi^{\varepsilon_\bosonA}_{\Delta_\bosonA}(z)}\otimes e^{-if_{a,N,m}(z)P_\bosonB}M_N\ket{\chi^{\varepsilon_{\bosonB}}_{\Delta_{\bosonB}}(y)}\otimes\ket{\Psi_{\Delta_\bosonC}^{\varepsilon_{\bosonC}}}\otimes\ket{\Psi^{\varepsilon_\bosonC}_{\Delta_\bosonC}}\otimes \ket{0}\ ,
		\end{align}
		where we use the symbol~$\approx$ to indicate that these states are close 
		(see Lemma~\ref{lem: Ua shor circuit error} for a rigorous statement). 
		Applying this to each term in the superposition~\eqref{eq:psionelinearcombinationstatement}
		and using the pairwise orthogonality of the states~$\left\{\ket{\chi^{\varepsilon_\bosonA}_{\Delta_\bosonA}(z)}\right\}_{z\in\mathbb{Z}}$ 
		we conclude that
		$\Psi^{(1)}=\Uan\Phi^{(1)}$ is  $\varepsilon^{(2)}$-close to the \hypertarget{eq:statefour}{state}
		\begin{align}
			\ket{\Psi^{(2)}}= \cpsi{2}\sum_{z=0}^{2R-1} 
			\eta_{\kappa_\bosonA}(z-R)\ket{\chi^{\varepsilon_\bosonA }_{\Delta_\bosonA }(z)}\otimes 
			e^{-if_{a,N,m}(z)P_\bosonB}M_N\ket{\gkp^{\varepsilon_\bosonB }_{\kappa_\bosonB ,\Delta_\bosonB }}\otimes\ket{\Psi_{\Delta_\bosonC}^{\varepsilon_\bosonC}}\otimes \ket{0}\ ,
		\end{align}
		where an estimate of $\varepsilon^{(2)}$ is given in Lemma~\ref{lem: approx_unitary}. 
		
		\noindent {\bf Closeness of~$\Psi^{(2)}$ and $\Psi^{(3)}$:}
		The state $M_N\ket{\gkp^{\varepsilon_\bosonB }_{\kappa_\bosonB ,\Delta_\bosonB }}$ is approximately invariant under translations by (small) integer multiplies of~$N$. This implies that the state
		\begin{align}
			e^{-if_{a,N,m}(z)P}M_N\ket{\gkp^{\varepsilon_\bosonB }_{\kappa_\bosonB ,\Delta_\bosonB }}&\approx\  
			e^{-i(f_{a,N,m}(z)\! \xmod N)P_\bosonB}M_N\ket{\gkp^{\varepsilon_\bosonB }_{\kappa_\bosonB ,\Delta_\bosonB }}\label{eq:intermediatemxn}
		\end{align}
		(see Lemma~\ref{lem:approximateshiftinvarianceGKP}).
		Because the pseudomodular power satisfies
		\begin{align}
			f_{a,N,m}(z)&\equiv a^z\pmod N\qquad\textrm{ for all }\qquad z\in \{0,\ldots, 2R-1\}\ ,
		\end{align}
		Eq.~\eqref{eq:intermediatemxn}
		can be rewritten as 
		\begin{align}
			e^{-if_{a,N,m}(z)P_\bosonB}M_N\ket{\gkp^{\varepsilon_\bosonB }_{\kappa_\bosonB ,\Delta_\bosonB }}&\approx
			e^{-i(a^z\!\xmod N)P_\bosonB}M_N\ket{\gkp^{\varepsilon_\bosonB }_{\kappa_\bosonB ,\Delta_\bosonB }}\label{eq:intermediateuseful}
		\end{align}
		for all $z\in\{0,\ldots, 2R-1\}$. With Eq.~\eqref{eq:intermediateuseful} and using the pairwise orthogonality of the states~$\{\ket{\chi^{\varepsilon_\bosonA }_{\Delta_\bosonA }(z)}\}_{z\in\mathbb{Z}}$, we conclude that the state~\hyperlink{eq:statefour}{$\Psi^{(2)}$} is $\varepsilon^{(3)}$-close to the state~$\Psi^{(3)}$ defined \hypertarget{eq:statefive}{as}
		\begin{align}
			\cpsi{3}\sum_{z=0}^{2R-1} 
			\eta_{\kappa_\bosonA}(z-R)\ket{\chi^{\varepsilon_\bosonA }_{\Delta_\bosonA }(z)}\otimes 
			e^{-i(a^z\xmod N)P_\bosonB}M_N\ket{\gkp^{\varepsilon_\bosonB }_{\kappa_\bosonB ,\Delta_\bosonB }}\otimes \ket{0}\ ,
		\end{align}
		with an error~$\varepsilon^{(3)}$ bounded by Lemma~\ref{lem: untruncation_error_tail}. 
		
		\noindent {\bf Closeness of~$\Psi^{(3)}$ and $\Psi^{(4)}$:}
		The state~\hyperlink{eq:statefive}{$\Psi^{(3)}$} is  close to the state
		\begin{align}
			\ket{\Psi^{(4)}}= \cpsi{4}\sum_{z\in \mathbb{Z}} 
			\eta_{\kappa_\bosonA}(z-R)\ket{\chi^{\varepsilon_\bosonA}_{\Delta_\bosonA }(z)}\otimes 
			e^{-i(a^z\xmod N)P_\bosonB}M_N\ket{\gkp^{\varepsilon_\bosonB }_{\kappa_\bosonB ,\Delta_\bosonB }}\otimes\ket{\Psi_{\Delta_\bosonC}^{\varepsilon_\bosonC}}\otimes \ket{0}
		\end{align}
		up to an error~$\varepsilon^{(4)}$
		estimated in Lemma~\ref{lem:shift_inariance_error}. 
		This follows from the fact that for the chosen parameter~$\kappa_\bosonA$, the envelope~$z\mapsto \eta_{\kappa_\bosonA}(z-R)$ is sufficiently small for $z\not\in [0, 2R-1]$.
		
		\noindent {\bf Closeness of~$\Psi^{(4)}$ and $\Psi^{(5)}$:}
		The state~$\Psi^{(4)}$ is close to the state
		\begin{align}
			\ket{\Psi^{(5)}}=\cpsi{5}\sum_{z \in \mathbb{Z}}
			\eta_{\kappa_\bosonA}(z-R)\ket{\chi_{\Delta_\bosonA }(z)}\otimes 
			e^{-i(a^z\xmod N)P_\bosonB}M_N\ket{\gkp^{\varepsilon_\bosonB }_{\kappa_\bosonB ,\Delta_\bosonB }}\otimes\ket{\Psi_{\Delta_\bosonC}^{\varepsilon_\bosonC}}\otimes \ket{0}
		\end{align}
		up to an error~$\varepsilon^{(5)}$ (see Lemma~\ref{lem:untruncation_error_peak}). 
		Indeed, the states~$\Psi^{(4)}$ and~$\Psi^{(5)}$
		only differ by the fact that the Gaussians in mode~$\bosonA$ are truncated in the state~$\Psi^{(5)} = \Phi_a$. Furthermore, for $\varepsilon>0$ sufficiently small compared to~$\Delta$, a centered Gaussian~$\Psi_{\Delta}$ is close to
		its (normalized) truncation~$\Psi_\Delta^{\varepsilon}$ to the interval~$[-\varepsilon,\varepsilon]$, see Lemma~\ref{lem: gaussian gaussian epsilon}.

		This concludes the proof of Proposition~\ref{prop:propositionone}, except for 
		the task of establishing bounds~$\varepsilon^{(1)},\ldots,\varepsilon^{(5)}$ on the pairwise distances as in Eq.~\eqref{eq:errorboundpsij}. 
	\end{proof}
	
	\subsection{Distance bounds from properties of approximate GKP states}
	In this section, we derive upper bounds on the distances
	\begin{align}
		\|\Psi^{(0)}-\Psi^{(1)}\|_1,\  \|\Psi^{(2)}-\Psi^{(3)}\|_1,\ \|\Psi^{(3)}-\Psi^{(4)}\|_1\ \textrm{ and }\ \|\Psi^{(4)}-\Psi^{(5)}\|_1\ .
	\end{align}
	These are simple consequences of  basic properties of approximate GKP states (discussed in Section~\ref{sec: approximate gkp}). The remaining bound on $\|\Psi^{(1)}-\Psi^{(2)}\|_1$ will be derived in Section~\ref{sec:approx function evaluation implication for algorithm}.
	
	The first bound of this kind relies on the fact that
	the state~$\ket{\gkp_{\kappa,\Delta}}$ is close to its truncated version~$\ket{\gkp_{\kappa,\Delta}^\varepsilon}$ for suitably chosen parameters (see Lemma~\ref{lem: overlap truncated gkp}), and the fact that most of the support of the state~$\ket{\gkp^\varepsilon_{\kappa,\Delta}}$ is contained in the interval~$[-r,r]$ for suitably large~$r>0$ (see Lemma~\ref{lem:proj norm}). Together with the fact that the Gaussian state $\ket{\Psi_{\Delta}}$ is close to its truncated version $\ket{\Psi^\varepsilon_{\Delta}}$ for suitably chosen parameters (see Lemma~\ref{lem: gaussian gaussian epsilon}).
	
	\begin{lemma}\label{lem:truncationerr} Let $n\ge 4$.
		Let $(m,R,\kappa_\bosonA,\Delta_\bosonA,\kappa_\bosonB,\Delta_\bosonB, \Delta_\bosonC)$ be as specified in Table~\ref{tab: parameters}, and let  $\varepsilon_\bosonA=\sqrt{\Delta_\bosonA}$, $\varepsilon_\bosonB=\sqrt{\Delta_\bosonB}$ and $\varepsilon_\bosonC=\sqrt{\Delta_\bosonC}$. Consider 
		the two states
		\begin{align}
			\begin{array}{c@{\hskip 0.2cm}c@{\hskip 0.2cm}c@{\hskip 0.1cm}r@{\hskip 0.1cm}c@{\hskip 0.1cm}c@{\hskip 0.1cm}c@{\hskip 0.1cm}c@{\hskip 0.1cm}c@{\hskip 0.1cm}c@{\hskip 0.1cm}c}
				\ket{\Psi^{(0)}} &=&
				\Uan\big(& 
				e^{-iRP_\bosonA}\ket{\gkp_{\kappa_\bosonA,\Delta_\bosonA }} &\otimes& M_N\ket{\gkp_{\kappa_\bosonB ,\Delta_\bosonB }}
				&\otimes& \ket{\Psi_{\Delta_\bosonC}} 
				&\otimes& \ket{0}&\hspace{-1.8ex}\big)\notag\\
				\ket{\Psi^{(1)}}&=&
				\Uan\big(&
				\cpsi{1}\cdot\Pi_{[-1/2,2R-1/2]}e^{-iRP_\bosonA}\ket{\gkp^{\varepsilon_\bosonA }_{\kappa_\bosonA,\Delta_\bosonA }}
				&\otimes&M_N\ket{\gkp^{\varepsilon_\bosonB }_{\kappa_\bosonB ,\Delta_\bosonB }}
				&\otimes&\ket{\Psi_{\Delta_\bosonC}^{\varepsilon_\bosonC}} 
				&\otimes&\ket{0} 
				& \big)\ .\notag
			\end{array}
		\end{align}
		Then
		\begin{align}
			\left\|\Psi^{(0)}-\Psi^{(1)}\right\|_1
			&\le 9 \cdot 2^{-2n} =:\varepsilon^{(1)}\ .
		\end{align}
	\end{lemma}
	\begin{proof}
		We consider the two states
		\begin{align}
			\ket{\Phi^{(0)}} &= e^{-iRP_\bosonA}\ket{\gkp_{\kappa_\bosonA,\Delta_\bosonA }}\otimes M_N\ket{\gkp_{\kappa_\bosonB ,\Delta_\bosonB }}\otimes\ket{\Psi_{\Delta_\bosonC}}\otimes \ket{0}\\
			\ket{\widetilde{\Phi}^{(0)}} &= e^{-iRP_\bosonA}\ket{\gkp^{\varepsilon_\bosonA }_{\kappa_\bosonA,\Delta_\bosonA }}\otimes M_N\ket{\gkp^{\varepsilon_\bosonB }_{\kappa_\bosonB ,\Delta_\bosonB }}\otimes\ket{\Psi^{\varepsilon_\bosonC}_{\Delta_\bosonC}}\otimes \ket{0}\ .
		\end{align}
		Using the unitary of~$e^{-iRP}$ and $M_N$, we then have 
		\begin{align}
			\left|\left\langle \Phi^{(0)}, \widetilde{\Phi}^{(0)}\right\rangle \right|^2 
			&= \left|\left\langle \gkp_{\kappa_\bosonA,\Delta_\bosonA }, \gkp_{\kappa_\bosonA,\Delta_\bosonA }^{\varepsilon_\bosonA }\right\rangle \right|^2 \cdot \left|\left\langle \gkp_{\kappa_\bosonB ,\Delta_\bosonB }, \gkp_{\kappa_\bosonB ,\Delta_\bosonB }^{\varepsilon_\bosonB }\right\rangle \right|^2 \cdot \left|\left\langle \Psi_{\Delta_\bosonC},\Psi^{\varepsilon_\bosonC}_{\Delta_\bosonC}\right\rangle\right|^2 \\
			&\ge (1-9\Delta_\bosonA)(1-9\Delta_\bosonB)(1-2\Delta_\bosonC)\\
			&\ge 1-9(\Delta_\bosonA  + \Delta_\bosonB ) -2\Delta_\bosonC \ ,\label{eq:lowerboundoverlapphitildephiz}
		\end{align}
		where we used Lemma~\ref{lem: overlap truncated gkp} (giving a lower bound on the overlap~$|\langle \gkp_{\kappa,\Delta},\gkp^{\varepsilon}_{\kappa,\Delta}|^2$) twice
		and Lemma~\ref{lem: gaussian gaussian epsilon} (giving a lower bound on~$|\langle \Psi_\Delta,\Psi^\varepsilon_\Delta\rangle|^2$) to obtain the first inequality. Combining~\eqref{eq:lowerboundoverlapphitildephiz} with the relation
		\begin{align} \label{eq: trace distance - fidelity}
			\left\lVert \psi - \varphi \right\rVert_1 = 2 \sqrt{1 - \abs{\left\langle \psi, \varphi \right\rangle}^2}\qquad\textrm{ for any two pure states $\psi$  and $\varphi$}\ ,
		\end{align}        
		we obtain
		\begin{align}
			\left\|\Phi^{(0)}-\widetilde{\Phi}^{(0)}\right \|_1 
			&\le 2\sqrt{9(\Delta_\bosonA  + \Delta_\bosonB )+2\Delta_\bosonC}\\
			&\le 6\left(\sqrt{\Delta_\bosonA }+\sqrt{\Delta_\bosonB }\right) + 3\sqrt{\Delta_\bosonC}\ , \label{eq:phi0phitild1}
		\end{align}
		where we used $\sqrt{x+y}\le \sqrt{x}+\sqrt{y}$ for all $x,y\ge 0$.

		Now consider the state
		\begin{align}
			\ket{\Phi^{(1)}}&=\cpsi{1}\cdot\Pi_{[-1/2,2R-1/2]}e^{-iRP}\ket{\gkp^{\varepsilon_\bosonA }_{\kappa_\bosonA,\Delta_\bosonA }}
			\otimes M_N\ket{\gkp^{\varepsilon_\bosonB }_{\kappa_\bosonB ,\Delta_\bosonB }}
			\otimes\ket{\Psi_{\Delta_\bosonC}^{\varepsilon_\bosonC}} 
			\otimes\ket{0}\ .
		\end{align}
		Then
		\begin{align}
			\ket{\Phi^{(1)}}=\cpsi{1}\cdot\Pi_{[-1/2,2R-1/2]}\ket{\widetilde{\Phi}^{(0)}}\ 
		\end{align}
		by definition. It follows that 
		\begin{align}
			\left| \left\langle\Phi^{(1)},\widetilde{\Phi}^{(0)}\right\rangle \right|^2 
			&=\cpsi{1}^2\cdot \abs{\langle \gkp_{\kappa_\bosonA,\Delta_\bosonA }^{\varepsilon_\bosonA },e^{iRP}\Pi_{[-1/2,2R-1/2]} e^{-iRP}\ket{\gkp_{\kappa_\bosonA,\Delta_\bosonA }^{\varepsilon_\bosonA }}}^2\\
			&=\cpsi{1}^2 \cdot\abs{\langle\gkp_{\kappa_\bosonA,\Delta_\bosonA }^{\varepsilon_\bosonA },\Pi_{[-R-1/2, R-1/2]}\ket{\gkp_{\kappa_\bosonA,\Delta_\bosonA }^{\varepsilon_\bosonA }}}^2\\
			&= \|\Pi_{[-R-1/2,R-1/2]}\gkp_{\kappa_\bosonA,\Delta_\bosonA }^{\varepsilon_\bosonA }\|^2\ ,\\
			&\ge \|\Pi_{[-R/2,R/2]}\gkp_{\kappa_\bosonA,\Delta_\bosonA }^{\varepsilon_\bosonA }\|^2\\
			&\ge 
			1-2 e^{-\frac{1}{4}\kappa_\bosonA^2 R^2}\ ,
		\end{align}
		where we used that
		\begin{align}
			e^{iRP}\Pi_{[-1/2,2R-1/2]} e^{-iRP}=\Pi_{[-R-1/2, R-1/2]} \ .
		\end{align}
		in the second step, the fact that 
		\begin{align}
			\cpsi{1}^2=\|\Pi_{[-1/2,2R-1/2]}e^{-iRP} \gkp_{\kappa_\bosonA,\Delta_\bosonA }^{\varepsilon_\bosonA }\|^{-2}=\|\Pi_{[-R-1/2,R-1/2]} \gkp_{\kappa_\bosonA,\Delta_\bosonA }^{\varepsilon_\bosonA }\|^{-2}\ .
		\end{align}
		in the third identity, the inclusion~$[-R/2,R/2]\subseteq[-R-1/2,R-1/2]$ to obtain the first inequality, and Lemma~\ref{lem:proj norm}
		(giving a lower bound on $\|\Pi_[-r,r]\gkp_{\kappa,\Delta}^\varepsilon\|^2$) to obtain the second inequality.
		With the relation~\eqref{eq: trace distance - fidelity} we conclude that
		\begin{align}
			\left\|\Phi^{(1)}-\widetilde{\Phi}^{(0)}\right\|_1&\le   2 \cdot \sqrt{2} e^{-\frac{1}  {8}\kappa_\bosonA^2 R^2} \le 3 e^{-\frac{1}{8}\kappa_\bosonA^2 R^2} \ .\label{eq:threthreetwo}
		\end{align}
		Combining~\eqref{eq:phi0phitild1} and~\eqref{eq:threthreetwo}
		with the triangle inequality yields
		\begin{align}
			\left\|\Psi^{(0)}-\Psi^{(1)}\right\|_1
			&\le  6\left(\sqrt{\Delta_\bosonA }+\sqrt{\Delta_\bosonB }\right) + 3\sqrt{\Delta_\bosonC} + 3e^{-\frac{1}{8}\kappa_\bosonA^2 R^2}\\
			&=6\left(2^{-8n}+2^{-9n^2}\right)+3\cdot 2^{-25n}+3e^{-\frac{1}{32}2^{2n}}\\
			&\le 6\cdot 2^{-2n}+2^{-2n} + 2^{-2n}+ 2^{-2n}\\
			&\le 9 \cdot 2^{-2n}
		\end{align}
		where we used the parameters
		$(\kappa_\bosonA,\Delta_\bosonA,\kappa_\bosonB,\Delta_\bosonB, \Delta_\bosonC, R)$ from  Table~\ref{tab: parameters}
		and the assumption that $n\ge 4$. The claim follows from this since we have 
		$\Psi^{(b)}=U_{a,N,m}\Phi^{(b)}$ for $b=0,1$ by definition of these states and because the $L^1$-distance is invariant under unitaries.
	\end{proof}

	The following lemma is a consequence of the fact that the state~$\ket{\gkp^\varepsilon_{\kappa,\Delta}}$
	is approximately invariant under small integer shifts, see Lemma~\ref{lem:approximateshiftinvarianceGKP} for a detailed statement. 
	\begin{lemma} \label{lem: untruncation_error_tail}
		Assume $n\ge 3$. Let
		$(m,R,\kappa_\bosonA,\Delta_\bosonA,\kappa_\bosonB,\Delta_\bosonB, \Delta_\bosonC)$ be as specified in Table~\ref{tab: parameters}. Let $\varepsilon_\bosonA=\sqrt{\Delta_\bosonA}$, $\varepsilon_\bosonB=\sqrt{\Delta_\bosonB}$ and $\varepsilon_\bosonC=\sqrt{\Delta_\bosonC}$.
		Consider the two states
		\begin{align*}
			\begin{array}{c@{\hskip 0.2cm}c@{\hskip 0.2cm}c@{\hskip 0.2cm}c@{\hskip 0.2cm}c@{\hskip 0.2cm}c@{\hskip 0.2cm}c@{\hskip 0.2cm}c@{\hskip 0.2cm}c@{\hskip 0.2cm}c@{\hskip 0.2cm}cc}
				\ket{\Psi^{(2)}} &=& \cpsi{2}\sum_{z=0}^{2R-1} 
				\eta_{\kappa_\bosonA}(z-R)\ket{\chi^{\varepsilon_\bosonA }_{\Delta_\bosonA }(z)}
				&\otimes& \hspace{3ex}e^{-if_{a,N,m}(z)P_\bosonB}M_N\ket{\gkp^{\varepsilon_\bosonB }_{\kappa_\bosonB ,\Delta_\bosonB }}&\otimes&\!\!\ket{\Psi_{\Delta_\bosonC}^{\varepsilon_\bosonC}}&\otimes&\ket{0}& ,\\
				\ket{\Psi^{(3)}}&=& \cpsi{3}\sum_{z=0}^{2R-1}
				\eta_{\kappa_\bosonA}(z-R)\ket{\chi^{\varepsilon_\bosonA }_{\Delta_\bosonA }(z)}\!\!&\otimes&\!\! e^{-i(a^z \xmod{N})P_\bosonB} 
				M_N\ket{\gkp^{\varepsilon_\bosonB }_{\kappa_\bosonB ,\Delta_\bosonB }} &\otimes& \ket{\Psi_{\Delta_\bosonC}^{\varepsilon_\bosonC}}&\otimes&\ket{0}& ,
			\end{array}
		\end{align*}
		where $f_{a,N,m}(z)$ is the pseudomodular power (cf. Eq.~\eqref{eq:pseudomodularpowerfirstdef}).
		Then 
		\begin{align}
			\left\|\Psi^{(2)}-\Psi^{(3)}\right\|_1
			\le 2^{-2n}
			=:\varepsilon^{(3)}\ .
		\end{align}
	\end{lemma}
	\begin{proof}
		By the pairwise orthogonality of the states $\{\chi^{\varepsilon_\bosonA }_{\Delta_\bosonA }(z)\}_{z\in\mathbb{Z}}$, we have
		\begin{align}
			\left\langle \Psi^{(2)}, \Psi^{(3)} \right\rangle 
			&= \cpsi{2} \cpsi{3}\!\!\sum_{z=0}^{2R-1} \!\eta_{\kappa_\bosonA}(z-R)^2\! \left\langle M_N \gkp^{\varepsilon_\bosonB }_{\kappa_\bosonB ,\Delta_\bosonB }, e^{i(f_{a,N,m}(z) - (a^z\!\xmod{N}))P_\bosonB} M_N \gkp^{\varepsilon_\bosonB }_{\kappa_\bosonB ,\Delta_\bosonB } \right\rangle\\
			&= \cpsi{2} \cpsi{3}\sum_{z=0}^{2R-1} \eta_{\kappa_\bosonA}(z-R)^2
			\langle \gkp^{\varepsilon_\bosonB }_{\kappa_\bosonB ,\Delta_\bosonB },
			e^{id(z)P_\bosonB}\gkp^{\varepsilon_\bosonB }_{\kappa_\bosonB ,\Delta_\bosonB }\rangle\label{eq:inserthereeqm}
		\end{align}
		where we used that $M_N^\dagger PM_N=P/N$ and introduced the abbreviation
		\begin{align}
			d(z) = \frac{f_{a,N,m}(z) - (a^z \xmod{N})}{N}\ .
		\end{align}
		The following properties of~$d(z)$ will be important for our analysis.
		\begin{claim}
			Let $z\in\{0,\ldots,2R-1\}$ be arbitrary. Then 
			\begin{enumerate}[(i)]
				\item \label{it:dzfirststatementxm}
				$d(z)\in\mathbb{Z}$, i.e., $d(z)$ is an integer and
				\item $|d(z)|\leq N^m$.\label{it:dzupperboundfmnm}
			\end{enumerate}
		\end{claim}
		\begin{proof}
			Statement~\eqref{it:dzfirststatementxm} follows from the fact that
			the pseudomodular power satisfies the relation
			\begin{align}
				f_{a,N,m}(z)\equiv a^z \xmod N\qquad\textrm{ for }\qquad z\in\{0,\ldots, 2^{m}-1\}\  .
			\end{align} 
			For our choice of parameters~$(m,R)$ (see Table~\ref{tab: parameters}), we have the relation $2^{m}-1=2R-1$, implying that 
			\begin{align}
				f_{a,N,m}(z)\equiv a^z \xmod N\qquad\textrm{ for }\qquad z\in\{0,\ldots, 2R-1\}\  .
			\end{align} 
			This shows that $f_{a,N,m}(z)-a^z$ is an integer multiple of~$N$ proving~\eqref{it:dzfirststatementxm}.
			
			To prove~\eqref{it:dzfirststatementxm}, recall that for any integer~$z\in\mathbb{N}_0$, the pseudomodular power $f_{a,N,m}(z)$ (see Eq.~\eqref{eq:pseudomodularpowerfirstdef}) is a 
			product of $m$~factors, each belonging to the set~$\mathbb{Z}_N=\{0,\ldots,N-1\}$. Thus 
			\begin{align}
				f_{a,N,m}(z)\leq N^{m}\qquad\textrm{ for every }\qquad z\in \mathbb{N}_0\ .
			\end{align} 
			In particular, we have 
			\begin{align}
				d(z)\leq \left|\frac{f_{a,N,m}(z)}{N}\right| + \left|a^z\xmod N\right|\leq \frac{1}{N}(N^{m}+N)\leq N^m\ 
			\end{align}
			as claimed.
		\end{proof}
		It follows from  property~\eqref{it:dzfirststatementxm} and Lemma~\ref{lem:approximateshiftinvarianceGKP} (giving a lower bound on~$\langle \gkp^\varepsilon_{\kappa,\Delta},e^{idP}\gkp^\varepsilon_{\kappa,\Delta}\rangle$ for integer~$d$) 
		that
		\begin{align}
			\langle \gkp^{\varepsilon_\bosonB }_{\kappa_\bosonB ,\Delta_\bosonB },
			e^{id(z)P_\bosonB}\gkp^{\varepsilon_\bosonB }_{\kappa_\bosonB ,\Delta_\bosonB }\rangle
			&\geq 1-\frac{\kappa_\bosonB ^2}{2} d(z)^2\\
			&\geq 1-\frac{\kappa_\bosonB ^2}{2} N^{2m}\qquad\textrm{ for all }z\in \{0,\ldots,2R-1\}\ .\label{eq:zlowerboundoverlapmv}
		\end{align}
		where we used property~\eqref{it:dzupperboundfmnm} to obtain the second inequality.

		Applying inequality~\eqref{eq:zlowerboundoverlapmv} to every summand in Eq.~\eqref{eq:inserthereeqm} gives 
		\begin{align}
			\left\langle \Psi^{(2)}, \Psi^{(3)} \right\rangle 
			&\geq      \cpsi{2} \cpsi{3}\left(1-\textfrac{\kappa_\bosonB ^2 N^{2m}}{2}\right)\cdot \sum_{z=0}^{2R-1}\eta_{\kappa_\bosonA}(z-R)^2\ .
		\end{align}
		Note that  by the pairwise orthogonality of the states $\{\chi^{\varepsilon_\bosonA }_{\Delta_\bosonA }(z)\}_{z\in\mathbb{Z}}$ and the definition of the states~$\Psi^{(2)}$ and~$\Psi^{(3)}$, we have     \begin{align}\cpsi{2}^{-2}=\cpsi{3}^{-2}&=\sum_{z=0}^{2R-1}\eta_{\kappa_\bosonA}(z-R)^2\ .
		\end{align}
		Thus
		\begin{align}
			\left\langle \Psi^{(2)}, \Psi^{(3)} \right\rangle
			&\ge 1-\kappa_\bosonB ^2N^{2m}/2 \ ,
			\label{eq: Psi2Psi3 overlap one}
		\end{align}
		The inequality $(1-x)^2\ge 1-2x$ for all $x\in\mathbb{R}$ and 
		the choice of $(m,\kappa_\bosonB)$ from Table~\ref{tab: parameters} give 
		\begin{align}
			\left|\left\langle \Psi^{(2)}, \Psi^{(3)} \right\rangle\right|^2
			&\ge 1-\kappa_\bosonB ^2 N^{2m}\\
			&\ge 1-2^{-2n^2} \ .
		\end{align}
		Hence  the relation between the overlap and the $L^1$-distance (cf.\ Eq.~\eqref{eq: trace distance - fidelity}) implies that 
		\begin{align}
			\left\|\Psi^{(2)}-\Psi^{(3)}\right\|_1
			&\le 2\cdot 2^{-n^2}\\
			&\le 2^{-2n} \qquad\textrm{for }n \ge 3\ . 
		\end{align}
	\end{proof}
	
	Another application of the fact that~$\gkp_{\kappa,\Delta}^\varepsilon\in L^2(\mathbb{R})$
	has most of its support in an interval of the form~$[-r,r]$ (with suitably chosen~$r$, see  Lemma~\ref{lem:proj norm}) is the following statement.
	\begin{lemma}\label{lem:shift_inariance_error}
		Let $n\ge 4$. 
		Let $(m,R,\kappa_\bosonA,\Delta_\bosonA,\kappa_\bosonB,\Delta_\bosonB, \Delta_\bosonC)$ be as specified in Table~\ref{tab: parameters}. Set  $\varepsilon_\bosonA=\sqrt{\Delta_\bosonA}$, $\varepsilon_\bosonB=\sqrt{\Delta_\bosonB}$ and $\varepsilon_\bosonC=\sqrt{\Delta_\bosonC}$.
		Consider the two states
		\begin{align}
			\begin{array}{c@{\hskip 0.2cm}c@{\hskip 0.2cm}c@{\hskip 0.2cm}c@{\hskip 0.2cm}c@{\hskip 0.15cm}c@{\hskip 0.15cm}c@{\hskip 0.15cm}c@{\hskip 0.15cm}c@{\hskip 0.15cm}c@{\hskip 0.15cm}c}
				\ket{\Psi^{(3)}}&=& \cpsi{3}\sum_{z=0}^{2R-1}
				\eta_{\kappa_\bosonA}(z-R)\ket{\chi^{\varepsilon_\bosonA }_{\Delta_\bosonA }(z)}&\otimes & e^{-i(a^z \xmod{N})P_\bosonB} 
				M_N\ket{\gkp^{\varepsilon_\bosonB }_{\kappa_\bosonB ,\Delta_\bosonB }} &\otimes& \ket{\Psi_{\Delta_\bosonC}^{\varepsilon_\bosonC}} &\otimes& \ket{0}\notag\\
				\ket{\Psi^{(4)}}&=& \cpsi{4}\sum_{z\in \mathbb{Z}}
				\eta_{\kappa_\bosonA}(z-R)\ket{\chi^{\varepsilon_\bosonA }_{\Delta_\bosonA }(z)}&\otimes & e^{-i(a^z\xmod{N})P_\bosonB} 
				M_N\ket{\gkp^{\varepsilon_\bosonB }_{\kappa_\bosonB ,\Delta_\bosonB }} &\otimes& \ket{\Psi_{\Delta_\bosonC}^{\varepsilon_\bosonC}} &\otimes& \ket{0}
			\end{array}\ .\notag
		\end{align}
		Then 
		\begin{align}
			\left\|\Psi^{(3)}-\Psi^{(4)}\right\|_1
			&\le 2^{-n}
			=:\varepsilon^{(4)}\ .
		\end{align}
	\end{lemma}
	\begin{proof}
		Due to the pairwise orthogonality of the  states~$\{\chi^{\varepsilon_\bosonA}_{\Delta_\bosonA}(z)\}_{z\in\mathbb{Z}}$ we have
		\begin{align}
			c_3^{-2}&=\sum_{z=0}^{2R-1}\eta_{\kappa_\bosonA}(z-R)^2\\
			c_4^{-2}&=\sum_{z\in\mathbb{Z}}\eta_{\kappa_\bosonA}(z-R)^2\ .
		\end{align}
		It follows (again using orthogonality) that 
		\begin{align}
			|\langle \Psi^{(3)},\Psi^{(4)}\rangle|^2 &=
			c_3^2c_4^2\left(\sum_{z=0}^{2R-1}
			\eta_{\kappa_\bosonA}(z-R)^2\right)^2\\
			&=
			\frac{\sum_{z=0}^{2R-1}\eta_{\kappa_\bosonA}(z-R)^2}{\sum_{z\in\mathbb{Z}}\eta_{\kappa_\bosonA}(z-R)^2}
			\\
			&=\frac{\sum_{z=-R}^{R-1}\eta_{\kappa_\bosonA}(z)^2}{\sum_{z\in\mathbb{Z}}\eta_{\kappa_\bosonA}(z)^2}\ .\label{eq:rzrmnvd}
		\end{align}
		Recalling the definition
		\begin{align}
			\ket{\gkp^{\varepsilon_\bosonA}_{\kappa_\bosonA,\Delta_\bosonA}}&=C_{\kappa_\bosonA,\Delta_\bosonA}\sum_{z\in \mathbb{Z}}\eta_{\kappa_A}(z)\ket{\chi^{\varepsilon_\bosonA}_{\Delta_\bosonA}(z)}
		\end{align}
		and observing that $\chi^{\varepsilon_\bosonA}_{\Delta_\bosonA}(z)$ has support contained in $(z-1/2,z+1/2)$, we can write~\eqref{eq:rzrmnvd} as 
		\begin{align}
			|\langle \Psi^{(3)},\Psi^{(4)}\rangle|^2 &=\left\| \Pi_{[-R-1/2,R-1/2]}\gkp_{\kappa_\bosonA, \Delta_\bosonA }^{\varepsilon_\bosonA }\right\|^2\\
			&\ge 
			\|\Pi_{[-R/2,R/2]}\gkp_{\kappa_\bosonA, \Delta_\bosonA }^{\varepsilon_\bosonA }\|^2 
		\end{align}
		where we used that $[-R/2,R/2]\subseteq[-R-1/2,R-1/2]$. With Lemma~\ref{lem:proj norm} (bounding the norm of $\Pi_{[-r,r]}\gkp^\varepsilon_{\kappa,\Delta}$) it follows that 
		\begin{align}
			\left|\left\langle \Psi^{(3)}, \Psi^{(4)} \right\rangle \right|^2 &\ge 
			1-2 e^{-\frac{1}{4}\kappa_\bosonA^2 R^2}
			\ .
		\end{align}
		Hence by the relation~\eqref{eq: trace distance - fidelity} between the overlap and the $L^1$-norm distance it follows that 
		\begin{align}
			\left\|\Psi^{(3)}-\Psi^{(4)}\right\|_1&\le 2\cdot \sqrt{2} e^{-\frac{1}{8}\kappa_\bosonA^2 R^2}\le 3 e^{-\frac{1}{8}\kappa_\bosonA^2 R^2} = 3e^{-\frac{1}{32} 2^{2n}} \le  2^{-2n}\ 
		\end{align} for 
		the choice of $(\kappa_\bosonA,R)$ as in Table~\ref{tab: parameters} and
		$n\ge 4$, as claimed.
	\end{proof}
	
	The following result only relies on the fact that
	the normalization constants~$C_{\kappa,\Delta}$ and $C_{\kappa}$ in the definition of~$\ket{\gkp_{\kappa,\Delta}}$ respectively~$\ket{\gkp_{\kappa,\Delta}^\varepsilon}$ are close to each other (see Lemma~\ref{lem: norm and ratio norm gkp gkp ep}). 
	
	\begin{lemma}\label{lem:untruncation_error_peak}
		Let $n\ge 1$.
		Let $(m,R,\kappa_\bosonA,\Delta_\bosonA,\kappa_\bosonB,\Delta_\bosonB, \Delta_\bosonC)$ be as specified in Table~\ref{tab: parameters}. Set $\varepsilon_\bosonA=\sqrt{\Delta_\bosonA}$, $\varepsilon_\bosonB=\sqrt{\Delta_\bosonB}$, $\varepsilon_\bosonC=\sqrt{\Delta_\bosonC}$.
		Consider the two states
		\begin{align}
			\begin{array}{c@{\hskip 0.2cm}c@{\hskip 0.2cm}c@{\hskip 0.2cm}c@{\hskip 0.2cm}c@{\hskip 0.15cm}c@{\hskip 0.15cm}c@{\hskip 0.15cm}c@{\hskip 0.15cm}c@{\hskip 0.15cm}c@{\hskip 0.15cm}c}
				\ket{\Psi^{(4)}}&=& \cpsi{4}\sum_{z\in \mathbb{Z}}
				\eta_{\kappa_\bosonA}(z-R)\ket{\chi^{\varepsilon_\bosonA }_{\Delta_\bosonA }(z)}&\otimes & e^{-i(a^z\xmod{N})P_\bosonB} 
				M_N\ket{\gkp^{\varepsilon_\bosonB }_{\kappa_\bosonB ,\Delta_\bosonB }} &\otimes& \ket{\Psi_{\Delta_\bosonC}^{\varepsilon_\bosonC}} &\otimes& \ket{0}& ,\\
				\ket{\Psi^{(5)}}&=&  \cpsi{5}\sum_{z\in \mathbb{Z}}
				\eta_{\kappa_\bosonA}(z-R)\ket{\chi_{\Delta_\bosonA }(z)}&\otimes & e^{-i(a^z\xmod{N})P_\bosonB} 
				M_N\ket{\gkp^{\varepsilon_\bosonB }_{\kappa_\bosonB ,\Delta_\bosonB }}  &\otimes& \ket{\Psi_{\Delta_\bosonC}^{\varepsilon_\bosonC}} &\otimes& \ket{0}& .
			\end{array}\notag
		\end{align}
		Then
		\begin{align}
			\left\|\Psi^{(4)}-\Psi^{(5)}\right\|_1&\le 2^{-2n} =:\varepsilon^{(5)} \ .
		\end{align}
	\end{lemma}
	\begin{proof}
		For brevity, let us introduce the states
		\begin{align}
			\ket{\theta(z)}&=e^{-i(a^z\xmod{N})P_\bosonB} 
			M_N\ket{\gkp^{\varepsilon_\bosonB }_{\kappa_\bosonB ,\Delta_\bosonB }} \otimes \ket{\Psi_{\Delta_\bosonC}^{\varepsilon_\bosonC}}\qquad\textrm{ for }\qquad z\in\mathbb{Z}\ .
		\end{align}
		It is easy to check that
		\begin{align}
			0\leq \langle \theta(z),\theta(z')\rangle &\leq 1\qquad\textrm{ for all }\qquad z,z'\in \mathbb{Z}\ .\label{eq:thetaoverlapnonnega}
		\end{align}
		(In particular, these inner products are real-valued.) 
		In terms of these states, we have 
		\begin{align}
			\ket{\Psi^{(4)}}&=c_4\sum_{z\in\mathbb{Z}}\eta_{\kappa_\bosonA}(z-R)\ket{\chi^{\varepsilon_\bosonA}_{\Delta_\bosonA}(z)}\otimes\ket{\theta(z)}\\
			\ket{\Psi^{(5)}}&=c_5\sum_{z\in\mathbb{Z}}\eta_{\kappa_\bosonA}(z-R)\ket{\chi_{\Delta_\bosonA}(z)}\otimes\ket{\theta(z)}\ .
		\end{align}
		Thus
		\begin{align}
			\langle \Psi^{(4)},\Psi^{(5)}\rangle &=c_4c_5\sum_{z,z'\in\mathbb{Z}}
			\eta_{\kappa_\bosonA}(z-R)\eta_{\kappa_\bosonA}(z'-R) \langle \chi^{\varepsilon_\bosonA}_{\Delta_\bosonA}(z),\chi_{\Delta_\bosonA}(z')\rangle \cdot \langle \theta(z),\theta(z')\rangle\\
			&\geq c_4c_5\sum_{z\in\mathbb{Z}}
			\eta_{\kappa_\bosonA}(z-R)^2 \langle \chi^{\varepsilon_\bosonA}_{\Delta_\bosonA}(z),\chi_{\Delta_\bosonA}(z)\rangle \label{eq:psiffl}
		\end{align}
		where we used~\eqref{eq:thetaoverlapnonnega}   and the fact that~
		\begin{align}
			\langle \chi^{\varepsilon_\bosonA}_{\Delta_\bosonA}(z),\chi^{\varepsilon_\bosonA}_{\Delta_\bosonA}(z')\rangle\geq 0\qquad\textrm{ for all}\qquad  z,z'\in\mathbb{Z}\ . \label{eq:positivmiym}
		\end{align}
		
		By the pairwise orthogonality of the states~$\left\{\ket{\chi^{\varepsilon_\bosonA}_{\Delta_\bosonA}(z)}\right\}_{z\in\mathbb{Z}}$, we have 
		\begin{align}
			c_4^{-2}&=\sum_{z\in\mathbb{Z}} \eta_{\kappa_\bosonA}(z-R)^2=\sum_{z\in\mathbb{Z}} \eta_{\kappa_\bosonA}(z)^2
		\end{align}
		and we conclude that 
		\begin{align}
			c_4=C_{\kappa_\bosonA}\label{eq:cvimamv}
		\end{align} is equal to the normalization constant of the state~$\gkp^{\varepsilon_\bosonA}_{\kappa_\bosonA,\Delta_\bosonA}$. On the other hand, we have
		\begin{align}
			c_5^{-2}&=\sum_{z,z'\in\mathbb{Z}}\eta_{\kappa_\bosonA}(z-R)\eta_{\kappa_\bosonA}(z'-R) \langle \chi_{\Delta_\bosonA}(z),\chi_{\Delta_\bosonA}(z')\rangle\cdot \langle \theta(z),\theta(z')\rangle\\
			&\leq \sum_{z,z'\in\mathbb{Z}}\eta_{\kappa_\bosonA}(z-R)\eta_{\kappa_\bosonA}(z'-R) \langle \chi_{\Delta_\bosonA}(z),\chi_{\Delta_\bosonA}(z')\rangle\\
			&= \sum_{z,z'\in\mathbb{Z}}\eta_{\kappa_\bosonA}(z)\eta_{\kappa_\bosonA}(z') \langle e^{iRP_\bosonA}\chi_{\Delta_\bosonA}(z),e^{iRP_\bosonA}\chi_{\Delta_\bosonA}(z')\rangle\\
			&= \sum_{z,z'\in\mathbb{Z}}\eta_{\kappa_\bosonA}(z)\eta_{\kappa_\bosonA}(z') \langle \chi_{\Delta_\bosonA}(z),\chi_{\Delta_\bosonA}(z')\rangle=C_{\kappa_\bosonA,\Delta_\bosonA}^{-2}\ ,
		\end{align}
		where the inequality follows from~\eqref{eq:thetaoverlapnonnega} and~\eqref{eq:positivmiym}, and where we used the fact that~$e^{iRP}$ is unitary.  We conclude that $c_5$ is related to the normalization constant~$C_{\kappa_\bosonA,\Delta_\bosonA}$ of the state~$\gkp_{\kappa_\bosonA,\Delta_\bosonA}$ by 
		\begin{align}
			c_5&\geq C_{\kappa_\bosonA,\Delta_\bosonA}\ .\label{eq:cfiveem}
		\end{align}
		Combining Eqs.~\eqref{eq:cvimamv},~\eqref{eq:cfiveem} and~\eqref{eq:psiffl} gives
		\begin{align}
			\langle \Psi^{(4)},\Psi^{(5)}\rangle &\geq C_{\kappa_\bosonA}C_{\kappa_\bosonA,\Delta_\bosonA}\sum_{z\in\mathbb{Z}}\eta_{\kappa_\bosonA}(z-R)^2 \langle \chi^{\varepsilon_\bosonA}_{\Delta_\bosonA}(z),\chi_{\bosonA}(z)\rangle\\
			&=C_{\kappa_\bosonA}C_{\kappa_\bosonA,\Delta_\bosonA}\sum_{z\in\mathbb{Z}}\eta_{\kappa_\bosonA}(z)^2 \langle \chi^{\varepsilon_\bosonA}_{\Delta_\bosonA}(z+R),\chi_{\bosonA}(z+R)\rangle\\
			&=C_{\kappa_\bosonA}C_{\kappa_\bosonA,\Delta_\bosonA}\sum_{z\in\mathbb{Z}}\eta_{\kappa_\bosonA}(z)^2 \langle \chi^{\varepsilon_\bosonA}_{\Delta_\bosonA}(z),\chi_{\bosonA}(z)\rangle\ ,
		\end{align}
		where we used  that 
		$\ket{\chi^{\varepsilon_\bosonA}_{\Delta_\bosonA}(z+R)}=e^{iRP_\bosonA}|\chi^{\varepsilon_\bosonA}_{\Delta_\bosonA}(z)\rangle$
		and similarly $\ket{\chi_{\Delta_\bosonA}(z+R)}=e^{iRP_\bosonA}|\chi_{\Delta_\bosonA}(z)\rangle$, as well as the fact that $e^{iRP}$ is unitary. Using that 
		~$\langle \chi^{\varepsilon_\bosonA}_{\Delta_\bosonA}(z),\chi_{\bosonA}(z)\rangle=\langle \Psi^{\varepsilon_\bosonA}_{\Delta_\bosonA},\Psi_{\Delta_\bosonA}\rangle\geq \sqrt{1-2\Delta_\bosonA}$ again by the unitarity of translations, as well as Lemma~\ref{lem: gaussian gaussian epsilon}, we conclude that
		\begin{align}
			\langle \Psi^{(4)},\Psi^{(5)}\rangle &\geq \sqrt{1-2\Delta_\bosonA}\cdot C_{\kappa_\bosonA}C_{\kappa_\bosonA,\Delta_\bosonA}\sum_{z\in\mathbb{Z}}\eta_{\kappa_\bosonA}(z)^2\\
			&=\sqrt{1-2\Delta_\bosonA}\cdot \frac{C_{\kappa_\bosonA,\Delta_\bosonA}}{C_{\kappa_\bosonA}}
		\end{align}
		by definition of~$C_{\kappa_\bosonA}$. 
		Since the choice of parameters~$(\kappa_\bosonA,\Delta_\bosonA)$ (see Table~\ref{tab: parameters}) ensures that $\kappa_\bosonA,\Delta_\bosonA\in (0,1/4)$ we can apply  Lemma~\ref{lem: norm and ratio norm gkp gkp ep}
		to bound the fraction~$\frac{C_{\kappa_\bosonA,\Delta_\bosonA}}{C_{\kappa_\bosonA}}$ of normalization constants, obtaining
		\begin{align}
			\abs{\left\langle \Psi^{(4)}, \Psi^{(5)} \right\rangle}^2 
			&\ge (1-2\Delta_\bosonA )\frac{\left(C_{\kappa_\bosonA,\Delta_\bosonA }\right)^2}{\left(C_{\kappa_\bosonA}\right)^2}\\
			&\ge (1-2\Delta_\bosonA )(1-7\Delta_\bosonA ) \\
			&\ge1-9\Delta_\bosonA \ ,
		\end{align}
		where we used the inequality $(1-x)(1-y)\ge 1-x-y$ for all $x,y>0$. The relation~\eqref{eq: trace distance - fidelity} relating the $L^1$-distance and the overlap therefore implies that 
		\begin{align}
			\left\|\Psi^{(4)}-\Psi^{(5)}\right\|_1&\le 6\sqrt{\Delta_\bosonA } = 6\cdot 2^{-8n}\le 2^{-2n}\ ,
		\end{align}
		where we used the choice of parameter $\Delta_\bosonA$ from Table~\ref{tab: parameters} and the assumption $n\ge1$. This is the claim.
	\end{proof}

	It remains to establish an upper bound on the distance~$\left\|\Psi^{(3)}-\Psi^{(4)}\right\|_1$. This requires an analysis of the unitary~$\Ua$, which we provide in the next section.

	\clearpage
	\section{Approximate function evaluation for CV systems}\label{sec: approx function evaluation}
	The goal of this section is to analyze 
	the action of the unitary~$\Ua$ (defined in Eq.~\eqref{eq: UaNm gates})
	on states that have support on non-integer positions.  
	We will show that --- in a sense made precise in this section --- this unitary 
	still (approximately) computes the pseudomodular power even if the positions of modes~$\bosonA$ and~$\bosonC$ are only close to integers. The corresponding statement is given in Lemma~\ref{lem: Ua shor circuit error}. 
	Its formulation relies on a framework we introduce to formalize the notion of a unitary on a CV system approximately computing a function. In particular, we establish continuity bounds describing, e.g., what our LSB-gate does when applied to (position-encoded) reals that are only close to integers.
	
	In the following, we are interested in unitaries acting on $L^2(\cX)$ where $\cX$ is a measure space. In particular, we consider the case $\cX = \mathbb{R}^m \times \{0,1\}$ 
	(note that $L^2(\mathbb{R}^m)\cong L^2(\mathbb{R})^{\otimes m}$) 
	and  use that 
	\begin{align}
		L^2(\mathbb{R}^m\times\{0,1\}) \cong L^2(\mathbb{R}^m) \otimes \mathbb{C}^2\cong  L^2(\mathbb{R}^m,\mathbb{C}^2)\ .
	\end{align} More concretely, the isomorphism between the first two of these spaces is characterized by the following construction, using an orthonormal basis~$\{\ket{0},\ket{1}\}$ of~$\mathbb{C}^2$. A function $\Psi \in L ^2(\mathbb{R}^m\times\{0,1\})$ is mapped to~$\ket{\Psi_0}\otimes\ket{0}+\ket{\Psi_1}\otimes\ket{1}$, equivalently referred to as a {\em spinor} $(\Psi_0, \Psi_1)$, where $\Psi_i = \Psi(\cdot, i) \in L^2(\mathbb{R}^m) $ for $i\in\{0,1\}$. For two spinors $(\Psi_0, \Psi_1)$, $(\Phi_0, \Phi_1)$  the inner product on $L^2(\mathbb{R}^m) \otimes \mathbb{C}^2$ is then given by $\langle(\Psi_0, \Psi_1),  (\Phi_0, \Phi_1)\rangle = \langle \Psi_0, \Phi_0\rangle + \langle \Psi_1, \Phi_1\rangle$, where $\langle \Psi_i, \Phi_i\rangle$ for $i\in\{0,1\}$ is the usual inner product on $L^2(\mathbb{R}^m)$.  Finally, an element~$\Psi\in L^2(\mathbb{R}^m,\mathbb{C}^2)$ corresponds to the spinor~$(\Psi_0,\Psi_1)$ where $\Psi_i(x)=\langle i,\Psi(x)\rangle$ is the coefficient of~$\ket{i}$ in~$\ket{\Psi(x)}\in \mathbb{C}^2$ for $i=0,1$ and $x\in\mathbb{R}^m$.

	In the following, we use these isomorphic spaces interchangeably. Moreover, we will write an element $\Psi$ in one of these spaces as
	\begin{align}
		\ket{\Psi} = \int_{\mathbb{R}^m} dx\left( \Psi_0(x) \ket{x}\otimes\ket{0} + \Psi_1(x) \ket{x}\otimes\ket{1}\right)\ ,
	\end{align}
	where $\ket{x}$ is the (formal) tensor product of position-eigenstates $\ket{x}=\ket{x_1}\otimes\cdots\otimes\ket{x_m}$ for $x\in\mathbb{R}^m$. 
	We define the support of a function $\Psi \in L^2(\mathbb{R}^m \times \{0,1\})$ as 
	\begin{align}
		\supp(\Psi) = \overline{\left\{(x,b) \in \mathbb{R}^m \times \{0,1\}\ |\  \Psi(x,b) \neq 0\right\}}\ 
	\end{align}
	where $\overline{A}$ denotes the closure of a set~$A$.
	
	\subsection{Exact function evaluation}
	In the following, we consider injective functions of the form
	\begin{align}
		\begin{matrix}
			f: &A\times B        &\rightarrow & \mathbb{R}^m\times\{0,1\}&\\
			& (x,b)& \mapsto & f(x,b)=(f_1(x,b),f_2(x,b))\ .
		\end{matrix} \label{eq:function}
	\end{align}
	where $A\subseteq\mathbb{R}^m$ is closed  and $B \subseteq \{0,1\}$. Typically, we consider the case where $A=A_1 \times A_2 \times \cdots \times A_m$ is also a product set.  Furthermore, we assume that for each $b\in B$, the function~$f_1(\cdot,b):A\to \mathbb{R}^m$ is a diffeomorphism. Such a function $f$ defines an isometry 
	\begin{align}
		U_f:\begin{matrix}
			L^2(A\times B) & \rightarrow & U_fL^2(A\times B)\subseteq L^2(\mathbb{R}^m\times \{0,1\})
		\end{matrix}   
	\end{align}
	on the subspace~$L^2(A\times B)\subseteq L^2(\mathbb{R}^m\times \{0,1\})$  of functions with support contained in $A\times B$.
	That is, the isometry~$U_f$ is defined by its action
	\begin{align}
		U_f\left(\ket{x}\otimes\ket{b}\right)&=J_{f_1(\cdot,b)}(x)^{1/2}    \left(\ket{f_1(x,b)}\otimes\ket{f_2(x,b)}\right)\quad\textrm{ for } (x, b) \in A\times B \ ,\label{Eq: pointwise normalization U_f}
	\end{align}
	where $J_{f_1(\cdot,b)}(x)=|\det D_{f_1(\cdot ,b)}(x)|$ denotes the determinant of the Jacobi-matrix~$D_{f_1(\cdot ,b)}$ of~$f_1(\cdot,b)$ at the point~$x$. The factor~$J_{f_1(\cdot,b)}(x)^{1/2}$
	in Eq.~\eqref{Eq: pointwise normalization U_f} ensures that $U_f$ maps the subspace~$L^2(A\times B)$ isometrically to the image~$U_fL^2(A\times B)$.

	We note that the image of
	\begin{align}
		\ket{\Psi} = \sum_{b\in B} \int_{A}dx\, \Psi(x,b)\ket{x}\otimes\ket{b}\qquad\textrm{ for }\qquad \Psi\in L^2(A\times B)
	\end{align}
	is
	\begin{align}
		U_f\ket{\Psi}&=
		\sum_{c \in\{0,1\}} \int_{\mathbb{R}^m} dy\, \left(U_f \Psi\right)(y,c) \ket{y}\otimes\ket{c} 
	\end{align} 
	where we introduced the function $(U_f \Psi)\in L^2(\mathbb{R}^m\times \{0,1\})$ defined as
	\begin{align}
		\left(U_f \Psi\right)(y,c) = 
		\begin{cases} 
			(J_{f_1}(\cdot,b))(x)^{-1/2}\Psi(x)\big|_{
				(x,b)=f^{-1}(y,c)}
			&\qquad \textrm{if }(y,c) \in f(A\times B)\\
			0&\qquad \textrm{otherwise}
		\end{cases}\ .
	\end{align}
	Note that by construction we have 
	$\supp(U_f\Psi) \subseteq  \overline{f(A\times B)}$ whenever $\supp(\Psi) \subseteq A\times B$, i.e., 
	\begin{align}
		U_f\Psi\in L^2(\overline{f(A\times B)})\qquad\textrm{ for every }\qquad \Psi \in L^2(A\times B)\ .\label{eq:closureproperty}
	\end{align}

	Note that the $x$-dependent prefactor~$J_{f_1(\cdot,b)}(x)^{1/2}$ in Eq.~\eqref{Eq: pointwise normalization U_f} also implies that the composition is well-defined, i.e., for two injective functions $f: A\times B \to A'\times B'\subseteq \mathbb{R}^m\times\{0,1\}$, $g: A'\times B' \to  \mathbb{R}^m\times\{0,1\}$ with $A\times B \subseteq\mathbb{R}^m\times\{0,1\}$
	and diffeomorphic restrictions $f_1(\cdot,b)$, $g_1(\cdot, b')$ for $b\in B$ respectively $b' \in B'$, we have
	\begin{align} \label{eq: U_f composition}
		U_g \cdot U_f = U_{g\circ f}~:~\begin{matrix}
			L^2(A\times B) & \rightarrow & U_{f\circ g}L^2(A\times B)\subseteq L^2(\mathbb{R}^m\times \{0,1\})
		\end{matrix}   
	\end{align}

	Two special types of functions $f$ are of particular interest for our analysis.
	First, suppose $f=g\times\mathsf{id}$ takes the form~$f(x,b)=(g(x),b)$ with a diffeomorphism $g: A \to g(A)$. Then $U_{g\times \mathsf{id}}$ acts trivially on the qubit and 
	\begin{align}
		U_{g\times\mathsf{id}}\ket{\Psi}&=\sum_{b\in \{0,1\}}
		\int_{g(A)}dy\,  J_g(g^{-1}(y))^{-1/2}\Psi(g^{-1}(y),b)\ket{y}\otimes\ket{b}\\
		&=\sum_{b\in \{0,1\}}
		\int_{g(A)}dy\, J_{g^{-1}}(y)^{1/2}\Psi(g^{-1}(y),b)\ket{y}\otimes\ket{b}\ 
	\end{align}
	for $\Psi\in L^2(A\times \{0,1\})$.
	
	Another case of interest is where
	$f=\id\times (h \oplus \id)$
	takes the form $f(x,b)=(x,h(x) \oplus b)$
	for a function $h:\mathbb{R}^m\rightarrow \{0,1\}$. In this case we have 
	\begin{align}
		U_{\id\times (h\oplus \id)}\ket{\Psi}&
		=\sum_{b\in \{0,1\}}\int_{\mathbb{R}^m} dx\,
		\Psi(x,b)\ket{x}\otimes\ket{h(x) \oplus b}\ .
	\end{align}
	This case captures functions that extract information from a CV system into a qubit system.
	
	\subsection{Framework for approximate function evaluation}
	\label{sec:appx framework}
	We will be interested in unitaries that approximately compute certain functions on states that are supported on the corresponding domain. To formalize this notion, we introduce the following definition. 
	\begin{definition} \label{def: approx computation}
		Let $A\times B\subseteq \mathbb{R}^m\times \{0,1\}$ be a closed set. Let $f:A\times B\rightarrow\mathbb{R}^m\times \{0,1\}$ be injective with diffeomorphic restrictions $f_1(\cdot, b)$ for $b \in B$. Let $\varepsilon\geq 0$.
		Let $U$ be a unitary on $L^2(\mathbb{R}^m\times \{0,1\})$. We say that the unitary~$U$ \emph{$\varepsilon$-approximately computes~$f$ on $A\times B$} if
		\begin{align}
			\left\|
			U\Psi-U_f\Psi
			\right\|_1\leq \varepsilon\qquad\textrm{ for any  state }\qquad \ket{\Psi}\in L^2(A\times B)\ .\label{eq:l2atimesbx}
		\end{align}
		We will call a real~$\varepsilon>0$ satisfying~\eqref{eq:l2atimesbx} an \emph{error bound} for~$U$.
	\end{definition}

	In the following, we will use diagrammatic notation to specify the domain and range of a unitary $U$ which 
	$\varepsilon$-approximately computes  a function~$f:A\times B\rightarrow A'\times B'$ for $A,A'\subseteq\mathbb{R}^m$ and $B,B'\subseteq \{0,1\}$:  We write
	\begin{align}
		\vcenter{\hbox{\includegraphics{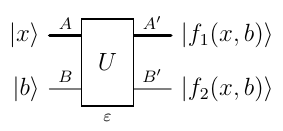}}}
	\end{align}
	We often consider product sets~$A=A_1\times\cdots\times A_m$.  If $U$ $\varepsilon$-approximately computes  a function 
	$f: (A_1\times\cdots \times A_m) \times B \to (A_1'\times \cdots \times A_m')\times B'$
	in the sense of Definition~\ref{def: approx computation}, we specify
	the sets~$A_1,\ldots,A_m\subseteq \mathbb{R}$, $B\in \{0,1\}$ 
	corresponding to the domain of~$f$, and the sets~$A'_1,\ldots,A'_m\subseteq \mathbb{R}$, $B'\in \{0,1\}$  corresponding to the range of~$f$ by labeling the associated input- and output wires, as follows:
	\begin{align}
		\vcenter{\hbox{\includegraphics{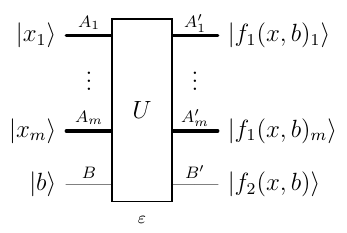}}}
		\label{eq: diagramatic notation}
	\end{align}
	We will omit $\varepsilon$ whenever $U$ computes the function $f$ exactly on~$A\times B$.

	The composition property~\eqref{eq: U_f composition} leads to the following lemma for the composition of two unitaries approximately computing two functions with compatible ranges and domains.

	\begin{lemma}  \label{lem: composition closeness}
		For $i\in \{1,2\}$, let $A_i\times B_i \subseteq \mathbb{R}^m \times \{0,1\}$ be a closed set, let $f^{(i)}:A_i\times B_i\rightarrow\mathbb{R}^m\times \{0,1\}$ and let $U^{(i)}$ be a unitary on~$L^2(\mathbb{R}^m\times \{0,1\})$ which $\varepsilon^{(i)}$-approximately~$f^{(i)}$ on $A_i\times B_i$. Assume further that 
		\begin{align}
			f^{(1)}(A_1\times B_1)\subseteq A_2\times B_2\ .
		\end{align}
		Then the unitary~$U^{(2)}U^{(1)}$ $(\varepsilon^{(1)}+\varepsilon^{(2)})$-approximately computes the function~$f^{(2)}\circ f^{(1)}$ on $A_1\times B_1$.
	\end{lemma}
	\begin{proof}
		Let us consider an arbitrary state $\Psi$ with $\supp(\Psi) \subseteq A_1\times B_1$. 
		Using~\eqref{eq:closureproperty} we have $\supp(U_{f^{(1)}}\Psi)\subseteq \overline{f^{(1)}(A_1\times B_1)} \subseteq A_2\times B_2$ by the assumption that $A_2\times B_2$ is closed. 
		Using the assumption about $U^{(2)}$, we therefore have 
		\begin{align} \label{eq: U_2U_1f - U_f2U_1 dist}
			\left\|
			U^{(2)} \left(U_{f^{(1)}} \Psi\right) - U_{f^{(2)}}\left( U_{f^{(1)}}\Psi \right)
			\right\|_1\leq \varepsilon_2\ .
		\end{align}
		By applying the triangle inequality and the invariance of the $L^1$-norm under unitaries, we obtain
		\begin{align}
			\left\|
			U^{(2)}\left( U^{(1)} \Psi \right)  -  U_{f^{(2)} \circ f^{(1)}}\Psi 
			\right\|_1  
			&\le \left\|
			U^{(2)} \left(U^{(1)} \Psi\right) -  U^{(2)} \left(U_{f^{(1)}} \Psi \right)
			\right\|_1 
			+ \left\|
			U^{(2)} \left( U_{f^{(1)}} \Psi   \right) -  
			U_{f^{(2)} \circ f^{(1)}}\Psi
			\right\|_1 \notag \\
			&= \left\|
			U^{(1)} \Psi   -  U_{f^{(1)}} \Psi  \right\|_1 +
			\left\| U^{(2)} \left( U_{f^{(1)}} \Psi   \right) - U_{f^{(2)}}\left( U_{f^{(1)}} \Psi \right)  
			\right\|_1 \notag \\
			&\le \varepsilon_1 + \varepsilon_2\ ,
		\end{align}
		where we used that  $ U_{f^{(2)}}\left( U_{f^{(1)}} \Psi \right) = U_{f^{(2)}\circ f^{(1)}} \Psi$ (see Eq.~\eqref{eq: U_f composition}), and 
		where the last inequality follows from the assumption on~$U^{(1)}$ combined with~\eqref{eq: U_2U_1f - U_f2U_1 dist}.
	\end{proof}

	\subsection{Approximate evaluation of functions  $h:\mathbb{R}^m\rightarrow \{0,1\}$}
	In this section, we present a sufficient condition for a unitary $U$ to $\varepsilon$-approximately compute a function $f: \mathbb{R}^m\times \{0,1\} \rightarrow \mathbb{R}^m\times \{0,1\}$ of the form $f= \id\times (h \oplus \id)$ with $h:\mathbb{R}^m\rightarrow \{0,1\}$.
	Such a family of functions is able to extract information about the CV system into the qubit while keeping the CV system unchanged. An example is our LSB-extraction gate (cf. Section~\ref{sec: ideal LSB}). We will use statements from this section to analyze its action on inputs close to integers (see Section~\ref{sec:approx analysis LSB}).
	
	To establish this sufficient condition, we need the following 
	\begin{lemma}\label{lem:overlapdistanceproperty}
		Let $\cX$ be a measure space. Let $0 < \varepsilon \le 1/2$ and $\phi\in L^2(\cX)$ with $\|\phi\|=1$ be arbitrary.
		For each $x\in\cX$, let $\Psi(x),\Psi'(x)\in\cH$ be normalized elements of a Hilbert space~$\cH$. Assume that
		\begin{align}
			\lvert 1 -  \langle \Psi'(x),\Psi(x)\rangle \rvert  \leq \varepsilon\qquad\textrm{ for all }\qquad x\in\cX\ .\label{eq:assumptionoverlapm}
		\end{align}
		Consider the states
		\begin{align}
			\ket{\Phi}&=\int_\cX dx\ \phi(x)\ket{x}\otimes\ket{\Psi(x)}\\
			\ket{\Phi'}&=\int_\cX dx\ \phi(x)\ket{x}\otimes\ket{\Psi'(x)}\ 
		\end{align}
		in the Hilbert space~$L^2(\cX, \cH)$. 
		Then
		\begin{align}
			|\langle \Phi,\Phi'\rangle|&\geq 1-2\varepsilon\ .
		\end{align}
	\end{lemma}
	\begin{proof}
		Eq.~\eqref{eq:assumptionoverlapm} implies the bounds
		\begin{align} \label{eq: bound Real Im part}
			\mathrm{Re}\langle \Psi'(x),\Psi(x)\rangle &\ge 1- \varepsilon \\
			\left| \mathrm{Im}\langle \Psi'(x),\Psi(x)\rangle\right|&\leq \varepsilon \ 
		\end{align}
		on the real- and imaginary parts of the overlap~$\langle \Psi'(x),\Psi(x)\rangle$ for $x\in\cX$.  Hence we can bound 
		\begin{align}
			\abs{\langle \Phi',\Phi\rangle}&= \abs{ \int dx \abs{\phi(x)}^2
				\langle \Psi'(x),\Psi(x)\rangle }\\
			&= \abs{\int dx \abs{\phi(x)}^2 \left(\mathrm{Re}
				\langle \Psi'(x),\Psi(x)\rangle
				+ i\/\mathrm{Im}\langle \Psi'(x),\Psi(x)\rangle \right) } \\
			&\geq \abs{ \int dx |\phi(x)|^2 \mathrm{Re}
				\langle \Psi'(x),\Psi(x)\rangle }   - \abs{ \int dx |\phi(x)|^2 \mathrm{Im}\langle \Psi'(x),\Psi(x)\rangle }  \\
			&\ge  \abs{ \int dx |\phi(x)|^2 \mathrm{Re}
				\langle \Psi'(x),\Psi(x)\rangle }   -  \int dx |\phi(x)|^2 \abs{\mathrm{Im}\langle \Psi'(x),\Psi(x)\rangle } \\
			& \geq 1 - 2\varepsilon \ ,
		\end{align}
		where the first inequality follows from the triangle inequality on $\mathbb{C}$  and the last inequality from~\eqref{eq: bound Real Im part}.
	\end{proof}
	The following lemma gives the desired sufficient condition. For ease of later application, we give a ``stable'' version including an additional system~$L^2(\mathbb{R}^\ell)$ where the function acts as the identity.
	\begin{lemma} \label{lem: approximate function discrete}
		Let $0 < \varepsilon \le 1/2$, let $A \subseteq \mathbb{R}^m$ be a closed 
		set and $f: A\times \{0,1\} \rightarrow  \mathbb{R}^m \times\{0, 1\}$ be a function of the form $f = \mathsf{id} \times (h \oplus \mathsf{id})$, i.e., $f(x,b) = (x, h(x) \oplus b)$. Let $\{\ket{\Omega(x,b)}\}_{(x,b)\in A\times\{0,1\}}\subseteq \mathbb{C}^2$ be a family of qubit states such that
		\begin{align} 
			|1 - \langle h(x) \oplus b, \Omega(x,b)\rangle| \le \varepsilon  \quad \textrm{for all } \quad (x,b) \in A\times \{0,1\}\ . \label{eq: approx function h assumption}
		\end{align}
		Let $U$ be a unitary acting on $L^2(\mathbb{R}^m)\otimes \mathbb{C}^2$ as 
		\begin{align}
			U\left(\ket{x}\otimes\ket{b}\right) = \ket{x} \otimes\ket{\Omega(x,b)}\qquad\textrm{ for }\qquad (x,b)\in A\times \{0,1\}\ .
		\end{align} Let $\ell \in \mathbb{N}_0$ be arbitrary.  
		Then the unitary 
		$I_{L^2(\mathbb{R}^\ell)}\otimes U$
		$(6\varepsilon^{1/4})$-approximately computes the function
		\begin{align}
			\begin{matrix}
				\id\times f:& \mathbb{R}^\ell\times A\times \{0,1\} & \rightarrow & \mathbb{R}^\ell \times \mathbb{R}^m\times \{0,1\}\\
				& (w,x,b) & \mapsto & (w,f(x,b))=(w,x,h(x)\oplus b)
			\end{matrix}
		\end{align}
		on the set~$\mathbb{R}^\ell\times A\times \{0,1\}$.
	\end{lemma}
	
	\begin{proof}
		We first consider the case $\ell=0$.
		Let $\Psi$ be a state with $\supp(\Psi) \subseteq A\times \{0,1\}$. Then  we have 
		\begin{align}
			U_f\ket{\Psi}&=\sum_{b\in \{0,1\}}\int_{\mathbb{R}^m} dx \Psi(x,b)\ket{x}\otimes \ket{h(x)\oplus b}\\
			U\ket{\Psi}&=\sum_{b\in \{0,1\}}\int_{\mathbb{R}^m} dx \Psi(x,b)\ket{x}\otimes\ket{\Omega(x,b)}
		\end{align}
		by definition. It follows that
		\begin{align}
			\langle U_f \Psi, U\Psi\rangle
			&=\sum_{b,b'\in \{0,1\}} \int_{\mathbb{R}^m} \overline{\Psi(x,b')}\Psi(x,b) \langle h(x)\oplus b', \Omega(x,b)\rangle\ .
		\end{align}
		Separating the pairs~$(b,b')\in \{0,1\}^2$ into those where $b=b'$ and those where~$b\neq b'=b\oplus 1$, we obtain
		\begin{align}
			\langle U_f \Psi, U\Psi\rangle&=I_1+I_2\label{eq:ufionetwo}
		\end{align}
		where
		\begin{align}
			I_1&:=\sum_{b\in \{0,1\}} \int_{\mathbb{R}^m} dx \, \abs{\Psi(x,b)}^2 \langle h(x) \oplus b, \Omega(x,b)\rangle\\
			I_2&:=\sum_{b\in \{0,1\}} \int_{\mathbb{R}^m} dx \, \overline{\Psi(x,b\oplus 1)}\Psi(x, b) \langle h(x) \oplus b \oplus 1, \Omega(x, b)\rangle\ .
		\end{align}
		We can rewrite the term~$I_1$ as an inner product
		\begin{align}
			I_1 = \left\langle \sum_{b\in \{0,1\}} \int_{\mathbb{R}^m} dx \Psi(x,b)\ket{x}\otimes\ket{b}\otimes\ket{h(x)\oplus b}, \sum_{b'\in \{0,1\}} \int_{\mathbb{R}^m} dx' \Psi(x',b')\ket{x'}\otimes\ket{b'}\otimes\ket{\Omega(x', b')}\right\rangle \ .\notag
		\end{align}
		By Lemma~\ref{lem:overlapdistanceproperty} setting $\mathcal{X}= \mathbb{R}^m \times \{0,1\}$, we obtain (by the assumption~\eqref{eq: approx function h assumption}) that
		\begin{align} \label{eq: bound I_1}
			|I_1| \ge 1 - 2 \varepsilon \ .
		\end{align}
		Note that because~$\ket{\Omega(x,b)}\in \mathbb{C}^2$ has unit norm, we have
		\begin{align}
			\left|\langle  h(x)\oplus b,  \Omega(x,b)\rangle \right|^2 +  \left|\langle  h(x)\oplus b \oplus 1, \Omega(x,b) \rangle\right|^2 = 1 \quad \textrm{ for all } \quad (x,b) \in A\times \{0,1\}\ .
		\end{align}
		Since $|1 - \langle h(x) \oplus b, \Omega(x,b)\rangle|\le \varepsilon$ that implies $\left|\langle \Omega(x,b), h(x)\oplus b\rangle\right|^2 \ge 1 - 2\varepsilon$, we have
		\begin{align}
			\left|\langle h(x)\oplus b \oplus 1,  \Omega(x,b)  \rangle\right| \le \sqrt{2 \varepsilon}\qquad\textrm{ for all }\qquad (x,b)\in A\times \{0,1\}\ .
		\end{align}
		Therefore we obtain
		\begin{align}
			|I_2| \le \sqrt{2 \varepsilon} \left(\sum_{b\in \{0,1\}} \int_{\mathbb{R}^m} dx \, \abs{\Psi(x,b\oplus 1)}\cdot\abs{\Psi(x, b)} \right) \ .
		\end{align}
		By the Cauchy-Schwarz inequality applied first to $L^2(\mathbb{R}^m)$ and then to $\mathbb{R}^2$, we can further bound
		\begin{align}
			\sum_{b\in \{0,1\}} \int_{\mathbb{R}^m} dx \, \abs{\Psi(x,b\oplus 1)}\cdot\abs{\Psi(x, b)} &\le  \sum_{b\in \{0,1\}} \left( \int_{\mathbb{R}^m} dx \,\abs{\Psi(x,b\oplus 1)}^2\right)^{1/2} \left( \int_{\mathbb{R}^m} dx \,\abs{\Psi(x,b)}^2\right)^{1/2} \notag\\
			&= \sum_{b\in \{0,1\}} \lVert \Psi(\cdot,b\oplus 1)\rVert_{L^2(\mathbb{R}^m)} \lVert \Psi( \cdot,b)\rVert_{L^2(\mathbb{R}^m)}\\
			&\le  \left(\sum_{b\in \{0,1\}} \lVert \Psi(\cdot,b)\rVert_{L^2(\mathbb{R}^m)}^2\right)^{1/2} \left(\sum_{b\in \{0,1\}} \lVert \Psi(\cdot,b\oplus 1)\rVert_{L^2(\mathbb{R}^m)}^2\right)^{1/2} \notag \\
			&= \sum_{b\in \{0,1\}} \lVert \Psi(\cdot,b)\rVert_{L^2(\mathbb{R}^m)}^2\\
			&= \|\Psi\|^2 = 1\ .
		\end{align}
		The last step uses the fact that $\Psi$ is normalized.
		Thus
		\begin{align}
			|I_2| \le \sqrt{2 \varepsilon}\ .\label{eq: bound I_2}
		\end{align}
		It follows from~\eqref{eq:ufionetwo}
		combined with~\eqref{eq: bound I_1} and~\eqref{eq: bound I_2} that 
		\begin{align}
			\left|\langle U_f \Psi, U\Psi\rangle\right| &\geq |I_1|-|I_2|\\
			&\geq 1 - (2\varepsilon + \sqrt{2\varepsilon} )\ .
		\end{align}
		The relation between the overlap and the $L^1$-norm distance stated in Eq.~\eqref{eq: trace distance - fidelity} and the bound $\varepsilon \le \sqrt{\varepsilon}$ for $0< \varepsilon \le 1/2$ yield
		\begin{align}
			\left\lVert U_f\Psi  -  U \Psi \right\rVert_1 \le 2 \sqrt{2(2\varepsilon + \sqrt{2\varepsilon})} < 6 \varepsilon^{1/4} \ .
		\end{align}
		This proves the claim for~$\ell=0$. 
		
		Finally, we prove the ``stable'' version. Let $\ell>0$. Let us define the function
		\begin{align}
			\begin{matrix}
				g: & \mathbb{R}^\ell \times A\times \{0,1\} & \rightarrow & \mathbb{R}^\ell \times \mathbb{R}^m\times \{0,1\}\\
				& (w,x,b) &\mapsto & g(w,x,b)=(w,x,h'(w,x)\oplus b)
			\end{matrix}
		\end{align}
		where $h'(w,x) = h(x)$. The unitary $I_{L^2(\mathbb{R}^\ell)} \otimes U$ acts as $(I_{L^2(\mathbb{R}^\ell)}\otimes U)\left(\ket{w}\otimes\ket{x}\otimes\ket{b}\right) = \ket{w}\otimes\ket{x}\otimes\ket{\Omega'(w,x,b)}$, where $\ket{\Omega'(w,x,b)} := \ket{\Omega(x,b)}$. Setting $A' := \mathbb{R}^\ell \times A$, we have 
		\begin{align}
			|1 - \langle h'(w,x)\oplus b, \Omega'(w,x,b)\rangle|\le\varepsilon \qquad\textrm{ for all }\qquad (w,x,b)\in A'\times \{0,1\}\ ,
		\end{align} Thus the argument for~$\ell>0$ reduces  to the one for the case $\ell=0$.
	\end{proof}

	\subsection{Approximate function evaluation by various building blocks}\label{sec: Aprox Analysis}
	The analysis of the circuit~$\cQ_{a,N}$ 
	relies on our framework for approximate function evaluation: We repeatedly apply Lemma~\ref{lem: composition closeness}
	to evaluate how certain composed unitaries act on approximate position-eigenstates.

	\subsubsection{Approximate evaluation of the least significant bit extraction}\label{sec:approx analysis LSB}
	Consider the following unitary $U^{\lsb}_{\bosonA\to\qubit}$ acting on a single bosonic mode~$\bosonA$ and a qubit~$\qubit$.
	\begin{align}
		\includegraphics{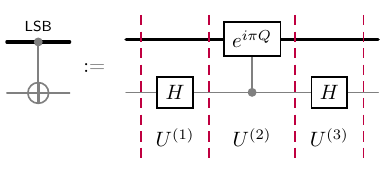}
	\end{align}
	This unitary is decomposed as
	\begin{align}
		U^{\lsb}_{\bosonA\to\qubit}=U^{(3)}U^{(2)}U^{(1)}\qquad\textrm{ where }
		\qquad
		\begin{matrix*}[l]U^{(1)}&=& I_\bosonA\otimes H_\qubit\\
			U^{(2)}&=& \mathsf{ctrl}_\qubit e^{i\pi Q_\bosonA} \\
			U^{(3)}&=& I_\bosonA\otimes H_\qubit \ .
		\end{matrix*} \label{eq: U LSB}
	\end{align}
	
	In Section~\ref{sec: ideal LSB} we argued  that when applied to a position-eigenstate~$\ket{x}$ with an integer position $x\in\mathbb{Z}$ and  a qubit in a computational basis state $\{\ket{0},\ket{1}\}$, the unitary~$U^{\mathsf{LSB}}_{\bosonA\to \qubit}$ computes the LSB of~$x$ into the qubit. That is, it acts as
	\begin{align}
		U^{\mathsf{LSB}}_{\bosonA\to\qubit}\left(\ket{x}_{\bosonA}\otimes\ket{b}_\qubit\right)&=\ket{x}_{\bosonA}\otimes\ket{x_0\oplus b}_\qubit \qquad\textrm{ for all }\qquad x\in\mathbb{Z}\textrm{ and } b\in \{0,1\}\ ,
	\end{align}
	where we used $x_0=x \xmod 2$ to denote the least significant bit of the integer~$x\in\mathbb{Z}$.

	We now consider the non-ideal case, where the input state is a position-eigenstate~$\ket{x}$ such that $x\in\mathbb{R}$ is not necessarily an integer, but is close to an integer up to an error~$\varepsilon<1/2$,
	that is, 
	$\min_{y\in\mathbb{Z}} |x-y| \le \varepsilon$.
	
	We will show that $U^{\mathsf{LSB}}_{\bosonA\to\qubit}$ approximately computes
	the LSB~$\round{x}_{0}$ of the rounded number~$\round{x}\in\mathbb{Z}$.
	Here we denote by $\round{x}\in\mathbb{Z}$ the closest integer to~$x\in\mathbb{R}$ (for $x=x'+1/2$, where $x'\in\mathbb{Z}$, we set $\round{x}=x'+1$).  Recall that for $\varepsilon>0$ and a set $M\subseteq\mathbb{Z}$, we define the set
	\begin{align}
		M (\varepsilon):=\{x\in\mathbb{R}\ |\ \exists y \in M: |x-y|\leq \varepsilon\}\ ,
	\end{align}
	which is  the Minkowski sum $M(\varepsilon)=M+[-\varepsilon,\varepsilon]$ of the set $M$ and the interval~$[-\varepsilon,\varepsilon]$. 
	
	The key property of the unitary $U^{\mathsf{LSB}}_{\bosonA\to\qubit}$ we use is the following.
	\begin{lemma}\label{lemma:lsb pointwise closeness}
		Let $U^{\lsb}_{\bosonA\to\qubit}$ be defined by Eq.~\eqref{eq: U LSB}. 
		Let $(x,b)\in\mathbb{Z}(\varepsilon)\times \{0,1\}$ be arbitrary. Then there  is a state~$\ket{\Omega(x,b)}\in \mathbb{C}^2$ such that 
		\begin{align}
			U^{\mathsf{LSB}}_{\bosonA\to\qubit}\left(\ket{x}_\bosonA \otimes
			\ket{b}_{\qubit}
			\right) = \ket{x}_\bosonA \otimes \ket{\Omega(x,b)}_\qubit\qquad\textrm{ and }\qquad 
			|1 - \langle \round{x}_{0}\oplus b, \Omega(x,b)\rangle| \le \frac{\varepsilon\pi}{2}\ .
		\end{align}
		Similarly, there is a state~$\ket{\Omega'(x,b)}$ such that \begin{align}      \left(U^{\mathsf{LSB}}_{\bosonA\to\qubit}\right)^\dagger\left(\ket{x}_\bosonA \otimes       \ket{b}_{\qubit}       \right) = \ket{x}_\bosonA \otimes \ket{\Omega'(x,b)}_\qubit  \qquad\textrm{ and }\qquad               |1 - \langle \round{x}_{0}\oplus b, \Omega'(x,b)\rangle| \le \frac{\varepsilon\pi}{2} . \end{align} 
	\end{lemma}
	Lemma~\ref{lemma:lsb pointwise closeness}
	shows that  the unitary~$U^{\mathsf{LSB}}_{\bosonA\to\qubit}$ satisfies the sufficient condition for approximate function evaluation given in Lemma~\ref{lem: approximate function discrete}. We can therefore conclude that $U^{\mathsf{LSB}}_{\bosonA\to\qubit}$ computes the LSB of the position $x$ of mode~$\bosonA$ even if $x$ is only $\varepsilon$-close to an integer, for $\varepsilon\in [0,1/2)$. We argue that this is the case even for a ``stabilized'' version of~$U^{\mathsf{LSB}}_{\bosonA\to\qubit}$, i.e., when additional modes are involved, see Corollary~\ref{cor:LSBtracedist}.

	\begin{proof} By definition, $U^{\lsb}_{\bosonA\to\qubit}$ consists of three unitaries $U^{\lsb}_{\bosonA\to\qubit}=U^{(3)}U^{(2)}U^{(1)}$ (see Eq.~\eqref{eq: U LSB}). We apply this  to the state $\ket{x}_\bosonA \otimes\ket{b}_{\qubit}$. 
		After applying the first unitary $U^{(1)} = I_\bosonA \otimes H_\qubit$ (i.e., a Hadamard on the qubit) we have
		\begin{align}
			U^{(1)}\left(\ket{x}_\bosonA \otimes\ket{b}_{\qubit} \right) &= \ket{x}_\bosonA \otimes\left(H_\qubit \ket{b}_{\qubit}\right)\\
			&= \ket{x}_\bosonA \otimes\frac{1}{\sqrt{2}}\left(\ket{0}_\qubit+e^{-i\pi b}\ket{1}_\qubit\right)\ .
		\end{align}
		Second, we apply $U^{(2)} = I_\boson \otimes \proj{0}_\qubit + e^{i\pi Q_\bosonA} \otimes \proj{1}_\qubit$ and get
		\begin{align}
			U^{(2)} U^{(1)} \left(\ket{x}_\bosonA \otimes\ket{b}_{\qubit} \right) = \ket{x}_\bosonA \otimes \frac{1}{\sqrt{2}}\left(\ket{0}_\qubit +  e^{i\pi(x-b)} \ket{1}_\qubit\right)\ .
		\end{align}
		Finally, application of the unitary $U^{(3)}= I_\bosonA \otimes H_\qubit $  leads to 
		\begin{align}
			U^{(3)} U^{(2)} U^{(1)}\left(\ket{x}_\bosonA  \otimes \ket{b}_{\qubit} \right) &= \ket{x}_\bosonA \otimes\frac{1}{\sqrt{2}}\left( \ket{+}_\qubit + e^{i\pi(x-b)} \ket{-}_\qubit\right) \\
			&= \ket{x}_\bosonA \otimes \frac{1}{2}\left( \left(1 + e^{i\pi(x-b)}\right) \ket{0}_\qubit +  \left(1 - e^{i\pi(x-b)}\right) \ket{1}_\qubit \right)\\
			&= \ket{x}_\bosonA \otimes \ket{\Omega(x, b)}_\qubit \ \label{eq: U LSB output state}
		\end{align}
		where we introduced the single-qubit state
		\begin{align}
			\ket{\Omega(x,b)} = \frac{1}{2}\left( \left(1 + e^{i\pi(x-b)}\right) \ket{0}_\qubit +  \left(1 - e^{i\pi (x-b)}\right) \ket{1}_\qubit   \right)\ .
		\end{align}
		Since we can write
		\begin{align}
			\ket{\round{x}_{0} \oplus b} = \frac{1}{2}\left((1+e^{i\pi(\round{x}_{0} + b)})\ket{0}+(1-e^{i\pi(\round{x}_{0} + b)})\ket{1}\right) \ , \label{eq: qubit x0 oplus b}
		\end{align}
		the overlap of the state~$\ket{\Omega(x,b)}$ and the state~$\ket{\round{x}_{0} \oplus b}$ is
		\begin{align}
			\langle \round{x}_{0} \oplus b, \Omega(x, b)\rangle& = \frac{1}{2}\left(1 + e^{i\pi (x-\round{x}_{0}-2b)}\right) \\
			& = \frac{1}{2}\left(1 + e^{i\pi (x-\round{x}_{0})}\right) \\
			& = \frac{1}{2}\left(1 + e^{i\pi (x-\round{x}_{0})}\cdot e^{-i\pi (\round{x}-\round{x}_{0})}\right) \\
			& = \frac{1}{2}\left(1 + e^{i\pi \dev(x)}\right)\ ,
		\end{align}
		where we used that $\round{x}-\round{x}_0$ is even to obtain the penultimate step and where $\dev(x) := x - \round{x}$ is the deviation from the closest integer. 
		We are interested in $(x,b)\in \mathbb{Z}(\varepsilon)\times \{0,1\}$. By definition $\abs{\dev(x)} \le \varepsilon$ for $x\in \mathbb{Z}(\varepsilon)$ thus we have
		\begin{align}
			\abs{1 - \langle \round{x}_{0} \oplus b, \Omega(x, b)\rangle} &= \abs{\frac{1 - e^{i\pi\dev(x)}}{2}} = \abs{\sin \frac{\pi\dev(x)}{2}} \le \abs{\sin \frac{\pi\varepsilon}{2}} \le \frac{\pi\varepsilon}{2}\ ,\label{eq:chaintein}
		\end{align}
		where we used the identity $\abs{\frac{1-e^{i x}}{2}}=\abs{\sin \frac{x}{2} }$, the fact that $x \mapsto \sin x$ is odd on $\mathbb{R}$ and monotonically increasing on $[0,\pi/2]$, and $\abs{\sin x}\le \abs{x}$ for all $x\in\mathbb{R}$. This concludes the proof of the claim for the unitary~$U^{\lsb}_{\bosonA\rightarrow\qubit}$.
		
		Now consider the adjoint $(U^{\lsb}_{\bosonA\to\qubit})^\dagger=(U^{(1)})^\dagger(U^{(2)})^{\dagger}(U^{(3)})^\dagger$. 
		By the same line of reasoning as for $U^{\lsb}_{\bosonA\to\qubit}$ we obtain
		\begin{align}
			(U^{(3)})^\dagger(U^{(2)})^{\dagger}(U^{(1)})^\dagger\left(\ket{x}_\bosonA  \otimes \ket{b}_{\qubit} \right) 
			&= \ket{x}_\bosonA \otimes \frac{1}{2}\left( \left(1 + e^{-i\pi(x+b)}\right) \ket{0}_\qubit +  \left(1 - e^{-i\pi(x+b)}\right) \ket{1}_\qubit \right)\notag \\
			&= \ket{x}_\bosonA \otimes \ket{\Omega'(x, b)}_\qubit \ ,\label{eq: U LSB dagger output state}
		\end{align}
		where
		\begin{align}
			\ket{\Omega'(x, b)}&=\frac{1}{2}\left( \left(1 + e^{-i\pi(x+b)}\right) \ket{0}_\qubit +  \left(1 - e^{-i\pi(x+b)}\right) \ket{1}_\qubit \right)\ .
		\end{align}
		By~\eqref{eq: qubit x0 oplus b} the overlap satisfies
		\begin{align}
			\langle \round{x}_{0} \oplus b, \Omega'(x, b)\rangle        & = \frac{1}{2}\left(1 + e^{-i\pi(x+\round{x}_{0})}\right) \\
			& = \frac{1}{2}\left(1 + e^{-i\pi(x+\round{x}_{0})}\cdot e^{i\pi (\round{x}+\round{x}_0)}\right) \\
			& = \frac{1}{2}\left(1 + e^{-i\pi \dev(x)}\right)\ ,
		\end{align}
		where we obtained the second equality by noting that $\round{x}+\round{x}_0$ is even for $x\in\mathbb{R}$. It follows (as in Eq.~\eqref{eq:chaintein}) that $\abs{1 - \langle \round{x}_{0} \oplus b, \Omega'(x, b)\rangle} \le \frac{\varepsilon\pi}{2}$ for $(x,b)\in \mathbb{Z}(\varepsilon)\times \{0,1\}$. 
	\end{proof}

	\begin{corollary}\label{cor:LSBtracedist}
		Let $\varepsilon \in (0,1/2)$ and let $\ell\in \mathbb{N}_0$. Consider the unitary~$U = I_\bosonA\otimes U^{\mathsf{LSB}}_{\bosonB\to\qubit}$, where~$\bosonA$ denotes an $\ell$-mode system,~$\bosonB$ is a single mode, and~$\qubit$ denotes a qubit system. Then 
		\begin{center}
			\includegraphics{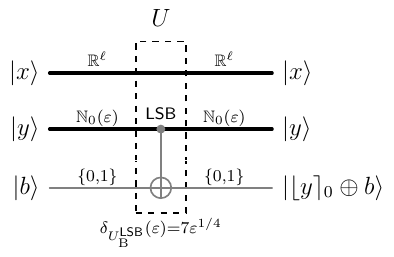}
		\end{center}
		in the diagrammatic notation (cf.~\eqref{eq: diagramatic notation}). 
		Furthermore, the unitary~$U^\dagger = I_\bosonA\otimes (U^{\mathsf{LSB}}_{\bosonB\to\qubit})^\dagger$ satisfies
		\begin{center}
			\includegraphics{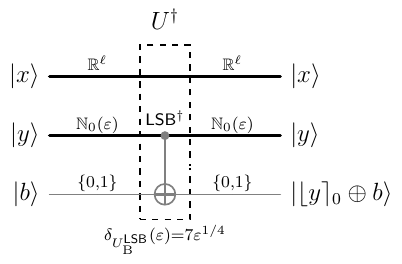}
		\end{center}
	\end{corollary}
	\begin{proof}
		Both claims follow immediately from Lemma~\ref{lem: approximate function discrete} applied to Lemma~\ref{lemma:lsb pointwise closeness} along with $6(\pi/2)^{1/4} \le 7$.
	\end{proof}
	
	\subsubsection{Approximate function evaluation by the unitaries $V_\alpha$, $\Vam$, and $\Ua$}
	In this  section we use our framework for approximate function evaluation to analyze the composed unitaries $V_\alpha$, $\Vam$, and $\Ua$ used in our algorithm.
	
	Let us recall the unitary $V_\alpha$ defined in Fig.~\ref{fig:composite gates}  of the main text for $\alpha > 0$. It is given by the circuit
	\begin{align}
		\vcenter{\hbox{\includegraphics[scale=1.1]{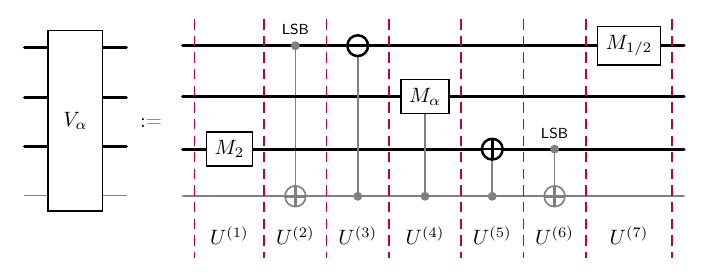}}}
		\label{eq: Valpha circuit}
	\end{align}
	i.e.,
	\begin{align*}   
		\left(V_\alpha\right)_{\bosonA\bosonB\bosonC\qubit}&=U^{(7)}U^{(6)}U^{(5)}U^{(4)}U^{(3)}U^{(2)}U^{(1)}
	\end{align*}
	where
	\begin{align}
		\begin{matrix*}[l]
			U^{(1)}&=&(M_{2})_{\bosonC}\\
			U^{(2)}&=&U^\lsb_{\bosonA\to\qubit}\\
			U^{(3)}&=&\mathsf{ctrl}_\qubit e^{iP_{\bosonA}}  \\
			U^{(4)}&=&\mathsf{ctrl}_{\qubit} (M_\alpha)_{\bosonB}
		\end{matrix*}\qquad\qquad \begin{matrix*}[l]
			U^{(5)}&=&\mathsf{ctrl}_\qubit e^{-iP_{\bosonC}}\\
			U^{(6)}&=&U^\lsb_{\bosonC\to\qubit}\\
			U^{(7)}&=&(M_{2^{-1}})_{\bosonA}\ .\\
			&&
		\end{matrix*} \label{eq: Valpha Us}
	\end{align}
	We depict the actions of the constituting unitaries on inputs that are close to integers on the modes~$\bosonA$ and~$\bosonC$ by the diagrams in Fig.~\ref{fig:V:Us} (cf.~\eqref{eq: diagramatic notation} for an explanation of the diagrammatic notation). These actions are obtained directly from the definitions of the respective unitaries and from Corollary~\ref{cor:LSBtracedist} for the LSB gates (by relabeling modes).
	
	\begin{figure}[!ht]
		\centering
		\begin{subfigure}[t]{0.47\textwidth}
			\raisebox{+20pt}{\includegraphics{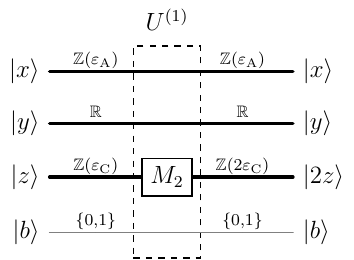}}
			\caption{Action of the unitary $U^{(1)}$.}
			\label{fig:V:U1}
		\end{subfigure}\hfill
		\begin{subfigure}[t]{0.47\textwidth}
			\includegraphics{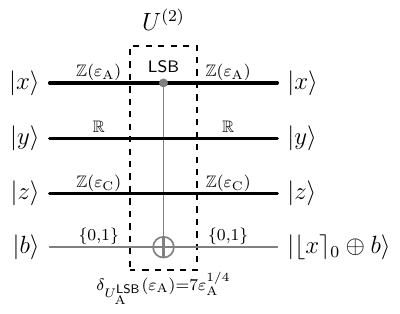}
			\caption{Action of $U^{(2)}$ (obtained by Corollary~\ref{cor:LSBtracedist}).}
			\label{fig:V:U2}
		\end{subfigure}\\[0.1cm]
		\begin{subfigure}[t]{0.47\textwidth}
			\raisebox{+4pt}{\includegraphics{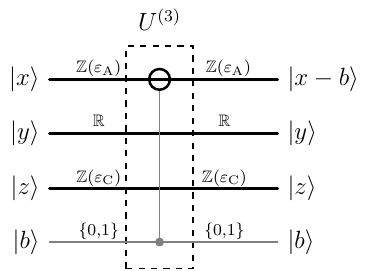}}
			\caption{Action of the unitary $U^{(3)}$.}
			\label{fig:V:U3}
		\end{subfigure}\hfill
		\begin{subfigure}[t]{0.47\textwidth}
			\includegraphics{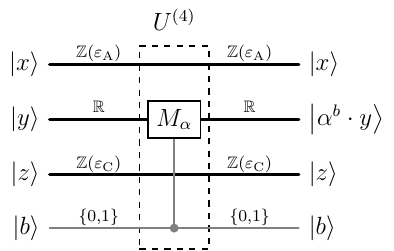}
			\caption{Action of the unitary $U^{(4)}$.}
			\label{fig:V:U4}
		\end{subfigure}\\[0.1cm]
		\begin{subfigure}[t]{0.47\textwidth}
			\raisebox{+22pt}{\includegraphics{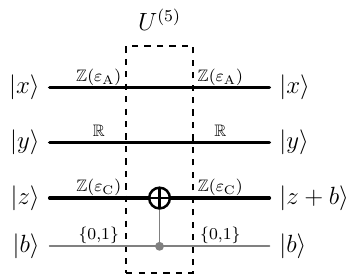}}
			\caption{Action of the unitary $U^{(5)}$.}
			\label{fig:V:U5}
		\end{subfigure}\hfill
		\begin{subfigure}[t]{0.47\textwidth}
			\includegraphics{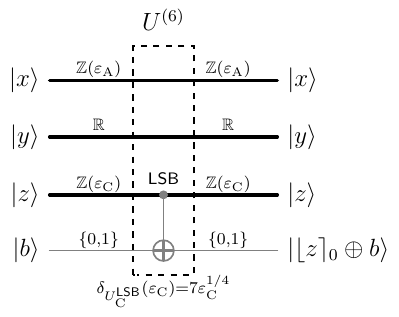}
			\caption{Action of $U^{(6)}$ (obtained by Corollary~\ref{cor:LSBtracedist}).}
			\label{fig:V:U6}
		\end{subfigure}
		\begin{subfigure}[b]{0.47\textwidth}
			\includegraphics{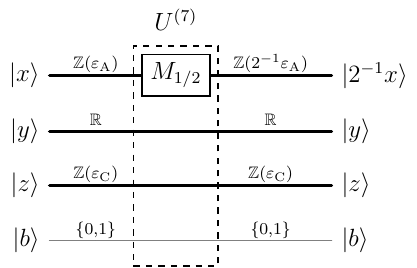}
			\caption{Action of the unitary $U^{(7)}$.}
			\label{fig:V:U7}
		\end{subfigure}
		\caption{The approximate action of the unitaries decomposing $V_{\alpha}=U^{(7)}U^{(6)}\cdots U^{(1)}$ displayed in the diagramatic notation~\eqref{eq: diagramatic notation} of our approximate function evaluation framework described in Section~\ref{sec:appx framework}.}
		\label{fig:V:Us}
	\end{figure}

	\begin{samepage}
		Using the elements expressed by Fig.~\ref{fig:V:Us}, we can show the following:
		\begin{lemma} \label{lem: V(a) unitary error}
			Let $\varepsilon_\bosonA\in(0,1/2)$ and $\varepsilon_\bosonC\in(0,1/4)$. Let $\alpha > 0$. Consider the unitary $V_\alpha$ defined by the circuit in Fig.~\ref{fig:composite gates}  of the main paper. Then 
			\begin{align}
				\hspace*{2cm}\vcenter{\hbox{\includegraphics{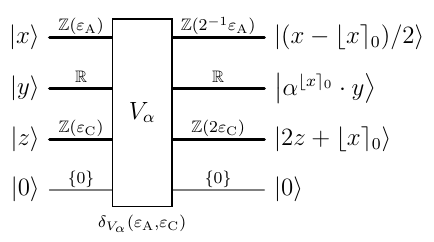}}}
			\end{align}	
			where 
			$\delta_{V_\alpha}(\varepsilon_\bosonA,\varepsilon_\bosonC)=7 \varepsilon_\bosonA^{1/4}+7(2\varepsilon_\bosonC)^{1/4}$. 
	\end{lemma}\end{samepage}
	We note that the error bound $\delta_{V_\alpha}(\varepsilon_\bosonA,\varepsilon_\bosonC)$ does not depend on the parameter  $\alpha$ of the unitary $V_\alpha$. Below we will only use the action of $V_\alpha$ restricted to inputs close to non-negative integers in modes $\bosonA$ and $\bosonC$, i.e., 
	\begin{align}
		\hspace*{2cm}
		\raisebox{-60pt}{\mbox{\includegraphics{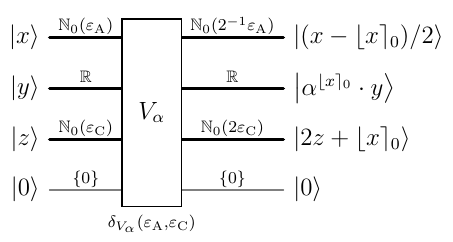}}}\label{eq: V alpha action restricted}
	\end{align}
	This is well-defined since $x\in\mathbb{N}_0(\varepsilon_\bosonA)$ implies that $(x-\round{x}_0)/2 \in \mathbb{N}_0(2^{-1}\varepsilon_\bosonA)$, and similarly, $z\in\mathbb{N}_0(\varepsilon_\bosonA)$ 
	implies that $2z+\round{x}_0\in \mathbb{N}_0(2\varepsilon_\bosonC)$ for any $x\in \mathbb{N}_0$. 
	\begin{proof}[Proof of Lemma~\ref{lem: V(a) unitary error}]
		Let us write $V^{(T)}=U^{(T)}\cdots U^{(1)}$ for $T\in \{2, \ldots,7\}$. We obtained the approximate action of each unitary $V^{(T)}$ by Lemma~\ref{lem: composition closeness} and list them in Fig.~\ref{fig:Valpha:Vs}. The claim follows by noting that $V^{(7)}=V_\alpha$ and by setting   (see Fig.~\ref{fig:v7last})
		\begin{align}
			\delta_{V_\alpha}&=\delta_{U_\bosonA^{\lsb}}(\varepsilon_\bosonA )+\delta_{U_\bosonC^{\lsb}}(2\varepsilon_\bosonC)\\
			&=7 \varepsilon_\bosonA^{1/4}+7(2\varepsilon_\bosonC)^{1/4}\ .
		\end{align}
	\end{proof}
	
	\begin{figure}[p]
		\centering
		\begin{subfigure}[t]{0.47\textwidth}
			\centering
			\hspace*{-0.8cm}\includegraphics{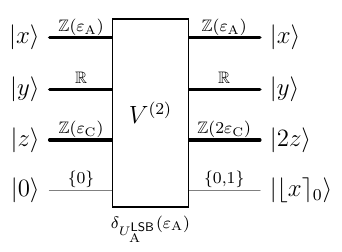}
			\caption{The action of the unitary $V^{(2)}$ is given by composition of $U^{(1)}$ and $U^{(2)}$ from Figs.~\ref{fig:V:U1} and~\ref{fig:V:U2}, respectively.}\label{fig:V:V2}
		\end{subfigure}\qquad
		\begin{subfigure}[t]{0.47\textwidth}
			\centering
			\hspace*{-0.4cm}\includegraphics{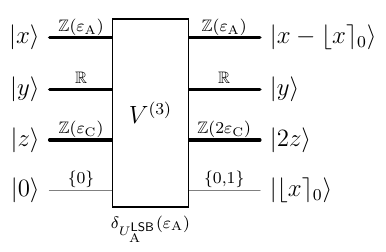}
			\caption{The action of the unitary $V^{(3)}$ is given by composition of $V^{(2)}$ from Fig.~\ref{fig:V:V2} with $U^{(3)}$ from Fig.~\ref{fig:V:U3}.}
			\label{fig:V:V3}
		\end{subfigure}\\[0.5cm]
		\begin{subfigure}[t]{0.47\textwidth}
			\centering
			\hspace*{-0.2cm}\includegraphics{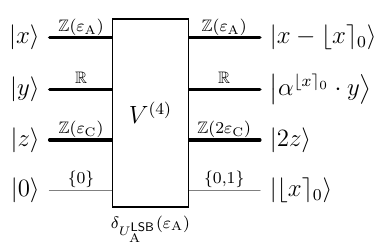}
			\caption{
				The action of $V^{(4)}$ is given by composing $V^{(3)}$ from Fig.~\ref{fig:V:V3} with $U^{(4)}$ from Fig.~\ref{fig:V:U4}.
			}\label{fig:V:V4}
		\end{subfigure}\qquad
		\begin{subfigure}[t]{0.47\textwidth}
			\centering
			\hspace*{-0.2cm}\includegraphics{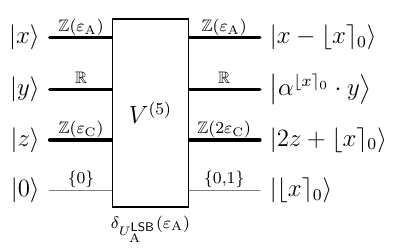}
			\caption{
				The action of $V^{(5)}$ is given by composing $V^{(4)}$ from Fig.~\ref{fig:V:V4} with $U^{(5)}$ from Fig.~\ref{fig:V:U5}.
			}\label{fig:V:V5}
		\end{subfigure}\\[0.5cm]
		\begin{subfigure}[t]{0.47\textwidth}
			\centering
			\includegraphics{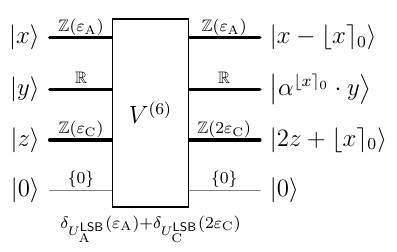}
			\caption{The action of $V^{(6)}$ is given by composing the unitary $V^{(5)}$ from Fig.~\ref{fig:V:V5} with the unitary $U^{(6)}$ from Fig.~\ref{fig:V:U6} and by noting that $\round{2z+\round{x}_0}_0=\round{2z}_0\oplus \round{x}_0=\round{x}_0$ for any $z\in\mathbb{Z}(\varepsilon_\bosonC)$ such that $2\varepsilon_\bosonC<1/2$.}\label{fig:V:V6}
		\end{subfigure}\qquad
		\begin{subfigure}[t]{0.47\textwidth}
			\centering
			\hspace*{0.6cm}\includegraphics{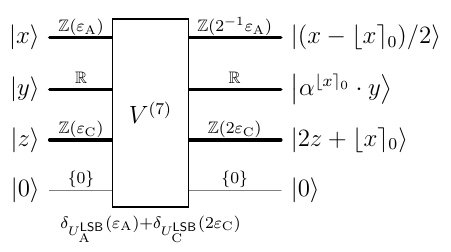}
			\caption{The action of the unitary $V^{(7)}$ is given by composing $V^{(6)}$ from Fig.~\ref{fig:V:V6} with $U^{(7)}$ from Fig.~\ref{fig:V:U7}.\label{fig:v7last}}
		\end{subfigure}
		\caption{Action of the unitaries $V^{(1)}$, \ldots, $V^{(7)}$ from the decomposition of the unitary~$V_{\alpha}$ obtained by repeated application of Lemma~\ref{lem: composition closeness}.}
		\label{fig:Valpha:Vs}
	\end{figure}

	We  also need the approximate action of the unitary $V_\alpha^\dagger$.
	\begin{lemma} \label{lem: V(a) dagger gate error}
		Let $\varepsilon_\bosonA\in(0,1/4)$ and $\varepsilon_\bosonC\in(0,1/2)$. Let $\alpha > 0$. Consider the unitary $V_\alpha^\dagger$, we have 
		\begin{align}
			\hspace*{2cm}
			\vcenter{\hbox{\includegraphics{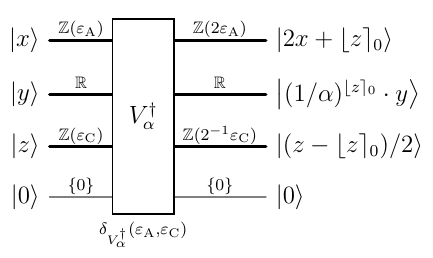}}}
		\end{align}	
		where $\delta_{V_\alpha^\dagger}(\varepsilon_\bosonA,\varepsilon_\bosonC)=7 (2\varepsilon_\bosonA)^{1/4}+7\varepsilon_\bosonC^{1/4}$.
	\end{lemma}
	Again, the error bound $\delta_{V^\dagger_\alpha}(\varepsilon_\bosonA,\varepsilon_\bosonC)$ does not depend on the parameter  $\alpha$. We will only need the action of $V^\dagger_\alpha$ restricted to inputs close to non-negative integers in modes $\bosonA$ and $\bosonC$, i.e.
	\begin{align}
		\hspace*{2cm} \raisebox{-65pt}{\mbox{\includegraphics{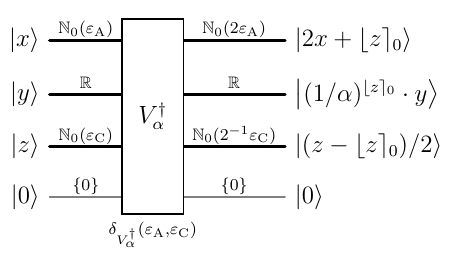}}} \label{eq: V alpha dagger action restricted}
	\end{align}
	This action is well-defined since 
	$x\in\mathbb{N}_0(\varepsilon_\bosonA)$ implies that~$2x+\round{z}_0\in\mathbb{N}_0(2\varepsilon_\bosonA)$ for any $z\in \mathbb{N}_0(\varepsilon_\bosonC)$, and 
	$z\in\mathbb{N}_0(\varepsilon_\bosonC)$ implies that $(z-\round{z}_0)/2 \in \mathbb{N}_0(2^{-1}\varepsilon_\bosonC)$.
	\begin{proof}[Proof of Lemma~\ref{lem: V(a) dagger gate error}]
		The adjoint of $V_\alpha$ from Eq.~\eqref{eq: Valpha circuit} is
		\begin{align}
			V_\alpha^\dagger&=(U^{(1)})^\dagger \cdots (U^{(7)})^\dagger
		\end{align}
		where the unitaries $U^{(1)},\ldots, U^{(7)}$ are defined in Eq.~\eqref{eq: Valpha Us}. The approximate action of the unitaries decomposing $V_\alpha^\dagger$ is given in Fig.~\ref{fig:V dagger:Us}. We will use Lemma~\ref{lem: composition closeness} to derive the composed action. Let us write $(V')^{(T)}=(U^{(7-T+1)})^\dagger\cdots (U^{(7)})^\dagger$ for $T\in\{2,\ldots, 7\}$. Each $(V')^{(T)}$ is a unitary composed of the first $T$ unitaries from the decomposition of $V_\alpha^\dagger$. The approximate action of each of them is given in Fig.~\ref{fig:Valpha dagger:Vs}. The claim follows from $(V')^{(7)}=V_\alpha^\dagger$ and by setting $\delta_{V_\alpha^{\dagger}}(\varepsilon_\bosonA,\varepsilon_\bosonC)=\delta_{U_\bosonC^{\lsb}}(\varepsilon_\bosonC)+\delta_{U_\bosonA^{\lsb}}(2\varepsilon_\bosonA )=7 (2\varepsilon_\bosonA)^{1/4}+7 \varepsilon_\bosonC^{1/4}$.
	\end{proof}

	\begin{figure}[p]
		\centering
		\begin{subfigure}[t]{0.47\textwidth}
			\raisebox{27pt}{\includegraphics{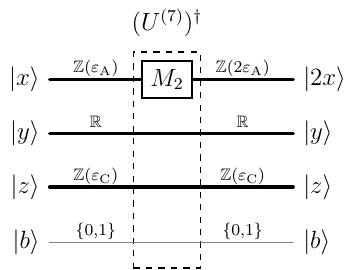}}
			\caption{Action of the unitary $(U^{(7)})^{\dagger}$.}
			\label{fig:Vdagger:U1}
		\end{subfigure}\qquad
		\begin{subfigure}[t]{0.47\textwidth}
			\includegraphics{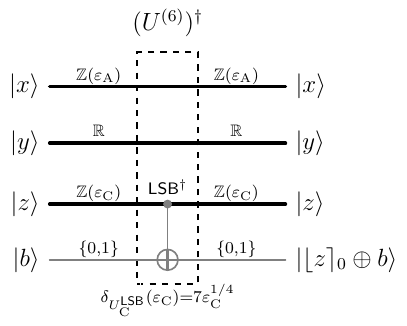}
			\caption{Action of $(U^{(6)})^\dagger$ (by Corollary~\ref{cor:LSBtracedist}).}
			\label{fig:Vdagger:U2}
		\end{subfigure}\\[0.3cm]
		\begin{subfigure}[t]{0.47\textwidth}
			\includegraphics{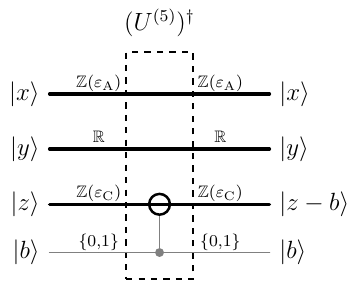}
			\caption{Action of the unitary $(U^{(5)})^{\dagger}$.}
			\label{fig:Vdagger:U3}
		\end{subfigure}\qquad
		\begin{subfigure}[t]{0.47\textwidth}
			\includegraphics{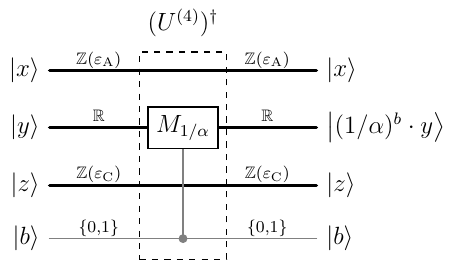}
			\caption{Action of the unitary $(U^{(4)})^\dagger$.}
			\label{fig:Vdagger:U4}
		\end{subfigure}\\[0.3cm]
		\begin{subfigure}[t]{0.47\textwidth}
			\raisebox{+20pt}{\includegraphics{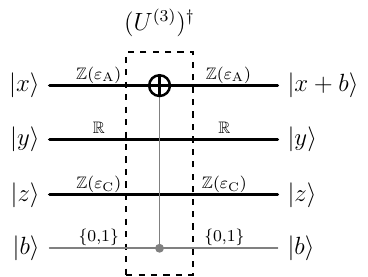}}
			\caption{Action of the unitary $(U^{(3)})^\dagger$.}
			\label{fig:Vdagger:U5}
		\end{subfigure}\qquad
		\begin{subfigure}[t]{0.47\textwidth}
			\includegraphics{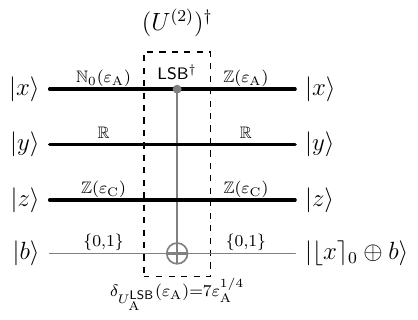}
			\caption{Action of $(U^{(2)})^\dagger$ (by Corollary~\ref{cor:LSBtracedist}).}
			\label{fig:Vdagger:U6}
		\end{subfigure}\\[0.3cm]
		\begin{subfigure}[b]{0.47\textwidth}
			\includegraphics{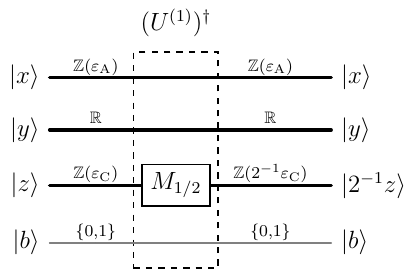}
			\caption{Action of the unitary $(U^{(1)})^\dagger$.}
			\label{fig:Vdagger:U7}
		\end{subfigure}\\
		\caption{The approximate action of the unitaries decomposing $V_{\alpha}^\dagger=(U^{(1)})^\dagger(U^{(2)})^\dagger\cdots (U^{(7)})^\dagger$ in the diagramatic notation~\eqref{eq: diagramatic notation}. The unitaries are ordered in sub-figures in the order they are applied in the circuit~$V_{\alpha}^\dagger$.}
		\label{fig:V dagger:Us}
	\end{figure}

	\begin{figure}[!htbp]
		\centering
		\begin{subfigure}[t]{0.47\textwidth}
			\centering
			\hspace*{-1cm}\includegraphics{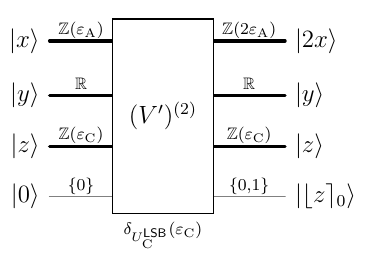}
			\caption{The action of the unitary $(V')^{(2)}$ is a composition of $(U^{(7)})^\dagger$ with $(U^{(6)})^\dagger$ from Figs.~\ref{fig:Vdagger:U1} and~\ref{fig:Vdagger:U2}, respectively.}\label{fig:Vdagger:V2}
		\end{subfigure}\qquad
		\begin{subfigure}[t]{0.47\textwidth}
			\centering
			\includegraphics{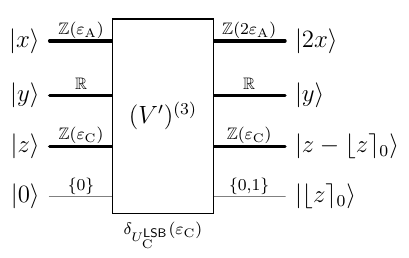}
			\caption{The action of the unitary $(V')^{(3)}$ is given by composing $(V')^{(2)}$ from Fig.~\ref{fig:Vdagger:V2} with $(U^{(5)})^\dagger$ from Fig.~\ref{fig:Vdagger:U3}.}
			\label{fig:Vdagger:V3}
		\end{subfigure}\\[0.5cm]
		\begin{subfigure}[t]{0.47\textwidth}
			\centering
			\includegraphics{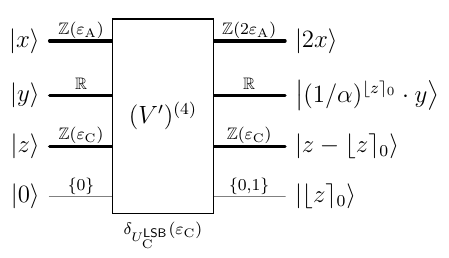}
			\caption{
				The action of $(V')^{(4)}$ is given by composing $(V')^{(3)}$ from Fig.~\ref{fig:Vdagger:V3} with $(U^{(4)})^\dagger$ from Fig.~\ref{fig:Vdagger:U4}.
			}\label{fig:Vdagger:V4}
		\end{subfigure}\qquad
		\begin{subfigure}[t]{0.47\textwidth}
			\centering
			\includegraphics{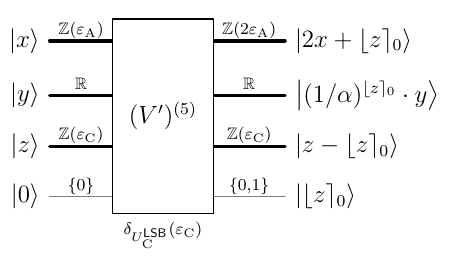}
			\caption{
				The action of $(V')^{(5)}$ is given by composing $(V')^{(4)}$ from Fig.~\ref{fig:Vdagger:V4} with $(U^{(3)})^\dagger$ from Fig.~\ref{fig:Vdagger:U5}.
			}\label{fig:Vdagger:V5}
		\end{subfigure}\\[0.5cm]
		\begin{subfigure}[t]{0.47\textwidth}
			\centering
			\includegraphics{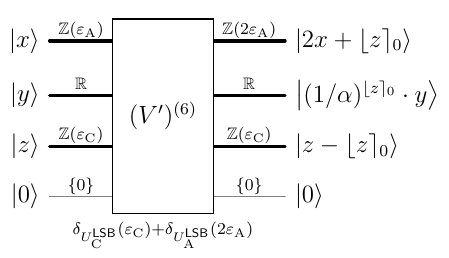}
			\caption{The action of $(V')^{(6)}$ is given by composing $(V')^{(5)}$ from Fig.~\ref{fig:Vdagger:V5} with the unitary $(U^{(2)})^\dagger$ from Fig.~\ref{fig:Vdagger:U6} and by the fact that $\round{2x+\round{z}_0}_0=\round{2x}_0\oplus \round{z}_0=\round{z}_0$ for any $x\in\mathbb{Z}(\varepsilon_\bosonC)$ with $2\varepsilon_\bosonC<1/2$.}\label{fig:Vdagger:V6}
		\end{subfigure}\qquad
		\begin{subfigure}[t]{0.47\textwidth}
			\centering
			\includegraphics{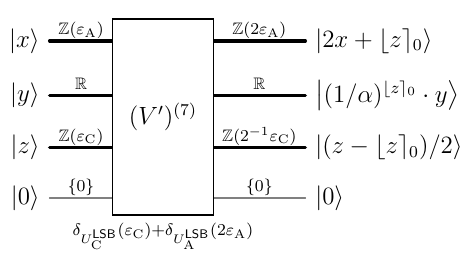}
			\caption{The action of $(V')^{(7)}$ is given by composing $(V')^{(6)}$ from Fig.~\ref{fig:Vdagger:V6} with $(U^{(1)})^\dagger$ from Fig.~\ref{fig:Vdagger:U7}.}
		\end{subfigure}
		\caption{Actions of the unitaries $(V')^{(1)}$, \ldots, $(V')^{(7)}$ from the decomposition of~$V_{\alpha}^\dagger$. These actions were obtained in a similar way as $V^{(1)}$, \ldots, $V^{(7)}$ from the decomposition of $V_\alpha$ in Fig.~\ref{fig:Valpha:Vs} (i.e., by Lemma~\ref{lem: composition closeness}).}
		\label{fig:Valpha dagger:Vs}
	\end{figure}

	To analyze the unitary $\Vam$ and its adjoint it is useful to first prove the following statement.
	\begin{lemma} \label{lem: V m V dagger m}
		Let $m\in\mathbb{N}$ and  $\vec{\alpha}=(\alpha_1,\ldots,\alpha_m)\in (0,\infty)^m$ and $\vec{\beta}=(\beta_1,\ldots,\beta_m)\in (0,\infty)^m$.
		Define the unitary
		\begin{align}
			V(\vec{\beta},\vec{\alpha})&=V^\dagger_{\beta_m}\cdots V^\dagger_{\beta_1}\cdot V_{\alpha_m}\cdots V_{\alpha_1}\  
		\end{align}
		and the functions
		\begin{align}
			\begin{matrix}
				f_{\vec{\alpha}}:&\mathbb{N}_0&\rightarrow & \mathbb{R}\\
				& x & \mapsto  & \prod_{i=1}^{m} \alpha_{i}^{x_{i-1}}
			\end{matrix}\qquad\quad\textrm{ and }\quad\qquad
			\begin{matrix}
				g_{\vec{\beta}}:&\mathbb{N}_0&\rightarrow & \mathbb{R}\\
				&x & \mapsto & \prod_{i=1}^{m} (1/\beta_i)^{x_{m-i}} 
			\end{matrix} \ \ \ ,
		\end{align}
		where $x=\sum_{i=0} 2^ix_i$.
		Let $\varepsilon_\bosonA\in(0,1/4)$ and let $\varepsilon_\bosonC\in(0,2^{-(m+2)})$.
		Then
		\begin{align}\hspace*{2.5cm}
			\vcenter{\hbox{\includegraphics{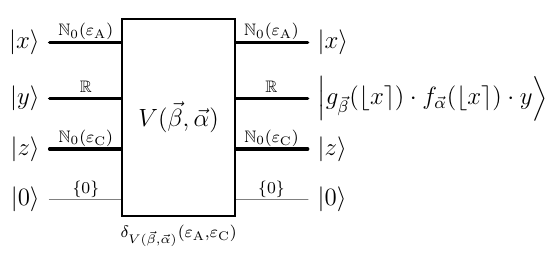}}}
		\end{align}
		where $\delta_{V(\vec{\beta},\vec{\alpha})}(\varepsilon_\bosonA ,\varepsilon_\bosonC ) = 88\cdot \varepsilon_\bosonA^{1/4} + 88\cdot (2^{m}\varepsilon_\bosonC)^{1/4}$.
	\end{lemma}
	\begin{proof} We will use Lemma~\ref{lem: composition closeness} to obtain the approximate action of $V(\vec{\beta},\vec{\alpha})$ from the approximate actions of its constituent unitaries. Set 
		\begin{align}
			(x^{(0)},y^{(0)},z^{(0)},\varepsilon_\bosonA^{(0)},\varepsilon_\bosonC^{(0)})&=(x,y,z,\varepsilon_\bosonA,\varepsilon_\bosonC)\ .
		\end{align}
		By Lemma~\ref{lem: V(a) unitary error}
		(that is, Eq.~\eqref{eq: V alpha action restricted}) we have
		\begin{align}
			\hspace*{-1.2cm}
			\vcenter{\hbox{\includegraphics{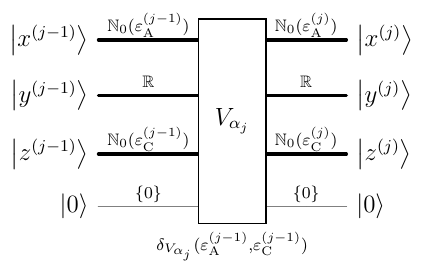}}}
		\end{align}	
		where
		\begin{align}
			\begin{array}{ccl}
				x^{(j)}&=&(x^{(j-1)}-\round{x^{(j-1)}}_0)/2\\
				y^{(j)}&=&\alpha_j^{\round{x^{(j-1)}}_0}\cdot y^{(j-1)}\\
				z^{(j)}&=&2z^{(j-1)}+\round{x^{(j-1)}}_0
			\end{array}\qquad\textrm{ and }\qquad 
			\begin{matrix}
				\varepsilon_\bosonA^{(j)}&=&2^{-1}\varepsilon_\bosonA^{(j-1)}\\
				\varepsilon_\bosonC^{(j)}&=&2 \varepsilon_{\bosonC}^{(j-1)}
			\end{matrix}\label{eq:Valpha input output behaviour}
		\end{align}
		for any $j\in \{1,\ldots,m\}$. Note that Lemma~\ref{lem: V(a) unitary error} is applicable here since $\varepsilon_{\bosonA}^{(j)}=2^{-j}\varepsilon_\bosonA$,  $\varepsilon^{(j)}_\bosonC=2^j\varepsilon_\bosonC$ and   $2^{-j}\varepsilon_\bosonA<1/4$,  $2^j\varepsilon_\bosonC<1/4$  by the  assumptions on $\varepsilon_\bosonA$ and $\varepsilon_\bosonC$,  for all $j\in\{1,\ldots,m\}$.  It follows inductively that
		\begin{align}
			\begin{array}{ccl}
				x^{(m)}&=&2^{-m}x-\sum_{j=0}^{m-1}2^{-m+j}\round{x}_j\\
				y^{(m)}&=&\left(\prod_{j=1}^{m} \alpha_j^{\round{x}_{j-1}}\right)\cdot y = f_{\vec{\alpha}}(\round{x})\cdot y \\
				z^{(m)}&=&2^m z+\sum_{j=0}^{m-1}2^{m-1-j}\round{x}_j
			\end{array}\qquad\textrm{ and }\qquad 
			\begin{matrix}
				\varepsilon_\bosonA^{(m)}&=&2^{-m}\varepsilon_\bosonA\\
				\varepsilon_\bosonC^{(m)}&=&2^m\varepsilon_\bosonC\
			\end{matrix} \label{eq: V alpha beta m step}
		\end{align}
		\newcommand*{\px}{\tilde{x}}
		\newcommand*{\py}{\tilde{y}}
		\newcommand*{\pz}{\tilde{z}}
		\newcommand*{\peps}{\tilde{\varepsilon}} 
		Let us set 
		\begin{align}    (\px^{(0)},\py^{(0)},\pz^{(0)},\peps_\bosonA^{(0)},\peps_\bosonC^{(0)})&=(x^{(m)},y^{(m)},z^{(m)},\varepsilon^{(m)}_\bosonA,\varepsilon^{(m)}_\bosonC)\ .\label{eq:pxpypzdef}
		\end{align}
		By Lemma~\ref{lem: V(a) dagger gate error} (more precisely, Eq.~\eqref{eq: V alpha dagger action restricted}), we have
		\begin{align}
			\hspace*{-1.2cm} \vcenter{\hbox{\includegraphics{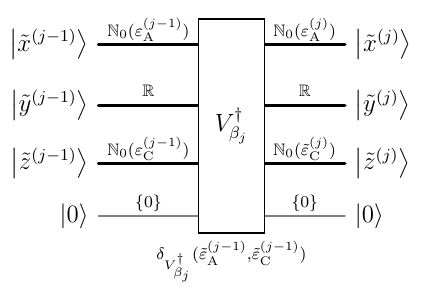}}}
		\end{align}	
		where
		\begin{align}
			\begin{array}{ccl}
				\px^{(j)}&=&2\px^{(j-1)}+\round{\pz^{(j-1)}}_0\\
				\py^{(j)}&=&(1/\beta_j)^{\round{\pz^{(j-1)}}_0}\cdot \py^{(j-1)}\\
				\pz^{(j)}&=&(\pz^{(j-1)}-\round{\pz^{(j-1)}}_0)/2
			\end{array}\qquad\textrm{ and }\qquad 
			\begin{matrix}
				\peps_\bosonA^{(j)}&=&2\peps_\bosonA^{(j-1)}\\
				\peps_\bosonC^{(j)}&=&2^{-1}\peps_{\bosonC}^{(j-1)}
			\end{matrix}\label{eq:Valpha input output behaviour tilde}
		\end{align}
		for any $j\in \{1,\ldots,m\}$. Inductively this implies that
		\begin{align}
			\peps_\bosonA^{(j)}=2^{j}\varepsilon_\bosonA^{(m)}=2^{-m+j}\varepsilon_\bosonA\qquad\textrm{ and }\qquad \peps_\bosonC^{(j)}=2^{-j}\varepsilon_\bosonC^{(m)}=2^{m-j}\varepsilon_{\bosonC}\qquad \textrm{ for all } j\in\{1,\ldots, m\}
		\end{align}
		where we used~\eqref{eq: V alpha beta m step}. Thus we can apply Lemma~\ref{lem: V(a) dagger gate error} since by assumptions on $\varepsilon_\bosonA$ and $\varepsilon_\bosonC$ we have $\peps_\bosonA^{(j)}<1/4$ and $\peps_\bosonC^{(j)}<1/4$ for all $j\in\{1,\ldots, m\}$.

		We claim that
		\begin{align}
			\begin{matrix}
				\px^{(j)}&=&x^{(m-j)}\\
				\pz^{(j)}&=&z^{(m-j)}
			\end{matrix}\qquad\textrm{ for all }\qquad j\in \{0,\ldots,m\}\ . \label{eq: px x pz z}
		\end{align}
		The claim is trivially true for~$j=0$ by definition (see Eq.~\eqref{eq:pxpypzdef}). 
		Inductively, we have for $j\in \{1,\ldots,m\}$
		\begin{align}
			\px^{(j)}&=2\px^{(j-1)}+\round{\pz^{(j-1)}}_0\qquad\ \ \textrm{ by definition (Eq.~\eqref{eq:Valpha input output behaviour tilde})}\\
			&=2x^{(m-(j-1))}+\round{z^{(m-(j-1))}}_0\qquad\textrm{ by the induction hypothesis}\\
			&=(x^{(m-(j-1)-1)}-\round{x^{(m-(j-1)-1)}}_0)+
			\round{2z^{(m-(j-1)-1)}+\round{x^{(m-(j-1)-1)}}_0}_0
			\quad\textrm{by definition (Eq.~\eqref{eq:Valpha input output behaviour})}\notag \\
			&=(x^{(m-(j-1)-1)}-\round{x^{(m-(j-1)-1)}}_0)+\round{x^{(m-(j-1)-1)}}_0\\
			&=x^{(m-j)}\ .
		\end{align}
		In the penultimate step, we used that $z^{(m-j)}\in \mathbb{N}_0(\varepsilon^{(m-j)}_\bosonC)\subseteq \mathbb{N}_0(\varepsilon)$ for some $\varepsilon<1/4$. Similarly, we have
		\begin{align}
			\pz^{(j)}&=(\pz^{(j-1)}-\round{\pz^{(j-1)}}_0)/2\qquad\qquad\ \ \textrm{ by definition (Eq.~\eqref{eq:Valpha input output behaviour tilde})}\\
			&=(z^{(m-(j-1))}-\round{z^{(m-(j-1))}}_0)/2 \qquad\textrm{ by the induction hypothesis}\\
			&=(2z^{(m-(j-1)-1)}+\round{x^{(m-(j-1)-1)}}_0-\round{2z^{(m-(j-1)-1)}+\round{x^{(m-(j-1)-1)}}_0}_0)/2\quad\textrm{by definition (Eq.~\eqref{eq:Valpha input output behaviour})}\notag\\
			&=(2z^{(m-j)}+\round{x^{(m-j)}}_0-\round{x^{(m-j)}}_0)/2\\
			&=z^{(m-j)}\ .
		\end{align}
		In the penultimate step, we used that $z^{(m-(j-1)-1)}=z^{(m-j)}\in \mathbb{N}_0(\varepsilon^{(m-j)}_\bosonC)\subseteq \mathbb{N}_0(\varepsilon)$ for some $\varepsilon<1/4$. 
		
		Next, we show  that 
		\begin{align}
			\py^{(j)}=\left(\prod_{i=1}^{j} (1/\beta_i)^{\round{x}_{m-i}}\right) \cdot y^{(m)}\qquad\textrm{ for all }\qquad j\in\{0,\ldots,m\}\ .\label{eq:pymrec}
		\end{align} Eq.~\eqref{eq:pymrec} is obtained inductively by using
		$\py^{(j)}=(1/\beta_j)^{\round{x}_{m-j}}\cdot \py^{(j-1)}$, which can be shown as follows: We have 
		\begin{align}
			\py^{(j)} &= (1/\beta_j)^{\round{\pz^{(j-1)}}_0}\cdot \py^{(j-1)} \qquad \textrm{ by definition (Eq.~\eqref{eq:Valpha input output behaviour tilde})}\\
			&= (1/\beta_j)^{\round{z^{(m-(j-1))}}_0}\cdot \py^{(j-1)} \qquad \textrm{ by~\eqref{eq: px x pz z}}\\
			&= (1/\beta_j)^{\round{2z^{(m-j)}+\round{x^{(m-j)}}_0}_0}\cdot \py^{(j-1)} \\
			&= (1/\beta_j)^{\round{x^{(m-j)}}_0}\cdot \py^{(j-1)}\label{eq:py to expr one}\\
			&= (1/\beta_j)^{\round{x}_{m-j}}\cdot \py^{(j-1)} \label{eq:xm to x bit}
		\end{align}
		where we obtained~\eqref{eq:py to expr one} from $z^{(m-j)}\in\mathbb{N}_0(\varepsilon_\bosonC^{(m-j)})\subseteq \mathbb{N}_{0}(\varepsilon)$ for some $\varepsilon<1/4$ and Eq.~\eqref{eq:xm to x bit} by 
		\begin{align}
			\round{x^{(j)}}_0=\round{x}_{j} \qquad\textrm{ for $j\in\{0,\ldots, m\}$} \ . \label{eq: xj to jth bit}   
		\end{align}
		This claim can be seen as follows. We have
		\begin{align}
			\round{x^{(j)}}_0 &= \round{(x^{(j-1)}-\round{x^{(j-1)}}_0)/2}_0 \qquad \textrm{ by definition (see~\eqref{eq:Valpha input output behaviour})}\\
			&=\round{x^{(j-1)}}_1 \label{eq:xm to bit one}
		\end{align}
		where we used the identity $\round{(x-\round{x}_0)/2}_0=\round{x}_1$ for any $x\in\mathbb{R}$. We obtain~\eqref{eq: xj to jth bit} by induction.

		In conclusion, we have shown the following action of the unitary $V(\vec{\beta},\vec{\alpha})$.
		\begin{align}
			\vcenter{\hbox{\includegraphics{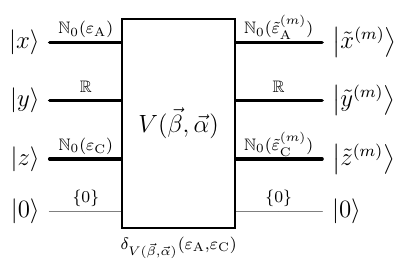}}}
		\end{align}
		where  
		\begin{align}
			\begin{array}{ccl}
				\px^{(m)}&=&x^{(0)} = x\\
				\py^{(m)}&=&\left(\prod_{j=1}^{m} (1/\beta_j)^{\round{x}_{m-j}}\right) \cdot y^{(m)} = g_{\vec{\beta}}(\round{x})\cdot f_{\vec{\alpha}}(\round{x})\cdot y \\
				\pz^{(m)}&=&z^{(0)} = z
			\end{array}\quad\textrm{ and }\quad 
			\begin{matrix}
				\tilde{\varepsilon}_\bosonA^{(m)}&=&\varepsilon_\bosonA\\
				\tilde{\varepsilon}_\bosonC^{(m)}&=&\varepsilon_\bosonC \ 
			\end{matrix}\ .
		\end{align}
		
		It remains to determine the error bound~$\delta_{V(\vec{\beta},\vec{\alpha})}(\varepsilon_\bosonA,\varepsilon_\bosonC)$. 
		By Lemma~\ref{lem: V(a) unitary error} we have the  error bound 
		\begin{align}
			\delta_{V^\dagger_{\alpha_j}}(\varepsilon_{\bosonA}^{(j-1)},  \varepsilon_\bosonC^{(j-1)}) = 7(\varepsilon_{\bosonA}^{(j-1)})^{1/4}+7\cdot(2\varepsilon_\bosonC^{(j-1)})^{1/4}\ ,
		\end{align}
		and by Lemma~\ref{lem: V(a) dagger gate error} we have
		\begin{align}
			\delta_{V^\dagger_{\beta_j}}(\tilde{\varepsilon}_{\bosonA}^{(j-1)},  \tilde{\varepsilon}_\bosonC^{(j-1)}) = 7(\tilde{\varepsilon}_{\bosonA}^{(j-1)})^{1/4}+7\cdot(2\tilde{\varepsilon}_\bosonC^{(j-1)})^{1/4} \ .
		\end{align}
		We obtain an error bound~$\delta_{V_{\alpha_m}\cdots V_{\alpha_1}}$ on the approximate action of the composition of unitaries $V_{\alpha_m}\cdots V_{\alpha_1}$ from Lemma~\ref{lem: composition closeness}, that is, we can use 
		\begin{align}
			\delta_{V_{\alpha_m}\cdots V_{\alpha_1}}&=\sum_{j=1}^{m} \delta_{V_{\alpha_j}}(\varepsilon_\bosonA^{(j-1)},\varepsilon_\bosonC^{(j-1)}) \\
			&=\sum_{j=1}^{m} \left( 7(\varepsilon_{\bosonA}^{(j-1)})^{1/4}+7(2\varepsilon_\bosonC^{(j-1)})^{1/4}\right)\\
			&=\sum_{j=1}^{m} \left(7(2^{-(j-1)} \varepsilon_\bosonA)^{1/4}+7(2\cdot 2^{j-1}\varepsilon_\bosonC)^{1/4}\right)\\
			&=7\varepsilon_\bosonA^{1/4} \sum_{j=0}^{m-1} \left(\frac{1}{2^{1/4}}\right)^j + 7(2\varepsilon_\bosonC)^{1/4}\sum_{j=0}^{m-1} \left(2^{1/4}\right)^j\\
			&=7\varepsilon_\bosonA^{1/4} \frac{1-2^{-m/4}}{1-2^{-1/4}}+7(2\varepsilon_\bosonC)^{1/4}\frac{1-2^{m/4}}{1-2^{1/4}} \\
			&= 7\varepsilon_\bosonA^{1/4} \frac{1-2^{-m/4}}{1-2^{-1/4}}+7 \varepsilon_\bosonC^{1/4}\frac{2^{m/4}-1}{1-2^{-1/4}} \\
			&\le 44\cdot \varepsilon_\bosonA^{1/4}+44\cdot(2^m\varepsilon_\bosonC)^{1/4}
			\ .\label{eq:ctrlMf delta m}
		\end{align}
		An error bound on the composition of unitaries $V^\dagger_{\beta_m}\cdots V^\dagger_{\beta_1}$ is obtained analogously as
		\begin{align}
			\delta_{V^\dagger_{\beta_m}\cdots V^\dagger_{\beta_1}}&=	\sum_{j=1}^{m} \delta_{V^\dagger_{\beta_j}}(\varepsilon_\bosonA^{(m+j)},\varepsilon_\bosonC^{(m+j)}) \\
			&=\sum_{j=1}^{m} \left( 7(2\cdot \varepsilon_{\bosonA}^{(m+j)})^{1/4}+7(\varepsilon_\bosonC^{(m+j)})^{1/4}\right)\\
			&=\sum_{j=1}^{m} \left( 7(2\cdot 2^{-m+(j-1)} \varepsilon_{\bosonA})^{1/4}+7(2^{m-(j-1)}\varepsilon_\bosonC)^{1/4}\right)\\
			&=7\cdot 2^{-(m-1)/4}\varepsilon_\bosonA^{1/4}\sum_{j=0}^{m-1} (2^{1/4})^{j}+7\cdot 2^{m/4}\varepsilon_\bosonC^{1/4} \sum_{j=0}^{m-1}\left(\frac{1}{2^{1/4}}\right)^{j} \\
			&=7\cdot 2^{-(m-1)/4}\varepsilon_\bosonA^{1/4}\frac{1-2^{m/4}}{1-2^{1/4}}+7\cdot 2^{m/4}\varepsilon_\bosonC^{1/4} \frac{1-2^{-m/4}}{1-2^{-1/4}}\\
			&=7\cdot \varepsilon_\bosonA^{1/4}\frac{1-2^{-m/4}}{1-2^{-1/4}}+7\cdot \varepsilon_\bosonC^{1/4} \frac{2^{m/4}-1}{1-2^{-1/4}}\\
			&\le 44 \cdot \varepsilon_A^{1/4}+44\cdot (2^m \varepsilon_\bosonC)^{1/4}
			\ .\label{eq:ctrlMf delta prime m}
		\end{align}
		Combining~\eqref{eq:ctrlMf delta m} and~\eqref{eq:ctrlMf delta prime m} yields
		\begin{align}
			\delta_{V_{\alpha_m}\cdots V_{\alpha_1}} + \delta_{V^\dagger_{\beta_m}\cdots V^\dagger_{\beta_1}} 
			&\le 88\cdot \varepsilon_\bosonA^{1/4} + 88\cdot (2^{m}\varepsilon_\bosonC)^{1/4}  \ . \label{eq: VaNm delta}
		\end{align}
		We set $\delta_{V(\vec{\beta},\vec{\alpha})}(\varepsilon_\bosonA ,\varepsilon_\bosonC ):=88\cdot \varepsilon_\bosonA^{1/4} + 88\cdot( 2^{m}\varepsilon_\bosonC)^{1/4}$.
		This completes the proof.
	\end{proof}

	We show the properties of the unitary $\Vam$ defined in Fig.~\ref{fig:composite gates}  of the main paper, whose definition we recall here:
	\begin{align}
		\raisebox{-47pt}{\mbox{\includegraphics{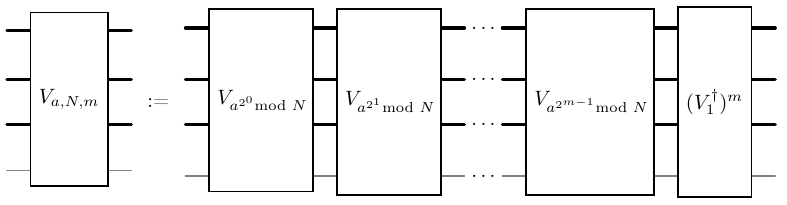}}}
		\label{eq: VaNm circuit}
	\end{align}
	where we used the $m$th power of the unitary $(V_1^\dagger)^m$ to indicate the $m$-fold application of the unitary $V_1^\dagger$.

	\begin{samepage}
		\begin{lemma} \label{lem: controlled M_abullet}
			Let $\varepsilon_\bosonA\in(0,1/4)$ and let $\varepsilon_\bosonC\in(0,2^{-(m+2)})$.
			Consider the unitary $\Vam$ described by the circuit in Fig.~\ref{fig:composite gates}  of the main text. 
			Then
			\begin{align}\hspace*{1.5cm}
				\vcenter{\hbox{\includegraphics{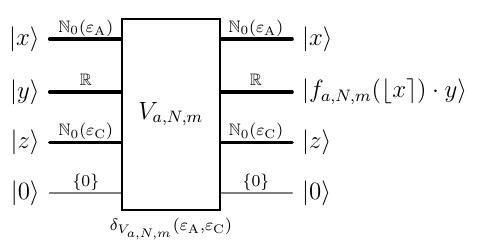}}}
			\end{align}
			where the function $f_{a,N,m}(\round{x})=\prod_{i=0}^{m-1} \left(a^{2^i}\xmod N \right)^{\round{x}_i}$ is the pseudomodular power and where $\delta_{V_{a,N,m}}(\varepsilon_\bosonA ,\varepsilon_\bosonC ) = 88\cdot \varepsilon_\bosonA^{1/4} + 88\cdot (2^{m}\varepsilon_\bosonC)^{1/4}$.
		\end{lemma}
	\end{samepage}
	\begin{proof}
		We apply Lemma~\ref{lem: V m V dagger m} with 
		\begin{align}
			\vec{\alpha}&=(\alpha_1,\ldots,\alpha_m) \qquad \textrm{ where } \alpha_j=a^{2^{j-1}}\xmod N \ ,\\
			\vec{\beta}&=(\beta_1,\ldots,\beta_m) \qquad \textrm{ where } \beta_j=1 \ 
		\end{align}
		for all $j\in\{1,\ldots m\}$. The claim follows since $\Vam=V(\vec{\beta},\vec{\alpha})$ by definition and since
		\begin{align}
			f_{\vec{\alpha}}(\round{x})&=\prod_{i=1}^{m} \left(a^{2^{i-1}}\xmod N \right)^{\round{x}_{i-1}}=f_{a,N,m}(\round{x}) \qquad \text{and}\\
			g_{\vec{\beta}}(\round{x})&=1\ .
		\end{align}
	\end{proof}
	
	We will need the following statement about the inverse $\Vam^\dagger$.
	\begin{lemma} \label{lem: controlled M_abullet dagger}
		Let $\varepsilon_\bosonA\in (0,1/4)$ and let $\varepsilon_\bosonC\in (0,2^{-(m+2)})$.
		Consider the unitary $\Vam$ described by the circuit in Fig.~\ref{fig:composite gates}  of the main text. 
		Then the action of its inverse $\Vam^\dagger$ is
		\begin{align}
			\hspace*{1.5cm}
			\raisebox{-70pt}{\mbox{\includegraphics{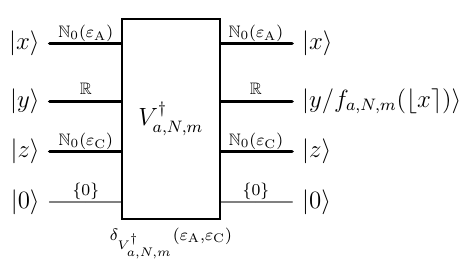}}}
			\ \ \ ,\label{eq:ctrlMf Wm+1}
		\end{align}
		where $\delta_{V_{a,N,m}^\dagger}(\varepsilon_\bosonA ,\varepsilon_\bosonC) = \delta_{V_{a,N,m}}(\varepsilon_\bosonA ,\varepsilon_\bosonC)= 88\cdot \varepsilon_\bosonA^{1/4} + 88\cdot (2^{m}\varepsilon_\bosonC)^{1/4}$
	\end{lemma}
	\begin{proof}
		By taking the adjoint of $\Vam$ in Eq.~\eqref{eq: VaNm circuit} we have 
		\begin{align}
			\Vam^\dagger = \left((V_{1}^\dagger)^m \cdot V_{a^{2^{m-1}}\xmod N}\cdots V_{a^{2^0}\xmod N}\right)^\dagger=V_{a^{2^0}\xmod N}^\dagger \cdots V_{a^{2^{m-1}}\xmod N}^\dagger\cdot (V_1)^m \ \ .
		\end{align}
		Thus we can apply Lemma~\ref{lem: V m V dagger m} with
		\begin{align}
			\vec{\alpha}&=(\alpha_1,\ldots,\alpha_m) \qquad \textrm{ where } \alpha_j = 1 \ , \\
			\vec{\beta}&=(\beta_1,\ldots,\beta_m) \qquad \textrm{ where } \beta_j=a^{2^{m-j}}\xmod N \ 
		\end{align}
		for all $j\in\{0,1\}$. The claim follows since
		\begin{align}
			f_{\vec{\alpha}}(\round{x})&=1  \qquad \text{and}\\
			g_{\vec{\beta}}(\round{x})&=\prod_{i=1}^{m} 1/\left(a^{2^{m-i}}\xmod N \right)^{\round{x}_{m-i}}=1/f_{a,N,m}(\round{x})\ .
		\end{align}
	\end{proof}

	We now analyse the gate 
	$\Ua$ implementing the controlled translation by a pseudomodular power 
	(see Fig.~\ref{fig:composite gates}  of the main paper) which we recall here.
	\begin{align}
		\vcenter{\hbox{\includegraphics{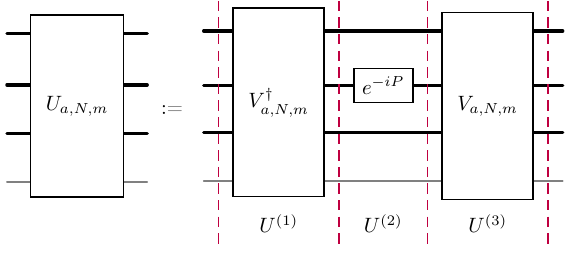}}}
	\end{align}
	This circuit 
	is composed of the unitaries
	\begin{align}
		\Ua=U^{(3)}U^{(2)}U^{(1)}\qquad\textrm{ where }
		\qquad
		\begin{matrix*}[l]
			U^{(1)}&=& \Vam^\dagger \\ 
			U^{(2)}&=& I_\bosonA \otimes e^{-iP_\bosonB}\otimes I_\bosonC\otimes I_\qubit\\
			U^{(3)}&=& \Vam 
		\end{matrix*}\ . \label{eq:UaNm decomposition}
	\end{align}
	The approximate actions of these unitaries on inputs which are close to integers on modes~$\bosonA$ and $\bosonC$ are given in Fig.~\ref{fig:UaNm Us approx action}.
	\begin{figure}[ht]
		\centering
		\begin{subfigure}[t]{0.47\textwidth}
			\includegraphics{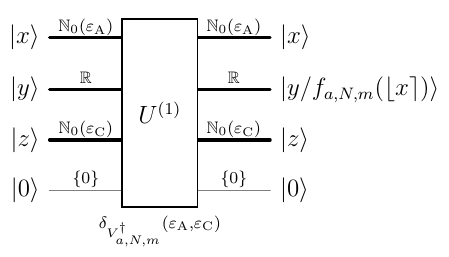}
			\caption{The action of the unitary $U^{(1)}$ is given by Lemma~\ref{lem: controlled M_abullet dagger}.
				Here $\delta_{V^\dagger_{a,N,m}}\!\!\!(\varepsilon_\bosonA,\varepsilon_\bosonC)=\delta_{V_{a,N,m}}(\varepsilon_\bosonA,\varepsilon_\bosonC)=88\cdot \varepsilon_\bosonA^{1/4} + 88\cdot (2^{m}\varepsilon_\bosonC)^{1/4}$.}\label{fig:UaNm U1}
		\end{subfigure}\hfill
		\begin{subfigure}[t]{0.47\textwidth}
			\raisebox{0.48cm}{\includegraphics{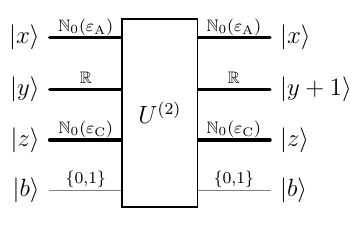}}
			\caption{The action of the unitary $U^{(2)}=I_\bosonA\otimes e^{-iP_\bosonB}\otimes I_\bosonC\otimes I_\qubit $. This unitary is implemented exactly.}\label{fig:UaNm U2}
		\end{subfigure}\\[0.5cm]
		\begin{subfigure}[t]{0.47\textwidth}
			\hspace*{0.3cm}\includegraphics{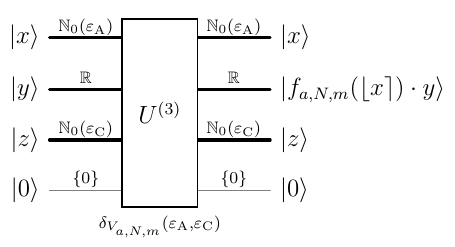}
			\caption{The action of the unitary $U^{(3)}$ given by Lemma~\ref{lem: controlled M_abullet} with $\delta_{V_{a,N,m}}(\varepsilon_\bosonA,\varepsilon_\bosonC)=88\cdot \varepsilon_\bosonA^{1/4} + 88\cdot (2^{m}\varepsilon_\bosonC)^{1/4}$.}\label{fig:UaNm U3}
		\end{subfigure}
		\caption{The approximate actions of the unitaries $U^{(1)}$, $U^{(2)}$ and $U^{(3)}$ constituting the unitary $\Uan=U^{(3)} U^{(2)} U^{(1)}$.}
		\label{fig:UaNm Us approx action}
	\end{figure}

	\begin{lemma}\label{lem: Ua shor circuit error} Let $\varepsilon_\bosonA\in(0,1/4)$ and $\varepsilon_\bosonC\in(0,2^{-(m+2)})$.
		Consider the unitary $\Ua$ described by the circuit in Fig.~\ref{fig:composite gates}  (of the main paper). Then
		\begin{center}
			\hspace*{1.25cm}\includegraphics{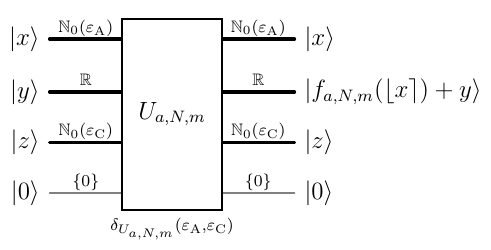}  
		\end{center}
		\vspace*{-0.3cm}
		where $\delta_{U_{a,N,m}}=176\cdot \varepsilon_\bosonA^{1/4} + 176\cdot (2^{m}\varepsilon_\bosonC)^{1/4}$.
	\end{lemma}
	\begin{proof}
		We use Lemma~\ref{lem: composition closeness} to obtain the approximate action of the combined unitaries $U^{(2)}U^{(1)}$ and $U^{(3)}U^{(2)}U^{(1)}$ from the decomposition of unitary $\Ua$. Those are displayed in Fig.~\ref{fig:UaNm accumulated actions}. The claim follows from the action of the unitary $U^{(3)}U^{(2)}U^{(1)}=\Ua$ shown in Fig.~\ref{fig:UaNm U3U2U1} and by the identity $\delta_{V_{a,N,m}^\dagger}(\varepsilon_\bosonA,\varepsilon_\bosonC)=\delta_{V_{a,N,m}}(\varepsilon_\bosonA,\varepsilon_\bosonC)=88\cdot \varepsilon_\bosonA^{1/4} + 88\cdot (2^{m}\varepsilon_\bosonC)^{1/4}$, see Fig.~\ref{fig:UaNm U1}.
	\end{proof}
	
	\begin{figure}[ht]
		\centering
		\begin{subfigure}[t]{\textwidth}
			\centering
			\hspace*{0.5cm}\includegraphics{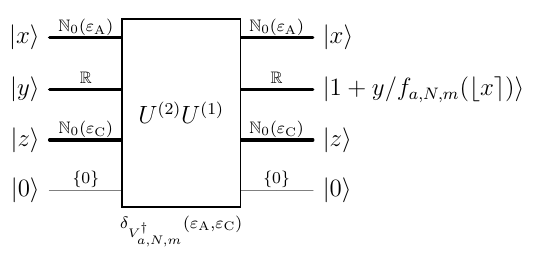}
			\caption{The action of the composed unitary $U^{(2)}U^{(1)}$ given by composing the action of $U^{(1)}$ from Fig.~\ref{fig:UaNm U1} with the action of $U^{(2)}$ from Fig.~\ref{fig:UaNm U2}.}
			\label{fig:UaNm U2U1}
		\end{subfigure}\\[0.5cm]
		\begin{subfigure}[t]{\textwidth}
			\centering
			\hspace*{0.3cm}\includegraphics{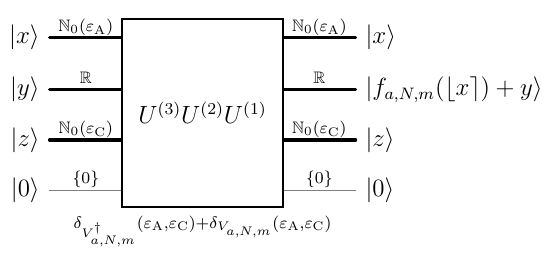}
			\caption{The action of the composed unitary $U^{(3)}U^{(2)}U^{(1)}=\Ua$ obtained by composing the action of the unitary $U^{(2)}U^{(1)}$ from Fig.~\ref{fig:UaNm U2U1} with the action of the unitary $U^{(3)}$ in Fig.~\ref{fig:UaNm U3}.}
			\label{fig:UaNm U3U2U1}
		\end{subfigure}
		\caption{Combined action of the unitaries $U^{(1)}$, $U^{(2)}$ and $U^{(3)}$ from Fig.~\ref{fig:UaNm Us approx action} composing $\Ua$ obtained by repeated application of Lemma~\ref{lem: composition closeness}.}
		\label{fig:UaNm accumulated actions}
	\end{figure}

	\subsection{Implications for the algorithm}\label{sec:approx function evaluation implication for algorithm}
	Finally, we can give a bound on the distance~$\left\|\Psi^{(1)}-\Psi^{(2)}\right\|_1$, the last remaining component of the proof of Proposition~\ref{prop:propositionone} (see Section~\ref{sec: prop one proof}). 
	
	\begin{lemma}\label{lem: approx_unitary} 
		Let $N\in\mathbb{N}$ be an $n$-bit integer. Let $(m,R,\kappa_\bosonA,\Delta_\bosonA,\kappa_\bosonB,\Delta_\bosonB, \Delta_\bosonC)$ be as specified in Table~\ref{tab: parameters} and let $\varepsilon_\bosonA=\sqrt{\Delta_\bosonA}$, $\varepsilon_\bosonB=\sqrt{\Delta_\bosonB}$ and $\varepsilon_\bosonC=\sqrt{\Delta_\bosonC}$. Consider the two states
		\begin{align}
			\ket{\Psi^{(1)}}&=\Uan\big(\cpsi{1}\cdot\Pi_{\scriptsize[-1/2,2R-1/2]} e^{-iRP}\ket{\gkp^{\varepsilon_\bosonA }_{\kappa_\bosonA,\Delta_\bosonA }}\otimes M_N\ket{\gkp^{\varepsilon_\bosonB }_{\kappa_\bosonB ,\Delta_\bosonB }}\otimes\ket{\Psi_{\Delta_\bosonC}^{\varepsilon_\bosonC}}\otimes \ket{0} \big)\\
			\intertext{and}
			\ket{\Psi^{(2)}} &= \cpsi{2}\sum_{z=0}^{2R-1} 
			\eta_{\kappa_\bosonA}(z-R)\ket{\chi^{\varepsilon_\bosonA }_{\Delta_\bosonA }(z)}\otimes
			e^{-if_{a,N,m}(z)P} M_N\ket{\gkp^{\varepsilon_\bosonB }_{\kappa_\bosonB ,\Delta_\bosonB }}\otimes \ket{\Psi_{\Delta_\bosonC}^{\varepsilon_\bosonC}}\otimes \ket{0} .
		\end{align}
		Then
		\begin{align}
			\left\|\Psi^{(1)}-\Psi^{(2)}\right\|_1
			&\le 352\cdot 2^{-2n}
			=:\varepsilon^{(2)} \ .
		\end{align}
	\end{lemma}
	
	\begin{proof}
		Since $\varepsilon_\bosonA<1/2$ can write 
		\begin{align}
			\ket{\Psi^{(1)}} &= U_{a,N,m}\left(\cpsi{1}\sum_{z=0}^{2R-1} 
			\eta_{\kappa_\bosonA}(z-R)\ket{\chi^{\varepsilon_\bosonA }_{\Delta_\bosonA }(z)}\otimes M_N\ket{\gkp^{\varepsilon_\bosonB }_{\kappa_\bosonB ,\Delta_\bosonB }}\otimes \ket{\Psi_{\Delta_\bosonC}^{\varepsilon_\bosonC}}\otimes \ket{0}\right)\ .
		\end{align}
		In Lemma~\ref{lem: Ua shor circuit error}
		we showed that the approximate action of the unitary~$\Ua$ is 
		\begin{center}
			\hspace*{1cm}\includegraphics{UaNm.pdf}
		\end{center}
		Since $\ket{\lambda+y}=e^{-i\lambda P}\ket{y}$ for any $\lambda,y\in \mathbb{R}$ we have that $\ket{f_{a,N,m}(\round{x})+y}=e^{-if_{a,N,m}(\round{x})P}\ket{y}$. Therefore by Definition~\ref{def: approx computation}, we have
		\begin{align}
			\left\|\Psi^{(1)}-\Psi^{(2)}\right\|_1
			& \le \delta_{U_{a,N,m}(\varepsilon_\bosonA,\varepsilon_\bosonC)}\\
			&= 176\cdot \varepsilon_\bosonA^{1/4} + 176\cdot (2^{m}\varepsilon_\bosonC)^{1/4} \\
			& = 176\cdot \varepsilon_\bosonA^{1/4}+176\cdot (2^{m}(2^{-m}\varepsilon_\bosonA))^{1/4}\\
			& = 352 \cdot \varepsilon_\bosonA^{1/4} \\
			& =352 \cdot 2^{-2n} \ ,
		\end{align}
		where we noted that $\varepsilon_\bosonC=2^{-m}\varepsilon_\bosonA$ and $\varepsilon_\bosonA = \sqrt{\Delta_\bosonA}$ and used the parameters
		$(m,R,\kappa_\bosonA,\Delta_\bosonA,\kappa_\bosonB,\\\Delta_\bosonB, \Delta_\bosonC)$ given in Table~\ref{tab: parameters}.
	\end{proof}

	\clearpage
	\section{Proof of Proposition~\ref{prop:propositiontwo}: Distribution of measurement outcomes}\label{sec: prop two proof}
	The goal of this section is to prove Proposition~\ref{prop:propositiontwo}. This proposition centers around the distribution~$p_{\Phi_a}$ on~$\mathbb{R}$ obtained by performing a homodyne $P$-quadrature measurement on mode~$\bosonA$ of the state
	$\Phi_a \in L^2(\mathbb{R})^{\otimes 3}\otimes \mathbb{C}^2$ defined in Eq.~\eqref{eq:stateoutputm} of Proposition~\ref{prop:propositionone} which we repeat here
	\begin{align}
		\!\!\!\ket{\Phi_a} &=c_{\Phi_a} \sum_{z\in \mathbb{Z}}
		\eta_{\kappa_\bosonA}(z-R)\ket{\chi_{\Delta_\bosonA }(z)}\otimes 
		e^{-i(a^z\xmod N)P_\bosonB}M_N\ket{\gkp^{\varepsilon_\bosonB }_{\kappa_\bosonB ,\Delta_\bosonB }}\otimes \ket{\Psi_{\Delta_\bosonC}^{\varepsilon_\bosonC}} \otimes \ket{0}\   \label{eq:stateoutputmrestated}
	\end{align}
	where~$c_{\Phi_a}>0$ is a normalization constant.
	Recall that by  Proposition~\ref{prop:propositionone}, this state is close to the output state of the circuit~$\cQ_{a,N}$. In particular, the following property of~$p_{\Phi_a}$ immediately implies an analogous property for the actual output distribution of the circuit (see Eq.~\eqref{eq: bound distribution by trace dist}). 
	
	The property of interest --- established by Proposition~\ref{prop:propositiontwo} --- is that the family of distributions~$\{p_{\Phi_a}\}_{a\in\mathbb{Z}_N^*}$ is suitable (for factoring) in the sense of Definition~\ref{def:suitabledistribution}, i.e., it satisfies
	\begin{align}
		\min_{d\in \mathbb{Z}_{r(a)}} \int_{\Gamma_{d}(a)} p_{\Phi_a}(w) dw\ge \Omega(1)\cdot \frac{1}{r(a)}\qquad\textrm{ for all }\qquad a\in\mathbb{Z}_N^*\ ,
	\end{align}
	where $r(a)$ is the multiplicative order of $a$ in $\mathbb{Z}_N^*$. For convenience, let us restate this claim.
	
	\propositiontwo*
	In order to compute the output distribution when measuring the momentum (i.e., $P$-quadrature) of mode~$\bosonA$ of the state $\Phi_a$, we proceed as follows: In Section~\ref{sec:fouriertransformPhi} we compute an expression of~$\Phi_a$ in momentum space (i.e., after a Fourier transform on $L^2(\mathbb{R})$). We will then analyze the Fourier-transform using Poisson's formula, and establish lower bounds on its modulus in certain intervals of interest, i.e., $\Gamma_{d}(a)$ for $d\in\mathbb{Z}_{r(a)}$, see  Section~\ref{sec:lowerboundsmodulusft}.
	Finally, in Section~\ref{sec:boundspdf}, we conclude the proof by  analyzing the probability density function~$p_{\Phi_a}$ given by Born's rule when measuring the $P$-quadrature of mode~$\bosonA$.

	\subsection{The Fourier transform of~$\Phi_a$\label{sec:fouriertransformPhi}}
	In this section, we compute the Fourier transform of~$\Phi_a$ given in~\eqref{eq:stateoutputmrestated}.
	For $a\in\mathbb{Z}_N^*$, it will be convenient to define the set 
	\begin{align}
		\Rem_{a,N}:=\{a^z \xmod N \mid z\in\mathbb{Z} \}\ .\label{eq: Rem a N}
	\end{align}
	That is, $\Rem_{a,N}\subseteq\mathbb{Z}_N^*$ is the subgroup generated by the element~$a$. 
	Observe that the set $\Rem_{a,N}$ has cardinality $|\Rem_{a,N}|=r(a)$ because $a\in\mathbb{Z}_N^*$ has order~$r(a)$.
	
	By straightforward computation, we show the following:
	\begin{lemma}\label{lem:fouriertransformandrhoa}
		We have 
		\begin{align}
			\ket{\Phi_a}=\sum_{k\in\Rem_{a,N}}\ket{\Theta_k}\otimes e^{-ikP_\bosonB}M_N\ket{\gkp^{\varepsilon_\bosonB }_{\kappa_\bosonB ,\Delta_\bosonB }} \otimes \ket{\Psi_{\Delta_\bosonC}^{\varepsilon_\bosonC}}\otimes \ket{0}\,\label{eq:thetakintroductions}
		\end{align}
		where for each~$k\in\Rem_{a,N}$, the  function~$\Theta_k$
		has the Fourier transform
		\begin{align}
			\widehat{\Theta}_k(w)&=c_{\Phi_a}\sum_{\substack{z\in\mathbb{Z}\\a^{z}\xmod N=k}}
			\eta_{\kappa_\bosonA}(z-R)\widehat{\Psi}_{\Delta_\bosonA }(w)e^{2\pi i zw}\ .
			\label{eq: F Theta k w}
		\end{align}
		Furthermore, the reduced density operator
		~$\rho_{\bosonA}=\tr_{\bosonB\bosonC\qubit} \proj{\Phi}$ after tracing out modes~$\bosonB$,$\bosonC$ and the qubit~$\qubit$ is equal to 
		\begin{align}
			\rho_{\bosonA} &=\sum_{k\in\Rem_{a,N}} \proj{\Theta_k}\ .\label{eq:rhoexpressionclaimb}
		\end{align}
	\end{lemma}
	We note that the vectors~$\{\ket{\Theta_k}\}_{k\in \Rem_{a,N}}$ in this expression are not normalized.
	\begin{proof}
		It is straightforward to verify that $\Phi_a$ in~\eqref{eq:stateoutputmrestated} can be written in the form of Eq.~\eqref{eq:thetakintroductions} with
		\begin{align}
			\ket{\Theta_k} &= c_{\Phi_a} \sum_{\substack{z\in \mathbb{Z}\\ a^z \xmod N=k}}\eta_{\kappa_\bosonA}(z-R)\ket{\chi_{\Delta_\bosonA }(z)}\qquad\textrm{ for }\qquad k\in\Rem_{a,N}\ .
		\end{align}
		By definition of the Fourier transform and $\chi_{\Delta_\bosonA}(\cdot)$, we have 
		\begin{align}
			\widehat{\Theta}_k(w)&=c_{\Phi_a}
			\sum_{\substack{z\in\mathbb{Z}\\a^{z}\xmod N=k}}
			\eta_{\kappa_\bosonA}(z-R)\int\Psi_{\Delta_\bosonA }(x-z)e^{2\pi i xw}dx\ .
		\end{align}
		Expression~\eqref{eq: F Theta k w} follows 
		from the fact that
		\begin{align}
			\int f(x-z)e^{2\pi i xw}dx&=\widehat{f}(w)e^{2\pi i z w}
		\end{align}
		for any $f\in L^1(\mathbb{R}) \cap L^2(\mathbb{R})$.
		
		The claim~\eqref{eq:rhoexpressionclaimb} about the reduced density operator~$\rho_{\bosonA}$ follows immediately from the pairwise orthogonality of the states $\left\{e^{-ikP_\bosonB}M_N\ket{\gkp^{\varepsilon_\bosonB }_{\kappa_\bosonB ,\Delta_\bosonB }}\right\}_{k\in \Rem_{a,N}}$ which we argue in the following. For $k\not= k'\in \Rem_{a,N}$ we have
		\begin{align}
			\langle e^{-ikP}M_N\gkp^{\varepsilon_\bosonB }_{\kappa_\bosonB ,\Delta_\bosonB },
			e^{-ik'P}M_N\gkp^{\varepsilon_\bosonB }_{\kappa_\bosonB ,\Delta_\bosonB }\rangle
			&=
			\langle \gkp^{\varepsilon_\bosonB }_{\kappa_\bosonB ,\Delta_\bosonB },
			M_N^\dagger e^{i(k'-k')P}M_N\gkp^{\varepsilon_\bosonB }_{\kappa_\bosonB ,\Delta_\bosonB }\rangle\\
			&=
			\langle \gkp^{\varepsilon_\bosonB }_{\kappa_\bosonB ,\Delta_\bosonB },
			M_N^\dagger e^{i(k-k')P/N}M_N\gkp^{\varepsilon_\bosonB }_{\kappa_\bosonB ,\Delta_\bosonB }\rangle \ .
		\end{align}
		It is easy to check that for $k\neq k'$ the real number
		$d=(k-k')/N$ is at distance at least~$1/N$ from the closest integer. Since  
		$\varepsilon_\bosonB = 2^{-9n^2}<2^{-(n+1)}\le 1/(2N)$ by the choice of parameters (see Table~\ref{tab: parameters}), it follows from
		Lemma~\ref{lem:orthogonalitytruncated gkpstates} that
		\begin{align}
			\langle \gkp^{\varepsilon_\bosonB }_{\kappa_\bosonB ,\Delta_\bosonB },
			M_N^\dagger e^{idP}M_N\gkp^{\varepsilon_\bosonB }_{\kappa_\bosonB ,\Delta_\bosonB }\rangle&=0 \ .\label{eq:dkkprimeN}
		\end{align}
	\end{proof}
	
	For later reference, we observe the following: Because of the pairwise orthogonality of the states  $\left\{e^{-ikP_\bosonB}M_N\ket{\gkp^{\varepsilon_\bosonB }_{\kappa_\bosonB ,\Delta_\bosonB }}\right\}_{k\in \Rem_{a,N}}$, 
	Lemma~\ref{lem:fouriertransformandrhoa} 
	implies that the probability density function $p_{\Phi_a}:\mathbb{R}\to\mathbb{R}$ of obtaining 
	an outcome
	$w\in\mathbb{R}$ when measuring the $P$-quadrature 
	of mode~$\bosonA$ in the state~$\Phi_a$ is given by 
	\begin{align}
		p_{\Phi_a}(w)=\sum_{k\in\Rem_{a,N}} \abs{\widehat{\Theta}_k(w)}^2\ .\label{eq: prob density}
	\end{align}
	
	\subsection{Bounds on the Fourier transform\label{sec:lowerboundsmodulusft}}
	We  show that~$|\widehat{\Theta}_k(w)|$ constituting the probability density function in Eq.~\eqref{eq: prob density} can be rewritten as follows. 
	\begin{lemma}\label{lem: F Theta k} Let $N\in \mathbb{N}$, $a\in\mathbb{Z}_N^*$, 
		and $k\in\Rem_{a,N}$ be arbitrary. Let $r=r(a)$ be the order of $a$ in $\mathbb{Z}_N^*$. Then 
		$\widehat{\Theta}_k(w)$ (see Eq.~\eqref{eq: F Theta k w}) satisfies
		\begin{align}
			\abs{\widehat{\Theta}_k(w)}
			&=\abs{c_{\Phi_a}}\frac{|\widehat{\Psi}_{\Delta_\bosonA }(w)|}{r}\abs{\sum_{z\in\mathbb{Z}} e^{2\pi i (R-\ind_a(k))z/r}\widehat{\eta}_{\kappa_\bosonA}(w + z/r)}\qquad\textrm{ for any }w\in\mathbb{R}\ ,\label{eq: F Theta k}
		\end{align}
		where for $k \in \Rem_{a,N}$ we denote the index of $k$ to the base $a$ modulo $N$ by $\ind_a(k)\in\mathbb{Z}_r$, i.e., $\ind_a(k)\in\mathbb{Z}_r$ is the smallest number in~$\{0,\ldots,r-1\}$ satisfying
		$a^{\ind_a(k)}\equiv k \pmod N$.
	\end{lemma}
	
	\begin{proof} Let $k \in \Rem_{a,N}$.
		Let $\mathcal{L}$ be the lattice defined as
		\begin{align}
			\cL&=\left\{z\in\mathbb{Z} \mid a^{z} \xmod N=1\right\}\ .
		\end{align}
		Since $r$ is the order of $a$ in $\mathbb{Z}_N^*$, we have $\cL=r\mathbb{Z}$. The corresponding dual lattice $\cL^* = \{ y \in \mathbb{R}: y \cdot z \in \mathbb{Z} \textrm{ for all } z \in \cL \}$
		is $\cL^*=\frac{1}{r}\mathbb{Z}$, which has lattice determinant (see~\cite{SMAntoniZygmund1935})
		\begin{align}
			\det \cL^*=\frac 1 r\ . \label{eq:latticedeterminantduallaticem}
		\end{align}
		We also have that the translated lattice $\ind_a(k)+\cL$ satisfies
		\begin{align}
			\ind_a(k)+\cL
			&=\left\{z\in\mathbb{Z} \mid a^{z} \xmod N=k\right\}\,\label{eq: cL shift}
		\end{align}
		by definition of~$\ind_a(k)$. 
		
		Using~\eqref{eq: cL shift}, we can rewrite $\widehat{\Theta}_k$ given in Eq.~\eqref{eq: F Theta k w} as
		\begin{align}
			\widehat{\Theta}_k(w)&= c_{\Phi_a}\sum_{z\in \ind_a(k)+\cL}
			\eta_{\kappa_\bosonA}(z-R)\widehat{\Psi}_{\Delta_\bosonA }(w)e^{2\pi i z w}\\
			&=
			c_{\Phi_a}
			e^{2\pi i w \cdot\ind_a(k)}\widehat{\Psi}_{\Delta_\bosonA }(w)
			\sum_{y\in\cL}f(y)\ , \label{eq: thetahat with lattice}
		\end{align}
		i.e., as  a sum over lattice points~$y\in\cL$. Here we introduced the function
		\begin{align}
			f(y)&=\eta_{\kappa_\bosonA}(y+\ind_a(k)-R)e^{2\pi i yw}\\
			&=\eta_{\kappa_\bosonA}(y+J)e^{2\pi iyw}\qquad\textrm{ where }\qquad J:=\ind_a(k)-R\ .\label{eq:jsumdefinitionx}
		\end{align}
		The Fourier transform of~$f$
		is
		\begin{align}
			\widehat{f}(v)&=\int f(y)e^{2\pi i y v}dy\\
			&=\int \eta_{\kappa_\bosonA}(y+J)e^{2\pi iyw} e^{2\pi i y v}dy\\
			&=\int \eta_{\kappa_\bosonA}(z) e^{2\pi i(z-J)w}e^{2\pi i(z-J)v}dz\\
			&=e^{-2\pi iJ(v+w)}\int \eta_{\kappa_\bosonA}(z) e^{2\pi iz(v+w)}dz\\
			&=e^{-2\pi iJ(v+w)}\widehat{\eta}_{\kappa_\bosonA}(v+w)\ .
		\end{align}
		where we used the variable substitution $z=y+J$ to obtain the third line.
		With the Poisson summation formula (see e.g.,~\cite{SMAntoniZygmund1935})
		\begin{align}
			\sum_{y\in\cL}f(y)&=(\det \cL^*)\sum_{v\in\cL^*}\widehat{f}(v)
		\end{align}
		and~\eqref{eq: thetahat with lattice} in combination with the definition of dual lattice $\cL^*=\frac1r \mathbb{Z}$ and its determinant from with~\eqref{eq:latticedeterminantduallaticem}, we thus obtain the expression
		\begin{align}
			\widehat{\Theta}_k(w)
			&=c_{\Phi_a} e^{2\pi i w\cdot \ind_a(k)}\widehat{\Psi}_{\Delta_\bosonA }(w)\frac{1}{r}\sum_{y\in \frac{1}{r}\mathbb{Z}} e^{-2\pi i J (y+w)}\widehat{\eta}_{\kappa_\bosonA}(y+w)\\
			&=c_{\Phi_a} e^{2\pi i w\cdot \ind_a(k)}\widehat{\Psi}_{\Delta_\bosonA }(w)\frac{1}{r}\sum_{z\in \mathbb{Z}} e^{-2\pi i J (z/r+w)}\widehat{\eta}_{\kappa_\bosonA}(z/r+w)\\
			&=c_{\Phi_a} e^{2\pi i \left(w\cdot \ind_a(k)-Jw\right)}\widehat{\Psi}_{\Delta_\bosonA }(w)\frac{1}{r}\sum_{z\in \mathbb{Z}} e^{-2\pi i J z/r}\widehat{\eta}_{\kappa_\bosonA}(z/r+w)
		\end{align}
		where we substituted $y=z/r$ to obtain the second line. With the definition of~$J$ in~\eqref{eq:jsumdefinitionx}, it follows that
		\begin{align}
			\abs{\widehat{\Theta}_k(w)}&=\abs{c_{\Phi_a}}\cdot\frac{|\widehat{\Psi}_{\Delta_\bosonA }(w)|}{r}\cdot \abs{\sum_{z\in\mathbb{Z}}
				e^{2\pi i(R-\ind_a(k))z/r}\widehat{\eta}_{\kappa_\bosonA}(w+z/r)}\ .
		\end{align}
	\end{proof}
	Using the expression~\eqref{eq: F Theta k} from Lemma~\ref{lem: F Theta k}, we can lower bound $|\widehat{\Theta}_k(w)|$ for very specific values of~$w$ as follows.

	\begin{lemma}\label{lem: bound on F theta k w} 
		Let $N\geq 2$, $a \in \mathbb{Z}_N^*$ and  $k\in\Rem_{a,N}$ be arbitrary. Let $r = r(a)$ be the order of~$a$ in $\mathbb{Z}_N^*$. Let the parameter $\kappa_A$ be defined by Table~\ref{tab: parameters}.
		Then 
		\begin{align}
			\abs{\widehat{\Theta}_k(\omega+m/r)}&\ge  \frac{e^{-\pi^2/2}}{\sqrt{\kappa_\bosonA}}\cdot |c_{\Phi_a}|\cdot \frac{|\widehat{\Psi}_{\Delta_\bosonA }(\omega+m/r)|}{r}
			\quad \textrm{ for all }\quad \omega\in \left[\frac{\kappa_\bosonA}{4},\frac{\kappa_\bosonA}{2}\right]\textrm{ and }m\in\mathbb{Z}\ .
		\end{align}
	\end{lemma}
	
	\begin{proof}
		Let $\varphi\in\mathbb{R}$ be arbitrary. 
		Then we have for $\omega\in \left[\frac{\kappa_\bosonA}{4},\frac{\kappa_\bosonA}{2}\right]$ and $m\in\mathbb{Z}$
		\begin{align}
			\sum_{z\in\mathbb{Z}}e^{i\varphi z}\widehat{\eta}_{\kappa_\bosonA}((\omega+m/r)+z/r)
			&=e^{-i m\varphi}\sum_{y\in\mathbb{Z}}e^{i\varphi y}\widehat{\eta}_{\kappa_\bosonA}(\omega+y/r)\\
			&=e^{-i m\varphi}\frac{\sqrt{2}\pi^{1/4}}{\sqrt{\kappa_\bosonA}}\sum_{y\in\mathbb{Z}}e^{i\varphi y}e^{-\textfrac{2\pi^2 (\omega+y/r)^2}{\kappa_\bosonA^2}} \label{eq:lowerbound abs value hat theta k inner sum}
		\end{align}
		where we used the variable substitution~$y=m+z$ in the first identity, and the fact that the Fourier transform of $\eta_{\kappa_\bosonA}$ is
		$\widehat{\eta}_{\kappa_\bosonA}(w) = \frac{\sqrt{2}\pi^{1/4}}{\sqrt{\kappa_\bosonA}} e^{-2 \pi^2 w^2/\kappa_\bosonA^2}$
		in the second identity. Defining
		\begin{align}
			\xi &=e^{i\varphi}e^{-4\pi^2 \omega/(\kappa_\bosonA^2 r)}\\
			b&=e^{-2\pi^2/(\kappa_\bosonA^2r^2)}\ ,
		\end{align}
		Eq.~\eqref{eq:lowerbound abs value hat theta k inner sum} can be rewritten as 
		\begin{align}
			\sum_{z\in\mathbb{Z}}e^{i\varphi z}\widehat{\eta}_{\kappa_\bosonA}((\omega+m/r)+z/r)&=e^{-im\varphi}\frac{\sqrt{2}\pi^{1/4}e^{-2\pi^2\omega^2/\kappa_\bosonA^2}}{\sqrt{\kappa_\bosonA}}\sum_{y\in\mathbb{Z}}\xi^y b^{y^2}\ .\label{eq:combncvda}
		\end{align}
		For our choice of parameters 
		we have 
		$\frac{\omega}{\kappa_\bosonA^2 r}\geq \frac{1}{4\kappa_\bosonA r}\geq \frac{N}{4}$ (because $\frac{1}{\kappa_\bosonA}>N^2$ and $r\leq N$). 
		It follows that~$\max\{|\xi|,|\xi|^{-1}\}=|\xi^{-1}|\geq e^{\pi^2 N}\geq \sqrt{2}+1$. Furthermore, we have $b<1$ by definition. This means that we can apply
		Lemma~\ref{lem:gausssumjacobitriple}, which gives a lower bound on a sum of the form~$\sum_{y\in\mathbb{Z}}\xi^y b^{y^2}$.
		Applied to~\eqref{eq:combncvda} this gives
		\begin{align}
			\left|\sum_{z\in\mathbb{Z}}e^{i\varphi z}\widehat{\eta}_{\kappa_\bosonA}((\omega+m/r)+z/r)\right|
			&\geq \frac{\sqrt{2}\pi^{1/4}e^{-2\pi^2\omega^2/\kappa_\bosonA^2}}{\sqrt{\kappa_\bosonA}}\left(1-2e^{-2\pi^2/(\kappa_\bosonA^2r^2)}
			\cdot e^{4\pi^2 \omega/(\kappa_\bosonA^2 r)}\right)\\
			&\geq \frac{\sqrt{2}\pi^{1/4}e^{-2\pi^2\omega^2/\kappa_\bosonA^2}}{\sqrt{\kappa_\bosonA}}
			\left(1-2e^{-\pi^2 N^2}\right)\ ,\label{eq:thetakexprmf}
		\end{align}
		where we used that
		\begin{align}        e^{-2\pi^2/(\kappa_\bosonA^2r^2)}e^{4\pi^2\omega/(\kappa_\bosonA^2r)}
			&\leq e^{-2\pi^2/(\kappa_\bosonA^2 r^2)}e^{2\pi^2/(\kappa_\bosonA r)}\qquad\textrm{ since } \omega\leq \kappa_\bosonA/2\\
			&= e^{-2\pi^2x(x-1)}\qquad\qquad \textrm{ with }x=1/(\kappa_\bosonA r)\geq N\\
			&\leq e^{-2\pi^2N(N-1)}\qquad \textrm{as $x\mapsto x(x-1)$ increases monotonically}\\
			&\leq e^{-\pi^2 N^2}\qquad\textrm{ for }N\geq 2\ .
		\end{align}
		With~\eqref{eq:combncvda} and~\eqref{eq:thetakexprmf}, as well as the choice  $\varphi=2\pi(R-\ind_a(k))/r$, we conclude that
		for $w=\omega+m/r$, we have 
		\begin{align}
			\left|
			\sum_{z\in\mathbb{Z}}e^{2\pi i (R-\ind_a(k))z/r}
			\widehat{\eta}_{\kappa_\bosonA}(w+z/r)\right|&=\frac{\sqrt{2}\pi^{1/4}e^{-2\pi^2\omega^2/\kappa_\bosonA^2}}{\sqrt{\kappa_\bosonA}}\left(1-2e^{-\pi^2N^2}\right)\geq \frac{e^{-2\pi^2\omega^2/\kappa_\bosonA^2}}{\sqrt{\kappa_\bosonA}}
		\end{align}
		where we used that $\sqrt{2}\pi^{1/4}(1-2e^{-\pi^2 N^2})\ge \sqrt{2}\pi^{1/4}(1-2e^{-\pi^2})\geq 1$.     
		Because of Eq.~\eqref{eq: F Theta k}, this means that 
		\begin{align}
			\abs{\widehat{\Theta}_k(w)}&\ge 
			\frac{e^{-2\pi^2\omega^2/\kappa_\bosonA^2}}{\sqrt{\kappa_\bosonA}}\cdot \abs{c_{\Phi_a}}\cdot\frac{|\widehat{\Psi}_{\Delta_\bosonA }(w)|}{r}\qquad\textrm{ for }\qquad w=\omega+m/r\ .
		\end{align}
		The claim follows from the fact that the function $\omega \mapsto e^{-2\pi^2\omega^2/\kappa_\bosonA^2}$ is monotonically decreasing for $\omega \ge 0$ combined with the assumption $\omega\in[\frac{\kappa_A}{4},\frac{\kappa_A}{2}]$.
	\end{proof}

	\subsection{Bounds on the probability distribution of outcomes\label{sec:boundspdf}}
	Recall the definition (see Eq.~\eqref{eq:Omega d})  of the set
	\begin{align}
		\Gamma_{d}(a):=\bigcup_{j\in \mathbb{Z}}\left[j+\frac d r - \frac{1}{2q}, j+\frac d r + \frac{1}{2q} \right]\ ,\label{eq:gammdadefrepeated}
	\end{align}
	for $a\in \mathbb{Z}_N^*$ and $d\in\mathbb{Z}_r$, where $r=r(a)$ is the order of $a$  and $q=\min\{2^k\mid k\in\mathbb{N}, 2^k>N^2\}$ is the smallest power of~$2$ greater than~$N^2$. Here we use  Lemma~\ref{lem: bound on F theta k w} to establish a lower bound on the probability density function~$p_{\Phi_a}$ integrated over 
	the set~$\Gamma_d(a)$.  
	We show that this quantity is lower bounded by a constant multiple of~$1/r$ independent of $a\in\mathbb{Z}_N^*$ and $d\in\mathbb{Z}_r$.  This will completes the proof of Proposition~\ref{prop:propositiontwo}.
	
	\begin{lemma}\label{lem:lowerbounddistributionweight}
		Let $(m,R,\kappa_\bosonA,\Delta_\bosonA,\kappa_\bosonB,\Delta_\bosonB, \Delta_\bosonC)$ be defined by Table~\ref{tab: parameters}. Let $\varepsilon_\bosonB = \sqrt{\Delta_\bosonB}$ and $\varepsilon_\bosonC = \sqrt{\Delta_\bosonC}$. Let $N\geq 4$ and $a\in\mathbb{Z}_N^*$ be arbitrary. Let $r=r(a)$ be the order of~$a$. 
		Then 
		\begin{align}
			\int_{\Gamma_{d}(a)} p_{\Phi_a}(w) dw
			&\ge \frac{e^{-\pi^2}}{16}\cdot \frac{1}{r} \qquad\textrm{ for any }\qquad d\in\mathbb{Z}_r\ .\label{eq:integralstolowerbound}
		\end{align}
	\end{lemma}
	\begin{proof}
		Let $a\in\mathbb{Z}_N^*$ be arbitrary.
		Fix an arbitrary $d\in\mathbb{Z}_r$, and an arbitrary~$k\in \Rem_{a,N}$  (cf.~Eq.~\eqref{eq: Rem a N}).  To lower bound each integral in~\eqref{eq:integralstolowerbound}
		we will replace the set~$\Gamma_d(a)$ by a suitable subset. Since $q>1$, the intervals in the definition~\eqref{eq:gammdadefrepeated} of~$\Gamma_d(a)$ are pairwise disjoint, i.e., we have 
		\begin{align}
			\left[j+\frac d r - \frac{1}{2q}, j+\frac d r + \frac{1}{2q} \right]\cap \left[j'+\frac d r - \frac{1}{2q}, j'+\frac d r + \frac{1}{2q} \right]&=\emptyset\qquad\textrm{ for }\qquad j\neq j'\ .
		\end{align}
		Since $\left[\frac{\kappa_\bosonA}{4}, \frac{\kappa_\bosonA}{2}\right]\subseteq\left[-\frac{1}{2q},\frac{1}{2q}\right]$ as $\kappa_\bosonA\le q^{-1}$ by the choice of parameters
		from Table~\ref{tab: parameters}, we have
		\begin{align}
			\bigcup_{j\in\mathbb{Z}} \left[j+\frac{d}{r}+\frac{\kappa_\bosonA}{4}, j+\frac{d}{r}+\frac{\kappa_\bosonA}{2} \right] \subseteq\Gamma_d(a)\ .
		\end{align}
		It follows that 
		\begin{align}
			\int_{\Gamma_d(a)} |\widehat{\Theta}_k(w)|^2 dw
			&\ge \sum_{j\in\mathbb{Z}}\int_{j+\frac d r + \frac{\kappa_{\bosonA}}{4}}^{j+\frac d r + \frac{\kappa_\bosonA}{2}} \abs{\widehat{\Theta}_k(w)}^2 dw\\
			&= \sum_{j\in\mathbb{Z}}\int_{\frac{\kappa_{\bosonA}}{4}}^{\frac{\kappa_\bosonA}{2}} \abs{\widehat{\Theta}_k\left(j+\frac{d}{r}+\omega\right)}^2 d\omega \qquad \textrm{ by substitution $w=j+\frac d r + \omega$}\\
			&\ge \frac{e^{-\pi^2}}{\kappa_\bosonA}\cdot \frac{|c_{\Phi_a}|^2}{r^2}\sum_{j\in \mathbb{Z}}\int_{\frac{\kappa_\bosonA}{4}}^{\frac{\kappa_\bosonA}{2}} \abs{\widehat{\Psi}_{\Delta_\bosonA }\left(j+\frac{d}{r}+\omega\right)}^2 d\omega\qquad\textrm{ by Lemma~\ref{lem: bound on F theta k w}}\\
			&\ge\frac{e^{-\pi^2}}{\kappa_\bosonA}\cdot\frac{|c_{\Phi_a}|^2}{r^2}\sum_{j\in \mathbb{Z}}\frac{\kappa_\bosonA}{4}\min_{\omega\in [\frac{\kappa_\bosonA}{4},\frac{\kappa_\bosonA}{2}]} \abs{\widehat{\Psi}_{\Delta_\bosonA }\left(j+\frac{d}{r}+\omega\right)}^2\\
			&= \frac{e^{-\pi^2}}{4}\cdot\frac{|c_{\Phi_a}|^2}{r^2}\sum_{j\in \mathbb{Z}}\min_{\omega\in [\frac{\kappa_\bosonA}{4},\frac{\kappa_\bosonA}{2}]} \abs{\widehat{\Psi}_{\Delta_\bosonA }\left(j+\frac{d}{r}+\omega\right)}^2\ .\label{eq: hat theta k bound}
		\end{align}
		We have
		\begin{align}
			\abs{j+\frac{d}{r}+\omega}
			&\le \abs{j}+\frac{d}{r}+\abs{\omega}\\
			&\leq \abs{j}+\frac d r + \frac{\kappa_\bosonA}{2}\qquad\textrm{ for }\qquad \omega\in [\kappa_\bosonA/4,\kappa_\bosonA/2]\\
			&\leq \abs{j}+\frac{r-1}{r} + \frac{\kappa_\bosonA}{2}\qquad\textrm{ because }d\in\mathbb{Z}_r\\
			&\leq \abs{j}+1\ , \label{eq: j dr omega bound}
		\end{align}
		where  in the last step, we used that $\kappa_\bosonA/2\le 1/r$ (by the choice of the parameters in Table~\ref{tab: parameters}).
		Because $\widehat{\Psi}_{\Delta_\bosonA}(w)= \sqrt{2}\pi^{1/4}\sqrt{\Delta_\bosonA} e^{-2\pi^2w^2\Delta_\bosonA^2}$ is symmetric and monotonically decreasing for~$w\geq 0$, it follows from~\eqref{eq: j dr omega bound} that for any $j\in\mathbb{Z}$ and $d\in\mathbb{Z}_r$ we have 
		\begin{align}
			\min_{\omega\in [\frac{\kappa_\bosonA}{4},\frac{\kappa_\bosonA}{2}]} \left|\widehat{\Psi}_{\Delta_\bosonA}\left(j+\frac{d}{r}+\omega\right)\right|^2&\geq  |\widehat{\Psi}_{\Delta_\bosonA }(|j|+1)|^2\\
			&=2\sqrt{\pi}\Delta_\bosonA  e^{-4\pi^2\Delta_\bosonA ^2(|j|+1)^2}\ .\label{eq:jsumdeltaajb}
		\end{align}
		Inserting~\eqref{eq:jsumdeltaajb} into~\eqref{eq: hat theta k bound} we conclude that 
		\begin{align}
			\int_{\Gamma_d(a)} |\widehat{\Theta}_k(w)|^2 dw&\ge 
			\frac{e^{-\pi^2}\sqrt{\pi}}{2}\cdot \frac{\Delta_\bosonA \cdot |c_{\Phi_a}|^2}{r^2}\sum_{j\in\mathbb{Z}}e^{-4\pi^2\Delta_\bosonA ^2 (|j|+1)^2}\\
			&=
			\frac{e^{-\pi^2}\sqrt{\pi}}{2}\cdot \frac{\Delta_\bosonA \cdot |c_{\Phi_a}|^2}{r^2}\sum_{j\in\mathbb{Z}}e^{-c (|j|+1)^2}\qquad \textrm{ where }\qquad c=4\pi^2\Delta_\bosonA ^2\\
			&\ge \frac{e^{-\pi^2}}{16}\cdot\frac{|c_{\Phi_a}|^2}{r^2}\ ,\label{eq:integrallowerbndthetakww}
		\end{align}
		where we used Corollary~\ref{cor: Delta sum abs bound} to obtain the last inequality. The latter can be applied because of our choice of parameters, which implies that $c<\pi/16$. Using the  lower bound on $|c_{\Phi_a}|^2$ given in Lemma~\ref{lem: c phi} below we obtain 
		\begin{align}
			\int_{\Gamma_d(a)} |\widehat{\Theta}_k(w)|^2 dw&\geq\frac{e^{-\pi^2}}{64}\cdot \frac{1}{r^2}\qquad\textrm{ for all }\qquad d\in\mathbb{Z}_r\textrm{ and } k\in\Rem_{a,N}\ .\label{eq:gammadarlowerbndv}
		\end{align}
		Combining Eq.~\eqref{eq: prob density}, i.e., the expression
		\begin{align}
			p_{\Phi_a}(w)=\sum_{k\in\Rem_{a,N}} \abs{\widehat{\Theta}_k(w)}^2\ .
		\end{align}
		with~\eqref{eq:gammadarlowerbndv}
		implies (by the linearity of the integral) that 
		\begin{align}
			\int_{\Gamma_{d}(a)} p_{\Phi_a}(w) dw& = \sum_{k\in\Rem_{a,N}} \int_{\Gamma_{d}(a)} \abs{\widehat{\Theta}_k(w)}^2dw\\
			&\geq 
			\frac{e^{-\pi^2}}{64}\cdot \frac{1}{r}
		\end{align}
		where we used that the set $\Rem_{a,N}$ consists  of~$|\Rem_{a,N}|=r$ elements. This is the claim.
	\end{proof}
	
	It remains to establish a bound on the normalization constant~$c_{\Phi_a}$ in the definition~\eqref{eq:stateoutputmrestated} of the state~$\Phi_a$.

	\begin{lemma}\label{lem: c phi}
		Let $(m,R,\kappa_\bosonA,\Delta_\bosonA,\kappa_\bosonB,\Delta_\bosonB, \Delta_\bosonC)$ be as specified  in Table~\ref{tab: parameters}. Let $\varepsilon_\bosonB = \sqrt{\Delta_\bosonB}$ and $\varepsilon_\bosonC = \sqrt{\Delta_\bosonC}$. Then 
		\begin{align}
			|c_{\Phi_a}|^2 \geq 1/4 \ .
		\end{align}
	\end{lemma}
	\begin{proof}
		By  definition of the state $\Phi_a$  (cf.~\eqref{eq:stateoutputmrestated}) and 
		the pairwise orthogonality of the family of states $\left\{e^{-ikP_\bosonB}M_N\ket{\gkp^{\varepsilon_\bosonB }_{\kappa_\bosonB ,\Delta_\bosonB }}\right\}_{k\in \Rem_{a,N}}$ (see~\eqref{eq:dkkprimeN}) we have
		\begin{align}
			|c_{\Phi_a}|^{-2}&=
			\sum_{y,z\in\mathbb{Z}}
			\eta_{\kappa_\bosonA}(y-R)\eta_{\kappa_\bosonA}(z-R)
			\langle \chi_{\Delta_\bosonA}(y),\chi_{\Delta_{\bosonA}}(z)\rangle\cdot \delta_{a^y \xmod N, a^z\xmod N}\ .\label{eq:upperbndmvz}
		\end{align}
		Since 
		\begin{align}
			\eta_{\kappa_\bosonA}(y-R)\eta_{\kappa_\bosonA}(z-R)
			\langle \chi_{\Delta_\bosonA}(y),\chi_{\Delta_{\bosonA}}(z)\rangle\geq 0\qquad\textrm{ for all }\qquad y,z\in\mathbb{Z}\ ,
		\end{align}
		we can omit the term~$\delta_{a^y \xmod N,a^z\xmod N}$ in Eq.~\eqref{eq:upperbndmvz} to get the upper bound 
		\begin{align}
			|c_{\Phi_a}|^{-2}&\leq \sum_{y,z\in\mathbb{Z}}
			\eta_{\kappa_\bosonA}(y-R)\eta_{\kappa_\bosonA}(z-R)
			\langle \chi_{\Delta_\bosonA}(y),\chi_{\Delta_{\bosonA}}(z)\rangle\ .
		\end{align}
		Recalling that~$R$ is an integer, we can use a substitution to obtain
		\begin{align}
			|c_{\Phi_a}|^{-2}&\leq \sum_{y,z\in\mathbb{Z}}
			\eta_{\kappa_\bosonA}(y)\eta_{\kappa_\bosonA}(z)\langle \chi_{\Delta_{\kappa_\bosonA}
			}(y+R),\chi_{\Delta_{\kappa_\bosonA}
			}(z+R)\rangle\\
			&=\sum_{y,z\in\mathbb{Z}}
			\eta_{\kappa_\bosonA}(y)\eta_{\kappa_\bosonA}(z)\langle \chi_{\Delta_{\kappa_\bosonA}
			}(y),\chi_{\Delta_{\kappa_\bosonA}
			}(z)\rangle\label{eq: C_Phi first}
		\end{align}
		where we used that $\langle \chi_{\Delta_\bosonA }(y+R),\chi_{\Delta_\bosonA }(z+R)\rangle = \langle e^{-iRP}\chi_{\Delta_\bosonA }(y), e^{-iRP}\chi_{\Delta_\bosonA }(z)\rangle = \langle \chi_{\Delta_\bosonA }(y),\chi_{\Delta_\bosonA }(z)\rangle$ because $e^{-iRP}$ is unitary.
		Comparing~\eqref{eq: C_Phi first} with the definition of the approximate GKP state~$\ket{\gkp_{\kappa,\Delta}}$ (cf.~\eqref{eq:gkp approximate eta chi}), we obtain
		\begin{align}
			|c_{\Phi_a}|^{-2} \le C_{\kappa_\bosonA,\Delta_\bosonA }^{-2}\ .\label{eq: c phi sec}
		\end{align}
		Our choice of parameters~$(m,R,\kappa_\bosonA,\Delta_\bosonA,\kappa_\bosonB,\Delta_\bosonB, \Delta_\bosonC)$ 
		(see Table~\ref{tab: parameters}) guarantees that
		$\Delta_{\bosonA}<1/8$ and $\kappa_\bosonA<1/8$. Thus Lemma~\ref{lem: norm and ratio norm gkp gkp ep} can be applied, giving the bound
		\begin{align}
			C_{\kappa_\bosonA,\Delta_\bosonA}^2 \geq 1/4\ .\label{eq:bndckmvac}
		\end{align}
		Combining~\eqref{eq: c phi sec} and~\eqref{eq:bndckmvac} gives the claim.
	\end{proof}

	\clearpage
	\section{Properties of approximate GKP states}\label{sec: properties appx GKP}
	In this section, we state some approximation results on Gaussians and their truncated versions, as well as analogous results for GKP-states.
	In separate work~\cite{brenneretalGKP2024}, we have used identical definitions, and established a few of the corresponding results. We refer to~\cite{brenneretalGKP2024} for proofs in these cases. 
	
	First, recall that the  for $\Delta>0$, we denote by $\Psi_\Delta(x)=\frac{1}{(\pi \Delta^2)^{1/4}}e^{-x^2/(2\Delta^2)}$  a centered Gaussian (corresponding to a squeezed vacuum state $\ket{\Psi_\Delta}=M_\Delta \ket{\vac}$). The following lemma shows that it is well approximated by its truncated version~$\Psi^\varepsilon_\Delta=
	\frac{\Pi_{[-\varepsilon,\varepsilon]}\Psi_\Delta}{\|\Pi_{[-\varepsilon,\varepsilon]}\Psi_\Delta\|}
	$ obtained by restricting its support to the interval~$[-\varepsilon,\varepsilon]$, and renormalizing.
	
	\begin{lemma}\cite[Lemma A.2 and Corollary A.3]
		{brenneretalGKP2024}\label{lem: gaussian gaussian epsilon}
		Let $\Delta>0$ and $\varepsilon\in [\sqrt{\Delta},1/2)$. Then 
		\begin{align}
			\left|\langle \Psi_\Delta, \Psi_\Delta^\varepsilon\rangle \right|^2 
			&\ge 1-2\Delta \  \label{eq: Psi Psi eps overlap restricted ep lemma}
		\end{align}
		and
		\begin{align}
			\left\|\Psi_\Delta - \Psi_\Delta^{\varepsilon}\right\|_1
			&\le 3 \sqrt{\Delta} \ .
		\end{align}
	\end{lemma}
	
	Recall the definitions of the approximate GKP state (defined in~\eqref{eq:gkp approximate eta chi})
	\begin{align}
		\ket{\gkp_{\kappa, \Delta}}&= C_{\kappa,\Delta} \sum_{z\in\mathbb{Z}}\eta_\kappa(z)\ket{\chi_\Delta (z)} \  ,\label{eq: approx GKP eta chi}
	\end{align}
	(with $\chi_\Delta(z)=\Psi_\Delta(x-z)$) 
	as well as its truncated variant with $\varepsilon<1/2$ (defined in~\eqref{eq:approximate gkp def eta xi eps})
	\begin{align}
		\ket{\gkp_{\kappa, \Delta}^{\varepsilon}}&= C_{\kappa} \sum_{z\in\mathbb{Z}}\eta_\kappa(z)\ket{\chi_\Delta^\varepsilon(z)} \ .
	\end{align}
	We use the following estimates on the normalization constants~$C_{\kappa,\Delta}$ and $C_{\kappa}$.
	\begin{lemma}\cite[Lemma A.8 (specialized)]{brenneretalGKP2024}\label{lem: norm and ratio norm gkp gkp ep} 
		We have 
		\begin{align}
			C_{\kappa,\Delta}^{2} & \ge 1/4 \qquad\textrm{ for }\qquad \kappa,\Delta\in(0,1/8)\ \label{eq:hmf}
		\end{align}
		and
		\begin{align}
			\frac{C_{\kappa, \Delta}^{2}}{C_{\kappa}^{2}}
			&\ge 
			1-7\Delta\qquad\textrm{ for }\qquad \kappa,\Delta\in (0,1/4)\ .\label{eq:lowerboundsecond}
		\end{align}
	\end{lemma}
	\begin{proof}
		By~\cite[Lemma A.8]{brenneretalGKP2024}, we have 
		\begin{align}
			C_{\kappa,\Delta}^{-2} & \le 1 + \frac{\kappa}{\sqrt{\pi}} + 2(\sqrt{2\pi}+\kappa)\Delta\qquad\textrm{ for }\kappa,\Delta>0\ ,
		\end{align}hence
		\begin{align}
			C_{\kappa,\Delta}^{2} & \ge 1 - \frac{\kappa}{\sqrt{\pi}} - 2(\sqrt{2\pi}+\kappa)\Delta\qquad\textrm{ for }\kappa,\Delta>0
		\end{align}
		by the inequality~$1/(1+x)\ge 1-x$ for all $x>0$. The claim~\eqref{eq:hmf} follows. 
		
		Furthermore, \cite[Lemma A.8]{brenneretalGKP2024} shows that 
		\begin{align}
			\frac{C_{\kappa, \Delta}^{2}}{C_{\kappa}^{2}}
			&\ge 1-\frac{2(\sqrt{2\pi}+\kappa)}{1-\kappa/\sqrt{\pi}}\Delta\qquad \textrm{ for }\kappa>0\ ,
		\end{align}
		which implies the claim~\eqref{eq:lowerboundsecond} for $\kappa<1/8$.
	\end{proof}

	The fact that the states~$\gkp_{\kappa,\Delta}$ and its truncated version~$\gkp_{\kappa, \Delta}^\varepsilon$  are close (for suitably chosen parameters) is expressed by the following statement.
	\begin{lemma}\cite[Lemma A.9 (specialized)]{brenneretalGKP2024}\label{lem: overlap truncated gkp}
		Let $\kappa\in(0,1/4)$, $\Delta>0$ and $\varepsilon\in [\sqrt{\Delta},1/2)$.
		Then
		\begin{align}
			\left|\left\langle \gkp_{\kappa, \Delta}, \gkp_{\kappa, \Delta}^\varepsilon \right\rangle \right|^2 \ge 1-9\Delta
		\end{align}
	\end{lemma}
	\begin{proof}
		\cite[Lemma A.9]{brenneretalGKP2024} shows that
		\begin{align}
			\left|\left\langle \gkp_{\kappa, \Delta}, \gkp_{\kappa, \Delta}^\varepsilon \right\rangle \right|^2 \ge 1-7\Delta - 2e^{-(\varepsilon/\Delta)^2} \qquad \textrm{ for $\kappa\in(0,1/4)$, $\Delta>0$ and $\varepsilon\in(0,1/2)$} \ .\label{eq: truncated GKP repeated}
		\end{align}
		The inequality $e^{-x}\le x^{-1}$ for $x>0$ together with the assumption on $\varepsilon\in[\sqrt{\Delta},1/2)$ imply the claim.
	\end{proof}

	We can show that the function~$\gkp^\varepsilon_{\kappa,\Delta}$ has most of its support on the interval~$[-r,r]$ for $r\gg 1/\kappa$, as expressed by the following lemma.
	\begin{lemma}\label{lem:proj norm}
		Let $r\ge 4$ and $\kappa>0$. Let $\varepsilon\in (0,1/2)$. Let $\Pi_{[-r,r]}$ be the projection onto the subspace of~$L^2(\mathbb{R})$ of functions having support on~$[-r,r]$. Then
		\begin{align}
			\left\|\Pi_{[-r,r]}\gkp^\varepsilon_{\kappa,\Delta}\right\|^2 \geq 
			1-2e^{-(\kappa r)^2}\ .
		\end{align}
	\end{lemma}
	\begin{proof}
		Recall the definition~\eqref{eq: truncated GKP repeated} of the state~$\ket{\gkp^\varepsilon_{\kappa,\Delta}}$. Because of the pairwise orthogonality of the states~$\{\ket{\chi^\varepsilon_\Delta(z)}\}_{z\in\mathbb{Z}}$, the normalization constant~$C_\kappa$ is given by
		\begin{align}
			C_{\kappa}^{-2}&=\sum_{z\in\mathbb{Z}}\eta_\kappa(z)^2\  .
		\end{align}
		On the other hand, we have 
		\begin{align}
			\left\lVert \Pi_{[-r,r]}\gkp_{\kappa, \Delta}^\varepsilon\right\rVert^2&=
			C_\kappa^2 \sum_{z,z'\in\mathbb{Z}}\eta_\kappa(z)\eta_\kappa(z')\langle \chi^\varepsilon_\Delta(z),\Pi_{[-r,r]}\chi^\varepsilon_\Delta(z')\rangle\\
			&\geq C_\kappa^2 \sum_{z=-t }^{t }\eta_\kappa(z)^2\qquad\textrm{ where }t:=\lfloor r\rfloor-1\ .
		\end{align}
		where we used that each function~$\chi_\Delta^\varepsilon(z)$ is supported on~$[z-\varepsilon,z+\varepsilon]$
		(implying that $\langle \chi^\varepsilon_\Delta(z),\Pi_{[-r,r]}\chi^\varepsilon_\Delta(z')\rangle=0$ unless $z=z'$), and the fact that each such function is non-negative (which implies that 
		$\langle \chi^\varepsilon_\Delta(z),\Pi_{[-r,r]}\chi^\varepsilon_\Delta(z)\rangle\leq \langle \chi^\varepsilon_\Delta(z),\chi^\varepsilon_\Delta(z)\rangle=\left\|\chi^\varepsilon_\Delta(z)\right\|^2=1$). 
		It follows that
		\begin{align}
			\left\| \Pi_{[-r,r]}\gkp_{\kappa, \Delta}^\varepsilon\right\|^2&
			\geq \frac{\sum_{z=-t}^t \eta_\kappa(z)^2}{\sum_{z\in\mathbb{Z}}\eta_\kappa(z)^2}\\
			&=\frac{\sum_{z=-t}^t\rho_s(z)}{\sum_{z\in\mathbb{Z}}\rho_s(z)}\qquad\textrm{ with }\qquad s=1/\kappa\ ,\label{eq:proj -rr gkp}
		\end{align}
		where we introduced the Gaussian function~$\rho_s:\mathbb{R}\rightarrow\mathbb{R}$ for $s>0$ as 
		\begin{align}
			\rho_s(z)&=e^{-\pi z^2/s^2}\ .
		\end{align}
		A random variable $X_s$ on $\mathbb{Z}$ sampled according to discrete Gaussian distribution with (width) parameter $s$ (see e.g.,~\cite{SMMicciancioRegev04}) is defined by
		\begin{align}
			\Pr\left[X_s=x\right]&=\frac{\rho_s(x)}{\rho_s(\mathbb{Z})}\qquad \textrm{where}\qquad
			\rho_s(\mathbb{Z}) = \sum_{z\in\mathbb{Z}}\rho_s(z)\ .
		\end{align}
		We thus have (by inserting $s=1/\kappa$ and $t=\lfloor r\rfloor -1$)
		\begin{align}
			\left\| \Pi_{[-r,r]}\gkp_{\kappa, \Delta}^\varepsilon\right\|^2
			&\ge \Pr[|X_{1/\kappa}|\le t]\\
			&= \Pr[|X_{1/\kappa}|\le \lfloor r\rfloor-1]\\
			&= \Pr[|X_{1/\kappa}|\le r - 1] \label{eq:proj -rr discrete prop}\\
			&\ge \Pr[|X_{1/\kappa}|< r - 1]\\
			&=1-\Pr[|X_{1/\kappa}|\ge r - 1] \ ,
		\end{align}
		where we obtained~\eqref{eq:proj -rr discrete prop} by the fact that the random variable $X_{1/\kappa}$ has support on $\mathbb{Z}$.
		
		From Lemma~\ref{thm:discretegaussianconcentration} below we obtain
		\begin{align}
			\left\| \Pi_{[-r,r]}\gkp_{\kappa, \Delta}^\varepsilon\right\|^2
			&\ge 1-2e^{-\frac{3\pi}{4}\left(\kappa (r-1) \right)^2}\ . \label{eq: proj -rr tailbound}\\
			&\ge 1-2e^{-\frac{3\pi}{8}\left(\kappa r\right)^2/2} \qquad \textrm{whenever } r\ge 4 \\
			&\ge 1-2e^{-(\kappa r)^2} \qquad\textrm{ since  } \textfrac{3\pi}{8}\ge 1 .
		\end{align}
		where we used that $(x-1)^2\ge x^2/2$ for all $x \ge 4$ to obtain the inequality in the second step. This is the claim.
	\end{proof}
	
	Ideal GKP states were originally introduced because of there invariance under shifts. Here we show that approximate GKP states are approximately invariant under integer shifts that are small compared to~$1/\kappa$. That is, we have:
	\begin{lemma}[Approximate shift-invariance of  approximate GKP states]\label{lem:approximateshiftinvarianceGKP}
		Let  $\varepsilon\in (0,1/2)$ and $\kappa,\Delta\in (0,1/4)$. 
		Then we have 
		\begin{align}
			\langle \gkp^{\varepsilon}_{\kappa,\Delta},e^{id P}\gkp^{\varepsilon}_{\kappa,\Delta}\rangle \geq 1-\frac{\kappa^2 d^2}{2}\qquad\textrm{  for any  }\qquad d \in\mathbb{Z}\ .
		\end{align}
	\end{lemma}
	\begin{proof}
		Using that $e^{id P}\ket{\chi^\varepsilon_\Delta(z)}=\ket{\chi^\varepsilon_\Delta(z+d)}$ we have, by the orthogonality of the states~$\ket{\chi^\varepsilon_\Delta(z)}\}_{z\in\mathbb{Z}}$ and the assumption that~$d$ is an integer, that
		\begin{align}
			\langle \gkp^{\varepsilon}_{\kappa,\Delta},e^{ id P}\gkp^{\varepsilon}_{\kappa,\Delta}\rangle
			&=C_{\kappa}^2 \sum_{z,z'\in\mathbb{Z}}
			\eta_{\kappa}(z)\eta_\kappa(z')
			\langle\chi^\varepsilon_{\Delta}(z),\chi^\varepsilon_\Delta(z'+d)\rangle\\
			&= C_\kappa^2 \sum_{z\in \mathbb{Z}}\eta_\kappa(z)\eta_\kappa(z-d)\ ,
		\end{align}
		where we used a variable substitution. 
		Since $C_\kappa^{-2}=\sum_{z\in\mathbb{Z}}\eta_\kappa(z)^2$  
		it follows that 
		\begin{align}
			\langle \gkp^{\varepsilon}_{\kappa,\Delta},e^{ id P}\gkp^{\varepsilon}_{\kappa,\Delta}\rangle
			&= \frac{\sum_{z\in \mathbb{Z}}\eta_\kappa(z)\eta_\kappa(z-d)}{\sum_{z\in\mathbb{Z}}\eta_\kappa(z)^2}\\
			&=\frac{\sum_{z\in \mathbb{Z}}e^{-\kappa^2z^2/2}e^{-\kappa^2(z-d)^2/2}}{\sum_{z\in\mathbb{Z}}(e^{-\kappa^2z^2/2})^2}\qquad \textrm{ by the definition of $\eta_\kappa(\cdot)$, see~\eqref{eq:eta kappa def}}
		\end{align}
		Hence the claim follows from Lemma~\ref{lem:sum gaussians ratio shift} (with $c=\kappa^2/2$ and~$\delta=d$).
	\end{proof}
	
	We also note that these states are orthogonal if shifted by 
	a real number with a large distance from the set of integers compared to the parameter~$\varepsilon$. That is, we have
	\begin{lemma}[Orthogonality of displaced GKP states]\label{lem:orthogonalitytruncated gkpstates}
		Let $\kappa,\Delta>0$, $\varepsilon\in (0,1/2)$ 
		and $d\in\mathbb{R}$ be such that
		\begin{align}
			|d-\round{d}|>2\varepsilon\ ,
		\end{align}
		i.e., the distance of~$d$ to $\mathbb{Z}$ is at least~$\varepsilon$. Then
		\begin{align}
			\langle \gkp^\varepsilon_{\kappa,\Delta},e^{-idP}\gkp^\varepsilon_{\kappa,\Delta}\rangle &=0\ .
		\end{align}
	\end{lemma}
	\begin{proof}
		This follows immediately from the fact that
		$\gkp^\varepsilon_{\kappa,\Delta}$ 
		is supported on~$\mathbb{Z}(\varepsilon)$, whereas 
		$e^{-idP}\gkp^\varepsilon_{\kappa,\Delta}$ 
		has support on the translated set~$(\mathbb{Z}+\{d\})(\varepsilon)$. If the 
		deviation of real~$d$ from the closest integer satisfies $|d-\round{d}|> 2\varepsilon$,
		these sets are disjoint and the claim follows.
	\end{proof}

	\clearpage
	\section{Combinatorial bounds}\label{sec: combinatorial bounds}
	In this section, we state
	various bounds on sums of Gaussians. Some of these are derived by elementary means (e.g., using monotonicity of the exponential function and Gaussian integrals).
	Others are more involved and based on more sophisticated results about lattices, including Poisson summation formula, or the Jacobi triple product identity.  We refer to~\cite{SMAntoniZygmund1935,SMJacobi, SMBanaszczyk1993,SMMicciancioRegev04,SMPeikert16,SMdadushdukas} for more details.

	\subsection{Elementary bounds on sums of Gaussians}
	\begin{corollary}\label{cor: Delta sum abs bound} Let $c\in (0,\pi/16)$. Then
		\begin{align}
			\sum_{z\in\mathbb{Z}} e^{-c(|z|+1)^2} \ge \frac{\sqrt{\pi}}{4\sqrt{c}}\,.
		\end{align}
	\end{corollary}
	\begin{proof}
		We show that for $c>0$
		\begin{align}
			\sum_{z\in\mathbb{Z}} e^{-c(|z|+1)^2} \ge \sqrt{\frac \pi c} - 3\ . \label{eq: sum gauss abs - 1 bound}
		\end{align}
		The claim is a simple corollary for $c\in (0,\pi/16)$.
		
		We have
		\begin{align}
			\sum_{z\in\mathbb{Z}} e^{-c(|z|+1)^2} &=  \left(\sum_{z\in\mathbb{Z}} e^{-cz^2} \right) - 1 - e^{-c} \label{eq: sum gauss abs - 1}
		\end{align}
		Since $x\mapsto e^{-cx^2}$ with $c>0$ is monotonously decreasing for $x\ge 0$ we have
		\begin{align}
			e^{-cz^2} &\ge  \int_{z}^{z+1} e^{-cx^2} dx \qquad\textrm{ for any $z\ge 0$} \ . \label{eq: gauss one step integral bound}
		\end{align}
		Since $x\mapsto e^{-cx^2}$ is even we have
		\begin{align}
			\sum_{z\in\mathbb{Z}} e^{-cz^2}& = 2\sum_{z\in\mathbb{N}_0} e^{-cz^2}-1 \ge \int e^{-cx^2}dx -1 \ge \sqrt{\frac \pi c} - 1\ .
		\end{align}
		We obtain~\eqref{eq: sum gauss abs - 1 bound} by inserting this bound into~\eqref{eq: sum gauss abs - 1} and by noting that $e^{-c}\le 1$.
	\end{proof}

	For $s>0$ and $x\in \mathbb{R}$ let $\rho_s(z)=e^{-\pi x^2/s^2}$ be a Gaussian function. Define
	\begin{align}
		f_s(t)&=\sum_{z\in\mathbb{Z}}\rho_s(z+t)=\rho_s(\mathbb{Z}+t)\qquad\textrm{ for }t\in\mathbb{R}\ \label{eq:fsdefinitionperiodicgaussian}
	\end{align}
	sometimes called the periodic Gaussian function over $\mathbb{Z}$ (cf.~\cite[Lecture 7]{SMdadushdukas}).
	
	We will need the inequality
	\begin{align}
		f_s(t)\geq f_s(0)e^{-\pi t^2/s^2}\qquad\textrm{ for all }t\in\mathbb{R}\ .\label{eq:perodicfunctionftransf}
	\end{align}
	\begin{proof}[Proof of~\eqref{eq:perodicfunctionftransf}]
		This proof is borrowed from~\cite[Lecture 7]{SMdadushdukas}, where a more general statement is shown. 
		We have
		\begin{align}
			f_s(t)&=\sum_{z\in\mathbb{Z}}e^{-\pi ((z+t)/s)^2}\\
			&=\sum_{z\in\mathbb{Z}}\frac{1}{2}\left(
			e^{-\pi ((z+t)/s)^2}+e^{-\pi ((t-z)/s)^2}\right)\ .\label{eq:sumfstxim}
		\end{align}
		But 
		\begin{align}
			\frac{1}{2}\left(
			e^{-\pi ((z+t)/s)^2}+e^{-\pi ((t-z)/s)^2}\right)
			&=e^{-\pi(t/s)^2}e^{-\pi (z/s)^2}
			\cdot \left(\frac{1}{2}e^{-2\pi zt/s^2}+\frac{1}{2}e^{2\pi zt/s^2}\right)\\
			&\geq e^{-\pi(t/s)^2}e^{-\pi (z/s)^2}\  ,\label{eq:intermediateusefulconvex}
		\end{align}
		where we used that $\frac{1}{2}(e^x+e^{-x})\geq e^{0}=1$ for all $x\in\mathbb{R}$ by convexity of the exponential function. Inserting~\eqref{eq:intermediateusefulconvex} into~\eqref{eq:sumfstxim} yields
		\begin{align}
			f_s(t)\geq e^{-\pi t^2/s^2}\sum_{z\in\mathbb{Z}}e^{-\pi z^2/s^2}\ ,
		\end{align}
		which is the claim.
	\end{proof}
	As an immediate application of Eq.~\eqref{eq:perodicfunctionftransf}, we give a bound on the ratio of two sums of Gaussians.
	\begin{lemma}\label{lem:sum gaussians ratio shift} Let $c>0$ and $\delta>0$.  Then 
		\begin{align}
			\frac{\sum_{z\in\mathbb{Z}} e^{-c z^2}e^{-c(z-\delta)^2}}{\sum_{z\in\mathbb{Z}} \left(e^{-c z^2}\right)^2 }&\ge e^{-c\delta^2}\geq 1-c\delta^2\ .
		\end{align}
	\end{lemma}
	\begin{proof}
		We have
		\begin{align}
			\frac{\sum_{z\in\mathbb{Z}} e^{-c z^2}e^{-c(z-\delta)^2}}{\sum_{z\in\mathbb{Z}} \left(e^{-c z^2}\right)^2 }
			&=e^{-c \delta^2/2}\cdot \frac{\sum_{z\in\mathbb{Z}} e^{-2c (z-\delta/2)^2}}{\sum_{z\in\mathbb{Z}} e^{-2c z^2} }\\
			&=e^{-c \delta^2/2}\cdot\frac{f_s(-\delta/2)}{f_s(0)}\qquad \textrm{ with } s=\sqrt{\pi/(2c)}\ ,
		\end{align}
		and $f_s$ as defined in Eq.~\eqref{eq:fsdefinitionperiodicgaussian}.
		Application of~\eqref{eq:perodicfunctionftransf} gives the lower bound
		\begin{align}
			\frac{f_s(-\delta/2)}{f_s(0)}\geq e^{-\pi (\delta/2)^2/s^2}=e^{-c\delta^2/2}\ ,
		\end{align}
		and the claim follows.
	\end{proof}

	\subsection{A tail bound for discrete Gaussian distributions}\label{sec:tailbound discrete gaussian}
	
	We need a tail bound for discrete Gaussian distribution on $\mathbb{Z}$. This is a specializations of standard results that we reproduce here for completeness. We refer to e.g.,~\cite[Lecture~7]{SMdadushdukas} and \cite{SMBanaszczyk1993, SMMicciancioRegev04, SMPeikert16} for more general results.
	
	To derive the bound we need the following statement relating the mass of discrete Gaussian distributions on $\mathbb{Z}$ with different width parameters.
	\begin{lemma}\label{lem:gaussiansumrelativeboundpoisson}
		For $s>0$, let $\rho_s(z)=\exp(-\pi z^2/s^2)$ and $\rho_s(\mathbb{Z})=\sum_{z\in\mathbb{Z}} \rho_s(z)$.
		We have
		\begin{align}
			\rho_{\alpha s}(\mathbb{Z})\leq \alpha \rho_s(\mathbb{Z})\qquad\textrm{ for all }s>0, \alpha\geq 1\ .
		\end{align}
	\end{lemma}
	\begin{proof} The proof is included for completeness, see e.g.,~\cite[Lecture~7]{SMdadushdukas} and~\cite{SMBanaszczyk1993}.
		This is a consequence of the Poisson summation formula~\cite{SMAntoniZygmund1935}, which states
		that for an aperiodic function $f:\mathbb{R}\rightarrow\mathbb{R}$, we have 
		\begin{align}
			\sum_{z\in\mathbb{Z}} f(z)&=\sum_{k\in\mathbb{Z}} \widehat{f}(k)\qquad\textrm{ where }\qquad \widehat{f}(k)=\int f(x)e^{2\pi ikx }dx\ . 
		\end{align}
		In particular, since the function~$f=\rho_t$ has Fourier transform 
		$\widehat{f}=t\rho_{1/t}$, we have
		\begin{align}
			\rho_t(\mathbb{Z})&=t\rho_{1/t}(\mathbb{Z})\qquad\textrm{ for any }\qquad t>0\ .
		\end{align}
		With $t=\alpha s$, we obtain
		\begin{align}
			\rho_{\alpha s}(\mathbb{Z})&=(\alpha s) \rho_{1/(\alpha s)}(\mathbb{Z})\\
			&=(\alpha s)\sum_{z\in\mathbb{Z}}e^{-\pi z^2 (\alpha s)^2}\\
			&\leq (\alpha s)\sum_{z\in\mathbb{Z}}e^{-\pi z^2  s^2}\qquad\textrm{ because }\alpha\geq 1\\
			&=\alpha s\rho_{1/s}(\mathbb{Z})\ .
		\end{align}
		Applying the Poisson summation formula again 
		gives~$s\rho_{1/s}(\mathbb{Z})=\rho_s(\mathbb{Z})$, implying the claim.
	\end{proof}

	We have the following tail bound for discrete Gaussian distributions on $\mathbb{Z}$.
	
	\begin{lemma}\label{thm:discretegaussianconcentration}
		Let $s > 0$.  Set $\rho_s(z)=\exp(-\pi z^2/s^2)$. Define a random variable~$X_s$ on~$\mathbb{Z}$ by
		\begin{align}
			\Pr\left[X_s=x\right]&=\frac{\rho_s(x)}{\rho_s(\mathbb{Z})}\ ,
		\end{align}
		where we write $\rho_s(\mathbb{Z})=\sum_{z\in\mathbb{Z}}\rho_s(z)$. 
		Then we have 
		\begin{align}
			\Pr\left[|X_s|\ge r \right]&\le 2\cdot e^{-\frac{3\pi}{4}(r/s)^2}\ 
		\end{align}
		for $r>0$.
	\end{lemma}
	
	\begin{proof}
		The proof follows~\cite[Lecture~7]{SMdadushdukas} to which we refer to for more general results, and more details. We use the tail estimate
		\begin{align}
			\Pr\left[A\geq t\right]\leq \frac{\ExpE\left[e^{\lambda A^2}\right]}{e^{\lambda t^2}}\qquad\textrm{ for any }t,\lambda>0\ ,
		\end{align}
		applicable to any random variable~$A$ taking on non-negative values. We apply this bound to $A:=\sqrt{2\pi }|X_s|$, $t:=\sqrt{2\pi}\cdot r$ and $\lambda:=\alpha/2$ and obtain
		\begin{align}
			\Pr\left[|X_s|\ge r\right]&\le \frac{\ExpE\left[e^{\alpha\pi X_s^2}\right]}{e^{\alpha\pi r^2}}\ .
		\end{align}
		We have
		\begin{align}
			\ExpE\left[e^{\alpha \pi X_s^2}\right]&=\frac{
				\sum_{x\in\mathbb{Z}}e^{\alpha\pi x^2}e^{-\pi x^2/s^2}
			}{\rho_s(\mathbb{Z})}\\
			&= \frac{\rho_{s/\sqrt{1-\alpha s^2}}(\mathbb{Z})}{\rho_s(\mathbb{Z})}\\
			&\leq \frac{1}{\sqrt{1-\alpha s^2}} \qquad \textrm{ whenever } 0<\alpha<1/s^2\ ,
		\end{align}
		where we obtained the last inequality by Lemma~\ref{lem:gaussiansumrelativeboundpoisson} 
		in the form~$\rho_{s/\sqrt{1-\alpha s^2}}(\mathbb{Z})\leq \frac{1}{\sqrt{1-\alpha s^2}} \rho_{s}(\mathbb{Z})$ since $1/\sqrt{1-\alpha s^2}\ge 1$ for $0<\alpha<1/s^2$. We have
		\begin{align}
			\Pr\left[|X_s| \ge r\right]&\le \frac{1}{\sqrt{1-\alpha s^2}} e^{-\alpha\pi r^2} \qquad \textrm{ whenever } 0<\alpha<1/s^2\ .
		\end{align}
		Choosing $\alpha:=3/(4s^2)< 1/s^2$ gives the claim
		\begin{align}
			\Pr\left[|X_s|\ge r\right] &\le 2\cdot e^{-\frac{3\pi}{4}(r/s)^2}\ .
			\label{eq:lasteqrsmv}
		\end{align}
	\end{proof}

	\subsection{A bound derived from Jacobi's triple product identity}
	We will also use the following lemma to lower bound sums of the form $\sum_{k \in \mathbb{Z}} \xi^{k} b^{k^2}$ when $\xi\in\mathbb{C}\backslash \{0\}$ and $b\in [0,1)$.
	\begin{lemma}\label{lem:gausssumjacobitriple}
		For $b\in [0,1)$ and $\xi\in\mathbb{C}\backslash \{0\}$ such that
		\begin{align}
			\max\{|\xi|,|\xi|^{-1}\}&\geq \sqrt{2}+1\,.\label{eq:lowerboundxinorm}
		\end{align} we have 
		\begin{align}
			\left|\sum_{k \in \mathbb{Z}} \xi^{k} b^{k^2}\right|&\ge \left(1-2b\max\{|\xi|,|\xi|^{-1}\}\right)\sum_{k=0}^\infty (b^4)^{k^2}\\
			&\ge \left(1-2b\max\{|\xi|,|\xi|^{-1}\}\right)\,.
		\end{align}
	\end{lemma}
	
	\begin{proof}
		Jacobi's triple product identity~\cite{SMJacobi} states that for any $0 \leq b<1$ and any $\xi \in \mathbb{C} \backslash \{0\}$,  we have
		\begin{align}
			\sum_{k \in \mathbb{Z}} \xi^{k} b^{k^2}&=\prod_{k=1}^\infty\left(1-b^{2 k}\right)\left(1+\xi b^{2 k-1}\right)\left(1+\xi^{-1} b^{2 k-1}\right)\ .\label{eq:jacobitriple}
		\end{align}
		We note that
		\begin{align}
			\left(1+\xi b^{2 k-1}\right)\left(1+\xi^{-1} b^{2 k-1}\right)
			&=1+(\xi+\xi^{-1}) b^{2k-1}+b^{2(2k-1)}\,
		\end{align}
		and thus
		\begin{align}
			\left|\left(1+\xi b^{2 k-1}\right)\left(1+\xi^{-1} b^{2 k-1}\right)\right|
			&\geq 1-|\xi+\xi^{-1}| b^{2k-1}+b^{2(2k-1)}\,.\label{eq:intermediatelowerboundb2kxi}
		\end{align}
		Let $r:=\max\{|\xi|,|\xi|^{-1}\}$.
		It is easy to check that 
		\begin{align}
			|\xi+\xi^{-1}|\in [r-r^{-1},r+r^{-1}]\,.\label{eq:containmentboundm}
		\end{align}
		In particular, 
		the assumption~\eqref{eq:lowerboundxinorm}
		implies that
		\begin{align}
			|\xi+\xi^{-1}|&\geq \left|r-r^{-1}\right|\geq 2\,.\label{eq:rrinvlowerboundtwo}
		\end{align}
		Since $|\xi+\xi^{-1}|\geq 2$, there is some $x\geq 1$ such that $x+x^{-1}=|\xi+\xi^{-1}|$. 
		Indeed, we can use the inverse~
		\begin{align}
			\begin{matrix}
				f^{-1}&:[2,\infty) & \rightarrow & [1,\infty)\\
				& z& \mapsto &f^{-1}(z)=\frac{1}{2}(z+\sqrt{z^2-4})
			\end{matrix}
		\end{align}
		of the function
		\begin{align}
			\begin{matrix}
				f&:[1,\infty) & \rightarrow & [2,\infty)\\
				& y & \mapsto & f(y):=y+y^{-1}\,,
			\end{matrix} 
		\end{align}
		and set 
		\begin{align}
			x:=f^{-1}(|\xi+\xi^{-1}|)\,.
		\end{align}
		Since~$f^{-1}$ is monotonically increasing in $z$, we have by~\eqref{eq:containmentboundm} that
		\begin{align}
			x\leq f^{-1}(r+r^{-1})=r\,,
		\end{align}
		as well as
		\begin{align}
			f^{-1}(r-r^{-1})&\leq x\,.
		\end{align}
		By a simple calculation, we also have
		\begin{align}
			f^{-1}(r-r^{-1})&\geq f^{-1}(2)=1\,,
		\end{align}
		where we used Eq.~\eqref{eq:rrinvlowerboundtwo} and the fact that~$f^{-1}$ is monotonically increasing. In conclusion, we have shown that 
		\begin{align}
			1&\leq x\leq r\,.\label{eq:xinequalityxm}
		\end{align}
		
		By definition, we have 
		\begin{align}
			(-x)+(-x)^{-1}&=-(x+x^{-1})=-|\xi+\xi^{-1}|\,.\label{eq:keyidentityxxi}
		\end{align}
		It follows that
		\begin{align}
			\left|\sum_{k \in \mathbb{Z}} \xi^{k} b^{k^2}\right|
			&=\prod_{k\geq 1} (1-b^{2k})\left|\left(1+\xi b^{2 k-1}\right)\left(1+\xi^{-1} b^{2 k-1}\right)\right|\qquad\textrm{ by \eqref{eq:jacobitriple}}\\
			&\geq \prod_{k=1}^\infty (1-b^{2k}) \left(1-|\xi+\xi^{-1}|b^{2k-1}+b^{2(2k-1)}\right)\qquad\textrm{ by~\eqref{eq:intermediatelowerboundb2kxi}}\\
			&=\prod_{k=1}^\infty (1-b^{2k}) \left(1+((-x)+(-x)^{-1})b^{2k-1}+b^{2(2k-1)}\right)\qquad\textrm{ by \eqref{eq:keyidentityxxi}}\\
			&=\prod_{k=1}^\infty (1-b^{2k}) 
			\left(1+(-x)b^{2k-1}\right)\left(1+(-x)^{-1}b^{2k-1}\right)\\
			&=\sum_{k\in\mathbb{Z}}(-x)^{k}b^{k^2}\qquad\textrm{ by Jacobi's triple product identity \eqref{eq:jacobitriple}}\\
			&=\sum_{k\in\mathbb{Z}}x^{2k}b^{(2k)^2}-\sum_{k\in\mathbb{Z}}x^{2k+1}b^{(2k+1)^2}\,.\label{eq:lastrelevantjcaobi}
		\end{align}
		We have (using that $x\ge 1$ by~\eqref{eq:xinequalityxm} and thus $x^{-2k}\geq 0$) that
		\begin{align}
			\sum_{k\in\mathbb{Z}}x^{2k}b^{(2k)^2}&=1+ \sum_{k=1}^\infty (x^{2k}+x^{-2k})b^{(2k)^2}\\
			&\geq 1+\sum_{k=1}^\infty x^{2k}b^{(2k)^2}\\
			&=\sum_{k=0}^\infty x^{2k}b^{(2k)^2}\,,\label{eq:expressionexplicitsumnv}
		\end{align}
		and
		\begin{align}
			\sum_{k\in\mathbb{Z}}x^{2k+1}b^{(2k+1)^2}
			&=\sum_{k=0}^\infty \left(x^{2k+1}+x^{-(2k+1)}\right)b^{(2k+1)^2}\\
			&\leq 2\sum_{k=0}^\infty x^{2k+1}b^{(2k+1)^2}\\
			&\leq 2\sum_{k=0}^\infty(b^{4k+1}x) x^{2k}b^{(2k)^2}\\
			&\leq 2bx\sum_{k=0}^\infty x^{2k}b^{(2k)^2}\,.\label{eq:lowerboundmznx}
		\end{align}
		By 
		combining~\eqref{eq:lastrelevantjcaobi} with~\eqref{eq:expressionexplicitsumnv} and~\eqref{eq:lowerboundmznx},
		we obtain
		\begin{align}
			\left|\sum_{k \in \mathbb{Z}} \xi^{k} b^{k^2}\right|
			&\geq(1-2bx) \sum_{k=0}^\infty x^{2k}b^{(2k)^2}\\
			&\ge (1-2bx) \sum_{k=0}^\infty (b^4)^{k^2}
		\end{align}
		where we used that $x\ge 1$, see~\eqref{eq:xinequalityxm}. Since we also have (also by~\eqref{eq:xinequalityxm}) that $x\leq \max\{|\xi|,|\xi|^{-1}\}$, the claim follows.
	\end{proof}
	
	\newpage
	\normalem
	\renewcommand{\refname}{References for the  Supplementary Material}
	{
	}

\end{document}